%% file: IVValidityforHeterogeneousCausalEffects.tex
\documentclass[11pt]{article}
\usepackage{titling}
\usepackage[OT1]{fontenc}
\usepackage{amsmath}
\usepackage{mathrsfs}
\usepackage{amsfonts}
\usepackage{amssymb}
\usepackage{diagbox}
\usepackage{footnote}
\usepackage{graphicx}
\usepackage{float}
\usepackage{epstopdf}
\usepackage[margin=1.25in]{geometry}
\usepackage{setspace}
\usepackage{tikz}
\usepackage{pgfplots}
\usepackage[round]{natbib}
\usepackage{multirow}
\usepackage[toc,page]{appendix}
\usepackage{hyperref}
\usepackage{caption}
\usepackage{subcaption}
\usepackage{enumitem}
\usepackage{booktabs}
\usepackage{chapterbib}
\providecommand{\U}[1]{\protect\rule{.1in}{.1in}}

\hypersetup{
bookmarks=true,
	colorlinks=true,
	linkcolor=blue,
	citecolor=blue,
	filecolor=magenta,
	urlcolor=blue,
	breaklinks=true,
}
\newtheorem{theorem}{Theorem}[section]
\newtheorem{assumption}{Assumption}[section]

\newtheorem{lemma}{Lemma}[section]

\newtheorem{remark}{Remark}[section]

\newenvironment{proof}[1][Proof]{\noindent\textbf{#1.} }{\ \rule{0.5em}{0.5em}}

\begin{document}
	\def\hnulldgptwopone{H0DGP2P1.txt}
	\def\hnulldgptwoqone{H0DGP2Q1.txt}
	\def\hnulldgptwopzero{H0DGP2P0.txt}
	\def\hnulldgptwoqzero{H0DGP2Q0.txt}
	
	\def\honedgponepone{H1DGP1P1.txt}
	\def\honedgponepzero{H1DGP1P0.txt}
	
	\def\honedgptwopone{H1DGP2P1.txt}
	\def\honedgptwopzero{H1DGP2P0.txt}
	
	\def\honedgpthreepone{H1DGP3P1.txt}
	\def\honedgpthreepzero{H1DGP3P0.txt}
	
	\def\honedgpfourpone{H1DGP4P1.txt}
	\def\honedgpfourpzero{H1DGP4P0.txt}
%\begin{spacing}{1.5}
\onehalfspacing
\include{IVValidityMain}

\clearpage

\include{IVValidityAppendix}

%\end{appendices}

%\end{spacing}
%abbrv plainnat named authordate1 apalike

\end{document}

%% file: IVValidityMain.tex
\title{Instrument Validity for Heterogeneous Causal Effects}
%\date{\mydate\today}
\author{Zhenting Sun\thanks{Correspondence to: No.\ 5 Yiheyuan Road, Haidian District, Beijing 100871, China. E-mail address: zhentingsun@nsd.pku.edu.cn.}\\China Center for Economic Research\\
	National School of Development\\ Peking University\\ China
}

\maketitle
\begin{abstract}
%With heterogeneous treatment effects, endogeneity in treatment variables can
%cause severe problems for identification and estimation of population average
%causal effects. By resorting to reliable instruments, people can identify average effects for
%compliers which is called the local average treatment effect (LATE).

This paper provides a general framework for testing instrument validity in heterogeneous causal effect models. The generalization includes the cases where the treatment can be multivalued ordered or unordered. Based on a series of testable implications, we propose a nonparametric test which is proved to be asymptotically size controlled and consistent. Compared to the tests in the literature, our test can be applied in more general settings and may achieve power improvement. Refutation of instrument validity by the test helps detect invalid instruments that may yield implausible results on causal effects. Evidence that the test performs well on finite samples is provided via simulations. We revisit the empirical study on return to schooling to demonstrate application of the proposed test in practice.  An extended continuous mapping theorem and an extended delta method, which may be of independent interest, are provided to establish the asymptotic distribution of the test statistic under null.
%We show that a valid instrument for a multivalued treatment may not remain valid if the treatment is coarsened.

\end{abstract}



\newpage

\section{Introduction}

The local average treatment effect (LATE) framework, introduced by the seminal works of \citet{imbens1994identification} and \citet{angrist1996identification}, is a commonly used approach in studies of instrumental variable (IV) models with treatment effect heterogeneity. The local quantile treatment effect (LQTE) is a concept similar to LATE. While LATE shows the treatment effect on the mean of the outcome, LQTE is more informative in regard to the effect on the outcome distribution.\footnote{See, for example, studies of LQTE in \citet{abadie2002bootstrap}, \citet{ananat2008effect}, \citet{cawley2012medical}, \citet{frolich2013unconditional}, and \citet{eren2014benefits}.} These causal effect models rely on several strong and sometimes controversial assumptions of IV validity: 1) The instrument should not affect the outcome directly; 2) it should be as good as random assignment; and 3) it affects the treatment in monotone fashion. Violations of these conditions can generally lead to unidentification and inconsistent estimation of treatment effects. Relevant surveys and discussion of this can be found in \citet{angrist2008mostly}, \citet{angrist2014mastering}, \citet{imbens2014instrumental}, \citet{imbens2015causal}, \citet{koenker2017handbook}, \citet{melly2017local}, and \citet{huber2018local}. Since the plausibility of the analyses of such models depends on IV validity, economics research has developed methods to examine these conditions based on testable implications.

Cases where the IV validity conditions may be violated can be found in empirical applications. For example, the college proximity was used as an instrument of education attainment in the study of \cite{NBERw4483}. If the education level is treated as a binary variable (four-year college degree), the validity of the college proximity is rejected by the test of \citet{kitagawa2015test} when no conditioning covariates are added in the model. \citet{mogstad2021causal} considered tuition and college proximity as multiple instruments for college attendance. They showed that, if a homogeneity condition does not hold for individuals, the validity of the multiple instruments will be violated. 
The quarter of birth instrument used in \citet{angrist1991does} is questionable because the exclusion restriction may not hold due to seasonal birth patterns \citep{bound1995problems,buckles2013season}. The monotonicity condition of IV validity fails in the selection with two-way flows example in \citet{lee2018identifying}.
%\footnote{See further discussion on IV validity assumptions in Section \ref{sec.IV assumptions discussion}.} 
%As shown in \citet{imbens1994identification} and \citet{angrist1995two}, the identification  and interpretation of causal effects rely on IV validity conditions. Violation of such conditions may lead to implausible results. 

\citet{kitagawa2015test} was the first paper to propose a test of
IV validity in heterogeneous causal effect models with a binary treatment based on the testable
implications in the literature. It was the first to show the sharpness of these testable implications. Their test, constructed using a bootstrap method, was shown to be asymptotically uniformly size controlled and consistent. \citet{mourifie2016testing} reformulated the testable implications used in \citet{kitagawa2015test} as conditional inequalities. They then
showed that these inequalities could be tested in the intersection bounds framework of
\citet{chernozhukov2013intersection} using the
Stata package provided by \citet{chernozhukov2014implementing}. The present paper provides a general framework for testing such IV validity assumptions. The proposed test can be applied in more general settings in which the treatment variable can be multivalued ordered or unordered\footnote{Studies of LATE with binary treatments  can be found in \citet{angrist1990lifetime}, \citet{angrist1991does}, and  \citet{vytlacil2002independence}. Those with multivalued treatments can be found in \citet{angrist1995two}, \citet{angrist1995split}, and \citet{vytlacil2006ordered}. Identification of causal effects in unordered choice (treatment) models can be found in \citet{heckman2006understanding}, \citet{heckman2007econometric}, \citet{heckman2008instrumental}, and \citet{heckman2018unordered}.} and the outcome variable can be unbounded\footnote{See \citet{reed2001pareto,reed2003pareto} and \citet{toda2012double} for the approximation of income distributions by members of the double Pareto parametric family.}. Also, the proposed test achieves power improvement by solving a technical issue and employing a novel bootstrap approach.  
\citet{huber2015testing} derived a testable
implication for a weaker LATE identifying condition, that is, that the potential
outcomes are mean independent of instruments, conditional on each selection type.\footnote{The condition of potential outcomes being mean independent of  instruments is not sufficient if we are concerned with distributional features of a complier's potential outcomes, such as the quantile treatment effects for compliers; see \citet{abadie2002instrumental} for details.}  The focus of the present paper is on full statistical independence of potential outcomes and instruments. 

%The null hypothesis for the testable implications used in \citet{kitagawa2015test} consists of a set of inequalities. They used an upper bound on the asymptotic distribution of the test statistic under null to construct the bootstrap critical value. The upper bound is identical to the asymptotic distribution when all the inequalities in the null are binding. Thus their test could be conservative. In the study described in the present paper, we solve a technical issue and establish the pointwise asymptotic distribution of the test statistic under null. Then we construct the critical value based on this asymptotic distribution, rather than on an upper bound, and therefore improve the power of the test.  

%We transform the proposed testable implication into an inequality involving the value of the supremum of a continuous map over a particular function space. 
A modified  variance-weighted Kolmogorov--Smirnov (KS) test statistic is employed in our test. As mentioned by \citet{kitagawa2015test}, variance-weighted KS statistics have been widely applied in the literature on conditional moment inequalities, such as in 
\citet{andrews2013inference}, \citet{armstrong2014weighted},
\citet{armstrong2016multiscale}, and \citet{chetverikov2017adaptive}. More
general KS statistics can be found in the stochastic dominance testing literature,
such as in \citet{abadie2002bootstrap}, 
\citet{barrett2003consistent}, \citet{horvath2006testing},
\citet{linton2010improved}, \citet{barrett2014consistent}, and \citet{donald2016improving}. To investigate the asymptotic properties of the proposed test, we introduce $L^r$ ($r\in\mathbb{N}$) spaces with which a series of fundamental results are established, such as the compactness of particular function spaces, the Glivenko--Cantelli and the Donsker results, and so on. Based on these results, we obtain the asymptotic behavior of the test statistic.\footnote{See further discussion before Theorem \ref{thm.weak convergence SK}.} The asymptotic properties of the proposed test are established accordingly.   
 
There are two major complications in deriving and approximating  the asymptotic distribution of the test statistic under null. First, the test statistic involves a nonsmooth (nondifferentiable) map of unknown parameters (underlying probability distributions), and the delta method fails to work. We provide an extended continuous mapping theorem and an extended delta method, which might be of independent interest, to overcome this difficulty. By showing that the conditions of the extended delta method are satisfied under several weak assumptions, we establish the null asymptotic distribution of the test statistic. Second, since the null asymptotic distribution involves a nonlinear function, the standard bootstrap method may fail to approximate this distribution consistently. {Discussion of this issue can be found in \citet{dumbgen1993nondifferentiable}, \citet{andrews2000inconsistency}, \citet{hirano2012impossibility}, \citet{hansen2017regression}, \citet{hong2014numerical}, and \citet{fang2014inference}.}  To achieve a consistent approximation, we
extend the bootstrap approach proposed by \citet{fang2014inference}\footnote{Other applications of this  bootstrap
method can be found in \citet{beare2015nonparametric},
\citet{Beare2016global}, \citet{Seo2016tests}, \citet{Beare2015improved},  and \citet{Beare2017improved}. A similar bootstrap approach can be found in \citet{hong2014numerical}.} and provide a valid bootstrap critical value. The test is found to be asymptotically size controlled and consistent. Evidence that the test performs well on finite samples is provided via simulations.

We now introduce the following notation, which will be used throughout the paper. We let $\leadsto$ denote Hoffmann--J\o rgensen weak convergence in a metric space. For
a set $\mathbb{D}$, denote the space of bounded functions on $\mathbb{D}$ by $\ell^{\infty}(\mathbb{D})$:
$
\ell^{\infty}\left(  \mathbb{D}\right)  =\left\{  f:\mathbb{D}\rightarrow
\mathbb{R}
:\left\Vert f\right\Vert _{\infty}<\infty\right\}$, { where } $\left\Vert
f\right\Vert _{\infty}=\sup_{x\in\mathbb{D}}\left\vert f\left(  x\right)
\right\vert
$.
If $\mathbb{D}$ is a topological space, let $C\left(  \mathbb{D}\right)  $ denote the set of
continuous functions on $\mathbb{D}$:
$
C\left(  \mathbb{D}\right)  =\left\{  f:\mathbb{D}\rightarrow
\mathbb{R}
:f\text{ is continuous}\right\}
$.
Let $(\Omega,\mathcal{A},\mathbb{P})$ be a probability space on which all random elements are well defined. Let $\mathcal{B}_{\mathbb{R}^m}$ denote the Borel $\sigma$-algebra on
$\mathbb{R}^m$ for all $m\in\mathbb{N}$. We use $\hat{Q}$ and $\hat{Q}^B$ to denote the empirical probability measure and the bootstrap empirical probability measure of each probability measure $Q$, respectively.

%\section{Instrument Validity Assumption}

\section{Setup and Testable Implications}\label{sec.setup}
\subsection{Binary Treatment}\label{subsec.setup binary}
 
To formally introduce the topic of interest, we first consider the heterogeneous causal
effect model of \citet{imbens1994identification}. Let $Y\in\mathbb{R}$ be the observable outcome variable, and let $D\in\left\{
0,1\right\}  $ be the observable treatment variable, where $D=1$ indicates that an
individual receives treatment. Let
$Z\in\left\{  0,1\right\}  $ be a binary instrumental variable. Let $Y_{dz}%
\in\mathbb{R}$ be the potential outcome variable\footnote{See \citet{rubin1974} and \citet{splawa1990application} for further discussion of the potential outcomes.} for $D=d$ and
$Z=z$, where $d,z\in\{0,1\}$. Similarly, let $D_{z}$ be
the potential treatment variable for $Z=z$.
The instrument validity assumption for binary treatment and binary IV is formalized as follows.

\begin{assumption}
\label{ass.IV validity for binary Z}IV validity for binary $D$ and binary $Z$:
\begin{enumerate}[label=(\roman*)]
\item Instrument Exclusion: For each $d\in\{0,1\}$, $Y_{d0}=Y_{d1}$ almost surely.

\item Random Assignment: The variable $Z$ is jointly independent of $\left(
Y_{00},Y_{01},Y_{10},Y_{11},D_{0},D_{1}\right)  $.

\item Instrument Monotonicity: The potential treatment response
indicators satisfy $D_{1}\geq D_{0}$ almost surely.

\end{enumerate}
\end{assumption}
Assumption \ref{ass.IV validity for binary Z} is from \citet{imbens1997estimating}, but it does not require strict instrument monotonicity. In this paper, we are not concerned with the strict monotonicity assumption, which is also known as the instrument relevance assumption.\footnote{As mentioned by \citet{kitagawa2015test}, the instrument relevance assumption can be assessed by inferring the coefficient in the first-stage regression of $D$ onto $Z$.}

For all Borel sets $B$ and $C$, we follow \citet{kitagawa2015test} and define probability
measures as follows:\footnote{For simplicity of notation, we implicitly assume that $(Y,D,Z)$ is $(\mathcal{A},\mathcal{B}_{\mathbb{R}^3})$-measurable.}
\begin{align*}
P_1\left(  B,C\right)     =\mathbb{P}\left(  Y\in B,D\in C|Z=1\right)  \text{ and } P_0\left(  B,C\right)    =\mathbb{P}\left(  Y\in B,D\in C|Z=0\right).
\end{align*}
Under Assumption
\ref{ass.IV validity for binary Z}(i), we can define a potential outcome variable $Y_{d}$ such that $Y_{d}=Y_{d0}=Y_{d1}$ almost surely. \citet{imbens1997estimating} showed that
for every Borel set $B$,
\begin{align}
&P_1\left(  B,\{1\}\right)  -P_0\left(  B,\{1\}\right)    =\mathbb{P}\left(  Y_{1}\in
B,D_{1}>D_{0}\right) \notag \\ &\text{ and } P_0\left(  B,\{0\}\right)  -P_1\left(  B,\{0\}\right)  =\mathbb{P}\left(  Y_{0}\in
B,D_{1}>D_{0}\right)  . \label{eq.probability difference}%
\end{align}
To see why \eqref{eq.probability difference} is true, we can write 
\begin{align*}
P_1\left(  B,\{1\}\right) & -P_0\left(  B,\{1\}\right)    =\mathbb{P}\left(  Y\in
B,D=1|Z=1\right)-\mathbb{P}\left(  Y\in
B,D=1|Z=0\right)\\
= & \,\mathbb{P}\left(  Y_{1}\in
B,D_{1}=1\right)-\mathbb{P}\left(  Y_{1}\in
B,D_{0}=1\right)=\mathbb{P}\left(  Y_{1}\in
B,D_{1}=1, D_{0}=0\right), 
\end{align*}
where the second equality follows from Assumptions
\ref{ass.IV validity for binary Z}(i) and \ref{ass.IV validity for binary Z}(ii) and the third equality follows from Assumption
\ref{ass.IV validity for binary Z}(iii). Similar reasoning yields the second equation in \eqref{eq.probability difference}.
Since the probabilities in \eqref{eq.probability difference} are nonnegative, we obtain the testable implication of Assumption \ref{ass.IV validity for binary Z} in 
\citet{balke1997bounds},  \citet{imbens1997estimating}, and \citet{heckman2005structural}: For all
$B\in\mathcal{B}_{%
%TCIMACRO{\U{211d} }%
%BeginExpansion
\mathbb{R}
%EndExpansion
}$,%
\begin{align}
P_1\left(  B,\{1\}\right)  -P_0\left(  B,\{1\}\right)   \geq0 \text{ and } P_0\left(  B,\{0\}\right)  -P_1\left(  B,\{0\}\right)  \geq0.
\label{eq.testable implication}%
\end{align}
To understand \eqref{eq.testable implication} graphically, suppose that $Y$ is a
continuous variable and that $p_z\left(  y,d\right)  $ is the derivative of the function $P_z\left(  (-\infty,y],\{d\}\right)$ with respect to $y$ for all $d,z\in\{0,1\}$. The following graphs show a case where \eqref{eq.testable implication} holds.
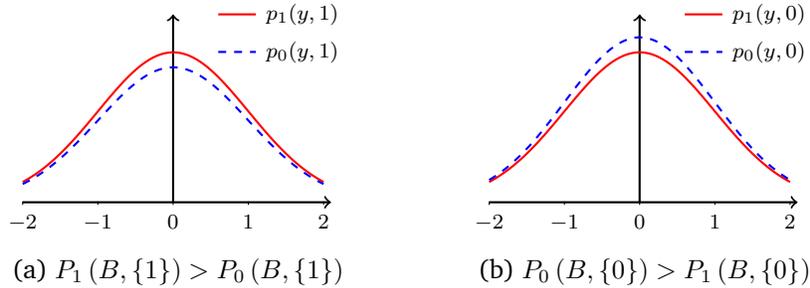
\begin{figure} [h]
\caption{A special case satisfying testable implication \eqref{eq.testable implication}}
\label{fig:nest}
\centering
\begin{subfigure}[b]{0.4\textwidth}
	\centering
\input{P1nestP0}
\subcaption{$P_1\left(B,\{1\}\right)>P_0\left(B,\{1\}\right)$}
\label{fig:nest a}
\end{subfigure}
\begin{subfigure}[b]{0.4\textwidth}
	\centering
\input{P0nestP1}
\subcaption{$P_0\left(B,\{0\}\right)>P_1\left(B,\{0\}\right)$}
\label{fig:nest b}
\end{subfigure}
\end{figure}

The first inequality in
\eqref{eq.testable implication} is shown in Figure \ref{fig:nest a}, where the derivative $p_1\left(  y,1\right)  $ is greater than $p_0\left(  y,1\right)  $ everywhere.
The second inequality in \eqref{eq.testable implication} is shown in Figure
\ref{fig:nest b}, where the derivative $p_0\left(  y,0\right)  $ is greater than
$p_1\left(  y,0\right)  $ everywhere. Additional graphical examples can be found in \citet{kitagawa2015test}.

\subsection{Multivalued Ordered Treatment} \label{subsec.multi}
Section \ref{subsec.setup binary}  discussed the case where the treatment and the instrument are both binary. 
In many applications, $D$ and $Z$ can be multivalued. See, for example, \citet{angrist1995two}, where the treatment variable is the number of years of schooling completed by a student and can take more than two values. 
Now suppose that $D\in\mathcal{D}=\left\{
d_{1},\ldots,d_J\right\}  $\footnote{The support $\mathcal{D}$ can be generalized to the case where $\mathcal{D}=\left\{d_{1},d_{2},\ldots\right\} $. See details in the supplementary appendix.} and $Z\in\mathcal{Z}=\left\{  z_{1},\ldots,z_{K}\right\}  $.  We let $d_{\max}$ be the
maximum value of $D$, and $d_{\min}$ the minimum value of $D$. Suppose there exist potential variables $Y_{dz}$ for $d\in\mathcal{D}$ and $z\in\mathcal{Z}$, and  $D_{z}$ for $z\in
\mathcal{Z}$. The IV validity assumption for multivalued treatment $D$ and multivalued
instrument $Z$ is then formalized as follows.

\begin{assumption}
\label{ass.IV validity for multivalued Z}IV validity for multivalued $D$ and multivalued $Z$:
\begin{enumerate}[label=(\roman*)]
\item Instrument Exclusion: For all $d\in\mathcal{D}$, $Y_{dz_{1}}=Y_{dz_{2}}
=\cdots=Y_{dz_{K}}$ almost surely.

\item Random Assignment: The variable $Z$ is jointly independent of $(
\tilde{Y},\tilde{D})  $, where
\begin{align*}
\tilde{Y}   =\left(  Y_{d_{1}z_{1}},\ldots,Y_{d_{1}z_{K}},\ldots,Y_{d_{J}z_{1}
},\ldots,Y_{d_{J}z_{K}}\right) \text{ and }\tilde{D} =\left(  D_{z_{1}},\ldots,D_{z_{K}}\right)  .
\end{align*}

\item Instrument Monotonicity: The potential treatment response variables satisfy $D_{z_{k+1}}\geq D_{z_{k}}$ almost surely for all
$k\in \{1,2,\ldots, K-1\}$.

\end{enumerate}
\end{assumption}
Assumption \ref{ass.IV validity for multivalued Z} is similar to Assumptions 1 and 2 of \citet{angrist1995two}. Theorems 1 and 2 of \citet{angrist1995two} showed that a weighted average of $K$ average causal responses can be identified under Assumption \ref{ass.IV validity for multivalued Z}. Since we allow multivalued $Z$, the monotonicity assumption needs to hold for each pair $(D_{z_{k}},D_{z_{k+1}})$.
The next lemma
establishes a testable implication of Assumption \ref{ass.IV validity for multivalued Z}.
% when the treatment variable has a maximum value and/or a minimum value.

\begin{lemma}
\label{ass.testable implication multivalue}A testable implication of
Assumption \ref{ass.IV validity for multivalued Z} is that for all $k$ with $1\le k \le K-1$,  all Borel sets $B$, and all $C=(-\infty,c]$ with $c\in\mathbb{R}$, the following hold:
\begin{align}\label{eq.testable implication multivalue}
&\mathbb{P}\left(  Y\in B,D=d_{\max}|Z=z_k\right)     \leq \mathbb{P}\left(  Y\in B,D=d_{\max}|Z=z_{k+1}\right)\nonumber\\
&\text{and }\mathbb{P}\left(  Y\in B,D=d_{\min}|Z=z_k\right)     \geq \mathbb{P}\left(  Y\in B,D=d_{\min}|Z=z_{k+1}\right); \\
&\mathbb{P}\left(  D\in C|Z=z_k\right)     \geq \mathbb{P}\left(  D\in C|Z=z_{k+1}\right). \label{eq.fosd multi}
\end{align}
\end{lemma}
Lemma \ref{ass.testable implication multivalue} generalized testable
implication \eqref{eq.testable implication} to the case where the treatment and the instrument can both be multivalued. The testable implication (first-order stochastic dominance) discussed by \citet{angrist1995two} for Assumption \ref{ass.IV validity for multivalued Z} is equivalent to  \eqref{eq.fosd multi}. Clearly, if $D$ and $Z$ are
 both binary as assumed in Section \ref{subsec.setup binary}, with $d_{\max}=1$ and $d_{\min}=0$, then
 \eqref{eq.testable implication multivalue} is equivalent to \eqref{eq.testable implication} and \eqref{eq.fosd multi} is implied by \eqref{eq.testable implication multivalue}.  \citet{liu2020two} proposed testable implication \eqref{eq.testable implication multivalue} for the case where $\mathcal{D}=\{0,1,2\}$ and $\mathcal{Z}=\{0,1,2\}$, and \eqref{eq.fosd multi} is also implied by \eqref{eq.testable implication multivalue} in this case. 
 Thus, \eqref{eq.testable implication multivalue} and \eqref{eq.fosd multi} together can be viewed as a generalized form of their condition.

\subsection{Unordered Treatment}\label{subsec.unordered}
Studies of identification of causal effects in unordered choice (treatment) models can be found in \citet{heckman2006understanding}, \citet{heckman2007econometric}, and \citet{heckman2008instrumental}. \citet{heckman2018unordered} showed that the assumptions\footnote{See \citet[pp.~2--3]{heckman2018unordered} for a discussion of these assumptions.} in the preceding literature could be relaxed, and they defined a new monotonicity condition for the identification of causal effects in such models. We follow \citet{heckman2018unordered} and suppose that the support  $\mathcal{D}$ of $D$ is an unordered set with $\mathcal{D=}\left\{d_{1},\ldots,d_{J}\right\}  $ and that the support $\mathcal{ Z}$ of $Z$ with $\mathcal{Z}=\{z_1,\ldots,z_K\}$ can be {unordered} as well. The unordered monotonicity condition proposed by \citet{heckman2018unordered} is as follows (Assumption A-3 of \citet{heckman2018unordered}).
\begin{assumption}\label{ass.unordered monotonicity}
The potential treatment response indicators satisfy the condition that for all	$d\in\mathcal{D}$ and all $z,z^{\prime}\in\mathcal{Z}$, $1\left\{  D_{z^{\prime}}=d\right\}  \geq1\left\{  D_{z}=d\right\}$ almost surely or $  1\left\{  D_{z^{\prime}}=d\right\}  \leq1\left\{  D_{z}=d\right\} $ almost surely.
\end{assumption}
It is worth noting that in Assumption \ref{ass.unordered monotonicity}, $D$ is allowed to be a vector random element. 
In the case where $D,Z\in\{0,1\}$, Assumption \ref{ass.unordered monotonicity} is equivalent to the assumption that $1\left\{  D_{1}=1\right\}  \geq1\left\{  D_{0}=1\right\}$ almost surely or $  1\left\{  D_{1}=1\right\}  \leq1\left\{  D_{0}=1\right\} $ almost surely. 
In practice, we often assume a specific direction in the assumption, such as $1\left\{  D_{1}=1\right\}  \geq1\left\{  D_{0}=1\right\}$ almost surely, which is equivalent to $D_{1}\geq D_{0}$ almost surely in Assumption \ref{ass.IV validity for binary Z}(iii). With the specific direction, we can prespecify a
set $\mathcal{C}\subset\mathcal{D}\times\mathcal{Z}\times\mathcal{Z}$ and assume that $1\left\{  D_{z^{\prime}}=d\right\}  \le 1\left\{  D_{z}=d\right\}$ almost surely for all $(d,z,z^{\prime})\in\mathcal{C}$. {For example, in the above case where $D,Z\in\{0,1\}$ and $1\left\{  D_{1}=1\right\}  \geq1\left\{  D_{0}=1\right\}$ almost surely, we let $\mathcal{C}=\{(0,0,1), (1,1,0)\}$.}
With this monotonicity condition of specified direction, we introduce the IV validity assumption for unordered treatment.\footnote{The test proposed in this paper can be extended for Assumption \ref{ass.unordered monotonicity} in which the direction is not specified. See details in Appendix \ref{sec.unordered treatment monotonicity with no deriction}.} 
 
\begin{assumption}\label{ass.IV validity for unordered D}
IV validity for unordered $D$ and unordered $Z$:
\begin{enumerate}[label=(\roman*)]
		\item Instrument Exclusion: For all $d\in\mathcal{D}$ and all $z,z'\in\mathcal{Z}$, $Y_{dz}=Y_{dz^{\prime}}$ almost surely.
		
		\item Random Assignment: The random element $Z$ is jointly independent of $(
		\tilde{Y},\tilde{D})  $, where
		\begin{align*}
		\tilde{Y}  =\left(  Y_{d_{1}z_{1}},\ldots,Y_{d_{1}z_{K}},\ldots,Y_{d_Jz_1},\ldots,Y_{d_Jz_K}\right) \text{ and }\tilde{D}   =\left(  D_{z_{1}},\ldots,D_{z_{K}}\right)  .
		\end{align*}

		\item Instrument Monotonicity: The potential treatment 
		elements satisfy the condition that $1\left\{  D_{z^{\prime}}=d\right\}  \le 1\left\{  D_{z}=d\right\}$ almost surely for all $(d,z,z^{\prime})\in\mathcal{C}$.
		
	\end{enumerate}
\end{assumption}
Under this assumption, we can define $Y_{d}$ such that $Y_d=Y_{dz}$ almost surely for all $z$, and hence
\begin{align*}
\mathbb{P}\left(  Y\in B,D=d|Z=z^{\prime}\right) = &\, E[1\{Y_d\in B\}\cdot 1\{D_{z^{\prime}}=d\}]\\
 \leq&\, E[1\{Y_d\in B\}\cdot 1\{D_{z}=d\}]=\mathbb{P}\left(Y\in B,D=d|Z=z\right)
\end{align*}
for all Borel sets $B$ and all $\left(  d,z,z^{\prime}\right)  \in\mathcal{C}$. 

\begin{lemma}\label{lemma.testable implication unorder}
	A testable implication of Assumption \ref{ass.IV validity for unordered D} is given by 
	\begin{align}\label{eq.testable implication unordered treatment}
	\mathbb{P}\left(  Y\in B,D=d|Z=z^{\prime}\right)  \le \mathbb{P}\left(Y\in B,D=d|Z=z\right)
	\end{align}
	for all Borel sets $B$ and all $\left(  d,z,z^{\prime}\right)  \in\mathcal{C}$, where $\mathcal{C}$ is a prespecified subset of $\mathcal{D}\times\mathcal{Z}\times\mathcal{Z}$. 
\end{lemma}

As shown in \citet{kitagawa2015test} and \citet{mourifie2016testing}, the testable implication in \eqref{eq.testable implication} is sharp. 
When the treatment or the instrument is multivalued (ordered or unordered), the cases could be complicated. \citet{kedagni2020generalized} considered testing the joint assumptions of instrument exclusion and statistical independence,  which are parts of (and different from) Assumption \ref{ass.IV validity for multivalued Z} and Assumption \ref{ass.IV validity for unordered D}. The exclusion condition of \citet{kedagni2020generalized} is the same as that in the present paper (Assumption \ref{ass.IV validity for multivalued Z}(i) and Assumption \ref{ass.IV validity for unordered D}(i)). The statistical independence condition of \citet{kedagni2020generalized} (the instrument $Z$ is jointly independent of $(Y_{d_1},\ldots, Y_{d_J})$) is weaker than (and implied by) the random assignment condition in the present paper (Assumption \ref{ass.IV validity for multivalued Z}(ii) and Assumption \ref{ass.IV validity for unordered D}(ii)). Thus, the underlying assumptions tested by \citet{kedagni2020generalized} are weaker than those tested by the present paper. 

\citet{kedagni2020generalized} provided sharp testable implications (the generalized instrumental inequalities) for the joint assumptions of instrument exclusion and statistical independence.
Consider a simple case where the outcome $Y\in \{0,1\}$, the treatment $D\in\mathcal{D}$ is multivalued, and the instrument $Z\in\mathcal{Z}$ is also multivalued. Suppose that the exclusion condition and the statistical independence condition hold. We can then define $Y_d=Y_{dz_1}=\cdots=Y_{dz_K}$ for every $d\in\mathcal{D}$. For each $y\in\{0,1\}$, every $d\in\mathcal{D}$, and every $z\in\mathcal{Z}$, we have that
\begin{align}\label{eq.conditional inequality KM}
    \mathbb{P}(Y=y,D=d|Z=z)\le\mathbb{P}(Y_d=y),
\end{align}
which implies that 
\begin{align}\label{eq.KM inequality 1}
    \max_{d\in\mathcal{D}}\sum_{y\in\{0,1\}}\max_{z\in\mathcal{Z}}\mathbb{P}(Y=y,D=d|Z=z)\le\sum_{y\in\{0,1\}}\mathbb{P}(Y_d=y)=1.
\end{align}
For all $y_1,\ldots,y_J\in\{0,1\}$,
\begin{align*}
&\mathbb{P}(Y_{d_1}=y_1,\ldots,Y_{d_J}=y_J)=\min_{z\in\mathcal{Z}}\mathbb{P}(Y_{d_1}=y_1,\ldots,Y_{d_J}=y_J|Z=z)\\
=&\,\min_{z\in\mathcal{Z}}\sum_{j=1}^J\mathbb{P}(Y_{d_1}=y_1,\ldots,Y_{d_J}=y_J,D=d_j|Z=z)\le \min_{z\in\mathcal{Z}}\sum_{j=1}^J\mathbb{P}(Y=y_j,D=d_j|Z=z).
\end{align*}
It then follows that
\begin{align}\label{eq.KM inequality 2}
&\sum_{y_1\in\{0,1\}}\cdots\sum_{y_J\in\{0,1\}}\min_{z\in\mathcal{Z}}\sum_{j=1}^J\mathbb{P}(Y=y_j,D=d_j|Z=z)\notag\\
\ge&\, \sum_{y_1\in\{0,1\}}\cdots\sum_{y_J\in\{0,1\}}\mathbb{P}(Y_{d_1}=y_1,\ldots,Y_{d_J}=y_J)=1.
\end{align}
Next, for every $j$ and every $y_j\in\{0,1\}$, 
\begin{align*}
    \mathbb{P}(Y_{d_j}=y_j)=&\,\sum_{y_1\in\{0,1\}}\cdots\sum_{y_{j-1}\in\{0,1\}}\sum_{y_{j+1}\in\{0,1\}}\cdots\sum_{y_J\in\{0,1\}}\mathbb{P}(Y_{d_1}=y_1,\ldots,Y_{d_J}=y_J)\\
    \le&\,\sum_{y_1\in\{0,1\}}\cdots\sum_{y_{j-1}\in\{0,1\}}\sum_{y_{j+1}\in\{0,1\}}\cdots\sum_{y_J\in\{0,1\}}\min_{z\in\mathcal{Z}}\sum_{\xi=1}^J\mathbb{P}(Y=y_{\xi},D=d_{\xi}|Z=z).
\end{align*}
With \eqref{eq.conditional inequality KM}, we have that
\begin{align}\label{eq.KM inequality 3}
    \max_{j\in\{1,\ldots,J\}}\max_{y_j\in\{0,1\}}\left\{\max_{z\in\mathcal{Z}}\mathbb{P}(Y=y_j,D=d_j|Z=z)-\varphi_j(y_j) \right\}\le 0,
\end{align}
where 
\begin{align*}
    \varphi_j(y_j)=\sum_{y_1\in\{0,1\}}\cdots\sum_{y_{j-1}\in\{0,1\}}\sum_{y_{j+1}\in\{0,1\}}\cdots\sum_{y_J\in\{0,1\}}\min_{z\in\mathcal{Z}}\sum_{\xi=1}^J\mathbb{P}(Y=y_{\xi},D=d_{\xi}|Z=z).
\end{align*}
The inequalities in \eqref{eq.KM inequality 1}--\eqref{eq.KM inequality 3} are the testable restrictions derived by \citet{kedagni2020generalized}, which are different from the proposed testable implications in \eqref{eq.testable implication multivalue}--\eqref{eq.testable implication unordered treatment}. \citet{kedagni2020generalized} suggest using the approach of \citet{chernozhukov2013intersection} to test the restrictions in \eqref{eq.KM inequality 1}--\eqref{eq.KM inequality 3}.\footnote{See Section 5 of \citet{kedagni2020generalized}.}

Though the underlying assumptions tested by \citet{kedagni2020generalized} are weaker than those tested by the present paper, no evidence has been found that, in general, the testable implications of \citet{kedagni2020generalized} are weaker than (or implied by) those proposed by the present paper. 
%It can be shown that as necessary conditions of the IV validity assumptions,\footnote{Exclusion and statistical independence are a part of the IV validity Assumption \ref{ass.IV validity for multivalued Z}.} 
Thus, to the best of our knowledge, the testable implications in \citet{kedagni2020generalized} and those in the present paper could be complementary to each other. That is, the proposed testable restrictions may not be sharp for the IV validity Assumptions \ref{ass.IV validity for multivalued Z} and \ref{ass.IV validity for unordered D}, and the proposed test may not be testing all possible restrictions. 
%the testable restrictions in \citet{kedagni2020generalized} do not imply those in the present paper, and the proposed implications in the present paper do not imply those in \citet{kedagni2020generalized}. 
In practice, we suggest that users first apply the method of \citet{chernozhukov2013intersection} to test the restrictions in \citet{kedagni2020generalized}, and then apply the proposed method to test the joint IV validity assumptions in the present paper. In this way, the test results could be more informative about which part of the IV validity assumptions may fail.

Another interesting question is that if we combine the inequalities of \citet{kedagni2020generalized} and those of the present paper together, are they sharp for the IV validity assumptions? To show this, we may draw on the sharpness results of \citet{kitagawa2015test}, \citet{mourifie2016testing}, and \citet{kedagni2020generalized}. However, this would not be straightforward because we now allow both the treatment and the instrument to be multivalued (ordered or unordered), and the IV validity assumptions involve more conditions (random assignment and monotonicity). Since this technical complication may be beyond the main context of the present paper, we leave it for future study as an independent topic.

\section{Test Formulation} \label{sec.multi D and Z}

To highlight the idea, we first introduce the test for the case where the treatment is multivalued  ordered, with support $\mathcal{D=}\left\{d_{1},\ldots,d_J\right\}  $.  The unordered treatment case will be discussed as an extension in Section \ref{sec.unordered}. Appendix \ref{sec.conditioning covariates setup} in the appendix extends the proposed test for the cases where conditioning covariates may be present. Also, we let $Z$ be multivalued with support $\mathcal{Z}=\{z_1,\ldots,z_K\}$.
The test is constructed based on the testable implication given in \eqref{eq.testable implication multivalue} and \eqref{eq.fosd multi}. Without loss of generality, we assume that $d_{\min}=0$ and $d_{\max}=1$. In practice, we can always normalize $d_{\min}$ and $d_{\max}$ to $0$ and $1$, respectively. Then \eqref{eq.testable implication multivalue} and \eqref{eq.fosd multi} are equivalent to 
\begin{align}\label{eq.testable implication inequalities multi}
&(-1)^{d}\cdot\{\mathbb{P}\left(  Y\in B,D=d|Z=z_{k+1}\right)-\mathbb{P}\left(  Y\in B,D=d|Z=z_k\right)\}  \leq 0  \notag\\
&\text{ and } \mathbb{P}\left( D\in C|Z=z_{k+1}\right)-\mathbb{P}\left( D\in C|Z=z_k\right)\le 0
\end{align}
for all $k$ with $1\le k\le K-1$, all closed intervals $B$ in $\mathbb{R}$, each $d\in\{0,1\}$, and all $C=(-\infty,c]$ with $c\in\mathbb{R}$. Here, \eqref{eq.testable implication multivalue} and \eqref{eq.fosd multi} originally require \eqref{eq.testable implication inequalities multi} to hold for all Borel sets $B$. Similar to {Lemma B.7} of \citet{kitagawa2015test}, {we can} show (by applying Lemma C1 of \citet{andrews2013inference}) that \eqref{eq.testable implication inequalities multi} holding for all closed intervals $B$ is equivalent to  \eqref{eq.testable implication inequalities multi} holding for all Borel sets $B$. 

By definition, for all $B,C\in\mathcal{B}_{\mathbb{R}}$ and all $k$ with $1\le k\le K$,
$
\mathbb{P}\left(  Y\in B,D\in C|Z=z_{k}\right)={\mathbb{P}\left(  Y\in B,D\in C,Z=z_{k}\right)
}/{\mathbb{P}\left(  Z=z_{k}\right)  }
$.
We now define function spaces
\begin{align}{\label{def.function spaces}}
&\mathcal{G}_{K}=\left\{  1_{\mathbb{R}\times\mathbb{R}  \times\left\{  z_{k}\right\}  }:k=1,2,\ldots
,K\right\},\notag\\
&\mathcal{G}=\left\{  \left(  1_{\mathbb{R}\times\mathbb{R}  \times\left\{  z_{k}\right\}  },1_{\mathbb{R}\times\mathbb{R}  \times\left\{  z_{k+1}\right\}  }\right)
:k=1,2,\ldots,K-1\right\},  \notag\\
&\mathcal{H}_{1}=\left\{  \left(  -1\right)  ^{d}\cdot1_{B\times\left\{
d\right\}  \times\mathbb{R}}:B\text{ is a closed interval in }\mathbb{R},
d\in\{0,1\}\right\},\notag\\
&\bar{\mathcal{H}}_1=\left\{  \left(  -1\right)  ^{d}\cdot1_{B\times\left\{
	d\right\}\times\mathbb{R}  }:B\text{ is a closed}, \text{open}, \text{or half-closed interval in }
\mathbb{R}
,d\in\left\{  0,1\right\}  \right\},\notag\\
&\mathcal{H}_{2}=\left\{  1_{\mathbb{R}\times C \times\mathbb{R}}:C=(-\infty,c],c\in\mathbb{R}\right\},\notag\\
& \bar{\mathcal{H}}_2=\left\{  1_{\mathbb{R}\times C \times\mathbb{R}  }: C=(-\infty,c] \text{ or } C=(-\infty,c),c\in\mathbb{R}  \right\},\notag\\
&\mathcal{H}=\mathcal{H}_{1}\cup\mathcal{H}_{2}, \text{ and } \bar{\mathcal{H}}=\bar{\mathcal{H}}_1\cup\bar{\mathcal{H}}_2.
\end{align}
Let $\mathcal{P}$ denote the set of probability measures on $(\mathbb{R}^{3},\mathcal{B}_{\mathbb{R}^{3}})$. We use an i.i.d.\ sample \linebreak $\{\left(  Y_i,D_i,Z_i \right)\}_{i=1}^{n}  $ which is distributed according to some probability distribution $Q$ in $\mathcal{P}$, that is, that the measure $Q(G)=\mathbb{P}((Y_i,D_i,Z_i)\in G)$ for all $G\in\mathcal{B}_{\mathbb{R}^3}$, to construct a test for the testable implication given in \eqref{eq.testable implication multivalue} and \eqref{eq.fosd multi} (or in \eqref{eq.testable implication inequalities multi}). 
For every $Q\in\mathcal{P}$ and every measurable function $v$, by an abuse of notation we define
\begin{align}\label{eq.Q map}
Q\left(  v\right)  =\int v\,\mathrm{d}Q.
\end{align}
Define, by convention (see, for example, \citet[p.~45]{folland2013real}), that 
\begin{align}\label{eq.0timesinfinity}
0\cdot \infty=0.
\end{align}
For each $Q\in\mathcal{P}$, the closure of $\mathcal{H}$ in $L^2(Q)$ is equal to $\bar{\mathcal{H}}$ (Lemma \ref{lemma.H complete}). For every $Q\in\mathcal{P}$ and every $\left(  h,g\right)  \in {\bar{\mathcal{H}}\times\mathcal{G}}$ with $g=(g_{1},g_{2})$, define
\begin{align}\label{eq.phi_Q}
\phi_Q\left(  h,g\right)    =\frac{Q\left(  h\cdot g_{2}\right)
}{Q\left(  g_{2}\right)  }-\frac{Q\left(  h\cdot g_{1}\right)  }{Q\left(
g_{1}\right)  }.
\end{align}
With \eqref{eq.0timesinfinity}, $\phi_Q$ is always well defined.  Then the null hypothesis equivalent to \eqref{eq.testable implication inequalities multi} is 
\begin{align}\label{eq.null 1}
H_0: \sup_{\left(  h,g\right)  \in
	{{\mathcal{H}}\times\mathcal{ G}}}\phi_Q\left( h,g\right)\le 0
\end{align}
if the underlying distribution of the data is $Q$. Since $Q(v)$ is continuous on  $L^2(Q)$, \eqref{eq.null 1} is equivalent to $\sup_{\left(  h,g\right)  \in
	{\bar{\mathcal{H}}\times\mathcal{ G}}}\phi_Q\left( h,g\right)\le 0$. 
The alternative hypothesis is naturally set to 
\begin{align*}
H_1:\sup_{\left(  h,g\right)  \in
	{{\mathcal{H}}\times\mathcal{ G}}}\phi_Q\left( h,g\right)> 0.
\end{align*} 
Define the sample analogue of $\phi_Q$ by
\begin{align*}
\hat{\phi}_Q\left(  h,g\right)  =\frac{\hat{Q}(h\cdot g_{2})  }{\hat{Q}(g_{2})}-\frac{\hat{Q}(h\cdot g_{1}) }{\hat{Q}(g_{1}) },
\end{align*}
where $\hat{Q}$ denotes the empirical probability measure of $Q$ such that for every
measurable function $v$,
\begin{align}\label{eq.defPn}
\hat{Q}\left(  v\right)  =\frac{1}{n}\sum_{i=1}^{n}v\left(  Y_{i},D_{i},Z_{i}\right),
\end{align}
and $\{(  Y_{i},D_{i},Z_{i})\}_{i=1}^n$ is the i.i.d.\ sample distributed according to $Q$.

The goal of this section is to construct a test for the $H_0$ in \eqref{eq.null 1}. To evaluate the ability of the test to provide size control, we consider a ``local'' sequence of probability distributions $\{P_n\}_{n=1}^{\infty}\subset\mathcal{P}$ under which the testable implication is true and $P_n$ converges to some probability measure $P\in\mathcal{P}$. We introduce the next two assumptions to formalize the above settings. 
\begin{assumption}
	\label{ass.independent data}$\left\{  \left(  Y_{i},D_{i},Z_{i}\right)
	\right\}  _{i=1}^{n} $ is an i.i.d.\ data set distributed according to probability distribution $P_n$ for each $n$, where $D_i$ and $Z_i$ are discrete variables with support $\mathcal{D}$ and $\mathcal{Z}$, respectively. 
\end{assumption}

\begin{assumption}\label{ass.probability path}
	There is a probability measure $P\in\mathcal{P}$ such that 
	\begin{equation}\label{eq.probaility path}
	\lim_{n\rightarrow\infty}\int\left(  \sqrt{n}\left\{  \mathrm{d}P_{n}^{1/2}-\mathrm{d}P^{1/2}\right\}  -\frac{1}{2}v_0\mathrm{d}P^{1/2}\right)  ^{2}=0
	\end{equation}
	for some measurable function $v_0$, where $\mathrm{d}P_{n}^{1/2}$ and $\mathrm{d}P^{1/2}$ denote the square roots of the densities of $P_{n}$ and $P$, respectively.

\end{assumption}
Assumptions \ref{ass.independent data} and \ref{ass.probability path} assume an i.i.d.\ sample whose distribution $P_n$ is allowed to {change} as $n$ increases, and to converge to some probability distribution $P$ as defined in (3.10.10) of \citet{van1996weak}. The local analysis of \citet{fang2014inference} considered the case where the value of the underlying parameter may be close to a point at which the map involved in the test statistic is only directionally differentiable (not fully differentiable).
A similar convergent distribution sequence was introduced to show the local size control of their test.\footnote{See Examples 2.1 and 2.2 of \citet{fang2014inference}.} As will be shown later, our test statistic involves a nondifferentiable (neither fully nor directionally differentiable) map. We follow \citet{fang2014inference} and assume such a convergent distribution sequence to show the local size control of our test.

Clearly, ${\mathcal{H}}\times\mathcal{ G}\subset L^2(P)\times(L^2(P)\times L^2(P))$. Under Assumption \ref{ass.probability path}, define a metric $\rho_{P}$ on $L^2(P)\times(L^2(P)\times L^2(P))$ by
\begin{equation}\label{eq.rho}
\rho_{P}\left(  \left(  h,g\right)  ,\left(  h^{\prime},g^{\prime}\right)
\right)  =\left\Vert h-h^{\prime}\right\Vert _{L^2\left(  P\right)
}+\left\Vert g_{1}-g_{1}^{\prime}\right\Vert _{L^2\left(  P\right)
}+\left\Vert g_{2}-g_{2}^{\prime}\right\Vert _{L^2\left(  P\right)  }
\end{equation}
for all
$\left(  h,g\right)  ,\left(  h^{\prime},g^{\prime}\right)  \in L^2(P)\times(L^2(P)\times L^2(P))$ with $g=(g_1,g_2)$ and $g^{\prime}=(g_1^{\prime},g_2^{\prime})$. By Lemma \ref{lemma.complete HG}, the closure of $\mathcal{H}\times\mathcal{G}$ in $L^2(P)\times(L^2(P)\times L^2(P))$ under $\rho_{P}$ is equal to ${\bar{\mathcal{H}}\times\mathcal{G}}$, where $\bar{\mathcal{H}}$ is defined in \eqref{def.function spaces}. Define
\begin{align*}
\Lambda(Q)=\prod_{k=1}^{K}Q\left(1_{\mathbb{R}\times\mathbb{R}\times\{z_k\}}\right) \text{  for all $Q\in\mathcal{P}$}, \text{ and } T_n=n\cdot\prod_{k=1}^{K}\hat{P}_n\left(1_{\mathbb{R}\times\mathbb{R}\times\{z_k\}}\right),
\end{align*}
where $\hat{P}_n$ is the empirical probability measure of $P_n$ defined as in \eqref{eq.defPn}. Under Assumption \ref{ass.probability path}, we mainly consider the nontrivial case where $\Lambda(P)>0$. 
Also, for every $Q\in\mathcal{P}$, define
\begin{align}\label{eq.stat variance multi}
\sigma_{Q}^2  (h,g)=\Lambda(Q) \cdot \left\{\frac{ Q\left(  h^2\cdot g_{2}\right)   }{Q^{2}\left(
	g_{2}\right)  }  -\frac{ Q^2\left(  h\cdot g_{2}\right)   }{Q^{3}\left(
	g_{2}\right)  }  
+\frac{ Q\left(	h^2\cdot g_{1}\right)   }{Q^{2}\left(  g_{1}\right)  }
-\frac{ Q^2\left(	h\cdot g_{1}\right)   }{Q^{3}\left(  g_{1}\right)  }\right\}
\end{align}
for all $(h,g)\in \bar{\mathcal{H}}\times\mathcal{G}$ with $g=(g_1,g_2)$, where $Q^m(v)=[Q(v)]^m$ for all $m\in\mathbb{N}$ and all measurable $v$. 

%\begin{assumption}\label{ass.min z positive}
%	$\min_{k\le K}\mathbb{P}(Z_i=z_k)>0$. 
%\end{assumption}

%Assumption \ref{ass.min z positive} simply guarantees that $Z_i$ can take all the possible values in $\mathcal{ Z}$ with positive probability, or in another word, $\mathcal{ Z}$ is the smallest set of possible values of $Z_i$. With this assumption and the setup above, we first introduce the following lemma.

\begin{lemma}
\label{lemma.weak convergence phi_K} Under Assumptions
\ref{ass.independent data} and \ref{ass.probability path}, $\sqrt{T_n}(  \hat{\phi}_{P_n}-\phi_{P})  \leadsto\mathbb{G}$ for some tight\footnote{In a metric space, tightness implies separability.} random element
$\mathbb{G}$ which almost surely has a uniformly $\rho_{P}$-continuous path, and for all $\left(  h,g\right)  \in\bar{\mathcal{H}}\times\mathcal{G}$ with $g=(g_1,g_2)$, the variance $Var\left(  \mathbb{G}\left(  h,g\right)  \right)$ is equal to the $\sigma_{P}^2  (h,g)$ given in \eqref{eq.stat variance multi}, where
\begin{align}\label{eq.sigma_square bound}
\sigma_P^{2}\left(  h,g\right)  \leq 1/4\cdot\max_{(g_1^{\prime},g_2^{\prime})\in\mathcal{G}}\left\{
\Lambda\left(  P\right)  /P\left(  g_{2}^{\prime}\right)  +\Lambda\left(  P\right)
/P\left(  g_{1}^{\prime}\right)  \right\}\le 1/2\cdot (K-1)^{-(K-1)},
\end{align}
and $K$ is the number of elements of $\mathcal{Z}$. In particular, $\sigma_P^{2}(h,g)\le1/4$ for all $(h,g)\in\bar{\mathcal{H}}\times\mathcal{ G}$ when $K=2$.
\end{lemma}
Lemma \ref{lemma.weak convergence phi_K} provides the pointwise ($P$ is fixed) asymptotic distribution of $\sqrt{T_n}(  \hat{\phi}_{P_n}-\phi_P)$ as $P_n$ converges to $P$ under Assumption \ref{ass.probability path}. We note that the pointwise asymptotic distribution of $\sqrt{T_n}(\hat{\phi}_{P_n}-\phi_P)$ is different from the asymptotic distribution of $\sqrt{T_n}(\hat{\phi}_{P_n}-\phi_{P_n})$ which can be obtained by Theorem 3.10.12 of \citet{van1996weak}. The weak convergence of $\sqrt{T_n}(\hat{\phi}_{Q}-\phi_{Q})$ uniform in $Q$ may be obtained under different assumptions following the notion of \citet[p.~168]{van1996weak}. We derive the pointwise asymptotic distribution in order to obtain the null asymptotic distribution of the test statistic using the proposed extended delta method. See the discussion after Theorem \ref{thm.weak convergence SK}.
Lemma \ref{lemma.weak convergence phi_K} also provides the asymptotic variance of $\sqrt{T_n}(  \hat{\phi}_{P_n}-\phi_P)$, which is uniformly bounded by $1$ for all $K>1$. We used the quantity $\sqrt{T_n}$ instead of $\sqrt{n}$ to establish the asymptotic distribution in order to achieve a parameter-free bound for the asymptotic variance as shown in \eqref{eq.sigma_square bound}.
The quantity $T_{n}$ is asymptotically equivalent to $n$ in the sense that $T_{n}/n\rightarrow\prod_{k=1}^{K}\mathbb{P}\left(  Z=z_{k}\right)  $ in probability. If we use $\sqrt{n}$, the bound of the asymptotic variance may involve the underlying parameter $P$. In the binary instrument case where $Z\in\{0,1\}$, we let $m_0=\sum_{i=1}^{n}1\left\{  Z_{i}=0\right\}  $ and $m_1=\sum_{i=1}^{n}1\left\{  Z_{i}=1\right\}  $. It then follows that $T_{n}=m_0 m_1/n$ which is used in the test of \citet{kitagawa2015test}. Suppose instead $Z\in\{0,1,2\}$, and we let $m_z=\sum_{i=1}^{n}1\left\{  Z_{i}=z\right\}  $ for $z\in\{0,1,2\}$. Then $T_{n}=m_0 m_1 m_2/n^2$.

The bound in \eqref{eq.sigma_square bound} will be useful when we construct the test statistic. 
By \eqref{eq.stat variance multi}, for every $\left(  h,g\right)  \in\bar{\mathcal{H}}\mathcal{\times G}$ with
$g=\left(  g_{1},g_{2}\right) $, define the sample analogue of
$\sigma_P^{2}\left(  h,g\right)  $ by
\begin{align}\label{eq.estimated stat variance multi}
\hat{\sigma}_{P_n}^{2}\left(  h,g\right)  =\frac{T_n}{n}\cdot\left\{\frac{ \hat{P}_n\left(  h^2\cdot g_{2}\right)   }{\hat{P}_n^{2}\left(
	g_{2}\right)  }  -\frac{ \hat{P}_n^2\left(  h\cdot g_{2}\right)   }{\hat{P}_n^{3}\left(
	g_{2}\right)  }  
+\frac{ \hat{P}_n\left(	h^2\cdot g_{1}\right)   }{\hat{P}_n^{2}\left(  g_{1}\right)  }
-\frac{ \hat{P}_n^2\left(	h\cdot g_{1}\right)   }{\hat{P}_n^{3}\left(  g_{1}\right)  }\right\}.
\end{align}
Note that for each $h\in\bar{\mathcal{H}}$ and each $g_{l}\in\mathcal{G}_{K}$, if $\hat{P}_{n}(g_{l})=0$ then $\hat{P}_{n}(h\cdot g_{l})=0$. By \eqref{eq.0timesinfinity}, $\hat{\sigma}_{P_n}^2$ is well defined. Similar to \eqref{eq.sigma_square bound}, we can find a bound for $\hat{\sigma}_{P_n}$ for every finite sample. It can be shown that for all $(h,g)$,
\begin{align}\label{eq.sigma_square_hat bound}
\hat{\sigma}_{P_n}^{2}\left(  h,g\right)  \leq 1/4\cdot\max_{(g_1^{\prime},g_2^{\prime})\in\mathcal{G}}\left\{
(T_n/n)  /\hat{P}_n\left(  g_{2}^{\prime}\right)  +(T_n/n)
/\hat{P}_n\left(  g_{1}^{\prime}\right)  \right\}\le 1/2\cdot (K-1)^{-(K-1)}.
\end{align}
Clearly, the bounds for $\sigma_P$ and $\hat{\sigma}_{P_n}$ will decrease as $K$ increases.

We may extend the idea of \citet{kitagawa2015test} and construct the test statistic to be 
\begin{align}\label{eq.test stat K}
\sup_{(h,g)\in{\mathcal{H}}\times\mathcal{G}} \frac{\sqrt{T_n}\hat{\phi}_{P_n}(h,g)}{\max\{\xi,\hat{\sigma}_{P_n}(h,g)\}}
\end{align}
for some positive number (trimming parameter) $\xi$. Here, $\xi$ plays two roles: (1) Since $\hat{\sigma}_{P_n}$ can be zero, $\xi$ bounds the denominator away from zero;\footnote{In practice, when the sample size is small, it is possible that we only have a small number of observations for $Z=z_k$ for some $k$. In this case, we can use $\sqrt{n}$ instead of $\sqrt{T_n}$ to construct the test statistic. We then use \eqref{eq.estimated stat variance multi} to find an empirical bound for $\hat{\sigma}_{P_n}$, and use this bound to determine the values of $\xi$. We could also redefine the instrument $Z$, in some cases, such that we have more observations for each possible value of the redefined instrument. For example, we may define the new instrument $\tilde{Z}=1\{Z\ge z_0\}$ for some $z_0$. However, this will change the definitions of all types of individuals (always takers, compliers, defiers, and never takers). In this case, we need to guarantee that the instrument used in the empirical analysis and the instrument used in the test are the same. The test result for $\tilde{Z}$ may be false for $Z$.} (2) as shown in the Monte Carlo studies of \citet{kitagawa2015test} and the present paper, different values of $\xi$, from small (close to 0) to large (close to 1), may lead to different powers of the test for the same data generating process (DGP), which could be close to 0. \citet{kitagawa2015test} suggests that if there is no prior knowledge
available about a likely alternative, the default choice of $\xi$ could be set to $0.07$  according to the simulation studies for the binary treatment and binary instrument case. They also suggest that users report test results using different values of $\xi$. This paper constructs the test statistic in a way that, loosely speaking, computes the weighted average of the test statistics in \eqref{eq.test stat K} over $\xi$.\footnote{In this way, we can avoid repeating the test using the same data set but different values of $\xi$ and making a decision based on all these results. The potential issue of multiple comparisons can be prevented accordingly.} If we put all the weight on one particular value of $\xi$, the test statistic degenerates to the test statistic in  \eqref{eq.test stat K}.

Let $\Xi$ be a predetermined closed  subset of $[0,1]$ such that $0\not\in\Xi$. The set $\Xi$ contains all the values of $\xi$ used for constructing the test statistic. Only one of the values greater than (or equal to) the bound in Lemma \ref{lemma.weak convergence phi_K}, say $1$, needs to be included in $\Xi$. The test statistic in \eqref{eq.test stat K} reduces to the unweighted KS statistic when $\xi=1$. 
Let $\nu$ be a positive measure on $\Xi$.
%Let $\Xi\subset [0,1]$ be a predetermined measurable set. 
\begin{assumption}\label{ass.nu}
 The measure $\nu$ satisfies that $0<\nu(\Xi)<\infty$ and ${S}_n\in L^{1}(\nu)$ for all $\omega\in\Omega$ and all $n$ with
 \begin{align*}
 {S}_n(\xi)= 	\sup_{(h,g)\in{\mathcal{H}}\times\mathcal{G}} \frac{\hat{\phi}_{P_n}(h,g)}{\max\{\xi,\hat{\sigma}_{P_n}(h,g)\}}.
 \end{align*}
\end{assumption}
Now we set the test statistic to
\begin{equation}\label{eq.test stat expansion}
TS_n=\int_{\Xi}\sup_{(h,g)\in{\mathcal{H}}\times\mathcal{G}} \frac{\sqrt{T_n}\hat{\phi}_{P_n}(h,g)}{\max\{\xi,\hat{\sigma}_{P_n}(h,g)\}}\,\mathrm{d}\nu(\xi).
\end{equation}
The measure $\nu$ could be a Dirac measure centered at some fixed $\xi\in\Xi$. This is equivalent to using a particular value for the trimming parameter to construct the test statistic as in \eqref{eq.test stat K}. Or $\nu$ could be a discrete or continuous probability measure that assigns probabilities to the elements of $\Xi$. This is equivalent to using a weighted average of the test statistics in \eqref{eq.test stat K} over $\xi$. By using \eqref{eq.test stat expansion}, we take into account the fact that the values of $\xi$ may influence the power of the test, and we can also avoid the multiple testing issue. See the discussion in Section \ref{sec.simulation} about the computational simplification of the test statistic in \eqref{eq.test stat expansion}. Define 
\begin{align}\label{eq.Psi_HG}
\Psi_{{\mathcal{H}}\times\mathcal{G}}=\left\{  \left(  h,g\right)
\in{\mathcal{H}}\times\mathcal{G}:\phi_P\left(  h,g\right)  =0\right\} \text{  and } \Psi_{\bar{\mathcal{H}}\times\mathcal{G}}=\left\{  \left(  h,g\right)
\in\bar{\mathcal{H}}\times\mathcal{G}:\phi_P\left(  h,g\right)  =0\right\}.
\end{align}
Since $1_{\left\{  a\right\}  \times\left\{  0\right\}  \times\mathbb{R}},-1_{\left\{  a\right\}  \times\left\{  1\right\}  \times\mathbb{R}}
\in{\mathcal{H}}$ for all $a\in\mathbb{R}$,  $\Psi_{{\mathcal{H}}\times\mathcal{G}}$ and $\Psi_{\bar{\mathcal{H}}\times\mathcal{G}}$ are not empty. 

In the following theorem, we establish the asymptotic distribution of the test statistic under null. We note that the $L^r$ ($r\in\mathbb{N}$) spaces play an important role in deriving this asymptotic distribution. For example, we show that $\bar{\mathcal{H}}$ is compact in $L^2(Q)$ for every $Q\in\mathcal{P}$ and $\bar{\mathcal{H}}\times\mathcal{G}$ is compact in $L^2(P)\times(L^2(P)\times L^2(P))$ under $\rho_P$ (constructed based on the $L^2$ norm). We obtain the Glivenko--Cantelli and the Donsker results using the $L^1$ and the $L^2$ norms. We also show that the weak limit $\mathbb{G}$ of $\sqrt{T_n}(\hat{\phi}_{P_n}-\phi_P)$ in Lemma \ref{lemma.weak convergence phi_K} has a continuous path under $\rho_P$.\footnote{See Appendix \ref{appen.auxiliary} for more detailed results.} The weak convergence in Theorem \ref{thm.weak convergence SK} is established by using these fundamental results.

\begin{theorem}
\label{thm.weak convergence SK}Suppose Assumptions \ref{ass.independent data}, \ref{ass.probability path}, and \ref{ass.nu} hold. If the $H_0$ in \eqref{eq.null 1} is true with $Q=P_n$ for all $n$, then
\begin{align}\label{eq.TS weak convergence}
TS_n \leadsto\int_{\Xi}\sup_{(h,g)\in{\Psi_{{\mathcal{H}}\times\mathcal{G}}}} \frac{\mathbb{G}(h,g)}{\max\{\xi,{\sigma}_{P}(h,g)\}}\,\mathrm{d}\nu(\xi),
\end{align}
where $\mathbb{G}$ is as in Lemma \ref{lemma.weak convergence phi_K}.

\end{theorem}
Theorem \ref{thm.weak convergence SK} provides the pointwise ($P$ is fixed) asymptotic distribution of the test statistic if the $H_0$ in \eqref{eq.null 1} is true with $Q=P_n$ for all $n$.\footnote{More precisely, the weak convergence in \eqref{eq.TS weak convergence} is under $P_n$.} To find this asymptotic distribution, we employed the pointwise weak convergence in Lemma \ref{lemma.weak convergence phi_K} and the extended delta method provided in Appendix \ref{appendix.main results}. Because of a nondifferentiability issue, the existing delta methods fail to work in establishing the weak convergence in \eqref{eq.TS weak convergence}. In Appendix \ref{appendix.main results}, we provide an extended continuous mapping theorem and an extended delta method elaborated by Theorems \ref{lemma.random continuous mapping} and \ref{lemma.random delta method}, respectively, to deal with this technical issue. See further discussion in Remark \ref{remark.application of extended delta method}. Theorem \ref{lemma.random continuous mapping} can be viewed as an extension of Theorem 1.11.1 of \citet{van1996weak}, and Theorem \ref{lemma.random delta method} can be viewed as an extension of Theorem 3.9.5 of \citet{van1996weak} and of Theorem 2.1 of \citet{fang2014inference}. For simplicity of notation, we let 
\begin{align*}
	\mathbb{T}=\int_{\Xi}\sup_{(h,g)\in{\Psi_{{\mathcal{H}}\times\mathcal{G}}}} \frac{\mathbb{G}(h,g)}{\max\{\xi,{\sigma}_{P}(h,g)\}}\,\mathrm{d}\nu(\xi) \text{  and } \mathbb{T}_0=\int_{\Xi}\sup_{(h,g)\in{\Psi_{\bar{\mathcal{H}}\times\mathcal{G}}}} \frac{\mathbb{G}_0(h,g)}{\max\{\xi,{\sigma}_{P}(h,g)\}}\,\mathrm{d}\nu(\xi),
\end{align*}
where $\mathbb{G}_0$ is some random element such that
$ \mathbb{G}=\mathbb{G}_0 + \Lambda(P)^{1/2}\mathcal{L}'_P(Q_0)$,\footnote{See more details in the proof of Theorem \ref{thm.test multi}.} where
$Q_0(v)=P(vv_0)$ for all suitable $v$, and for all $(h,g)\in\bar{\mathcal{H}}\times\mathcal{G}$,
\begin{align*}
	\mathcal{L}_{P}^{\prime}\left(  Q_0\right) & \left(  h,g\right)  =\\
	&\frac{Q_0\left(
		h\cdot g_{2}\right)  P\left(  g_{2}\right)  -P\left(  h\cdot g_{2}\right)
		Q_0\left(  g_{2}\right)  }{P^{2}\left(  g_{2}\right)  }-\frac{Q_0\left(  h\cdot
		g_{1}\right)  P\left(  g_{1}\right)  -P\left(  h\cdot g_{1}\right)  Q_0\left(
		g_{1}\right)  }{P^{2}\left(  g_{1}\right)  }.
	\end{align*}
It can be shown that $\mathcal{L}'_P(Q_0)\le 0$ on $\Psi_{\bar{\mathcal{H}}\times\mathcal{G}}$ under $H_0$, and thus $\mathbb{G}\le\mathbb{G}_0$.
Following the proof of Theorem \ref{thm.weak convergence SK}, we can show that 
\begin{align*}
\mathbb{T}=\int_{\Xi}\sup_{(h,g)\in{\Psi_{\bar{\mathcal{H}}\times\mathcal{G}}}} \frac{\mathbb{G}(h,g)}{\max\{\xi,{\sigma}_{P}(h,g)\}}\,\mathrm{d}\nu(\xi).
\end{align*}
It then follows that $\mathbb{T}\le\mathbb{T}_0$.
When $P_n$ is fixed at some $P$ for all $n$, then $v_0=0$ and $\mathbb{G}_0=\mathbb{G}$.

\subsection{Bootstrap-based Inference}\label{subsec.bootstrap}

It was shown that the asymptotic distribution in \eqref{eq.TS weak convergence} involves the set $\Psi_{{\mathcal{H}}\times\mathcal{G}}$ that depends on the underlying probability measure $P$. Therefore, we need to find a ``valid'' estimator
$\widehat{\Psi_{{\mathcal{H}}\times\mathcal{G}}}$ for $\Psi_{{\mathcal{H}}\times\mathcal{G}}$ in order to consistently approximate the asymptotic distribution. 
By the definition of
$\Psi_{{\mathcal{H}}\times\mathcal{G}}$ in \eqref{eq.Psi_HG}, we construct $\widehat{\Psi_{{\mathcal{H}}\times\mathcal{G}}}$  by
\begin{align}\label{eq.Psi_hat}
\widehat{\Psi_{{\mathcal{H}}\times\mathcal{G}}}=\left\{  \left(  h,g\right)
\in{\mathcal{H}}\times\mathcal{G}:\sqrt{T_n}\left\vert \frac{\hat{\phi}_{P_n}(h,g)}{\max\{\xi_0,\hat{\sigma}_{P_n}(h,g)\}}\right\vert  \leq\tau
_{n}\right\}
\end{align}
with $\tau_{n}\rightarrow\infty$ and $\tau_{n}/\sqrt{n}\rightarrow0$ as
$n\rightarrow\infty$, where $\xi_0$ is a small positive number. We suggest using $\xi_0=0.001$ in practice.\footnote{It can be shown that ${\widehat{\Psi_{{\mathcal{H}}\times\mathcal{G}}}}$ can also be used to approximate the asymptotic distribution when $\mathcal{D}=\{d_1,d_2,\ldots\}$. See \eqref{eq.equi bootstrap test stat}.} 
This is a method similar to that which is used in \citet{Beare2015improved} and  \citet{Beare2017improved} to estimate contact sets in independent contexts. {See \citet{linton2010improved} and \citet{lee2018testing} for further discussion of estimation of contact sets.}  
Each $(h,g)$ 
is included in $\widehat{\Psi_{{\mathcal{H}}\times\mathcal{G}}}$ if $\sqrt{T_n}|\hat{\phi}_{P_n}(h,g)|$ is no more
than $\tau_n$ estimated standard deviations from zero. 
As mentioned by \citet{Beare2017improved}, we can
effectively use pointwise confidence intervals to select points in this way.
\subsubsection{Test Procedure}\label{subsbusec.test procedure}

We implement the test in the following sequence of steps:
\begin{enumerate}[label=(\arabic*)]
\item Obtain the bootstrap sample $\{  (  \hat{Y}_{i},\hat{D}_{i},\hat{Z}_{i})  \}  _{i=1}^{n}$ drawn independently with replacement from the
sample $\left\{  \left(  Y_{i},D_{i},Z_{i}\right)  \right\}  _{i=1}^{n}$.

\item Calculate the bootstrap version of $\hat{\phi}_{P_n}$ by
\begin{align}\label{eq.phi hat star}
\hat{\phi}_{P_n}^{B}\left(  h,g\right)  =\frac{\hat{P}_{n}^{B}\left(  h\cdot
g_{2}\right)  }{\hat{P}_{n}^{B}\left(  g_{2}\right)  }-\frac{\hat{P}_{n}^{B}\left(
h\cdot g_{1}\right)  }{\hat{P}_{n}^{B}\left(  g_{1}\right)  },
\end{align}
let $T_n^B=n\cdot\prod_{k=1}^{K}\hat{P}_n^B\left(1_{\mathbb{R}\times\mathbb{R}\times\{z_k\}}\right)$, and calculate the bootstrap version of $\hat{\sigma}_{P_n}$ by
\begin{align}\label{eq.sigma hat star}
\hat{\sigma}_{P_n}^{B}\left(  h,g\right)  =\sqrt{\frac{T_n^B}{n}}\cdot\sqrt{\frac{ \hat{P}_{n}^{B}\left(  h^2\cdot g_{2}\right)   }{\hat{P}_n^{B}\left(g_{2}\right)^2  }  -\frac{ \hat{P}_n^{B }\left(  h\cdot g_{2}\right)^2   }{\hat{P}_n^{B}\left(
		g_{2}\right)^3  }  
	+\frac{ \hat{P}_n^{B}\left(	h^2\cdot g_{1}\right)   }{\hat{P}_n^{B}\left(  g_{1}\right)^2  }
	-\frac{ \hat{P}_n^{B}\left(	h\cdot g_{1}\right)^2   }{\hat{P}_n^{B}\left(  g_{1}\right)^3  }  }
\end{align}
for all
$(h,g)\in\bar{\mathcal{H}}\times\mathcal{G}$ with $g=(g_1,g_2)$, where $\hat{P}_{n}^{B}\left(  v\right)  =n^{-1}\sum_{i=1}^{n}v(  \hat{Y}_{i},\hat{D}_{i},\hat{Z}_{i})  $ for all measurable $v$.  We note that \eqref{eq.sigma_square_hat bound} also provides a bound for $\hat{\sigma}^{B}_{P_n}$.

\item Calculate the bootstrap version of the test statistic by  
\begin{align}\label{eq.bootstrap test stat}
TS_n^B=\int_{\Xi}\sup_{(h,g)\in{\widehat{\Psi_{{\mathcal{H}}\times\mathcal{G}}}}} \frac{\sqrt{T_n^B}(  \hat{\phi}_{P_n}^{B}(h,g)-\hat{\phi}_{P_n}(h,g))}{\max\{\xi,\hat{\sigma}_{P_n}^B(h,g)\}}\,\mathrm{d}\nu(\xi).
\end{align}
Since the asymptotic distribution in \eqref{eq.TS weak convergence} involves a nonlinear map, the bootstrap test statistic  in \eqref{eq.bootstrap test stat} was constructed following the idea of \citet{fang2014inference}. {The nonlinearity of the map may cause inconsistencies in the ``standard'' bootstrap approximation. See \citet{dumbgen1993nondifferentiable},  \citet{andrews2000inconsistency}, and \citet{fang2014inference} for details.} Because of the denominator ${\max\{\xi,\hat{\sigma}_{P_n}^B(h,g)\}}$, our approach is an extension of that of \citet{fang2014inference}. The calculation of \eqref{eq.bootstrap test stat} can be simplified in practice as discussed in Section \ref{sec.simulation} for \eqref{eq.test stat expansion}. %See Section \ref{sec.simulation} for details regarding Monte Carlo simulations. 
\item Repeat steps (1), (2), and (3) $n_B$ times independently, for (say) $n_B=1000$. Given
the nominal significance level $\alpha$, calculate the bootstrap critical value
$\hat{c}_{1-\alpha}$ by
\begin{align}\label{eq.c hat}
\hat{c}_{1-\alpha}=\inf\left\{  c:\mathbb{P}\left(  TS_n^B  \leq c\bigg|\{(Y_{i},D_{i},Z_{i})\}_{i=1}^{n}\right)  \geq
1-\alpha\right\}  .
\end{align}
In practice, we approximate $\hat{c}_{1-\alpha}$ by computing the $1-\alpha$ quantile of the $n_B$ independently generated bootstrap statistics, with $n_B$ chosen as large as is computationally convenient.

\item The decision rule for the test is: Reject $H_{0}$ if $TS_n>\hat{c}_{1-\alpha}$.

\end{enumerate}

The following theorem presents the asymptotic properties of the proposed test. Under Assumption \ref{ass.probability path}, Theorem \ref{thm.test multi}(i) provides the local size control of the test. As discussed in \citet{fang2014inference}, the asymptotic distribution of the test statistic may discontinuously depend on the parameter of interest, if the map involved in the test statistic is not fully differentiable. However, the finite sample distribution of the test statistic often continuously depends on the parameter of interest. \citet{imbens2004confidence} emphasize that this discrepancy may cause poor finite sample properties of the test. As suggested by \citet{fang2014inference}, a local analysis can help better approximate the finite sample properties of the test when the parameter of interest is close to a point at which the map is not fully differentiable. Our test statistic involves a  nondifferentiable map, and Theorem \ref{thm.test multi}(i) provides evidence for the good finite sample size property of the test. 
\begin{theorem}
\label{thm.test multi}Suppose Assumptions \ref{ass.independent data}, \ref{ass.probability path}, and \ref{ass.nu} hold. 
\begin{enumerate}[label=(\roman{*})]
\item If the $H_0$ in \eqref{eq.null 1} is true with $Q=P_n$ for all $n$
and the CDF of $\mathbb{T}_0 $ is increasing and continuous at its $1-\alpha$ quantile $c_{1-\alpha}$, then
	$\lim_{n\rightarrow\infty}\mathbb{P}(TS_n  >\hat{c}_{1-\alpha}) \le\alpha$.
	If, in addition, $P_n=P$ for all large $n$, then $\lim_{n\rightarrow\infty}\mathbb{P}(
	TS_n  >\hat{c}_{1-\alpha}) =\alpha$.

\item If the $H_0$ in \eqref{eq.null 1} is false with $Q=P$ and $P_n=P$  for all large $n$, then \\
$\lim_{n\rightarrow\infty}\mathbb{P}(
TS_n  >\hat{c}_{1-\alpha})  =1$.

\end{enumerate}
\end{theorem}
It is implied by Theorem 11.1 of  \citet{davydov1998local} that in (i) of Theorem \ref{thm.test multi}, the CDF of $\mathbb{T}_0$ is differentiable and has a positive derivative everywhere except at countably many points in its support, provided that  $\mathbb{T}_0\neq0$. If $\mathbb{T}_0=0$ at null configurations, our test statistic converges to zero in probability and so does the critical value. Theorem \ref{thm.test multi} does not show clearly how the rejection rate of the test
will behave asymptotically in this case. As discussed in \citet{Beare2017improved}, this is a common
theoretical limitation for irregular testing problems. Tests based on the machinery of \citet{fang2014inference}, and also those based on generalized
moment selection \citep{andrews2010inference,andrews2013inference}, may encounter this issue. One practical resolution is to replace the
bootstrap critical value $\hat{c}_{1-\alpha}$ with $\max\{\hat{c}_{1-\alpha},\eta\}$ or $\hat{c}_{1-\alpha}+\eta$,
where $\eta$ is some small positive constant. See, for instance, \citet[p.~13]{donald2016improving}. Simulation results in Table \ref{tab:Rej H0 multi (DGP2)} showed that the empirical rejection rates of our
test  with $\eta=0$ ($\tau_n=2$) are well controlled by the nominal significance level when $\mathbb{T}_0=0$ under null configurations.\footnote{In Section \ref{sec.simulation}, the value of $\tau_n$ is chosen to be $2$.} 

Theorem \ref{thm.test multi}(i) shows that the test is locally size controlled for every convergent distribution sequence satisfying the null. The convergent probability distributions $\{P_n\}$ depend on $n$, that is, the underlying distribution $P_n$ of the data can be different for every $n$. As $n\to\infty$, $P_n$ (satisfying the null) converges to $P$ under Assumption \ref{ass.probability path}. Theorem \ref{thm.weak convergence SK} provides the pointwise ($P$ is fixed) asymptotic distribution of the test statistic $TS_n$ for this convergent sequence of probability distributions $P_n\to P$. With this pointwise asymptotic distribution, we then obtain the local size control along such a probability distribution sequence:  $\lim_{n\to\infty}\mathbb{P}(TS_n  >\hat{c}_{1-\alpha})\le\alpha$.
When both $D$ and $Z$ are binary, \citet{kitagawa2015test} and the present paper consider testing the same null and alternative hypotheses. \citet{kitagawa2015test} obtains the uniform size control for their test under different conditions. That is, 
$\limsup_{n\to\infty}\sup_{Q\in\mathcal{P}_0}\mathbb{P}(TS_n  >\hat{c}^K_{1-\alpha})\le\alpha$, where $\mathcal{P}_0$ denotes the set of probability distributions in $\mathcal{P}$ that satisfy $H_0$, and the superscript ``$K$'' denotes the critical value of \citet{kitagawa2015test} (the test statistic of \citet{kitagawa2015test} is equivalent to that in the present paper when both $D$ and $Z$ are binary). Since Theorem \ref{thm.test multi}(i) assumes $P_n\in\mathcal{P}_0$, clearly we have that for every $P_n$,
\begin{align*}
    \mathbb{P}(TS_n  >\hat{c}^K_{1-\alpha})\le \sup_{Q\in\mathcal{P}_0}\mathbb{P}(TS_n  >\hat{c}^K_{1-\alpha}),
\end{align*}
which indicates that the uniform size control of \citet{kitagawa2015test} implies local size control. 
In general,  without additional assumptions, the local size control of the proposed test does not directly imply the uniform size control over the class of data generating processes in the null.

\subsection{Binary Treatment and Binary Instrument: Power Improvement}\label{sec.binary}

In this section, we consider the special case where the treatment $D$ and the instrument $Z$ are both binary.  
%\citet{kitagawa2015test} constructed a test for the instrument validity assumption based on testable implication \eqref{eq.testable implication} when $D$ and $Z$ are both binary.
We show how to achieve power improvement over the test of \citet{kitagawa2015test} based on the results of \citet{kitagawa2015test} and those from Section \ref{subsec.bootstrap}. As shown in Section \ref{sec.setup}, the null hypothesis for the testable implications consists of a set of inequalities. \citet{kitagawa2015test} used an upper bound on the asymptotic distribution of the test statistic under null to construct the bootstrap critical value. The upper bound is identical to the asymptotic distribution when all the inequalities in the null are binding. Therefore, their test could be conservative. The present paper establishes the asymptotic distribution of the test statistic under null. We then construct the critical value based on this asymptotic distribution, rather than on an upper bound, and therefore improve the power of the test.  
%compare the results from Section \ref{subsec.bootstrap} with those of \citet{kitagawa2015test}. 

Let $z_1=0$, $z_2=1$, $d_{1}=0$, and $d_{2}=1$. 
%All the results in Section \ref{subsec.bootstrap} hold in this case, and 
The test statistic in \eqref{eq.test stat expansion} is now numerically equal to the one constructed by \citet{kitagawa2015test} if we let $\nu$ be a Dirac measure. Recall that the instrument is allowed to be multivalued under the constructions in Section \ref{sec.multi D and Z}.\footnote{For the case where the treatment is binary and the instrument is multivalued, \citet{kitagawa2015test} constructed the test statistic by first computing the normalized differences of two empirical probability measures between neighboring pairs of values of instruments (ordered according to the propensity score), and then taking the maximum value of all these differences. Since these differences can be mutually correlated, it would not be straightforward to obtain the asymptotic distribution of their test statistic and approximate its null distribution by bootstrap.}

%The testing strategy in this paper is different from that of \citet{kitagawa2015test}. To make this point clear, we consider a simple case where $P_n=P$ for all $n$ and the $H_0$ in \eqref{eq.null 1} is true with $Q=P$.\footnote{Our test achieves size control under Assumption \ref{ass.probability path}  (the convergence of a ``local'' sequence of probability distributions), while the test of \cite{kitagawa2015test} achieves uniform size control under different conditions. Assuming a fixed $P$ makes the comparison more explicit.} We establish the asymptotic distribution in \eqref{eq.TS weak convergence} and use it to construct the critical value, while \citet{kitagawa2015test} used an upper bound of the asymptotic distribution to construct the critical value.  
We consider a simple case where $P_n=P$ for all $n$ and the $H_0$ in \eqref{eq.null 1} is true with $Q=P$.
As introduced in Section \ref{sec.setup}, we follow \citet{kitagawa2015test} and define probability measures
\begin{align*}
P_1\left(  B,C\right)     =\mathbb{P}\left(  Y\in B,D\in C|Z=1\right)  \text{ and } P_0\left(  B,C\right)    =\mathbb{P}\left(  Y\in B,D\in C|Z=0\right)
\end{align*}
for all $B,C\in\mathcal{B}_{\mathbb{R}}$. Now we define 
\begin{align*}
\mathcal{F}_b=\left\{(-1)^{d}\cdot 1_{B\times\{d\}}:B \text{ is a closed interval}, d\in\{0,1\}\right\},
\end{align*}
and write $P_d(f)=\int f \,\mathrm{d}P_d$ for all measurable $f$ and each $d\in\{0,1\}$. \citet{kitagawa2015test} showed that their critical value converged to the $1-\alpha$ quantile of the distribution\linebreak $\sup_{f\in\mathcal{F}_b}\mathbb{G}_{H}(f)/(\xi\vee\sigma_{H}(f))$, where $H=\lambda P_1+(1-\lambda)P_0$, $\lambda=\mathbb{P}(Z=1)$, $\mathbb{G}_H$ is an $H$-Brownian bridge, and $\sigma_H(f)$ is the standard deviation of $\mathbb{G}_H(f)$, that is, $\sigma_H^2(f)=H(f^2)-H^2(f)$. Let $\mathcal{F}_b^{\ast}=\left\{f\in\mathcal{ F}_b:P_0(f)=P_1(f)\right\}$. Then it is easy to show that $H(f)=P_0(f)=P_1(f)$ for all $f\in\mathcal{ F}_b^{\ast}$. Let $\nu$ be a Dirac measure centered at some $\xi$. It can be shown that 
\begin{align}\label{eq.upper bound of limiting distribution}
\sup_{f\in\mathcal{ F}_b}\frac{\mathbb{G}_H(f)}{\xi\vee\sigma_H(f)}\ge\sup_{f\in\mathcal{ F}_b^{\ast}}\frac{\mathbb{G}_H(f)}{\xi\vee\sigma_H(f)}\overset{L}{=}\mathbb{T},
\end{align}
where $\mathbb{T}$ is the asymptotic distribution of the test statistic in \eqref{eq.TS weak convergence} and ``$\overset{L}{=}$'' means equivalence in distribution. 

\citet{kitagawa2015test} constructed a pooled-data bootstrap approximation for the Gaussian process ${\mathbb{G}_H}/({\xi\vee\sigma_H})$, denoted by ${\mathbb{G}_H^{B}}/({\xi\vee\sigma_H^B})$, and then computed the bootstrap test statistic by $\sup_{f\in\mathcal{ F}_b}{\mathbb{G}_H^B(f)}/({\xi\vee\sigma_H^B(f)})$.
This bootstrap statistic approximates the distribution of $\sup_{f\in\mathcal{ F}_b} {\mathbb{G}_H(f)}/({\xi\vee\sigma_H(f)})$. For the case where $D$ and $Z$ are both binary, we suggest modifying the test in Section \ref{subsec.bootstrap} to achieve power improvement over the test of \citet{kitagawa2015test}.\footnote{The modification may also be applied to the case where $D$ is multivalued and $Z$ is binary.} Specifically, we first estimate $\mathcal{F}_b^{\ast}$ by a subset of $\mathcal{F}_b$, denoted by $\widehat{\mathcal{F}_b^{\ast}}$, in a way similar to \eqref{eq.Psi_hat}. Then we follow the bootstrap approach of \citet{kitagawa2015test} to construct ${\mathbb{G}_H^B}$ and ${\sigma_H^B}$, and construct the bootstrap test statistic by $\sup_{f\in\widehat{\mathcal{ F}_b^{\ast}}}{\mathbb{G}_H^B(f)}/({\xi\vee\sigma_H^B(f)})$. 
%In this way, the bootstrap statistic is numerically smaller than that of \citet{kitagawa2015test}, and  it approximates the distribution of $\mathbb{T}$ (equivalently, $\sup_{f\in\mathcal{ F}_b^{\ast}}{\mathbb{G}_H(f)}/({\xi\vee\sigma_H(f)})$) under $H_0$.
Clearly, the proposed critical value is always smaller than that of \citet{kitagawa2015test} because $\widehat{\mathcal{F}_b^{\ast}}\subset\mathcal{F}_b$. It can also be shown that our critical value converges to the $1-\alpha$ quantile of  $\sup_{f\in\mathcal{ F}_b^{\ast}}{\mathbb{G}_H(f)}/({\xi\vee\sigma_H(f)})$ (equivalently, $\mathbb{T}$) under $H_0$.
Since the test statistic in \eqref{eq.test stat expansion} is numerically equivalent to that of \citet{kitagawa2015test}, this shows that the power of the test can be improved by the use of our approach. This improvement is against all alternatives according to the construction of the critical value. See the simulation evidence in Appendix \ref{sec.comparison binary}.

The test of \citet{mourifie2016testing} for the inequalities in \eqref{eq.testable implication} employed the intersection bounds framework of
\citet{chernozhukov2013intersection}. As shown in Proposition 1 of \citet{mourifie2016testing},\footnote{See also Theorem 6 of  \citet{chernozhukov2013intersection}.} the limiting rejection rate under null is equal to the nominal significance level $\alpha$ only when all the inequalities in the null are binding, and below the nominal
level elsewhere in the null. This result is similar to that of \citet{kitagawa2015test}, because only when all the inequalities are binding, the (contact) set $\mathcal{F}_b^{\ast}$ is equal to $\mathcal{F}_b$. If we are at a point in the null where the inequalities are not all binding, then the tests of \citet{kitagawa2015test} and \citet{mourifie2016testing} would have limiting rejection
rates below the nominal level, and thus lack power against nearby points
in the alternative. Theorem \ref{thm.test multi} in the present paper shows that the proposed test can achieve the nominal level over a larger region in the null, where the inequalities in the testable implication could not all be binding, thereby improving power.

%\section{Extensions}

\subsection{Unordered Treatment}\label{sec.unordered}
With testable implication \eqref{eq.testable implication unordered treatment}, we
define the function space
\begin{align}\label{eq.function spaces unordered}
&\mathcal{H}\times \mathcal{G}=\left\{\left(1_{B\times\{d\}\times\mathbb{R}},\left( 1_{\mathbb{R}\times \mathbb{R}\times\left\{  z\right\}
},1_{\mathbb{R}\times \mathbb{R} \times\left\{  z^{\prime}\right\}  }\right)
\right):B\text{ is a closed interval}, \left(  d,z,z^{\prime}\right)  \in
\mathcal{C}\right\}.
\end{align}
For every probability measure $Q$ with \eqref{eq.Q map}, we define $\phi_Q$ by
$
\phi_{Q}\left(  h,g\right)  ={Q\left(  h\cdot g_2 \right) }/{Q\left(  g_2 \right)  }-{Q\left(  h\cdot g_1\right)  }/{Q\left(  g_1 \right)  }
$
for every $(h,g)\in\mathcal{H}\times\mathcal{G}$ with $g=(g_1,g_2)$. Testable implication \eqref{eq.testable implication unordered treatment} is equivalent to the $H_0$ in
\[
H_0: \sup_{(h,g)\in\mathcal{H}\times\mathcal{G}}\phi_{Q}\left( h,g\right) \le 0 \text{ and } H_1: \sup_{(h,g)\in\mathcal{H}\times\mathcal{G}}\phi_{Q}\left(  h,g\right) > 0
\]
if $Q$ is the underlying probability distribution of the data. Then we can follow the test procedure in Section \ref{subsbusec.test procedure} to conduct the test with the function space $\mathcal{H}\times\mathcal{G}$ defined in \eqref{eq.function spaces unordered}.  

%\subsection{Conditioning Covariates}\label{sec.conditioning covariates}

\section{Simulation Evidence}\label{sec.simulation}

We first designed Monte Carlo simulations for the case where $D$ and $Z$ are both
multivalued random variables such that $D\in\{0,1,2\}$ and $Z\in\{0,1,2\}$. Additional Monte Carlo studies can be found in Appendix \ref{sec.appendix Monte Carlo Comparison}. Each
simulation consisted of $1000$ Monte Carlo iterations and $1000$ bootstrap iterations. To expedite the simulation, we employed the warp-speed
method of \citet{giacomini2013warp}. The nominal significance level $\alpha$ was set to $0.05$. As shown in \eqref{eq.sigma_square bound} and \eqref{eq.sigma_square_hat bound}, $\sigma_P^2$ and $\hat{\sigma}_{P_n}^2$ are bounded by $(1/2)\cdot (K-1)^{-(K-1)}$, where $K=3$ in our setting. 
The simulations constructed in this section are similar to those in \citet{kitagawa2015test}. 
In each simulation, the measure $\nu$ was set to be a Dirac measure $\delta_{\xi}$ centered at one of the following values of $\xi$: $0.07$, $0.1$, $0.13$, $0.16$, $0.19$, $0.22$, $0.25$, $0.28$, $0.3$, and $1$, or to be a probability measure $\bar{\nu}_{\xi}$ that assigns equal probabilities (weights) to the values of $\xi$ listed above. Four values of $\xi$ were used in the simulations of \citet{kitagawa2015test}: $0.07$, $0.22$, $0.3$, and $1$, where $0.07\approx\sqrt{0.005(1-0.005)}$, $0.22\approx\sqrt{0.05(1-0.05)}$, and $0.3=\sqrt{0.1(1-0.1)}$. As shown in \eqref{eq.estimated stat variance multi}, for every $(h,g)\in\bar{\mathcal{H}}\times\mathcal{G}$ with $g=(g_1,g_2)$,
\begin{align*}
\hat{\sigma}_{P_n}^{2}\left(  h,g\right)  =&\,\frac{T_n}{n}\cdot\left\{\frac{ \hat{P}_n\left(  h^2\cdot g_{2}\right)   }{\hat{P}_n^{2}\left(
	g_{2}\right)  }  -\frac{ \hat{P}_n^2\left(  h\cdot g_{2}\right)   }{\hat{P}_n^{3}\left(
	g_{2}\right)  }  
+\frac{ \hat{P}_n\left(	h^2\cdot g_{1}\right)   }{\hat{P}_n^{2}\left(  g_{1}\right)  }
-\frac{ \hat{P}_n^2\left(	h\cdot g_{1}\right)   }{\hat{P}_n^{3}\left(  g_{1}\right)  }\right\}\\
=&\,\frac{\prod_{k=1}^{K}\hat{P}_n\left(1_{\mathbb{R}\times\mathbb{R}\times\{z_k\}}\right)}{\hat{P}_n\left(
	g_{2}\right)}\frac{ \hat{P}_n\left(  h^2\cdot g_{2}\right)   }{\hat{P}_n\left(
	g_{2}\right)  }\left\{1-\frac{ \hat{P}_n\left(  h^2\cdot g_{2}\right)   }{\hat{P}_n\left(
	g_{2}\right)  }\right\}\\
	&+\frac{\prod_{k=1}^{K}\hat{P}_n\left(1_{\mathbb{R}\times\mathbb{R}\times\{z_k\}}\right)}{\hat{P}_n\left(
	g_{1}\right)}\frac{ \hat{P}_n\left(  h^2\cdot g_{1}\right)   }{\hat{P}_n\left(
	g_{1}\right)  }\left\{1-\frac{ \hat{P}_n\left(  h^2\cdot g_{1}\right)   }{\hat{P}_n\left(
	g_{1}\right)  }\right\}.
\end{align*}
The values of $\xi\in\{0.07,0.22,0.3\}$ take the form of $\sqrt{\pi(1-\pi)}$ where $\pi\in\{0.005,0.05,0.1\}$. As discussed in \citet{kitagawa2015test}, $\pi$ can be interpreted as that if both
${ \hat{P}_n\left(  h^2\cdot g_{1}\right)   }/{\hat{P}_n\left(
	g_{1}\right)  }$ and ${ \hat{P}_n\left(  h^2\cdot g_{2}\right)   }/{\hat{P}_n\left(
	g_{2}\right)  }$ are less than $\pi$, then the weight becomes the inverse of $\xi$ instead of the inverse of the estimated standard deviation. As $\pi$ gets larger, less weight is put on $\hat{\phi}_{P_n}$ for smaller probability events, and vice versa. In the following simulations, we chose the values of $\xi$ following the choice of \citet{kitagawa2015test}.
In empirical practice, application-based simulations can be applied to choose $\Xi$ and $\nu$, which is illustrated in Section \ref{sec.choice of xi}.

When calculating the
supremum in the test statistic $TS_n$ in \eqref{eq.test stat expansion}, we followed
the numerical computation approach used by \citet{kitagawa2015test}. Specifically,
we calculated the supremum using only the closed intervals $B$ with the values of $\{Y_i\}_{i=1}^{n}$
observed in the data as the endpoints, that is, $B=[a,b]$ with $a,b\in\{Y_1,Y_2,\ldots,Y_n\}$ and $a\le b$. It is not hard to show that the test statistic calculated in this way is equal to that in \eqref{eq.test stat expansion}. We also used such closed intervals to calculate the bootstrap test statistic $TS^B_n$ in \eqref{eq.bootstrap test stat}. From all such intervals, we found those that satisfy the inequality in \eqref{eq.Psi_hat} and used them to calculate the supremum of $ {\sqrt{T_n^B}(  \hat{\phi}_{P_n}^{B}-\hat{\phi}_{P_n})  }/\max\{\xi,{\hat{\sigma}_{P_n}^{B}}\}$ for each $\xi$ listed above. 

\subsection{Size Control and Tuning Parameter Selection}\label{sec.size}
The first set of simulations was designed to investigate the size of the test and the selection of the tuning parameter. As shown in \eqref{eq.Psi_hat}, the estimate $\widehat{\Psi_{{\mathcal{H}}\times\mathcal{G}}}$ involves a tuning parameter $\tau_n$ with $\tau_n\to\infty$ and $\tau_n/\sqrt{n}\to0$ as $n\to\infty$. In practice, we need to use a particular value of $\tau_n$ for each sample size $n$. For this set of simulations, we set $n$ to $3000$ and $\tau_n$ to $0.1$, $0.5$, $1$, $2$, $3$, $4$, and $\infty$. For $\tau_n=\infty$, $\widehat{\Psi 		_{{\mathcal{H}}\times\mathcal{G}}}={{\mathcal{H}}	\times\mathcal{G}}$ and the test is conservative. We compared the rejection rates obtained using each of these values of $\tau_n$ and decided which value would be a good option for sample sizes close to $3000$. We let $U\sim\mathrm{Unif}(0,1)$, $V\sim\mathrm{Unif}(0,1)$, $N_0\sim \mathrm{N}(0,1)$, $N_1\sim \mathrm{N}(1,1)$, $N_2\sim \mathrm{N}(2,1)$, $Z=2\times 1\{U \le 0.5\}+1\{0.5<U \le 0.7\}$  ($\mathbb{P}(Z=2)=0.5$), $D_z=2\times 1\{V\le 0.33\}+ 1\{0.33  <V\le 0.66\}$ for $z=0,1,2$,  $D=\sum_{z=0}^2 1\{Z=z\}\times D_z$, and $Y=\sum_{d=0}^{2}1\{D=d\}\times N_d$. All the variables $U$, $V$, $N_0$, $N_1$, and $N_2$ are mutually independent. Clearly, Assumption \ref{ass.IV validity for multivalued Z} holds in this case with $z_1=0$, $z_2=1$, and $z_3=2$.

Table \ref{tab:Rej H0 multi} shows the results of the simulations. The rejection rates were influenced by the values of $\tau_n$ and $\xi$. For each measure $\nu$, a smaller $\tau_n$ yields greater rejection rates, because a smaller $\tau_n$ leads to a smaller critical value according to \eqref{eq.Psi_hat}. 
For $\tau_n=2$, all the rejection rates were close to those for $\tau_n=\infty$ (the conservative case).  Similar to the pattern of the results shown in \citet{kitagawa2015test}, some rejection rates for $\tau_n=2$ with $\delta_{\xi}$ centered at particular values of $\xi$ were slightly upwardly biased compared to the nominal size. Overall, however, the results showed good performance of the test in terms of size control. When sample sizes are less than or close to $3000$, we suggest using $\tau_n=2$ in practice to achieve good size control without a significant power loss.  When the sample size increases, $\tau_n$ should be increased accordingly. It is also worth noting that when we used the measure $\bar{\nu}_{\xi}$, the rejection rates could be well controlled by the nominal significance level. Thus if we have no additional information about the choice of $\xi$,  $\bar{\nu}_{\xi}$ can be a default choice for us. 

\input{Rejection_Rates_under_H0_Multi.tex} 

\subsection{Rejection Rates against Fixed Alternatives}
The second set of simulations was designed to investigate the power of the test. Six data generating processes (DGPs) in total were considered, and Assumption \ref{ass.IV validity for multivalued Z} did not hold with $z_1=0$, $z_2=1$, and $z_3=2$. Sample sizes were set to $n=200$, $600$, $1000$, $1100$, and $2000$. The probability $\mathbb{P}(Z=2)=r_n$, with $r_n=1/2$, $1/6$, $1/2$, $1/11$, and $1/2$ for the corresponding sample sizes. We set $\tau_n$ to $2$, as suggested in the preceding set of simulations. DGPs (1)--(4) are the cases where \eqref{eq.testable implication multivalue} was violated and \eqref{eq.fosd multi} was not violated, and DGPs (5) and (6) are the cases where both \eqref{eq.testable implication multivalue} and \eqref{eq.fosd multi} were violated. We let $U\sim\mathrm{Unif}(0,1)$, $V\sim\mathrm{Unif}(0,1)$, $W\sim\mathrm{Unif}(0,1)$, and $Z=2\times 1\{U \le r_n\}+1\{r_n<U \le r_n+0.2\}$.

For DGPs (1)--(4), we let $D_z=2\times 1\{V\le 0.45\}+ 1\{0.45 <V\le 0.55\}$ for $z=0,1,2$,   $D=\sum_{z=0}^2 1\{Z=z\}\times D_z$, $N_{00}\sim \mathrm{N}(0,1)$, $N_{10}\sim \mathrm{N}(0,1)$, and $N_{dz}\sim \mathrm{N}(0,1)$ for $d=0,1,2$ and $z=1,2$.
\begin{enumerate}[label=(\arabic*):]
	
	\item  $N_{20}\sim \mathrm{N}(-0.7,1)$ and $Y=\sum_{z=0}^{2}1\{Z=z\}\times(\sum_{d=0}^{2} 1\{D=d\}\times N_{dz})$.
	
	\item $N_{20}\sim \mathrm{N}(0,1.675^2)$ and $Y=\sum_{z=0}^{2}1\{Z=z\}\times(\sum_{d=0}^{2} 1\{D=d\}\times N_{dz})$.
	
	\item $N_{20}\sim \mathrm{N}(0,0.515^2)$ and $Y=\sum_{z=0}^{2}1\{Z=z\}\times(\sum_{d=0}^{2} 1\{D=d\}\times N_{dz})$.
	
	\item $N_{20a}\sim \mathrm{N}(-1,0.125^2)$, $N_{20b}\sim \mathrm{N}(-0.5,0.125^2)$, $N_{20c}\sim \mathrm{N}(0,0.125^2)$,\\ $N_{20d}\sim \mathrm{N}(0.5,0.125^2)$, $N_{20e}\sim \mathrm{N}(1,0.125^2)$, $N_{20}=1\{W\le 0.15\}\times N_{20a}+1\{0.15<W\le 0.35\}\times N_{20b}+1\{0.35<W\le 0.65\}\times N_{20c}+1\{0.65<W\le 0.85\}\times N_{20d}+1\{W>0.85\}\times N_{20e}$, and $Y=\sum_{z=0}^{2}1\{Z=z\}\times(\sum_{d=0}^{2} 1\{D=d\}\times N_{dz})$.
	
\end{enumerate} 
For DGPs (5) and (6), we let $N_0\sim \mathrm{N}(0,1)$, $N_1\sim \mathrm{N}(1,1)$, and $N_2\sim \mathrm{N}(2,1)$.
\begin{enumerate}[resume,label=(\arabic*):]
	
	\item $D_0=2\times 1\{V\le 0.6\}+ 1\{0.6 <V\le 0.8\}$, $D_1=2\times 1\{V\le 0.33\}+ 1\{0.33 <V\le 0.66\}$, $D_2=D_1$, $D=\sum_{z=0}^2 1\{Z=z\}\times D_z$, and $Y=\sum_{d=0}^{2} 1\{D=d\}\times N_{d}$.
    
    \item $D_0=2\times 1\{V\le 0.33\}+ 1\{0.33 <V\le 0.66\}$, $D_1=2\times 1\{V\le 0.6\}+ 1\{0.6 <V\le 0.8\}$, $D_2=D_0$, $D=\sum_{z=0}^2 1\{Z=z\}\times D_z$, and $Y=\sum_{d=0}^{2} 1\{D=d\}\times N_{d}$.
    
\end{enumerate}

All the variables $U$, $V$, $N_{00}$, $N_{10}$, $N_{20}$, $N_{01}$, $N_{11}$, $N_{21}$, $N_{02}$, $N_{12}$, $N_{22}$, $N_0$, $N_1$, and $N_2$ were set to be mutually independent. We briefly explain how DGPs (1)--(4) violate \eqref{eq.testable implication multivalue}, which is shown graphically in Figure \ref{fig:H1 true multi}. We let $p_z(y,d)$ be the derivative of $\mathbb{P}(Y\in(-\infty,y],D=d|Z=z)$ with respect to $y$ for all $d,z\in\{0,1,2\}$. Similar to Figure  \ref{fig:nest}, if \eqref{eq.testable implication multivalue} were true, then we would have $p_0(y,2) \le p_1(y,2) \le p_2(y,2)$ everywhere. For DGPs (1)--(4),  $p_1(y,2)=p_2(y,2)$ held for all $y$, but $p_0(y,2) \le p_1(y,2)$ did not hold on some range of $\mathbb{R}$. DGPs (5) and (6) are the cases where the monotonicity assumption did not hold and both \eqref{eq.testable implication multivalue} and \eqref{eq.fosd multi} were violated.

\begin{figure}[h]
	\centering
	\caption{Curves of {{$p_0\left(y,2\right)$}} (dashed) and {{$p_1\left(y,2\right)$}} (solid) for DGPs (1)--(4) }
	\begin{subfigure}[b]{0.235\textwidth}
		\centering
		\input{GraphH1DGP1.tex}
		\subcaption{DGP (1)}
	\end{subfigure}
	\begin{subfigure}[b]{0.235\textwidth}
		\centering
		\input{GraphH1DGP2.tex}
		\subcaption{DGP (2)}
	\end{subfigure}
	\begin{subfigure}[b]{0.235\textwidth}
		\centering
		\input{GraphH1DGP3.tex}
		\subcaption{DGP (3)}
	\end{subfigure}
	\begin{subfigure}[b]{0.235\textwidth}
		\centering
		\input{GraphH1DGP4.tex}
		\subcaption{DGP (4)}
	\end{subfigure}
	\label{fig:H1 true multi}
\end{figure}
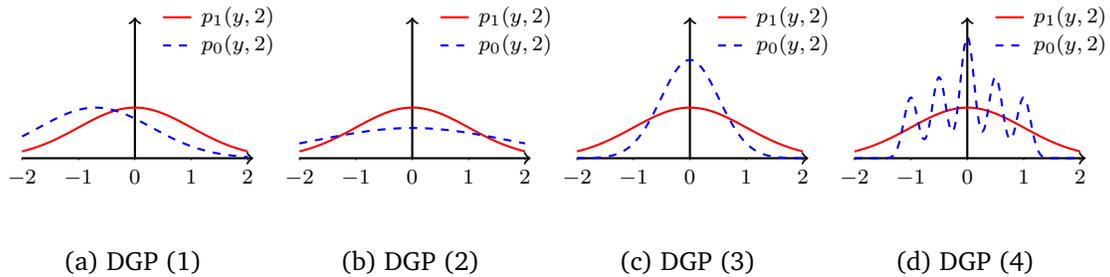

Table \ref{tab:Rej H1 multi} shows the rejection rates under DGPs (1)--(6), that is, the power of the test. For each DGP and each measure $\nu$, the rejection rate increased as the sample size $n$ was increased. The results for $\nu=\bar{\nu}_{\xi}$ showed that if we have no information about the choice of $\xi$, using the weighted average of the statistics over $\xi$ is a desirable option. When $n>200$, the rejection rates for using $\nu=\bar{\nu}_{\xi}$ were at a relatively high level compared to the results for using a Dirac measure.  

\input{Rejection_Rates_under_H1_Multi.tex}

\section{Empirical Application}\label{sec.empirical}
We revisit one empirical example discussed by \citet{kitagawa2015test} to show the performance of the proposed test in practice. The example is from \citet{NBERw4483}, who used college proximity as an instrument of years of schooling to study the causal link between education and earnings. The data are from the Young Men Cohort of the National Longitudinal Survey. In the original study of \citet{NBERw4483}, the educational level $D$ is a multivalued treatment variable, while \citet{kitagawa2015test} treated it as a binary treatment variable $T$ with $T=1\{D\ge 16\}$. The results of the test of \citet{kitagawa2015test} showed that the instrument was not valid when no covariates were controlled.  

We use the originally defined treatment variable $D$ to reconduct the test. 
Specifically, the treatment $D$ is education attainment observed in 1976 (the variable ``ed76''), the instrument $Z$ is whether an individual grew up near a 4-year college (the variable ``nearc4''), and the outcome is log wage observed in 1976 (the variable ``lwage76'') in the data set. The available sample size is 3010. We follow the setup in Section \ref{sec.multi D and Z} with $\mathcal{D}=\{1,2,\ldots,18\}$ and $\mathcal{Z}=\{0,1\}$. The instrument $Z=1$ implies that an individual grew up near a 4-year college. 
{Table} \ref{tab:ApplicationCard} shows the $p$-values obtained from our test using each measure $\nu$. From these results, we conclude that we do not reject the validity of instrument $Z$. In Section \ref{sec.choice of xi}, we show more results by using application-based simulations to choose $\Xi$ and $\nu$. The results are similar to those in Table  \ref{tab:ApplicationCard}.

\input{ApplicationResultsCard.tex} 

The testable implication used by \citet{kitagawa2015test} for binary $T$ is that
\begin{align}\label{eq.testable implication application}
&\mathbb{P}\left(  Y\in B,T=0|Z=1\right)-\mathbb{P}\left(  Y\in B,T=0|Z=0\right)  \leq 0  \notag\\
&\text{ and } \mathbb{P}\left(  Y\in B,T=1|Z=1\right)-\mathbb{P}\left(  Y\in B,T=1|Z=0\right)  \ge 0
\end{align}
for all closed intervals $B$.
The inequalities in \eqref{eq.testable implication application} are equivalent to the following for all closed intervals $B$:
\begin{align}\label{eq.testable implication application 2}
&\mathbb{P}\left(  Y\in B,D<16|Z=1\right)-\mathbb{P}\left(  Y\in B,D<16|Z=0\right)  \leq 0  \notag\\
&\text{ and } \mathbb{P}\left(  Y\in B,D\ge 16|Z=1\right)-\mathbb{P}\left(  Y\in B,D\ge 16|Z=0\right)  \ge 0,
\end{align}
which are different from those in the testable implication given in \eqref{eq.testable implication multivalue} and \eqref{eq.fosd multi} and are not implied by Assumption \ref{ass.IV validity for multivalued Z}. Thus a valid instrument $Z$ for multivalued $D$ which satisfies the testable implication given in \eqref{eq.testable implication multivalue} and \eqref{eq.fosd multi} may not satisfy the inequalities in \eqref{eq.testable implication application}, that is, $Z$ may not remain valid for binary (or coarsened) $T$. This provides a possible explanation for why we accept $Z$ but \citet{kitagawa2015test} rejected it.

\section{Conclusion}
In this paper, we provided a general framework for testing instrument validity in heterogeneous causal effect models. 
We generalized the testable implications of the instrument validity assumptions in the literature, and based on them we  proposed a nonparametric bootstrap test. An extended continuous mapping theorem and an extended delta method were provided to establish the asymptotic distribution of the test statistic, which may be of independent interest.  The proposed test can be applied in more general settings and may achieve power improvement. 

\input{IVValidityAppendixInText}

\section*{Acknowledgments}
This article is a revised version of the first
chapter of the author's doctoral thesis at UC San Diego.
I am deeply grateful to Brendan K. Beare, Zheng Fang, Andres Santos, Yixiao Sun, and Kaspar W\"{u}thrich for their constant support on this paper. I thank the editors and three anonymous referees for their constructive suggestions that help improve the paper significantly. I thank Shengtao Dai, Tongyu Li, and Xingyu Li for their excellent work as research assistants. I also thank Roy Allen, Qihui Chen, Asad Dossani, Graham Elliott, Ivan Fernandez-Val, Wenzheng Gao, James D. Hamilton, Jungbin Hwang, Toru Kitagawa, Sungwon Lee, Juwon Seo, Xiaojun Song, and all seminar participants for their insightful suggestions and comments.  
This work was supported by the National Natural Science Foundation of China [grant number 72103004].

\bibliographystyle{apalike}
\bibliography{reference1}

%% file: P1nestP0.tex
\begin{tikzpicture}[scale=1]
		\tikzstyle{vertex}=[font=\small,circle,draw,fill=yellow!20]
		\tikzstyle{edge} = [font=\scriptsize,draw,thick,-]
		\draw[black, thick, ->] (0,0) -- (0,2.5);
		\draw[black, thick, ->] (-2,0) -- (2.1,0);
		
	\draw[red, thick, -] (0.6,2.5) -- (1.1,2.5);
	\node [right,font=\scriptsize] at (1.1,2.5) {$p_1(y,1)$};
	\draw[blue, thick,dashed, -] (0.6,2) -- (1.1,2);
	\node [right,font=\scriptsize] at (1.1,2) {$p_0(y,1)$};
		
		\draw (0,0.3pt) -- (0,-1pt)
		node[anchor=north,font=\scriptsize] {$0$};
		\draw (-2,0.3pt) -- (-2,-1pt)
		node[anchor=north,font=\scriptsize] {$-2$};
		\draw (-1,0.3pt) -- (-1,-1pt)
		node[anchor=north,font=\scriptsize] {$-1$};
		\draw (1,0.3pt) -- (1,-1pt)
		node[anchor=north,font=\scriptsize] {$1$};
		\draw (2,0.3pt) -- (2,-1pt)
		node[anchor=north,font=\scriptsize] {$2$};
		%		\node[right,font=\footnotesize] at (1.1,0) {$t_n$};
%		\draw (0.3pt,0) -- (-0.8pt,0);
%		\draw (0.3pt,{1/5}) -- (-0.8pt,{1/5});
%		\draw (0.3pt,{2/5}) -- (-0.8pt,{2/5});
%		\draw (0.3pt,{3/5}) -- (-0.8pt,{3/5});
%		\draw (0.3pt,{4/5}) -- (-0.8pt,{4/5});
%		\draw (0.3pt,1) -- (-0.8pt,1);
%		\node[left,font=\scriptsize] at (0,0) {$0$};
%		\node[left,font=\scriptsize] at (0,{1/5}) {$.2$};
%		\node[left,font=\scriptsize] at (0,{2/5}) {$.4$};
%		\node[left,font=\scriptsize] at (0,{3/5}) {$.6$};
%		\node[left,font=\scriptsize] at (0,{4/5}) {$.8$};
%		\node[left,font=\scriptsize] at (0,1) {$1$};	
		\draw[red,thick] plot[smooth] file {\hnulldgptwopone};	
		\draw[blue,dashed,thick] plot[smooth] file {\hnulldgptwoqone};		
	%	\node[font=\small] at (0.5,-0.4) {(f) $\tau=0.1$};
%		\node[right] at (1.1,0) {$\gamma$};
\end{tikzpicture}

%% file: P0nestP1.tex
\begin{tikzpicture}[scale=1]
		\tikzstyle{vertex}=[font=\small,circle,draw,fill=yellow!20]
		\tikzstyle{edge} = [font=\scriptsize,draw,thick,-]
		\draw[black, thick, ->] (0,0) -- (0,2.5);
		\draw[black, thick, ->] (-2,0) -- (2.1,0);
		
		\draw[red, thick, -] (0.6,2.5) -- (1.1,2.5);
	    \node [right,font=\scriptsize] at (1.1,2.5) {$p_1(y,0)$};
	    \draw[blue, thick,dashed, -] (0.6,2) -- (1.1,2);
	    \node [right,font=\scriptsize] at (1.1,2) {$p_0(y,0)$};
		
		\draw (0,0.3pt) -- (0,-1pt)
		node[anchor=north,font=\scriptsize] {$0$};
		\draw (-2,0.3pt) -- (-2,-1pt)
		node[anchor=north,font=\scriptsize] {$-2$};
		\draw (-1,0.3pt) -- (-1,-1pt)
		node[anchor=north,font=\scriptsize] {$-1$};
		\draw (1,0.3pt) -- (1,-1pt)
		node[anchor=north,font=\scriptsize] {$1$};
		\draw (2,0.3pt) -- (2,-1pt)
		node[anchor=north,font=\scriptsize] {$2$};
		%		\node[right,font=\footnotesize] at (1.1,0) {$t_n$};
%		\draw (0.3pt,0) -- (-0.8pt,0);
%		\draw (0.3pt,{1/5}) -- (-0.8pt,{1/5});
%		\draw (0.3pt,{2/5}) -- (-0.8pt,{2/5});
%		\draw (0.3pt,{3/5}) -- (-0.8pt,{3/5});
%		\draw (0.3pt,{4/5}) -- (-0.8pt,{4/5});
%		\draw (0.3pt,1) -- (-0.8pt,1);
%		\node[left,font=\scriptsize] at (0,0) {$0$};
%		\node[left,font=\scriptsize] at (0,{1/5}) {$.2$};
%		\node[left,font=\scriptsize] at (0,{2/5}) {$.4$};
%		\node[left,font=\scriptsize] at (0,{3/5}) {$.6$};
%		\node[left,font=\scriptsize] at (0,{4/5}) {$.8$};
%		\node[left,font=\scriptsize] at (0,1) {$1$};	
		\draw[red,thick] plot[smooth] file {\hnulldgptwopzero};	
		\draw[blue,dashed,thick] plot[smooth] file {\hnulldgptwoqzero};		
	%	\node[font=\small] at (0.5,-0.4) {(f) $\tau=0.1$};
%		\node[right] at (1.1,0) {$\gamma$};
\end{tikzpicture}

%% file: Rejection_Rates_under_H0_Multi.tex
\begin{table}[h]
	
	\centering
	\caption{Rejection Rates under $H_{0}$ for Multivalued $D$ and Multivalued $Z$}
	\label{tab:Rej H0 multi}
	\scalebox{0.9}{
		\begin{tabular}{   c  c  c  c  c  c  c  c  c  c  c  c  }
			\hline
			\hline
			\multirow{2}{*}{$\tau_n$} & \multicolumn{10}{c}{$\xi$ for $\delta_{\xi}$} &\multirow{2}{*}{$\bar{\nu}_{\xi}$} \\
			\cline{2-11}
		  & 0.07 & 0.1 & 0.13 & 0.16 & 0.19 & 0.22 & 0.25 & 0.28 & 0.3 & 1 & \\
			\hline

		   $0.1$    &0.122&0.108&0.096&0.096&0.108&0.092&0.092&0.092&0.092&0.092&0.108\\
		   $0.5$    &0.092&0.070&0.068&0.074&0.064&0.069&0.069&0.069&0.069&0.069&0.075\\
		   $1$      &0.079&0.060&0.047&0.068&0.056&0.058&0.061&0.061&0.061&0.061&0.054\\
		   $2$      &0.073&0.050&0.037&0.050&0.050&0.055&0.048&0.048&0.048&0.048&0.047\\
		   $3$      &0.073&0.048&0.037&0.050&0.050&0.049&0.048&0.048&0.048&0.048&0.047\\
		   $4$      &0.073&0.048&0.037&0.050&0.050&0.049&0.048&0.048&0.048&0.048&0.047\\
		   $\infty$ &0.073&0.048&0.037&0.050&0.050&0.049&0.048&0.048&0.048&0.048&0.047\\

			\hline
		\end{tabular}
	}

\end{table}

%% file: GraphH1DGP1.tex
\begin{tikzpicture}[scale=0.75]
		\tikzstyle{vertex}=[font=\small,circle,draw,fill=yellow!20]
		\tikzstyle{edge} = [font=\scriptsize,draw,thick,-]
		\draw[black, thick, ->] (0,0) -- (0,2.5);
		\draw[black, thick, ->] (-2,0) -- (2.1,0);
		
		\draw[red, thick, -] (0.5,2.5) -- (1,2.5);
        \node [right,font=\scriptsize] at (1,2.5) {$p_1(y,2)$};
        \draw[blue, thick,dashed, -] (0.5,2) -- (1,2);
        \node [right,font=\scriptsize] at (1,2) {$p_0(y,2)$};
		
		\draw (0,0.3pt) -- (0,-1pt)
		node[anchor=north,font=\scriptsize] {$0$};
		\draw (-2,0.3pt) -- (-2,-1pt)
		node[anchor=north,font=\scriptsize] {$-2$};
		\draw (-1,0.3pt) -- (-1,-1pt)
		node[anchor=north,font=\scriptsize] {$-1$};
		\draw (1,0.3pt) -- (1,-1pt)
		node[anchor=north,font=\scriptsize] {$1$};
		\draw (2,0.3pt) -- (2,-1pt)
		node[anchor=north,font=\scriptsize] {$2$};
		%		\node[right,font=\footnotesize] at (1.1,0) {$t_n$};
%		\draw (0.3pt,0) -- (-0.8pt,0);
%		\draw (0.3pt,{1/5}) -- (-0.8pt,{1/5});
%		\draw (0.3pt,{2/5}) -- (-0.8pt,{2/5});
%		\draw (0.3pt,{3/5}) -- (-0.8pt,{3/5});
%		\draw (0.3pt,{4/5}) -- (-0.8pt,{4/5});
%		\draw (0.3pt,1) -- (-0.8pt,1);
%		\node[left,font=\scriptsize] at (0,0) {$0$};
%		\node[left,font=\scriptsize] at (0,{1/5}) {$.2$};
%		\node[left,font=\scriptsize] at (0,{2/5}) {$.4$};
%		\node[left,font=\scriptsize] at (0,{3/5}) {$.6$};
%		\node[left,font=\scriptsize] at (0,{4/5}) {$.8$};
%		\node[left,font=\scriptsize] at (0,1) {$1$};	
		\draw[red,thick] plot[smooth] file {\honedgponepone};	
		\draw[blue,dashed,thick] plot[smooth] file {\honedgponepzero};		
	%	\node[font=\small] at (0.5,-0.4) {(f) $\tau=0.1$};
%		\node[right] at (1.1,0) {$\gamma$};
\end{tikzpicture}

%% file: GraphH1DGP2.tex
\begin{tikzpicture}[scale=0.75]
		\tikzstyle{vertex}=[font=\small,circle,draw,fill=yellow!20]
		\tikzstyle{edge} = [font=\scriptsize,draw,thick,-]
		\draw[black, thick, ->] (0,0) -- (0,2.5);
		\draw[black, thick, ->] (-2,0) -- (2.1,0);
		
		\draw[red, thick, -] (0.5,2.5) -- (1,2.5);
        \node [right,font=\scriptsize] at (1,2.5) {$p_1(y,2)$};
        \draw[blue, thick,dashed, -] (0.5,2) -- (1,2);
        \node [right,font=\scriptsize] at (1,2) {$p_0(y,2)$};
		
		\draw (0,0.3pt) -- (0,-1pt)
		node[anchor=north,font=\scriptsize] {$0$};
		\draw (-2,0.3pt) -- (-2,-1pt)
		node[anchor=north,font=\scriptsize] {$-2$};
		\draw (-1,0.3pt) -- (-1,-1pt)
		node[anchor=north,font=\scriptsize] {$-1$};
		\draw (1,0.3pt) -- (1,-1pt)
		node[anchor=north,font=\scriptsize] {$1$};
		\draw (2,0.3pt) -- (2,-1pt)
		node[anchor=north,font=\scriptsize] {$2$};
		%		\node[right,font=\footnotesize] at (1.1,0) {$t_n$};
%		\draw (0.3pt,0) -- (-0.8pt,0);
%		\draw (0.3pt,{1/5}) -- (-0.8pt,{1/5});
%		\draw (0.3pt,{2/5}) -- (-0.8pt,{2/5});
%		\draw (0.3pt,{3/5}) -- (-0.8pt,{3/5});
%		\draw (0.3pt,{4/5}) -- (-0.8pt,{4/5});
%		\draw (0.3pt,1) -- (-0.8pt,1);
%		\node[left,font=\scriptsize] at (0,0) {$0$};
%		\node[left,font=\scriptsize] at (0,{1/5}) {$.2$};
%		\node[left,font=\scriptsize] at (0,{2/5}) {$.4$};
%		\node[left,font=\scriptsize] at (0,{3/5}) {$.6$};
%		\node[left,font=\scriptsize] at (0,{4/5}) {$.8$};
%		\node[left,font=\scriptsize] at (0,1) {$1$};	
		\draw[red,thick] plot[smooth] file {\honedgptwopone};	
		\draw[blue,dashed,thick] plot[smooth] file {\honedgptwopzero};		
	%	\node[font=\small] at (0.5,-0.4) {(f) $\tau=0.1$};
%		\node[right] at (1.1,0) {$\gamma$};
\end{tikzpicture}

%% file: GraphH1DGP3.tex
\begin{tikzpicture}[scale=0.75]
		\tikzstyle{vertex}=[font=\small,circle,draw,fill=yellow!20]
		\tikzstyle{edge} = [font=\scriptsize,draw,thick,-]
		\draw[black, thick, ->] (0,0) -- (0,2.5);
		\draw[black, thick, ->] (-2,0) -- (2.1,0);
		
		\draw[red, thick, -] (0.5,2.5) -- (1,2.5);
        \node [right,font=\scriptsize] at (1,2.5) {$p_1(y,2)$};
        \draw[blue, thick,dashed, -] (0.5,2) -- (1,2);
        \node [right,font=\scriptsize] at (1,2) {$p_0(y,2)$};
		
		\draw (0,0.3pt) -- (0,-1pt)
		node[anchor=north,font=\scriptsize] {$0$};
		\draw (-2,0.3pt) -- (-2,-1pt)
		node[anchor=north,font=\scriptsize] {$-2$};
		\draw (-1,0.3pt) -- (-1,-1pt)
		node[anchor=north,font=\scriptsize] {$-1$};
		\draw (1,0.3pt) -- (1,-1pt)
		node[anchor=north,font=\scriptsize] {$1$};
		\draw (2,0.3pt) -- (2,-1pt)
		node[anchor=north,font=\scriptsize] {$2$};
		%		\node[right,font=\footnotesize] at (1.1,0) {$t_n$};
%		\draw (0.3pt,0) -- (-0.8pt,0);
%		\draw (0.3pt,{1/5}) -- (-0.8pt,{1/5});
%		\draw (0.3pt,{2/5}) -- (-0.8pt,{2/5});
%		\draw (0.3pt,{3/5}) -- (-0.8pt,{3/5});
%		\draw (0.3pt,{4/5}) -- (-0.8pt,{4/5});
%		\draw (0.3pt,1) -- (-0.8pt,1);
%		\node[left,font=\scriptsize] at (0,0) {$0$};
%		\node[left,font=\scriptsize] at (0,{1/5}) {$.2$};
%		\node[left,font=\scriptsize] at (0,{2/5}) {$.4$};
%		\node[left,font=\scriptsize] at (0,{3/5}) {$.6$};
%		\node[left,font=\scriptsize] at (0,{4/5}) {$.8$};
%		\node[left,font=\scriptsize] at (0,1) {$1$};	
		\draw[red,thick] plot[smooth] file {\honedgpthreepone};	
		\draw[blue,dashed,thick] plot[smooth] file {\honedgpthreepzero};		
	%	\node[font=\small] at (0.5,-0.4) {(f) $\tau=0.1$};
%		\node[right] at (1.1,0) {$\gamma$};
\end{tikzpicture}

%% file: GraphH1DGP4.tex
\begin{tikzpicture}[scale=0.75]
		\tikzstyle{vertex}=[font=\small,circle,draw,fill=yellow!20]
		\tikzstyle{edge} = [font=\scriptsize,draw,thick,-]
		\draw[black, thick, ->] (0,0) -- (0,2.5);
		\draw[black, thick, ->] (-2,0) -- (2.1,0);
		
		\draw[red, thick, -] (0.5,2.5) -- (1,2.5);
		\node [right,font=\scriptsize] at (1,2.5) {$p_1(y,2)$};
		\draw[blue, thick,dashed, -] (0.5,2) -- (1,2);
		\node [right,font=\scriptsize] at (1,2) {$p_0(y,2)$};
		
		\draw (0,0.3pt) -- (0,-1pt)
		node[anchor=north,font=\scriptsize] {$0$};
		\draw (-2,0.3pt) -- (-2,-1pt)
		node[anchor=north,font=\scriptsize] {$-2$};
		\draw (-1,0.3pt) -- (-1,-1pt)
		node[anchor=north,font=\scriptsize] {$-1$};
		\draw (1,0.3pt) -- (1,-1pt)
		node[anchor=north,font=\scriptsize] {$1$};
		\draw (2,0.3pt) -- (2,-1pt)
		node[anchor=north,font=\scriptsize] {$2$};
		%		\node[right,font=\footnotesize] at (1.1,0) {$t_n$};
%		\draw (0.3pt,0) -- (-0.8pt,0);
%		\draw (0.3pt,{1/5}) -- (-0.8pt,{1/5});
%		\draw (0.3pt,{2/5}) -- (-0.8pt,{2/5});
%		\draw (0.3pt,{3/5}) -- (-0.8pt,{3/5});
%		\draw (0.3pt,{4/5}) -- (-0.8pt,{4/5});
%		\draw (0.3pt,1) -- (-0.8pt,1);
%		\node[left,font=\scriptsize] at (0,0) {$0$};
%		\node[left,font=\scriptsize] at (0,{1/5}) {$.2$};
%		\node[left,font=\scriptsize] at (0,{2/5}) {$.4$};
%		\node[left,font=\scriptsize] at (0,{3/5}) {$.6$};
%		\node[left,font=\scriptsize] at (0,{4/5}) {$.8$};
%		\node[left,font=\scriptsize] at (0,1) {$1$};	
		\draw[red,thick] plot[smooth] file {\honedgpfourpone};	
		\draw[blue,dashed,thick] plot[smooth] file {\honedgpfourpzero};		
	%	\node[font=\small] at (0.5,-0.4) {(f) $\tau=0.1$};
%		\node[right] at (1.1,0) {$\gamma$};
\end{tikzpicture}

%% file: Rejection_Rates_under_H1_Multi.tex
\begin{table}[h]
	
	\centering
	\caption{Rejection Rates under $H_{1}$ for Multivalued $D$ and Multivalued $Z$}
	\scalebox{0.85}{
		\begin{tabular}{  c  c  c  c  c  c  c  c  c  c  c  c  c  }
			\hline
			\hline
			\multirow{2}{*}{DGP} & \multirow{2}{*}{$n$}& \multicolumn{10}{c}{$\xi$ for $\delta_{\xi}$} &\multirow{2}{*}{$\bar{\nu}_{\xi}$} \\
			\cline{3-12}
			&  & 0.07 & 0.1 & 0.13 & 0.16 & 0.19 & 0.22 & 0.25 & 0.28 & 0.3 & 1 & \\
			\hline

			\multirow{5}{*}{(1)}	
			&200&0.060&0.140&0.175&0.200&0.185&0.155&0.153&0.153&0.153&0.153&0.159\\
			&600&0.672&0.683&0.616&0.482&0.323&0.230&0.214&0.214&0.214&0.214&0.516\\
			&1000&0.606&0.729&0.790&0.792&0.775&0.738&0.715&0.715&0.715&0.715&0.777\\
			&1100&0.889&0.859&0.720&0.504&0.314&0.216&0.217&0.217&0.217&0.217&0.658\\
			&2000&0.969&0.988&0.993&0.987&0.989&0.979&0.975&0.975&0.975&0.975&0.991\\		
			\hline

			\multirow{5}{*}{(2)}	
			
			&200&0.030&0.060&0.074&0.076&0.076&0.069&0.072&0.072&0.072&0.072&0.064\\
			&600&0.347&0.168&0.069&0.054&0.059&0.059&0.056&0.056&0.056&0.056&0.083\\
			&1000&0.404&0.379&0.294&0.146&0.088&0.059&0.062&0.062&0.062&0.062&0.153\\
			&1100&0.434&0.123&0.054&0.059&0.059&0.059&0.060&0.060&0.060&0.060&0.084\\
			&2000&0.896&0.897&0.775&0.521&0.269&0.177&0.154&0.154&0.154&0.154&0.635\\
			
			\hline
			
			\multirow{5}{*}{(3)}	
			
			&200&0.087&0.177&0.240&0.307&0.325&0.297&0.290&0.290&0.290&0.290&0.262\\
			&600&0.695&0.719&0.728&0.693&0.577&0.466&0.434&0.434&0.434&0.434&0.673\\
			&1000&0.660&0.743&0.826&0.856&0.880&0.887&0.875&0.875&0.875&0.875&0.878\\
			&1100&0.884&0.924&0.899&0.773&0.622&0.516&0.517&0.517&0.517&0.517&0.840\\
			&2000&0.968&0.985&0.991&0.995&0.995&0.998&0.999&0.999&0.999&0.999&0.999\\
			
			\hline
			
			\multirow{5}{*}{(4)}	
			&200&0.038&0.099&0.147&0.155&0.148&0.138&0.135&0.135&0.135&0.135&0.146\\
			&600&0.402&0.376&0.366&0.290&0.207&0.209&0.189&0.189&0.189&0.189&0.304\\
			&1000&0.331&0.433&0.407&0.406&0.444&0.475&0.477&0.477&0.477&0.477&0.483\\
			&1100&0.498&0.526&0.492&0.355&0.203&0.137&0.137&0.137&0.137&0.137&0.403\\
			&2000&0.597&0.704&0.710&0.725&0.741&0.769&0.791&0.791&0.791&0.791&0.796\\
			
			\hline
		
			\multirow{5}{*}{(5)}	
			&200&0.365&0.487&0.589&0.626&0.685&0.752&0.780&0.780&0.780&0.780&0.699\\
			&600&0.980&0.990&0.995&0.997&0.998&0.998&0.998&0.998&0.998&0.998&0.998\\
			&1000&0.994&0.998&0.999&0.999&1.000&1.000&1.000&1.000&1.000&1.000&1.000\\
			&1100&1.000&1.000&1.000&1.000&1.000&1.000&1.000&1.000&1.000&1.000&1.000\\
			&2000&1.000&1.000&1.000&1.000&1.000&1.000&1.000&1.000&1.000&1.000&1.000\\
			
			\hline
			
			\multirow{5}{*}{(6)}
			&200&0.372&0.482&0.545&0.616&0.659&0.701&0.711&0.711&0.711&0.711&0.664\\
			&600&0.704&0.823&0.904&0.929&0.962&0.981&0.988&0.988&0.988&0.988&0.965\\
			&1000&0.992&0.999&1.000&1.000&1.000&1.000&1.000&1.000&1.000&1.000&1.000\\
			&1100&0.912&0.957&0.979&0.984&0.990&0.995&0.995&0.995&0.995&0.995&0.990\\
			&2000&1.000&1.000&1.000&1.000&1.000&1.000&1.000&1.000&1.000&1.000&1.000\\

			\hline
		\end{tabular}
	}
	
	\label{tab:Rej H1 multi}
\end{table}

%% file: ApplicationResultsCard.tex
\begin{table}[h]
	
	\centering
	\caption{$p$-values Obtained from the Proposed Test for Each Measure $\nu$}
	\scalebox{0.9}{
		\begin{tabular}{   c  c  c  c  c  c  c  c  c  c  c  }
			\hline
			\hline
			 \multicolumn{10}{c}{$\xi$ for $\delta_{\xi}$} &\multirow{2}{*}{$\bar{\nu}_{\xi}$} \\
			\cline{1-10}
		  0.07 & 0.1 & 0.13 & 0.16 & 0.19 & 0.22 & 0.25 & 0.28 & 0.3 & 1 & \\
			\hline
			0.958&0.975&0.975&0.975&0.975&0.975&0.975&0.975&0.975&0.975&0.973\\

			\hline
		\end{tabular}
	}
	
	\label{tab:ApplicationCard}
\end{table}

%% file: IVValidityAppendixInText.tex
\setcounter{equation}{0}
\renewcommand{\theequation}{\thesection.\arabic{equation}}
\section*{Appendix}

\appendix

%\begin{flushleft}
%\LARGE	\textbf{Appendix}
%\end{flushleft}

\section{Extended Continuous Mapping Theorem and Extended Delta Method}\label{appendix.main results}
We follow \citet{van1996weak} to introduce some notation we use multiple times in the appendix. Let $\left(  \Omega,\mathcal{A},\mathbb{P}\right)  $ be an arbitrary
probability space. For an arbitrary map $T:\Omega\rightarrow\mathbb{\bar{R}}$,
we define the outer integral or outer expectation of $T$ with respect to
$\mathbb{P}$ by
\[
E^{\ast}\left[  T\right]  =\inf\left\{  E\left[  U\right]  :U\geq
T,U:\Omega\rightarrow\mathbb{\bar{R}}\text{ measurable and }E\left[  U\right]
\text{ exists}\right\}  .
\]
The outer probability of an arbitrary subset $B$ of $\Omega$ is
\[
\mathbb{P}^{\ast}\left(  B\right)  =\inf\left\{  \mathbb{P}\left(  A\right)  :A\supset
B,A\in\mathcal{A}\right\}  .
\]
The inner integral (or inner expectation) and the inner probability can
be defined as
\[
E_{\ast}\left[  T\right]  =-E^{\ast}\left[  -T\right] \text{  and }\mathbb{P}_{\ast}\left(
B\right)  =1-\mathbb{P}^{\ast}\left(  \Omega\setminus B\right)  ,
\]
respectively. We denote a minimal measurable majorant of $T$ (resp.\ a maximal
measurable minorant) by $T^{\ast}$ (resp.\ $T_{\ast}$), which always
exists by Lemma 1.2.1 of \citet{van1996weak}. Suppose $T$ is a real-valued map defined on an arbitrary product probability space $\left(  \Omega_1\times\Omega_2,\mathcal{A}_1\times\mathcal{A}_2 ,\mathbb{P}_1\times\mathbb{P}_2\right)  $. We write $E^{\ast}[T]$ for the outer expectation as before, and for every $\omega_1$, we define 
\begin{align}\label{eq.conditional expectation}
	E^{\ast}_2[T](\omega_1)=\inf \int U(\omega_2)\,\mathrm{d}\mathbb{P}_2(\omega_2),
\end{align}
where the infimum is taken over all measurable functions $U:\Omega_2\to \bar{\mathbb{R}}$ with $U(\omega_2)\ge T(\omega_1,\omega_2)$ for all $\omega_2$ such that $\int U \,\mathrm{d}\mathbb{P}_2$ exists. Then $E^{\ast}_1[E^{\ast}_2[T]]$ is the outer integral of the function $E^{\ast}_2[T]:\Omega_1\to \bar{\mathbb{R}}$, and we call $E^{\ast}_1[E^{\ast}_2[T]]$ the repeated outer expectation. We define the repeated inner expectation $E_{1\ast}[E_{2\ast}[T]]$  analogously.\footnote{Additional technical details about the repeated expectations can be found in \citet[pp.~10--12]{van1996weak}.} 
\begin{theorem}[Extended continuous mapping]
	\label{lemma.random continuous mapping}Let $\mathbb{D}$ and $\mathbb{E}$ be
	metric spaces with metrics $d$ and $e$, respectively. Let $\mathbb{D}_{0}\subset\mathbb{D}$.  Let $X$ be Borel measurable and take values
	in $\mathbb{D}_{0}$. Suppose, in addition, that either of the following conditions holds:
	\begin{enumerate}[label=(\alph*)]
		\item Let $\mathbb{D}_{n}\subset \mathbb{ D}$.  Let $X_{n}:\Omega\to\mathbb{D}$ with $X_n(\omega)\in\mathbb{D}_{n}$ for all $\omega\in\Omega$ and all $n$. Let $g_n$ be a random map with $g_{n}(\omega):\mathbb{D}_{n}\to \mathbb{E}$ (for every $\omega\in\Omega$, $g_n(\omega)$ is a map on $\mathbb{D}_n$). The random map $g_n$ satisfies the condition that for every $\varepsilon>0$ there is a measurable set $A\subset\Omega$ with $\mathbb{P}(A)\ge1-\varepsilon$ such that if
		$x_{n}\rightarrow x$ with $x_{n}\in\mathbb{D}_{n}$ and $x\in\mathbb{D}_{0}$,
		then $g_{n}\left(  x_{n}\right)  $ converges to
		$g\left(  x\right)  $ uniformly on $A$ ($\sup_{\omega\in A}e(g_n(\omega)(x_n),g(x))\to0$),\footnote{This is a condition similar to almost uniform convergence. See Definition 1.9.1(ii) of \citet{van1996weak}. By Lemma 1.9.2(iii) of \citet{van1996weak}, almost uniform convergence is equivalent to outer almost sure convergence if the limit is Borel measurable.}  where
		$g:\mathbb{D}_{0}\rightarrow\mathbb{E}$ is a fixed (deterministic) map. Also, $X$ is separable.
		
		\item Let $\mathbb{D}_{n}(\omega)\subset\mathbb{D}$ for all $\omega\in\Omega$ and all $n$. Let $X_{n}:\Omega\to\mathbb{D}$ with $X_n(\omega)\in\mathbb{D}_{n}(\omega)$ for all $\omega\in\Omega$ and all $n$. Let $g_n$ be a random map with $g_{n}(\omega):\mathbb{D}_{n}(\omega)\to \mathbb{E}$ (for every $\omega\in\Omega$, $g_n(\omega)$ is a map on $\mathbb{D}_n(\omega)$).  The random map $g_n$ satisfies the condition that for every $\varepsilon>0$ there is a measurable set $A\subset\Omega$ with $\mathbb{P}(A)\ge1-\varepsilon$ such that for every subsequence $\{x_{n_m}\}$, if
		$x_{n_m}\rightarrow x$ with $x_{n_m}\in\mathbb{D}_{n_m}(\omega_{n_m})$, $\omega_{n_m}\in A$, and $x\in\mathbb{D}_{0}$,
		then $g_{n_m}(\omega_{n_m})\left(  x_{n_m}\right)  $ converges to
		$g\left(  x\right)  $,
		where $g:\mathbb{D}_{0}\rightarrow\mathbb{E}$ is a fixed continuous map.

	\end{enumerate}
	Then we have that
	\begin{enumerate}[label=(\roman{*})]
		\item $X_{n}\leadsto X$ implies that $g_{n}\left(  X_{n}\right)  \leadsto g\left(  X\right)$;
		
		\item If $X_{n}$ converges to $X$ in outer probability,\footnote{See Definition 1.9.1(i) of convergence in outer probability in \citet{van1996weak}.} then $g_{n}\left(  X_{n}\right)$  converges to $g\left(  X\right)$ in outer probability;
		
		\item If $X_{n}$ converges to $ X$ outer almost surely,\footnote{See Definition 1.9.1(iii) of outer almost sure convergence in \citet{van1996weak}.} then $g_{n}\left(  X_{n}\right)$  converges to $g\left(  X\right)$ outer almost surely.
	\end{enumerate}
	
\end{theorem}

\begin{remark}
	Theorem \ref{lemma.random continuous mapping} is an extension of Theorem 1.11.1 (extended continuous mapping) of \citet{van1996weak}. Theorem 1.11.1 of \citet{van1996weak} assumes that every $g_n$ is a fixed map. Theorem \ref{lemma.random continuous mapping} allows every $g_n$ to be random. Theorem \ref{lemma.random continuous mapping}(i) will be used to establish Theorem \ref{lemma.random delta method} (extended delta method). 
\end{remark}
\begin{proof}[Proof of Theorem \ref{lemma.random continuous mapping}]
	Suppose Condition (a) holds. Assume the weakest of the three assumptions: the one in (i) that $X_n\leadsto X$.
	\textbf{First}, let $\mathbb{D}_{\infty}$ be the set of all $x$ for which there exists a sequence $\{x_n\}$ with $x_n\in\mathbb{ D}_n$ and $x_n\to x$. By the representation theorem (see, for example, Theorem 9.4 of \citet{pollard1990empirical} or Theorem 1.10.4 of \citet{van1996weak}), along the lines of the second paragraph in the proof of Theorem 1.11.1 of \citet{van1996weak}, we can show that $\mathbb{P}_{\ast}({X}\in\mathbb{ D}_{\infty})=1$. 
	\textbf{Second}, fix $\varepsilon$ and a measurable set $A$ with $\mathbb{P}(A)\ge 1-\varepsilon$ that satisfies the assumptions, and suppose there is some subsequence such that $x_{n^{\prime}}\to x$ with $x_{n^{\prime}}\in\mathbb{ D}_{n^{\prime}}$ for all ${n^{\prime}}$ and $x\in\mathbb{ D}_0\cap\mathbb{ D}_{\infty}$. Since $x\in\mathbb{ D}_{\infty}$, there is a sequence $y_n\to x$ with $y_n\in\mathbb{ D}_n$ for all $n$. Fill out the subsequence $x_{n^{\prime}}$ to an entire sequence by putting $x_n=y_n$ for all $n\notin \left\{n^{\prime}\right\}$. Then by assumption, $g_n(x_n)\to g(x)$ uniformly on $A$ on this entire sequence, hence also on the subsequence, that is, $g_{n^{\prime}}(x_{n^{\prime}})\to g(x)$ uniformly on $A$.
	\textbf{Third}, let $x_m\to x$ in $\mathbb{D}_0\cap\mathbb{D}_{\infty}$. For every $m$, there is a sequence $y_{m,n}\in\mathbb{D}_n$ with $y_{m,n}\to x_m$ as $n\to \infty$. Fix a small $\varepsilon>0$ and a measurable set $A$ with $\mathbb{P}(A)\ge 1-\varepsilon$ that satisfies the assumptions. Now we have that $g_{n}(y_{m,n})\to g(x_m)$ uniformly on $A$.  For every $m$, take $n_m$ such that $|y_{m,n_m}- x_m|<1/m$ and $|g_{n_m}(y_{m,n_m})- g(x_m)|<1/m$ uniformly on $A$ and such that $n_m$ is increasing in $m$. Then $y_{m,n_m}\to x$, and hence $g_{n_m}(y_{m,n_m})\to g(x)$ uniformly on $A$. Since $|g(x_m)-g(x)|\le |g_{n_m}(y_{m,n_m})- g(x_m)|+|g_{n_m}(y_{m,n_m})- g(x)|$ uniformly on $A$, we have $|g(x_m)-g(x)|\to 0$. Thus $g$ is continuous on $\mathbb{D}_0\cap\mathbb{D}_{\infty}$.
	
	For simplicity of notation, we will write $\mathbb{D}_0$ for $\mathbb{D}_0\cap\mathbb{D}_{\infty}$. Without loss of generality, we assume that $X$ takes its values in $\mathbb{D}_0$. Since $g$ is continuous on $\mathbb{D}_0$ now, $g(X)$ is Borel measurable.  
	
	(i). Let $F$ be an arbitrary
	closed set in $\mathbb{E}$. By the assumptions, for every $\varepsilon>0$
	there is a measurable set $A\subset\Omega$ with $\mathbb{P}\left(  A\right)
	\geq1-\varepsilon$ such that if $x_{n}\rightarrow x$ with $x_{n}\in\mathbb{D}_{n}$
	and $x\in\mathbb{D}_{0}$, then $g_{n}\left(  x_{n}\right)  $ converges  to
	$g\left(  x\right)  $ uniformly on $A$, that is, $\sup_{\omega\in A}\left\vert g_{n}\left(
	\omega\right)  \left(  x_{n}\right)  -g\left(  x\right)  \right\vert
	\rightarrow0$. Fix $\varepsilon$ and $A$. Then
	\begin{equation}\label{eq.subset}
		\cap_{k=1}^{\infty}\overline{\cup_{m=k}^{\infty}\cup_{\omega\in A}\left(
			g_{m}\left(  \omega\right)  \right)  ^{-1}\left(  F\right)  }\subset
		g^{-1}\left(  F\right)  \cup{\left(  \mathbb{D}-\mathbb{D}_{0}\right)  }.
	\end{equation}
	Suppose $x$ is an element of the set on the left-hand side of \eqref{eq.subset}. For every $n$, there exist  $n^{\prime}> n$, $\omega_{n^{\prime}}\in A$, and $x_{n^{\prime}}\in g_{n^{\prime}}(\omega_{n^{\prime}})^{-1}(F)\subset\mathbb{D}_{n^{\prime}}$ such that $d(x_{n^{\prime}},x)\leq1/n$. Therefore, there is a subsequence $x_{n_m}\in
	g_{n_m}(\omega_{n_m})^{-1}(F)\subset\mathbb{D}_{n_m}$ with $\omega_{n_m}\in A$ such that $n_m\uparrow\infty$ and $x_{n_m}\rightarrow x$ as $m\rightarrow
	\infty$. By the definition of $A$, either $g_{n_m}(\omega
	_{n_m})(x_{n_m})\rightarrow g(x)$ or $x\notin\mathbb{D}_{0}$. Since $F$ is closed, this implies that $g(x)\in F$ or $x\notin\mathbb{D}_{0}$.
	Then for every $k$,
	\begin{align}\label{eq.gn 1}
		\limsup_{n\rightarrow\infty}\mathbb{P}^{\ast}\left(  g_{n}\left(
		X_{n}\right)  \in F\right)\leq&\limsup_{n\rightarrow\infty}\mathbb{P}^{\ast}\left(  \left\{ \left\{ X_{n}
		\in\overline{\cup_{m=k}^{\infty}g_{m}^{-1}\left(  F\right)  }\right\}\cap A\right\}
		\cup A^c  \right) \notag \\
		=&\limsup_{n\rightarrow\infty}E\left[  \left(  1\left\{ \left\{  X_{n}\in
		\overline{\cup_{m=k}^{\infty}g_{m}^{-1}\left(  F\right)  }\right\}\cap A\right\}
		\vee1\left\{ A^c\right\}  \right)  ^{\ast}\right]  ,
	\end{align}
	where the equality is from Lemmas 1.2.3(i) and 1.2.1 of \citet{van1996weak}. Then by Lemmas
	1.2.2(viii), 1.2.1, and 1.2.3(i) of \citet{van1996weak},
	\begin{align}\label{eq.gn 2}
		& E\left[  \left(  1\left\{ \left\{ X_{n}\in\overline{\cup_{m=k}^{\infty}g_{m}%
			^{-1}\left(  F\right)  }\right\}\cap A\right\}  \vee1\left\{ A^c\right\}  \right)
		^{\ast}\right]  \notag\\
		=&\,E\left[  \left(  1\left\{ \left\{ X_{n}\in\overline{\cup_{m=k}^{\infty}g_{m}%
			^{-1}\left(  F\right)  }\right\}\cap A\right\}  \right)  ^{\ast}\vee\left(  1\left\{A^c\right\}  \right)  \right]\notag\\
		\leq&\,\mathbb{P}^{\ast}\left( \left\{ X_{n}\in\overline{\cup_{m=k}^{\infty}g_{m}%
			^{-1}\left(  F\right)  }\right\}\cap A\right)  +\mathbb{P}\left( A^{c}\right).
	\end{align}
	By \eqref{eq.gn 1} and \eqref{eq.gn 2}, together with Theorem 1.3.4(iii) (portmanteau) of \citet{van1996weak}, we have
	\begin{align*}
		\limsup_{n\rightarrow\infty}\mathbb{P}^{\ast}\left(  g_{n}\left(
		X_{n}\right)  \in F\right)  \leq& \limsup_{n\rightarrow\infty}
		\mathbb{P}^{\ast}\left( \left\{  X_{n}\in\overline{\cup_{m=k}^{\infty}g_{m}
			^{-1}\left(  F\right)  }\right\}\cap A\right)  +\mathbb{P}\left(
		A^{c}\right) \\
		\leq&\limsup_{n\rightarrow\infty}\mathbb{P}^{\ast}\left(  X_{n}\in
		\overline{\cup_{m=k}^{\infty}\cup_{\omega\in A}\left(  g_{m}\left(
			\omega\right)  \right)  ^{-1}\left(  F\right)  }\right)  +\varepsilon\notag\\ 
		\leq&\, \mathbb{P}\left(  X\in\overline{\cup_{m=k}^{\infty}\cup_{\omega\in
				A}\left(  g_{m}\left(  \omega\right)  \right)  ^{-1}\left(  F\right)
		}\right)  +\varepsilon.
	\end{align*}
	Letting $k\rightarrow\infty$ together with \eqref{eq.subset} gives
	\begin{align*}
		\limsup_{n\rightarrow\infty}\mathbb{P}^{\ast}\left(  g_{n}\left(
		X_{n}\right)  \in F\right) & \leq\mathbb{P}\left(  X\in\cap_{k=1}^{\infty
		}\overline{\cup_{m=k}^{\infty}\cup_{\omega\in A}\left(  g_{m}\left(
			\omega\right)  \right)  ^{-1}\left(  F\right)  }\right) +\varepsilon\notag\\ &\leq\mathbb{P}\left(  g\left(  X\right)  \in F\right)  +\varepsilon.
	\end{align*}
	Since $\varepsilon$ can be arbitrarily small, we can conclude that
	$\limsup_{n\rightarrow\infty}\mathbb{P}^{\ast}\left(  g_{n}\left(
	X_{n}\right)  \in F\right)  \leq\mathbb{P}\left(  g\left(  X\right)  \in
	F\right)  $. By Theorem 1.3.4(iii) of \citet{van1996weak} again, $g_{n}\left(  X_{n}\right)
	\leadsto g\left(  X\right)  $. 
	
	(ii). Choose $\delta_{n}\downarrow0$ with $\mathbb{P}%
	^{\ast}\left(  d\left(  X_{n},X\right)  \geq\delta_{n}\right)  \rightarrow0$.
	Fix $\varepsilon>0$. Let $A\subset\Omega$ be a measurable set with $\mathbb{P}\left(  A\right) 
	\ge1-\varepsilon$ that satisfies the assumptions. Let $B_{n}(\omega)$ be the set of all $x$ such that there is a 
	$y\in\mathbb{D}_{n}$ with $d\left(  y,x\right)  <\delta_{n}$ and
	$e\left(  g_{n}\left(  \omega\right)  \left(  y\right),g\left(  x\right)  \right)  >\varepsilon$. Let $B_{n}=\cup_{\omega\in A}B_{n}(\omega)$. Suppose $x\in B_{n}$ for infinitely
	many $n$.  Then there are sequences $\omega_{n_m}\in A$ and $x_{n_{m}}\in\mathbb{D}_{n_{m}}$ with $x_{n_{m}}\rightarrow x$ such that
	$  e\left(  g_{n_{m}}(\omega_{n_m})\left(  x_{n_{m}}\right)  ,g\left(  x\right)  \right)  >\varepsilon$ for each $m$. This implies that $x_{n_m}\to x$  with $x_{n_m}\in \mathbb{ D}_{n_m}$ but that $g_{n_m}(x_{n_m})$ does not converge to $g(x)$ uniformly on $A$. Thus by assumption, $x\notin\mathbb{D}_{0}$. Note that
	$x\in\limsup B_{n}$ is equivalent to $x\in B_{n}$ for infinitely many $n$.
	Thus we can conclude that $\limsup B_{n}\cap\mathbb{D}_{0}=\varnothing$.
	Since $g$ is continuous on $\mathbb{D}_{0}$, $B_{n}\cap\mathbb{D}_{0}$ is
	relatively open in $\mathbb{D}_{0}$ and hence relatively Borel. This is
	because if $z\in\mathbb{D}_{0}$ is close enough to $x\in B_{n}\cap\mathbb{D}_{0}$, then $d(y,z)\le d(y,x)+d(x,z)<\delta_n$ and $e\left(  g_{n}\left(  \omega\right)  \left(  y \right)  ,g\left(
	z\right)  \right)  \geq e\left(  g_{n}\left(
	\omega\right)  \left(  y\right)  ,g\left(  x\right)  \right)  -e\left(
	g\left(  z\right)  ,g\left(  x\right)  \right)  >\varepsilon$. 
	Since $X$ takes values in $\mathbb{ D}_0$ by assumption, by Lemma 1.2.3(i) of \citet{van1996weak},
	\begin{align*}
		\mathbb{P}^{\ast}\left(  X\in B_{n}\right)  =E^{\ast}\left[
		1\left\{  X\in B_{n}\right\}      \right]  =E\left[  1\left\{  X\in B_{n}\cap \mathbb{ D}_0\right\} \right].
	\end{align*}
	Also, by the dominated convergence theorem, 
	\begin{align*}
		&E\left[  1\left\{  X\in B_{n}\cap \mathbb{ D}_0\right\} \right]     \leq E\left[  1\left\{  X\in\cup
		_{m=n}^{\infty}(B_{m}\cap \mathbb{ D}_0)\right\}  \right]  \\
		\rightarrow &\, E\left[  1\left\{  X\in\cap_{n=1}^{\infty}\cup_{m=n}^{\infty
		}(B_{m}\cap \mathbb{ D}_0)\right\}  \right] =\mathbb{P}\left(  X\in\limsup  B_{n}\cap\mathbb{D}_{0}
		\right) =0.
	\end{align*}
	This implies that  $\mathbb{P}^{\ast}\left(  X\in B_{n}\right) \to 0$ as $n\to \infty$. Now we have that 
	\begin{align*}
		\mathbb{P}^{\ast}\left( e(g_n(X_n),g(X) )>\varepsilon \right)\le&\,\mathbb{P}^{\ast}\left( \left\{ e(g_n(X_n),g(X) )>\varepsilon\right\}\cap A \right) +\mathbb{P}\left(  A^c \right)\\
		\le&\, \mathbb{P}^{\ast}\left( X\in B_n\text{ or } d(X_n,X)\ge \delta_n \right)+\varepsilon \to \varepsilon.
	\end{align*}
	Since $\varepsilon$ is arbitrary, the claim holds.
	
	(iii). By Lemmas 1.9.3(i) and 1.9.2(iii) of \citet{van1996weak}, it suffices to prove that $\sup_{m\geq n}e\left(  g_{m}\left(
	X_{m}\right)  ,g\left(  X\right)  \right)  $ converges to $0$ in outer
	probability. Choose $\delta_{n}\downarrow0$ with $\mathbb{P}^{\ast
	}\left(  \sup_{m\geq n}d\left(  X_{m},X\right)  \geq\delta_{n}\right)
	\rightarrow0$. Fix $\varepsilon>0$. Let $A\subset\Omega$ be a measurable set with $\mathbb{P}\left(  A\right) 
	\ge1-\varepsilon$ such that if
	$x_{n}\rightarrow x$ with $x_{n}\in\mathbb{D}_{n}$ and $x\in\mathbb{D}_{0}$,
	then $g_{n}\left(  x_{n}\right)  $ converges to
	$g\left(  x\right)  $ uniformly on $A$. Let $B_{n}(\omega)$ be the set of all $x$ such that there are $m\ge n$ and
	$y\in\mathbb{D}_{m}$ with $d\left(  y,x\right)  <\delta_{n}$ and
	$e\left(  g_{m}\left(  \omega\right)  \left(  y\right),g\left(  x\right)  \right)  >\varepsilon$. Let $B_{n}=\cup_{\omega\in A}B_{n}(\omega)$. Then we can finish the proof along the lines of the proof of (ii).
	
	Suppose Condition (b) holds. 
	Repeat the proofs of (i), (ii), and (iii) under Condition (a) with the properties of $g_n$ and $g$ under Condition (b).
	For (ii), let $B_{n}(\omega)$ be the set of all $x$ such that there is a 
	$y\in\mathbb{D}_{n}(\omega)$ with $d\left(  y,x\right)  <\delta_{n}$ and
	$e\left(  g_{n}\left(  \omega\right)  \left(  y\right),g\left(  x\right)  \right)  >\varepsilon$. 
	For (iii), let $B_{n}(\omega)$ be the set of all $x$ such that there are $m\ge n$ and
	$y\in\mathbb{D}_{m}(\omega)$ with $d\left(  y,x\right)  <\delta_{n}$ and
	$e\left(  g_{m}\left(  \omega\right)  \left(  y\right),g\left(  x\right)  \right)  >\varepsilon$.
	The key difference is that Condition (a) requires that $X_n(\omega)\in\mathbb{D}_n$ for all $\omega$ holds for some fixed $\mathbb{ D}_n$. Condition (b) only requires that $X_n(\omega)\in\mathbb{D}_n(\omega)$ for all $\omega$ holds for some random $\mathbb{D}_n$ which can take different values $\mathbb{D}_n(\omega)$ for different $\omega$. On the other hand, Condition (b) strengthens the properties of $g_n$ and $g$ so that the claims hold as well. 
\end{proof}

\begin{theorem}[Extended delta method]\label{lemma.random delta method} 
	Let $\mathbb{D}$ and $\mathbb{E}$ be metrizable topological vector spaces, and let $r_{n}$ be constants with $r_{n}
	\rightarrow\infty$. Let $\hat{\phi}_{n}:\Omega\rightarrow\mathbb{D}%
	_{\mathcal{F}}\subset\mathbb{D}$ be a random element for every $n$. Let $\mathbb{D}_{0}\subset\mathbb{D}$.
	\begin{enumerate}[label=(\roman{*})]
		\item Let $\mathcal{F}:\mathbb{D}_{\mathcal{F}}\rightarrow\mathbb{E}$ satisfy the condition that for every $\varepsilon>0$, there is a measurable set $A\subset\Omega$ with $\mathbb{P}(A)\ge 1-\varepsilon$ such that for some map $\mathcal{F}_{\phi}^{\prime}$ on $\mathbb{ D}_0$,
		\begin{align*}
			r_{n}(  \mathcal{F}(  \hat{\phi}_{n}+r_{n}^{-1}h_{n})
			-\mathcal{F}(  \hat{\phi}_{n})  )  \rightarrow\mathcal{F}
			_{\phi}^{\prime}\left(  h\right)
		\end{align*}
		uniformly on $A$ for every convergent sequence $\{h_{n}\}\subset \mathbb{D}$ with $\hat{\phi}_{n}(\omega)+r_{n}^{-1}h_{n} 
		\in\mathbb{D}_{\mathcal{F}}$ for all $n$ and all $\omega$ and $h_{n}\rightarrow
		h\in\mathbb{D}_{0}$. If $X_{n}:\Omega\rightarrow
		\mathbb{D}_{\mathcal{F}}$ are maps \textbf{with} $ X_{n}(\omega)-\hat{\phi}_{n}(\omega)+\hat{\phi}_{n}(\omega^{\prime})\in\mathbb{D}_{\mathcal{ F}}$ for all $\omega,\omega^{\prime}\in\Omega$ and $r_{n}(  X_{n}-\hat{\phi}_{n})  \leadsto X$, where $X$ is separable and
		takes its values in $\mathbb{D}_{0}$, then $r_{n}(  \mathcal{F}(
		X_{n})  -\mathcal{F}(  \hat{\phi}_{n})  )
		\leadsto\mathcal{F}_{\phi}^{\prime}\left(  X\right)  $. Moreover, if
		$\mathcal{F}_{\phi}^{\prime}$ is continuous on all of $\mathbb{D}$, then
		$r_{n}(  \mathcal{F}(  X_{n})  -\mathcal{F}(  \hat{\phi
		}_{n})  )  -\mathcal{F}_{\phi}^{\prime}(  r_{n}(
		X_{n}-\hat{\phi}_{n})  )  $ converges to zero in outer probability.
		
		\item Let $\mathcal{F}:\mathbb{D}_{\mathcal{F}}\rightarrow\mathbb{E}$ satisfy the condition that for every $\varepsilon>0$, there is a measurable set $A\subset\Omega$ with $\mathbb{P}(A)\ge 1-\varepsilon$ such that for some \textbf{continuous} map $\mathcal{F}_{\phi}^{\prime}$ on $\mathbb{ D}_0$, 
		\begin{align*}
			r_{n_m}\{  \mathcal{F}(  \hat{\phi}_{n_m}(\omega_{n_m})+r_{n_m}^{-1}h_{n_m})
			-\mathcal{F}(  \hat{\phi}_{n_m}(\omega_{n_m})) \}  \rightarrow\mathcal{F}
			_{\phi}^{\prime}\left(  h\right)
		\end{align*}
		for every convergent subsequence $\{h_{n_m}\}\subset \mathbb{D}$ with $\hat{\phi}_{n_m}(\omega_{n_m})+r_{n_m}^{-1}h_{n_m} 
		\in\mathbb{D}_{\mathcal{F}}$, $\omega_{n_m}\in A$, and $h_{n_m}\rightarrow
		h\in\mathbb{D}_{0}$. If $X_{n}:\Omega\rightarrow
		\mathbb{D}_{\mathcal{F}}$ are maps with  $r_{n}(  X_{n}-\hat{\phi}_{n})  \leadsto X$, where $X$
		takes its values in $\mathbb{D}_{0}$, then $r_{n}(  \mathcal{F}(
		X_{n})  -\mathcal{F}(  \hat{\phi}_{n})  )
		\leadsto\mathcal{F}_{\phi}^{\prime}\left(  X\right)  $. Moreover, if
		$\mathcal{F}_{\phi}^{\prime}$ is continuous on all of $\mathbb{D}$, then
		$r_{n}(  \mathcal{F}(  X_{n})  -\mathcal{F}(  \hat{\phi
		}_{n})  )  -\mathcal{F}_{\phi}^{\prime}(  r_{n}(
		X_{n}-\hat{\phi}_{n})  )  $ converges to zero in outer probability.
	\end{enumerate}
\end{theorem}

\begin{remark}
	Theorem \ref{lemma.random delta method} is an extension of Theorem 3.9.5 (delta method) of \citet{van1996weak}. Here, $\hat{\phi}_n$ is allowed to be random, which is the key difference between the two theorems. Theorem \ref{lemma.random delta method} is used to establish the asymptotic distribution of the test statistic under null.
\end{remark}

\begin{proof}[Proof of Theorem \ref{lemma.random delta method}]
	(i). The proof mainly relies on the results of Theorem \ref{lemma.random continuous mapping}. 
	Define
	$\mathbb{D}_{n} (\omega)=\{  h\in\mathbb{D}:\hat{\phi}_{n}\left(
	\omega\right)  +r_{n}^{-1}h\in\mathbb{D}_{\mathcal{F}} \}$ for every $n$ and every $\omega\in\Omega$. Let $\mathbb{D}_{n}=\cap_{\omega\in\Omega}\mathbb{D}_n(\omega)$.
	Define $g_{n}(\omega)\left(  h\right)  =r_{n}(  \mathcal{F}(\hat{\phi}_{n}(\omega)+r_{n}^{-1}h)  -\mathcal{F}(  \hat{\phi}_{n}(\omega)))  $ for every $n$, every $\omega\in\Omega$, and every $h\in\mathbb{ D}_n$. Here, $g_n$ is a random map because of $\hat{\phi}_n$. For every $n$ and every $\omega\in\Omega$, $g_{n}(\omega):\mathbb{D}_n\to\mathbb{E}$.  By the
	assumptions, for every $\varepsilon>0$ there is a measurable set $A\subset\Omega$ with $\mathbb{P}(A)\ge 1-\varepsilon$ such that if $h_n\in\mathbb{D}_n$ with $h_n\to h\in\mathbb{D}_0$, then $g_{n}\left(  h_n\right)  \rightarrow\mathcal{F}_{\phi}^{\prime
	}\left(  h\right)  $ uniformly on $A$. Also, $r_{n}(  X_{n}(\omega)-\hat{\phi}_{n}(\omega))\in\mathbb{D}_n$ for all $\omega$ by assumption. Now
	by Theorem \ref{lemma.random continuous mapping}(i) (under Condition (a)),
	\[r_{n}(  \mathcal{F}(
	X_{n})  -\mathcal{F}(  \hat{\phi}_{n})  )=
	g_{n}(  r_{n}(  X_{n}-\hat{\phi}_{n})  )  \leadsto
	\mathcal{F}_{\phi}^{\prime}\left(  X\right)  .
	\]
	Moreover, suppose $\mathcal{F}_{\phi}^{\prime}$ is continuous on all of
	$\mathbb{D}$, and let $f_{n}\left(  h\right)  =(
	g_{n}\left(  h\right)  ,\mathcal{F}_{\phi}^{\prime}\left(  h\right))  $ for every $h
	\in\mathbb{D}_{n}$. By Theorem
	\ref{lemma.random continuous mapping}(i) again,
	\[
	\left({r_{n}(  \mathcal{F}(  X_{n})  -\mathcal{F}(
		\hat{\phi}_{n})  )  },{\mathcal{F}_{\phi}^{\prime}(
		r_{n}(  X_{n}-\hat{\phi}_{n})  )  }\right)=f_n(  r_{n}(  X_{n}-\hat{\phi}_{n})
	)  \leadsto\left({\mathcal{F}_{\phi}^{\prime}},{\mathcal{F}_{\phi
		}^{\prime}}\right)\left(  X\right)  .
	\]
	Thus by Theorem 1.3.6 (continuous mapping) of \citet{van1996weak},\linebreak
	$r_{n}(  \mathcal{F}(  X_{n})  -\mathcal{F}(  \hat{\phi
	}_{n})  )  -\mathcal{F}_{\phi}^{\prime}(  r_{n}(
	X_{n}-\hat{\phi}_{n})  )  \leadsto 0$. The claim follows from Lemma 1.10.2(iii) of \citet{van1996weak}.
	
	(ii). Together with the continuity of $\mathcal{F}^{\prime}_{\phi}$, by arguments similar to the proof of (i), we can show that the claim holds by Theorem \ref{lemma.random continuous mapping}(i) (under Condition (b)). 
\end{proof}

\setcounter{equation}{0}
\renewcommand{\theequation}{\thesection.\arabic{equation}}
\section{{Conditioning Covariates}}\label{sec.conditioning covariates setup}
In this section, we consider the case where conditioning covariates may be present, that is, the random assignment assumption holds conditional on some covariates. Suppose $X$ is a conditioning covariate vector with dimension $d_X$, let $\mathcal{X}$ be the set of possible values
of $X$, and let $\mathcal{X}=\left\{  x_{1},\ldots,x_{L}\right\}  $.  

First, consider the case introduced in Section \ref{subsec.multi} where the treatment and the instrument are both multivalued (and ordered). 
A testable implication with conditioning covariates is as follows.
\begin{lemma}\label{lemma.testable implication conditioning ordered}
	A testable implication of
	the {conditional version} of Assumption \ref{ass.IV validity for multivalued Z}
	is that
	\begin{align}\label{eq.testable implication covariates}
		&\mathbb{P}\left(  Y\in B,D=d_{\max}|Z=z_k,X=x_l\right)     \leq \mathbb{P}\left(  Y\in B,D=d_{\max}|Z=z_{k+1},X=x_l\right)\nonumber\\
		&\text{and } \mathbb{P}\left(  Y\in B,D=d_{\min}|Z=z_k,X=x_l\right)     \geq \mathbb{P}\left(  Y\in B,D=d_{\min}|Z=z_{k+1},X=x_l\right);  \notag\\
		&\mathbb{P}\left(  D\in C|Z=z_k,X=x_l\right)     \geq \mathbb{P}\left(  D\in C|Z=z_{k+1},X=x_l\right)
	\end{align}
	for all $k$ with $1\le k\le K-1$, all $l$ with $1\le l \le L$, all $B\in\mathcal{B}_{\mathbb{R}}$, and all $C=(-\infty,c]$ with $c\in\mathbb{R}$.	
\end{lemma}
Suppose $d_{\min}=0$ and $d_{\max}=1$ without loss of generality. 
Define function spaces
\begin{align}{\label{def.function spaces covariates}}
	& \mathcal{G}=\left\{  \left(  1_{\mathbb{R}
		\times\mathbb{R}  \times\left\{  z_{k}\right\}  \times\left\{
		x_{l}\right\}  },1_{\mathbb{R}\times\mathbb{R}  \times\left\{  z_{k+1}\right\}  \times\left\{
		x_{l}\right\}  }\right)  :k=1,\ldots, K-1,l=1,\ldots, L\right\},  \notag\\
	&\mathcal{H}_{1}=\left\{  \left(  -1\right)  ^{d}\cdot1_{B\times\left\{
		d\right\}  \times\mathbb{R}\times\mathbb{R}^{d_X}}:B\text{ is a closed interval},
	d\in\{0,1\}\right\},\notag\\
	&\mathcal{H}_{2}=\left\{  1_{\mathbb{R}\times C \times\mathbb{R}\times\mathbb{R}^{d_X}}:C=(-\infty,c],c\in\mathbb{R}\right\}, \text{ and }\mathcal{H}=\mathcal{H}_{1}\cup\mathcal{H}_{2}  .
\end{align}
For every probability measure $Q$ with \eqref{eq.Q map}, we define $\phi_Q$ by
$
\phi_{Q}\left(  h,g\right)  ={Q\left(  h\cdot g_2 \right) }/{Q\left(  g_2 \right)  }-{Q\left(  h\cdot g_1\right)  }/{Q\left(  g_1 \right)  }
$
for every $(h,g)\in\mathcal{H}\times\mathcal{G}$ with $g=(g_1,g_2)$. Testable implication \eqref{eq.testable implication covariates} is equivalent to the $H_0$ in
\[
H_0: \sup_{(h,g)\in\mathcal{H}\times\mathcal{G}}\phi_{Q}\left( h,g\right) \le 0 \text{ and } H_1: \sup_{(h,g)\in\mathcal{H}\times\mathcal{G}}\phi_{Q}\left( h,g\right) > 0
\]
if $Q$ is the underlying probability distribution of the data. Then we can follow the test procedure in Section \ref{subsbusec.test procedure} to conduct the test with the function space $\mathcal{H}\times\mathcal{G}$ defined by the $\mathcal{H}$ and the $\mathcal{G}$ in \eqref{def.function spaces covariates}.

Second, consider the case  introduced in Section \ref{subsec.unordered} where the treatment and the instrument can both be unordered. 
A testable implication with conditioning covariates is as follows.
\begin{lemma}\label{lemma.testable implication conditioning unordered}
	A testable implication of the conditional version of Assumption \ref{ass.IV validity for unordered D} is given by
	\begin{align}\label{eq.testable implication unordered treatment with conditioning covariates}
		\mathbb{P}\left(  Y\in B,D=d|Z=z^{\prime},X=x_l\right)  \le \mathbb{P}\left(Y\in B,D=d|Z=z,X=x_l\right)
	\end{align}
	for all Borel sets $B$, all $\left(  d,z,z^{\prime}\right)  \in\mathcal{C}$, and all $l$ with $1\le l\le L$, where $\mathcal{C}$ is a prespecified subset of $\mathcal{D}\times\mathcal{Z}\times\mathcal{Z}$. 
\end{lemma}
The inequality in \eqref{eq.testable implication unordered treatment with conditioning covariates} is similar to the generalized regression monotonicity (GRM) hypothesis in \citet{hsu2019testing}. The major difference is that $Z$ is allowed to be unordered in \eqref{eq.testable implication unordered treatment with conditioning covariates}.  
Define the function space
\begin{align}{\label{def.function spaces covariates unordered}}
	\mathcal{H}\times\mathcal{G}=\left\{
	\begin{array}
		[c]{c}%
		\left(  1_{B\times\{d\}\times\mathbb{R}\times\mathbb{R}^{d_{X}}},\left(
		1_{\mathbb{R}\times\mathbb{R}\times\left\{  z\right\}  \times\left\{
			x_{l}\right\}  },1_{\mathbb{R}\times\mathbb{R}\times\left\{  z^{\prime
			}\right\}  \times\left\{  x_{l}\right\}  }\right)  \right)  :B\text{ is a
			closed interval},\\
		(d,z,z^{\prime})\in\mathcal{C},l=1,\ldots,L
	\end{array}
	\right\}.
\end{align}
For every probability measure $Q$ with \eqref{eq.Q map}, we define $\phi_Q$ by
$
\phi_{Q}\left(  h,g\right)  ={Q\left(  h\cdot g_2 \right) }/{Q\left(  g_2 \right)  }-{Q\left(  h\cdot g_1\right)  }/{Q\left(  g_1 \right)  }
$
for every $(h,g)\in\mathcal{H}\times\mathcal{G}$ with $g=(g_1,g_2)$. Testable implication \eqref{eq.testable implication unordered treatment with conditioning covariates} is equivalent to the $H_0$ in 
\[
H_0: \sup_{(h,g)\in\mathcal{H}\times\mathcal{G}}\phi_{Q}\left( h,g\right) \le 0 \text{ and } H_1: \sup_{(h,g)\in\mathcal{H}\times\mathcal{G}}\phi_{Q}\left( h,g\right) > 0
\]
if $Q$ is the underlying probability distribution of the data. Then we can follow the test procedure in Section \ref{subsbusec.test procedure} to conduct the test with the function space $\mathcal{H}\times\mathcal{G}$ defined in \eqref{def.function spaces covariates unordered}.

%% file: IVValidityAppendix.tex
\setcounter{page}{1}
\setcounter{equation}{0}
\setcounter{footnote}{0}
\setcounter{table}{0}
\renewcommand{\theequation}{\thesection.\arabic{equation}}
\renewcommand{\thetable}{\thesection.\arabic{table}}
\begin{center}

\LARGE{Instrument Validity for Heterogeneous Causal Effects\\
	Online Supplementary Appendix }\\
[0.75cm]
\large{Zhenting Sun\\
	China Center for Economic Research\\
	National School of Development\\ Peking University\\
	zhentingsun@nsd.pku.edu.cn
}\\
\bigskip
\today
	
\end{center}

\bigskip

%\beginsupplement
%\begin{appendices}
%The supplementary appendix consists of two sections. Section \ref{appen.auxiliary} provides the proofs of the main results in the text. Section \ref{sec.appendix Monte Carlo Comparison} shows the power comparisons between the proposed test and the test of \citet{kitagawa2015test} via Monte Carlo simulations. 

%\appendix

For the multivalued ordered treatment case, we assume $\mathcal{D}=\left\{
d_{1},d_{2},\ldots\right\}  $ in the proofs to obtain more general results. Assumption \ref{ass.IV validity for multivalued Z} with $\mathcal{D}=\{d_1,d_2,\dots\}$ is 
\begin{enumerate}[label=(\roman*)]
	\item Instrument Exclusion: For all $d\in\mathcal{D}$, $Y_{dz_{1}}=Y_{dz_{2}}
	=\cdots=Y_{dz_{K}}$ almost surely.
	
	\item Random Assignment: The variable $Z$ is jointly independent of $(
	\tilde{Y},\tilde{D})  $, where
	\begin{align*}
		\tilde{Y}   =\left(  Y_{d_{1}z_{1}},\ldots,Y_{d_{1}z_{K}},Y_{d_{2}z_{1}
		},\ldots,Y_{d_{2}z_{K}},\dots\right) \text{ and }\tilde{D} =\left(  D_{z_{1}},\ldots,D_{z_{K}}\right).
	\end{align*}

	\item Instrument Monotonicity: The potential treatment response variables satisfy $D_{z_{k+1}}\geq D_{z_{k}}$ almost surely for all
	$k\in \{1,2,\ldots, K-1\}$.
	
\end{enumerate}
Without loss of generality, we may assume that both $d_{\min}$ and $d_{\max}$ exist with $d_{\min}=0$ and $d_{\max}=1$ for simplicity. If $d_{\min}$ and $d_{\max}$ exist, we can always normalize $d_{\min}$ and $d_{\max}$ to $0$ and $1$, respectively. Then the function spaces defined in \eqref{def.function spaces} can be used for $\mathcal{D}=\left\{
d_{1},d_{2},\ldots\right\}  $.
All the results hold for $\mathcal{D}=\left\{d_{1},\ldots,d_J\right\}  $. 

\section{Proofs of Main Results}\label{appen.auxiliary}

We first introduce the following notation. For every $A\subset\bar{\mathcal{H}}\times\mathcal{G}$, define a map $\mathcal{S}_{A }:\ell^{\infty}\left( \Xi\times \bar{\mathcal{H}}\times\mathcal{G} \right)
\rightarrow\ell^{\infty}(\Xi)$ by 
\[
\mathcal{S}_{A}\left(  \psi\right)(\xi)
=\sup_{\left(  h,g\right)  \in A }\psi\left( \xi, h,g\right)
\]
for all $\psi\in\ell^{\infty}\left( \Xi\times \bar{\mathcal{H}}\times\mathcal{G} \right)  $. For simplicity of notation, we will write ${\mathcal{S}}$ for $\mathcal{S}_{\bar{\mathcal{H}}\times\mathcal{G}}$.
Define $\mathcal{M}:\ell^{\infty}(  {\bar{\mathcal{H}}\times\mathcal{G}} )\to \ell^{\infty}(  \Xi\times\bar{\mathcal{H}}\times\mathcal{G} )$ by 
\begin{align}\label{eq.M}
	\mathcal{M}(\varphi)(\xi,h,g)=\max\{\xi,\varphi(h,g)\}
\end{align}
for all $\varphi\in\ell^{\infty}(  {\bar{\mathcal{H}}\times\mathcal{G}} )$ and all $(\xi,h,g)\in  \Xi\times\bar{\mathcal{H}}\times\mathcal{G} $. %Let $\Xi\subset [0,1]$ be a predetermined measurable set. 
Note that for every finite sample set, 
\begin{align}\label{eq.S equi}
	\mathcal{S}_{\mathcal{H}\times\mathcal{G}}(  {\hat{\phi}}_{P_n}/{\mathcal{M}(\hat{\sigma}_{P_n})})=\mathcal{S}(  {\hat{\phi}}_{P_n}/{\mathcal{M}(\hat{\sigma}_{P_n})}).
\end{align}
Define a function $\mathcal{I}:L^1(\nu)\to\mathbb{R}$ by $\mathcal{I}(f)=\int_{\Xi}f\,\mathrm{d}\nu$
for all $f\in L^1(\nu)$. 
Now we can write the test statistic in \eqref{eq.test stat expansion} as 
\begin{equation}\label{eq.test stat}
	\sqrt{T_n}\mathcal{I}\circ\mathcal{S}_{\mathcal{H}\times\mathcal{G}}\left(  \frac{\hat{\phi}_{P_n}}{\mathcal{M}(\hat{\sigma}_{P_n})}\right) .
\end{equation}

\begin{lemma}
\label{lemma.H complete}
Let $\mathcal{P}$ be the set of probability measures defined in Section \ref{sec.multi D and Z}. Let $\mathcal{H}_1$, $\bar{\mathcal{H}}_1$, $\mathcal{H}_2$, $\bar{\mathcal{H}}_2$, $\mathcal{H}$, and $\bar{\mathcal{H}}$ be as in \eqref{def.function spaces}.
Then for every $Q\in\mathcal{P}$, the closures of $\mathcal{H}_1$ and $\mathcal{H}_2$ in $L^2(Q)$ are equal to $\bar{\mathcal{H}}_1$ and $\bar{\mathcal{H}}_2$, respectively. Also, the closure of $\mathcal{H}$ in $L^2(Q)$ is equal to $\bar{\mathcal{H}}$ for every $Q\in\mathcal{P}$.

\end{lemma}

\begin{proof}[Proof of Lemma \ref{lemma.H complete}]
Let 
$\mathcal{H}_{1d}=\{  \left(  -1\right)  ^{d}\cdot1_{B\times\left\{
	d\right\}\times\mathbb{R}  }:B\text{ is a closed interval in }
\mathbb{R} \}$
for $d\in\{0,1\}$. We first show that the closure of $\mathcal{H}_{1d}$ in $L^2(Q)$ is equal to
\begin{align*}
\bar{\mathcal{H}}_{1d}=\left\{  \left(  -1\right)  ^{d}\cdot1_{B\times\left\{
	d\right\}\times\mathbb{R}  }:B\text{ is a closed, open, or half-closed interval in }
\mathbb{R} \right\}.
\end{align*}
If this is true, the first claim of the Lemma follows from $\bar{\mathcal{H}}_1=\bar{\mathcal{H}}_{10}\cup\bar{\mathcal{H}}_{11}$.

Suppose there is a sequence $\left\{  h_{n}\right\}  \subset
\mathcal{H}_{1d}$ such that $\left\Vert h_{n}-h\right\Vert _{L^2\left(Q\right)  }\rightarrow0$ for some $h\in L^2(Q)$. Then $h_{n}$ is a Cauchy sequence, that is, $\left\Vert h_{n}-h_{m}\right\Vert _{L^2\left(Q\right)  }\rightarrow0$ as $n,m\rightarrow\infty$. By the definition of $\mathcal{H}_{1d}$, $h_{n}=\left(  -1\right)  ^{d}\cdot1_{B_{n}\times\left\{  d\right\}\times\mathbb{R}  }$,
where $B_{n}$ is a closed interval in $\mathbb{R}$. 	
It is possible that $\int1_{B_{n}\times\left\{  d \right\} \times\mathbb{R} }
\,\mathrm{d}Q\rightarrow0$, and in this case there is a $B=\left\{  a\right\}  $ for some $a\in
\mathbb{R}	$ such that $Q\left(  B\times\mathbb{R}\times\mathbb{R}\right)  =0$ and  $h_{n}%
\rightarrow(-1)^d\cdot1_{B\times\left\{  d\right\} \times\mathbb{R} }\in\mathcal{H}_{1d}$.
If $\int1_{B_{n}\times\left\{  d \right\} \times\mathbb{R} }
\,\mathrm{d}Q\not\to0$, then there is an $\varepsilon>0$ such that for all $n_{\varepsilon}>0$, there
is an $n> n_{\varepsilon}  $ such that $\left\Vert h_{n} \right\Vert
_{L^2\left(Q\right)  }^2>\varepsilon$. For a $\delta_{1}\ll\varepsilon$, there is an $N_{1}$ such that $\left\Vert h_{n}-h_{m}\right\Vert _{L^2\left(
Q\right)  }^2<\delta_{1}$ for all $m,n>N_{1}$. Thus there is an $n_{1}>N_{1}$
such that $\left\Vert h_{n_{1}}\right\Vert _{L^2\left(Q\right)  }^2>\varepsilon$ and $\left\Vert h_{n}-h_{n_{1}}\right\Vert _{L^2\left(
Q\right)  }^2<\delta_{1}$ for all $n>N_{1}$. Now let $\delta_2$ be such that
$0<\delta_{2}\ll \delta_{1}$. Then there is an $N_{2}>n_{1}$ such that $\left\Vert h_{n}%
-h_{m}\right\Vert _{L^2\left(Q\right)  }^2<\delta_{2}$ for all $m,n>N_{2}$.
Thus there is an $n_{2}>N_{2}$ such that $\left\Vert h_{n_{2}}\right\Vert _{L^2\left(
Q\right)  }^2>\varepsilon$ and $\left\Vert h_{n}-h_{n_{2}}\right\Vert
_{L^2\left(Q\right)  }^2<\delta_{2}$ for all $n>N_{2}$. In this way, we can find a sequence
$\left\{  h_{n_{k}}\right\}  _{k}$ with 
$h_{n_{k}}=\left(  -1\right)
^{d}\cdot1_{B_{n_{k}}\times\left\{  d\right\} \times\mathbb{R} }$, 
$\left\Vert h_{n_k}\right\Vert _{L^2\left(Q\right)  }^2>\varepsilon$, $\left\Vert h_{n}-h_{n_{k}}\right\Vert _{L^2\left(Q\right)  }^2<\delta_{k}$ for all $n>n_{k}$, and $\delta_{k}\downarrow0$. Let $B^{\infty}={\cup_{j=1}^{\infty}\cap
_{k=j}^{\infty}B_{n_{k}}}$. For every $K$, $\left\Vert h_{n_{k}}-h_{n_{K}%
}\right\Vert _{L^2\left(Q\right)  }^2<\delta_{K}$ for all $k>K$. Notice that for every $K^{\prime}>K$,
\begin{align*}
&\Vert h_{n_{K}}-\left(  -1\right)  ^{d}\cdot 1_{(\cap_{k=K^{\prime}}^{\infty}B_{n_{k}}) \times\left\{
d\right\} \times\mathbb{R} }\Vert _{L^2\left( Q\right)  }^2 
=\int \vert 1_{B_{n_{K}} \times\left\{
	d\right\} \times\mathbb{R} }-1_{(\cap_{k=K^{\prime}}^{\infty}B_{n_{k}}) \times\left\{
	d\right\} \times\mathbb{R} }\vert^2 \,\mathrm{d}Q\\
=&\int 1_{B_{n_{K}}\setminus(\cap_{k=K^{\prime}}^{\infty}B_{n_{k}}) \times\left\{
	d\right\} \times\mathbb{R} } \,\mathrm{d}Q+\int 1_{(\cap_{k=K^{\prime}}^{\infty}B_{n_{k}})\setminus B_{n_{K}} \times\left\{
	d\right\} \times\mathbb{R} } \,\mathrm{d}Q.
\end{align*}
Because $B_{n_{k}}$ is a closed interval for all $k$, we have that for every $K^{\prime\prime}\ge K^{\prime}$, there exist $L_1$ and $L_2$ with $K^{\prime}\le L_1 \le L_2 \le K^{\prime\prime}$ such that $\cup_{k=K^{\prime}}^{K^{\prime\prime}}(B_{n_{K}}\setminus B_{n_{k}})=(B_{n_K}\setminus B_{n_{L_1}})\cup (B_{n_K}\setminus B_{n_{L_2}})$. 
Then since 
\begin{align*}
\left\Vert h_{n_{k}}-h_{n_{K}
}\right\Vert _{L^2\left(Q\right)  }^2= Q(B_{n_{K}}\setminus B_{n_{k}} \times\left\{
d\right\} \times\mathbb{R})+Q(B_{n_{k}}\setminus B_{n_{K}} \times\left\{
d\right\} \times\mathbb{R})<\delta_{K}
\end{align*}
for all $k>K$, we have
\begin{align*}
\int 1_{B_{n_{K}}\setminus(\cap_{k=K^{\prime}}^{\infty}B_{n_{k}}) \times\left\{
	d\right\} \times\mathbb{R} } \,\mathrm{d}Q=&\, Q(B_{n_{K}}\setminus(\cap_{k=K^{\prime}}^{\infty}B_{n_{k}}) \times\left\{
d\right\} \times\mathbb{R})\\
=&\, Q(\cup_{k=K^{\prime}}^{\infty}(B_{n_{K}}\setminus B_{n_{k}} )\times\left\{
d\right\} \times\mathbb{R})
\le 2\delta_{K}.
\end{align*}
Similarly, it is easy to show that 
$\int 1_{(\cap_{k=K^{\prime}}^{\infty}B_{n_{k}})\setminus B_{n_{K}} \times\left\{
	d\right\} \times\mathbb{R} } \,\mathrm{d}Q\le2\delta_{K}$.
Thus it follows that
\begin{align*}
\Vert h_{n_{K}}-\left(  -1\right)  ^{d}\cdot 1_{(\cap_{k=K^{\prime}}^{\infty}B_{n_{k}}) \times\left\{
	d\right\} \times\mathbb{R} }\Vert _{L^2\left( Q\right)  }^2 
\le 4\delta_{K},
\end{align*}
which is true for all $K^{\prime}>K$. Letting $K^{\prime}\to\infty$, by the dominated convergence theorem ($B^{\infty}={\cup_{j=1}^{\infty}\cap
_{k=j}^{\infty}B_{n_{k}}}$) we have 
\[
\Vert h_{n_{K}}-\left(  -1\right)  ^{d}\cdot1_{B^{\infty}\times\left\{
d\right\} \times\mathbb{R} }\Vert _{L^2\left(Q\right)  }^2\le 4\delta_K.
\]
This implies that 
$\Vert h_{n_{K}}-\left(  -1\right)  ^{d}\cdot1_{B^{\infty}\times\left\{
d\right\} \times\mathbb{R} }\Vert _{L^2\left(Q\right)  }\rightarrow0\text{ as
}K\rightarrow\infty$,
because $\delta_{K}\downarrow 0$. Finally, we have
\begin{align*}
\Vert h_{n}-\left(  -1\right)  ^{d}\cdot1_{B^{\infty}\times\left\{
d\right\} \times\mathbb{R} }\Vert _{L^2\left(Q\right)  } \le \Vert h_{n}-h_{n_{K}}  \Vert_{L^2(Q)} + \Vert h_{n_{K}}-\left(  -1\right)  ^{d}\cdot1_{B^{\infty}\times\left\{
d\right\} \times\mathbb{R} }\Vert _{L^2\left(Q\right)  }\rightarrow 0.
\end{align*}
Clearly, $B^\infty$ can be a closed, open, or half-closed interval in $\mathbb{R}$. Also, every element of $\bar{\mathcal{H}}_{1d}$ is equal to the limit of a sequence of elements of $\mathcal{H}_{1d}$ under the $L^2(Q)$ norm. Thus the closure of $\mathcal{H}_{1d}$ in $L^2(Q)$ is equal to $\bar{\mathcal{H}}_{1d}$ for every $Q\in\mathcal{P}$.
Similarly, we can show that the closure of $\mathcal{H}_2$ in $L^2(Q)$ is equal to $\bar{\mathcal{H}}_2$ for every $Q\in\mathcal{P}$. As a result, the closure of $\mathcal{H}=\mathcal{H}_1\cup \mathcal{H}_2$ in $L^2(Q)$ is equal to $\bar{\mathcal{H}}=\bar{\mathcal{H}}_1\cup\bar{\mathcal{H}}_2$ for every $Q\in\mathcal{P}$.
\end{proof}

\begin{lemma}
\label{lemma.VC class}Let $\mathcal{H}_1$ and $\mathcal{H}_2$ be defined as in \eqref{def.function spaces}. Then $\mathcal{H}_1$ is a VC class\footnote{See the definition of VC class of functions in \citet[p.~141]{van1996weak}.} with VC index $V\left(
\mathcal{H}_1\right)  =3$, and $\mathcal{H}_2$ is a VC class with VC index $V\left(
\mathcal{H}_2\right)  =2$.
\end{lemma}

\begin{proof} [Proof of Lemma \ref{lemma.VC class}] 
All the functions $h\in\mathcal{H}_1$ take the form $h=-1_{B\times\left\{
1\right\}\times\mathbb{R}  }$ or $h=1_{B\times\left\{  0\right\} \times\mathbb{R} }$, where $B$ is a closed
interval. If $h=-1_{B\times\left\{  1\right\} \times\mathbb{R} }$, the subgraph of $h$ is
\[
C_{1B}=\left\{  \left(  y,w,z,t\right)  \subset \mathbb{R}^4
:t<-1_{B\times\left\{  1\right\}\times\mathbb{R}  }\left(  y,w,z\right)  \right\}  .
\]
If $h=1_{B\times\left\{  0\right\} \times\mathbb{R} }$, the subgraph of $h$ is
\[
C_{0B}=\left\{  \left(  y,w,z,t\right)  \subset\mathbb{R}^4:t<1_{B\times\left\{  0\right\} \times\mathbb{R} }\left(  y,w,z\right)  \right\}  .
\]
Let $\mathcal{C=}\left\{  C_{dB}:B\text{ is a closed interval in
}\mathbb{R}, d\in\left\{0,1\right\}\right\}.$

Suppose there are two different points $a_{1}=\left(  y_{1},w_{1},z_1,t_{1}\right)
,a_{2}=\left(  y_{2},w_{2},z_2,t_{2}\right)  \in\mathbb{R}^4$ with $y_{1}<y_{2}$, $w_{1}=w_{2}=0$, and $0\leq t_{1},t_{2}<1$. Then there
is a point $\bar{y}\in\left(  y_{1},y_{2}\right)  $. Let $B_{0}=\{\bar{y}\}  $, $B_{1}=\left[
y_{1},\bar{y}\right]  $, $B_{2}=\left[  \bar{y},y_{2}\right]  $, and
$B_{3}=\left[  y_{1},y_{2}\right]  $. Now we have
$\varnothing= C_{0B_{0}}\cap\left\{  a_{1},a_{2}\right\}$, $\left\{  a_{1}\right\}  =C_{0B_{1}}\cap\left\{  a_{1},a_{2}\right\}$, $\left\{
a_{2}\right\}  =C_{0B_{2}}\cap\left\{  a_{1},a_{2}\right\}  $, and $\left\{
a_{1},a_{2}\right\}  =C_{0B_{3}}\cap\left\{  a_{1},a_{2}\right\}$.
Thus $\mathcal{C}$ shatters $\{a_1,a_2\}$.

Suppose now there are three different points  $a_{1}=\left(  y_{1},w_{1},z_1,t_{1}\right)$, $a_{2}=\left(  y_{2},w_{2},z_2,t_{2}\right)$, $a_{3}=\left(  y_{3},w_{3},z_3,t_{3}\right)$ in $\mathbb{R}^4$. Without loss of generality, suppose
$t_{1}\leq t_{2}\leq t_{3}<1$, so that it is possible for $\mathcal{C}$ to pick out $\{a_j\}$ for each $j\in\{1,2,3\}$. 
\begin{enumerate}[label=(\arabic*)]
	\item Suppose $t_{1}\geq0$. In this case, we need $w_{1}=w_{2}=w_{3}=0$ in order to pick out $\left\{  a_{j}\right\}  $ for each $j$. Without loss of
	generality, suppose $y_{1}\leq y_{2}\leq y_{3}$. If we want $\mathcal{C}$ to pick out
	$\left\{  a_{1},a_{3}\right\}  $, we need to find a closed interval $B$ such
	that $y_{1},y_{3}\in B$, in which case $a_{1},a_{3}\in C_{0B}$. However, $a_{2}\in
	C_{0B}$ for all such $B$.
	
	\item Suppose $t_{1}<0$, $t_{2}\geq0$. Then we need $w_{2}=w_{3}=0$ in order to pick out
	$\left\{  a_{j}\right\}  $ for each $j\in\{2,3\}$ by using $C_{0B}$ for some closed
	interval $B$. But in this case, $\mathcal{C}$ can never pick out $\left\{
	a_{2}\right\}  $, $\left\{  a_{3}\right\}  $, or $\left\{  a_{2},a_{3}\right\}  $,
	since for every closed interval $B$, $a_{1}\in C_{0B}$.
	
	\item Suppose $t_{1},t_{2}<0$, $t_{3}\geq0$. Then we need $w_{3}=0$ in order to pick out
	$\left\{  a_{3}\right\}  $ by using $C_{0B}$ for some closed interval $B$. In
	this case, $\mathcal{C}$ can never pick out $\{a_{3}\}$, since for every closed
	interval $B$, $a_{1},a_{2}\in C_{0B}$.
	
	\item Suppose $t_{1},t_{2},t_{3}<0$. Then for every closed interval $B$, $a_{1},a_{2},a_{3}\in
	C_{0B}$. If we want $\mathcal{C}$ to pick out 
	$\left\{  a_{j},a_{j^{\prime}}\right\}  $ for all $j\neq j^{\prime
	}$, we need to use $C_{1B}$. If $w_{j}\neq1$, then for every $B$, $a_{j}\in
	C_{1B}$. Thus we consider $w_{1}=w_{2}=w_{3}=1$.
	\begin{enumerate}
		\item Suppose $-1\leq t_{1},t_{2},t_{3}<0$. Without loss of generality, we assume that
		$y_{1}\leq y_{2}\leq y_{3}$. But now if we want $\mathcal{C}$ to pick out $\left\{
		a_{2}\right\}  $, we need to find a closed interval $B$ such that $y_{1}%
		,y_{3}\in B$ but $y_{2}\not \in B$, which is not possible.
		
		\item Suppose $t_{j}<-1$ for some $j\in\{1,2,3\}$. In this case, $a_{j}\in C_{1B}$ for every
		closed interval $B$. 
	\end{enumerate}

\end{enumerate}

Therefore, we conclude that $\mathcal{H}_1$ is a VC class
with VC index $V\left(  \mathcal{H}_1\right)  =3$. Similarly, we can show that $\mathcal{H}_2$ is a VC class
with VC index $V\left(  \mathcal{H}_2\right)  =2$.
\end{proof}

\begin{lemma}\label{lemma.totally bounded H}
	Let ${\mathcal{H}}$ be defined as in \eqref{def.function spaces}. Then ${\mathcal{H}}$ is totally bounded under $\left\Vert \cdot\right\Vert _{L^{r}\left( Q\right)  }$ for every probability measure $Q\in\mathcal{P}$ and every $r\ge 1$.
\end{lemma}

\begin{proof}[Proof of Lemma \ref{lemma.totally bounded H}]
Let $N\left(  \varepsilon, \mathcal{H}_j, L^{r}\left(Q\right)  \right)  $ denote
	the covering number under the $L^r(Q)$ norm for $\mathcal{H}_j$ for $j\in\{1,2\}$ and all $\varepsilon>0$, where $\mathcal{H}_j$ is defined as in \eqref{def.function spaces}. Since $\mathcal{H}_1$ and $\mathcal{H}_2$ are VC classes by Lemma \ref{lemma.VC class} with $V(\mathcal{H}_1)=3$ and $V(\mathcal{H}_2)=2$, by Theorem 2.6.7 of \citet{van1996weak} with envelope
	function $F=1$ and $r\ge 1$ we have that for every
	probability measure $Q$,
	\begin{align*}
	N\left(  \varepsilon,\mathcal{H}_1,L^{r}\left(Q\right)  \right)  \leq  K_1 3  \left(  16e\right)  ^{3  }\left(1/\varepsilon\right)  ^{2r  } \text{ and }
	N\left(  \varepsilon,\mathcal{H}_2,L^{r}\left(Q\right)  \right)  \leq  K_2 2  \left(  16e\right)  ^{2  }\left(1/\varepsilon\right)  ^{r  }
	\end{align*}
	for universal constants $K_1,K_2\ge1$ and every $\varepsilon\in(0,1)$. Since $\mathcal{H}=\mathcal{H}_1\cup\mathcal{H}_2$, we have
	\begin{align}\label{eq.H totally bounded}
	N\left(  \varepsilon,\mathcal{H},L^{r}\left(Q\right)  \right)  \leq & K_1 3 \left(  16e\right)  ^{3  }\left(
	1/\varepsilon\right)  ^{2r  }+ K_2 2 \left(  16e\right)  ^{2  }\left(
	1/\varepsilon\right)  ^{r  },
	\end{align}
	which implies that $\mathcal{H}$ is totally bounded. 
\end{proof}

\begin{lemma}\label{lemma.totally bounded H_hat}
 Let $\bar{\mathcal{H}}$ be as in \eqref{def.function spaces}. Then $\bar{\mathcal{H}}$ is compact under $\left\Vert \cdot\right\Vert _{L^2\left(Q\right)  }$ for every $Q\in\mathcal{P}$.
\end{lemma}

\begin{proof}[Proof of Lemma \ref{lemma.totally bounded H_hat}]
By Lemma \ref{lemma.totally bounded H}, ${\mathcal{H}}$ is totally bounded under $\left\Vert \cdot\right\Vert _{L^2\left(Q\right)  }$ for all $Q\in\mathcal{P}$. Suppose that $\mathcal{H}\subset \bigcup_{j\in J}B_{\varepsilon/2}(h_j)$, where $J$ is a finite index set and $B_{\varepsilon/2}(h_j)$ is an open ball with center $h_j$ and radius $\varepsilon/2$ under $\Vert\cdot\Vert_{L^2(Q)}$.  By Lemma \ref{lemma.H complete}, $\bar{\mathcal{H}}$ is equal to the closure of $\mathcal{H}$ in $L^2(Q)$. Clearly, 
$\bar{\mathcal{H}}\subset \bigcup_{j\in J}\overline{B_{\varepsilon/2}(h_j)}\subset \bigcup_{j\in J}B_{\varepsilon}(h_j)$, and therefore
\begin{align}\label{eq.H_bar totally bounded}
N(\varepsilon,\bar{\mathcal{H}},L^2(Q))\le N(\varepsilon/2,\mathcal{H},L^2(Q)),
\end{align}
which, together with \eqref{eq.H totally bounded}, implies that $\bar{\mathcal{H}}$ is totally bounded. Since $L^2(Q)$ is complete, $\bar{\mathcal{H}}$ is compact in $L^2(Q)$.
\end{proof}

Let $\bar{\mathcal{H}}$ and $\mathcal{G}_{K}$ be
defined as in \eqref{def.function spaces}. Let $\mathcal{V}=\left\{h\cdot f:h\in\mathcal{\bar{H}},f\in\mathcal{G}_{K}\right\}  $. Then define
\begin{align}\label{eq.tilde V}
\tilde{\mathcal{V}}=\mathcal{V}\cup\mathcal{ G}_K.
\end{align}

\begin{lemma}
	\label{lemma.V Donsker}  The function space $\tilde{\mathcal{V}}$ is Donsker and pre-Gaussian uniformly in $Q\in\mathcal{P}$.
\end{lemma}

\begin{proof}[Proof of Lemma \ref{lemma.V Donsker}] 
	For every $\delta>0$ and
	every $Q\in\mathcal{P}$, define
	\begin{align*}
	\tilde{\mathcal{V}}_{\delta,Q}=   \left\{  v-v^{\prime}:v,v^{\prime}\in
	\tilde{\mathcal{V}},\left\Vert v-v^{\prime}\right\Vert _{L^{2}\left(Q\right)
	}<\delta\right\}  \text{ and }
	\tilde{\mathcal{V}}_{\infty}^{2}=   \left\{  \left(  v-v^{\prime}\right)
	^{2}:v,v^{\prime}\in\tilde{\mathcal{V}}\right\}  .
	\end{align*}
	First, we show that $\tilde{\mathcal{V}}_{\delta,Q}$ is $Q$-measurable\footnote{See Definition 2.3.3 of $Q$-measurable class in \citet{van1996weak}.} for all
	$Q\in\mathcal{P}$. Similar to the construction of $\mathcal{H}$, we
	construct function spaces by
	\begin{align*}
	\mathcal{H}_{q1}= &  \left\{
	\left(  -1\right)  ^{d}\cdot1_{B\times\left\{  d\right\}  \times\mathbb{R}	}:
	B=\left[  a,b\right],a,b\in\mathbb{Q},a\le b, d\in\left\{  0,1\right\} \right\} ,\\
	\mathcal{H}_{q2}= &  \left\{  1_{\mathbb{R}\times C\times\mathbb{R}%
	}:C=\left(  -\infty,c\right] 
	,c\in\mathbb{Q}\right\},\text{ and }\mathcal{H}_{q}=  \mathcal{H}_{q1}\cup\mathcal{H}_{q2},
	\end{align*}
	where $\mathbb{Q}$ denotes the set of all rational numbers. Now define
	\begin{align*}
	\tilde{\mathcal{V}}_{q}=\left\{  h\cdot f:h\in\mathcal{H}_{q},f\in\mathcal{G}
	_{K}\right\}\cup\mathcal{ G}_K \text{ and }
	\tilde{\mathcal{V}}_{q\delta,Q}=\left\{  v-v^{\prime}:v,v^{\prime}\in\tilde{\mathcal{V}}_{q},\left\Vert v-v^{\prime}\right\Vert _{L^{2}\left(Q\right)  }
	<\delta\right\}  .
	\end{align*}
	By construction, $\mathcal{ G}_K$ is a finite set. Since $\mathbb{Q}$ is countable (and therefore the set of ordered pairs of elements of $\mathbb{Q}$ is countable), $\mathcal{H}_{q1}$ and $\mathcal{H}%
	_{q2}$ are countable (and therefore $\mathcal{H}_{q}$ and $\tilde{\mathcal{V}}_{q}$ are
	countable). 
	
	Clearly, $\tilde{\mathcal{V}}_{q\delta,Q}$ is a countable subset of $\tilde{\mathcal{V}}
	_{\delta,Q}$. For every $v\in\tilde{\mathcal{V}}$, there is a sequence $\{v_{m}\}
	\subset\tilde{\mathcal{V}}_{q}$ such that $v_{m} \rightarrow v  $ pointwise, because $\mathbb{Q}$ is dense in $\mathbb{R}$. For
	example, if $v=(-1)^{d}\cdot1_{\left(  \sqrt{2},\sqrt{3}\right]
		\times\{d\}\times\mathbb{R}}\cdot1_{\mathbb{R\times R\times}\left\{
		z_{k}\right\}  }$, we can find $v_{m}=(-1)^{d}\cdot1_{[a_{m},b_{m}
		]\times\{d\}\times\mathbb{R}}\cdot1_{\mathbb{R\times R\times}\left\{
		z_{k}\right\}  }$ with $a_{m}\downarrow\sqrt{2}$,
	$b_{m}\downarrow\sqrt{3}$, and $a_{m},b_{m}\in\mathbb{Q}$. Suppose
	$v-v^{\prime}\in\tilde{\mathcal{V}}_{\delta,Q}$ and $v_{m},v_{m}^{\prime}%
	\in\tilde{\mathcal{V}}_{q}$ such that $v_{m}  \rightarrow v  $ and $v_{m}^{\prime}  \rightarrow v^{\prime}$ pointwise. It is easy to show that $\left\Vert v_{m}-v_{m}
	^{\prime}\right\Vert _{L^{2}\left(Q\right)  }<\delta$ for large $m$, that is,
	$v_{m}-v_{m}^{\prime}\in\tilde{\mathcal{V}}_{q\delta,Q}$ for large $m$. By Example
	2.3.4 of \citet{van1996weak}, $\tilde{\mathcal{V}}_{\delta,Q}$ is $Q$-measurable, and
	this is true for all $\delta>0$. Similarly, $\tilde{\mathcal{V}}_{\infty}^{2}$ is $Q$-measurable.
	
	By the construction of $\tilde{\mathcal{V}}$, $F=1$ is a measurable envelope function
	with $\int F^{2}\,\mathrm{d}Q<\infty.$ Also, $\lim_{M\rightarrow\infty}\sup_{Q\in\mathcal{P}}\int F^2\cdot1\left\{F>M\right\}  \,\mathrm{d}Q=0$.
	For all $H\in\mathcal{P}$ and all 
	$\varepsilon\geq2$,
	\begin{align}\label{eq.V covering number eps>1}
	N\left(  \varepsilon\left\Vert F\right\Vert _{L^{2}\left(  H\right)
	},\tilde{\mathcal{V}},L^{2}\left(  H\right)  \right)  =N\left(  \varepsilon
	,\tilde{\mathcal{V}},L^{2}\left(  H\right)  \right)  =1.
	\end{align}
	For all $H\in\mathcal{P}$ and all $\varepsilon>0$, 
	\begin{align}\label{eq.V covering number eps<1}
	N\left(  \varepsilon,\mathcal{V},L^{2}\left(  H\right)  \right)  \leq
	N\left(  \frac{\varepsilon}{2},\mathcal{\bar{H}},L^{2}\left(  H\right)
	\right)  \cdot N\left(  \frac{\varepsilon}{2},\mathcal{G}_{K},L^{2}\left(
	H\right)  \right)  
	 \le K\cdot N\left(  \frac{\varepsilon}{2},\mathcal{\bar{H}},L^{2}\left(
	H\right)  \right),
	\end{align}
	where $K$ is the number of elements in $\mathcal{G}_K$.
	Thus by the definition of $\tilde{\mathcal{V}}$ in \eqref{eq.tilde V}, 
	\begin{align}\label{eq.covering number tilde V}
	N\left(  \varepsilon,\tilde{\mathcal{V}},L^{2}\left(  H\right)  \right)  \leq K\cdot N\left(  \frac{\varepsilon}{2},\mathcal{\bar{H}},L^{2}\left(
	H\right)  \right)+K
	\end{align}
	for all $H\in\mathcal{P}$ and all $\varepsilon>0$.
	Let $\mathcal{Q}$ denote the set of finitely
	discrete probability measures.
	The results in \eqref{eq.H totally bounded}, \eqref{eq.H_bar totally bounded}, \eqref{eq.V covering number eps>1}, and \eqref{eq.covering number tilde V}  imply that
	\begin{align*}
	&  \int_{0}^{\infty}\sup_{H\in\mathcal{Q}}\sqrt{\log N\left(  \varepsilon
		\left\Vert F\right\Vert _{L^{2}\left(  H\right)  },\tilde{\mathcal{V}},L^{2}\left(
		H\right)  \right)  }\,\mathrm{d}\varepsilon = \int_{0}^{2}\sup_{H\in\mathcal{Q}}
	\sqrt{\log N\left(  \varepsilon,\tilde{\mathcal{V}},L^{2}\left(  H\right)  \right)  }\,\mathrm{d}\varepsilon\\
	\leq &  \int_{0}^{2}\sqrt{\log\left\{  K\cdot(K_{1}+K_{2})\cdot3\cdot\left(
		16e\right)  ^{3}\left(  {4}/{\varepsilon}\right)  ^{4}+K\right\}
	}\,\mathrm{d}\varepsilon<\infty.
	\end{align*}
	The claim of the Lemma
	follows from Theorem 2.8.3 of \citet{van1996weak}.
\end{proof}

\begin{lemma}
\label{lemma.uniform Glivenko-Cantelli}The function space $\tilde{\mathcal{V}}$ defined in \eqref{eq.tilde V} is Glivenko--Cantelli
uniformly in $Q\in\mathcal{P}$.
\end{lemma}

\begin{proof}[Proof of Lemma \ref{lemma.uniform Glivenko-Cantelli}] 
Similar to the proof of Lemma \ref{lemma.V Donsker}, we can show that $\tilde{\mathcal{V}}$
is $Q$-measurable for every $Q\in\mathcal{P}$. With $F=1$ being an
envelope function of $\tilde{\mathcal{V}}$, we have $\lim_{M\rightarrow\infty}\sup_{Q\in\mathcal{P}}\int F\cdot1\left\{
F>M\right\}  \,\mathrm{d}Q=0$. Similar to the proofs of Lemmas \ref{lemma.H complete}, \ref{lemma.totally bounded H_hat}, and \ref{lemma.V Donsker}, we can show that for every $Q\in\mathcal{P}$ and every $\varepsilon>0$, the closure of $\mathcal{H}$ in ${L^1(Q)}$ is equal to $\bar{\mathcal{H}}$, $N(\varepsilon,\bar{\mathcal{H}},L^1(Q))\le N(\varepsilon/2,\mathcal{H},L^1(Q))$, and $N(  \varepsilon,\tilde{\mathcal{V}},L^{1}\left(  Q\right)  )  \leq K\cdot N\left(  {\varepsilon}/{2},\mathcal{\bar{H}},L^{1}\left(
Q\right)  \right)+K$.
Then by \eqref{eq.H totally bounded}, we can show that
$\sup_{H\in\mathcal{Q}_{n}}\log N(  \varepsilon\left\Vert F\right\Vert
_{L^{1}\left(  H\right)  },\tilde{\mathcal{V}},L^{1}\left(  H\right) )
=o\left(  n\right)$ with the envelope function $F=1$,
where $\mathcal{Q}_{n}$ is the collection of all possible realizations of
empirical measures of $n$ observations. Then by Theorem 2.8.1 in
\citet{van1996weak}, $\tilde{\mathcal{V}}$ is Glivenko--Cantelli uniformly in
$Q\in\mathcal{P}$.
\end{proof}

\begin{lemma}
\label{lemma.totally bounded HG}Let $\mathcal{H}$ and $\mathcal{G}$ be defined as in \eqref{def.function spaces}, let $\rho_{P}$ be as in \eqref{eq.rho}, and define $\overline{\mathcal{H}\times\mathcal{G}}$ as the closure of $\mathcal{H}\times\mathcal{G}$ in $L^2(P)\times(L^2(P)\times L^2(P))$ under $\rho_{P}$. Then $N\left(  \varepsilon,\overline{\mathcal{H}	\times\mathcal{G}},\rho_{P}\right)  =O\left(  {1}/{\varepsilon^{4}}\right)
$ as $\varepsilon\rightarrow0$.
\end{lemma}

\begin{proof}[Proof of Lemma \ref{lemma.totally bounded HG}] 
By the constructions of $\mathcal{H}\times\mathcal{G}$ and the metric $\rho
_{P}$,
\[
N\left(  \varepsilon,\mathcal{H}\times\mathcal{G},\rho_{P}\right)  \leq
  N\left(  \frac{\varepsilon}{3},\mathcal{H},L^2\left(
P\right)  \right) \cdot \left[N\left(  \frac{\varepsilon}{3},\mathcal{G}_{K},L^2\left(
P\right)  \right)  \right]^2,
\]
where $\mathcal{G}_K$ is defined as in \eqref{def.function spaces}.
By the construction of $\mathcal{ G}_K$, $N\left(  \varepsilon/3,\mathcal{G}_{K},L^{2}\left(  P\right)  \right) \leq K$,
where $K$ is the number of elements in $\mathcal{ G}_K$. This, together with \eqref{eq.H totally bounded}, implies that
$
N\left(  \varepsilon,\mathcal{H}\times\mathcal{G},\rho_{P}\right)  =O\left(
1/\varepsilon^{4}\right)  \text{ as }\varepsilon\rightarrow0.
$
Similar to \eqref{eq.H_bar totally bounded}, 
\[
N\left(  \varepsilon,\overline{\mathcal{H}	\times\mathcal{G}},\rho_{P}\right) \le
N\left(  \frac{\varepsilon}{2},\mathcal{H}\times\mathcal{G},\rho_{P}\right)  =O\left(
\frac{1}{\varepsilon^{4}}\right)  \text{ as }\varepsilon\rightarrow0.
\]
\end{proof}

\begin{lemma}
\label{lemma.complete HG}Let $\mathcal{H}$ and $\mathcal{G}$ be defined as in \eqref{def.function spaces}, and let $\rho_{P}$ be as in \eqref{eq.rho}. Then $\overline{\mathcal{H}\times\mathcal{G}}$, the closure of $\mathcal{H}\times\mathcal{G}$ under $\rho_{P}$ in Lemma \ref{lemma.totally bounded HG}, is compact and
$\overline{\mathcal{H}\times\mathcal{G}}=\bar{\mathcal{H}}\times{\mathcal{G}}$, 
where $\bar{\mathcal{H}}$ is defined as in \eqref{def.function spaces}.
\end{lemma}
\begin{proof}[Proof of Lemma \ref{lemma.complete HG}] 
The first claim follows from Lemma \ref{lemma.totally bounded HG} and the fact that $L^2(P)\times(L^2(P)\times L^2(P))$ is complete under $\rho_P$. The second claim holds by the constructions of $\rho_{P}$ and $\mathcal{G}$. 
\end{proof}

\begin{proof}[Proof of Lemma \ref{ass.testable implication multivalue}]
	Suppose Assumption
	\ref{ass.IV validity for multivalued Z} holds with $\mathcal{D}=\{d_1,d_2,\ldots\}$. 
	Then we can define $Y_d$ by
	$Y_{d}=Y_{dz_{1}}=Y_{dz_{2}}=\cdots=Y_{dz_{K}}$ almost surely
	for all $d\in\mathcal{D}$.
	First, suppose $d_{\max}$ exists.  
	Under Assumption \ref{ass.IV validity for multivalued Z}, for all $k$ with $1\le k \le K-1$ and all Borel sets $B$,
	\begin{align*}
		&\mathbb{P}\left(  Y\in B,D=d_{\max}|Z=z_k\right) = \mathbb{P}\left(  Y_{d_{\max}}\in B,D_{z_{k}}=d_{\max}\right)  \\
		=&\sum
		_{j}\mathbb{P}\left(  Y_{d_{\max}}\in B,D_{z_{k}}=d_{\max},D_{z_{k+1}}
		=d_{j}\right)  =\mathbb{P}\left(  Y_{d_{\max}}\in B,D_{z_{k}}=d_{\max},D_{z_{k+1}}
		=d_{\max}\right)
	\end{align*}
	and
	\begin{align*}
		\mathbb{P}\left(  Y\in B,D=d_{\max}|Z=z_{k+1}\right) &=\mathbb{P}\left(  Y_{d_{\max}}\in B,D_{z_{k+1}}=d_{\max}\right) \\
		& =\sum
		_{j}\mathbb{P}\left(  Y_{d_{\max}}\in B,D_{z_{k}}=d_{j},D_{z_{k+1}}=d_{\max
		}\right). 
	\end{align*}
	Thus
	$  \mathbb{P}\left(  Y\in B,D=d_{\max}|Z=z_{k+1}\right)
	\ge \mathbb{P}\left(  Y\in B,D=d_{\max}|Z=z_{k}\right) $.
	Second, suppose $d_{\min}$ exists. Then similarly,
	$  \mathbb{P}\left(  Y\in B,D=d_{\min}|Z=z_k\right)
	\ge \mathbb{P}\left(  Y\in B,D=d_{\min}|Z=z_{k+1}\right)$.
\end{proof}

\begin{remark}
	Lemma \ref{lemma.testable implication conditioning ordered} can be proved analogously. The proofs of Lemmas \ref{lemma.testable implication unorder} and \ref{lemma.testable implication conditioning unordered} are trivial. 
\end{remark}

%\begin{remark}
%	Lemmas \ref{lemma.testable implication unorder}, \ref{lemma.testable implication conditioning ordered}, and \ref{lemma.testable implication conditioning unordered} can be proved analogously.
%\end{remark}

%\section{Main Results}

%\section{Results in Section \ref{sec.multi D and Z}}

\begin{lemma}\label{lemma.L HD}
	Let $\mathbb{D}_{\mathcal{L}}=\{R\in\ell^{\infty}(\tilde{\mathcal{V}} ): R(h\cdot g_l)/R(g_l) \text{ exists for all }h\in\bar{\mathcal{H}}\text{ and all }g_l\in\mathcal{G}_K\}$. Define $\mathcal{L}:\mathbb{D}_{\mathcal{L}}\subset\ell^{\infty}(\tilde{\mathcal{V}})\rightarrow\ell^{\infty}\left(
	\bar{\mathcal{H}}\times\mathcal{G}\right)  $ by
	\[
	\mathcal{L}\left(  R\right)  \left(  h,g\right)  
	=\frac{R\left(  h\cdot g_{2}\right)  }{R\left(  g_{2}\right)  }-\frac{R\left(
		h\cdot g_{1}\right)  }{R\left(  g_{1}\right)  }
	\]
	for all $R\in \mathbb{D}_{\mathcal{L}}$ and all $(h,g)\in \bar{\mathcal{H}}\times\mathcal{G}$ with $g=(g_1,g_2)$.
	Then $\mathcal{L}$ is uniformly Hadamard differentiable\footnote{See the definitions of Hadamard differentiability and uniform Hadamard differentiability in \citet[pp.~372--375]{van1996weak}.} along every sequence $P_n \to P$ in $\mathbb{ D}_{\mathcal{ L}}$, tangentially to $\ell^{\infty}(\tilde{\mathcal{V}})$,  with the derivative $\mathcal{L}_P^{\prime}$ defined by 
	\begin{align}\label{eq.L HDD}
	\mathcal{L}_{P}^{\prime}\left(  H\right)  \left(  h,g\right)  =\frac{H\left(
		h\cdot g_{2}\right)  P\left(  g_{2}\right)  -P\left(  h\cdot g_{2}\right)
		H\left(  g_{2}\right)  }{P^{2}\left(  g_{2}\right)  }-\frac{H\left(  h\cdot
		g_{1}\right)  P\left(  g_{1}\right)  -P\left(  h\cdot g_{1}\right)  H\left(
		g_{1}\right)  }{P^{2}\left(  g_{1}\right)  }
	\end{align}
	for all $H\in\ell^{\infty}(\tilde{\mathcal{V}})$.\footnote{By \eqref{eq.0timesinfinity}, $\mathcal{L}_{P}^{\prime}$ is well defined.}
\end{lemma}
\begin{remark}
	By the definition of $\mathcal{ L}$, $\mathcal{L}(Q)=\phi_Q$ for all $Q\in\mathcal{P}$. We will apply Lemma \ref{lemma.L HD} along with the suitable delta method to deduce the asymptotic distributions of $\sqrt{n}(\hat{\phi}_{P_n}-\phi_{P})$ and the bootstrap version of this random element. 
\end{remark}
\begin{proof}[Proof of Lemma \ref{lemma.L HD}]
	For all $t_n\to0$, $P_n\to P$, and $H_n\to H$ in $\ell^{\infty}(\tilde{\mathcal{V}})$ such that $P_n\in \mathbb{D}_{\mathcal{L}}$ and $P_n+t_n H_n\in \mathbb{D}_{\mathcal{L}}$, we have that for each $(h,g)\in \bar{\mathcal{H}}\times\mathcal{G}$ with $g=(g_1,g_2)$,
	\begin{align*}
	&  \mathcal{L}\left(  P_{n}+t_{n}H_{n}\right)  \left(  h,g\right)
	-\mathcal{L}\left(  P_{n}\right)  \left(  h,g\right)  \\
	=&\frac{t_{n}H_{n}\left(  h\cdot g_{2}\right)  P_{n}\left(  g_{2}\right)
		-t_{n}P_{n}\left(  h\cdot g_{2}\right)  H_{n}\left(  g_{2}\right)  }{\left(
		P_{n}+t_{n}H_{n}\right)  \left(  g_{2}\right)  P_{n}\left(  g_{2}\right)
	}-\frac{t_{n}H_{n}\left(  h\cdot g_{1}\right)  P_{n}\left(  g_{1}\right)
		-t_{n}P_{n}\left(  h\cdot g_{1}\right)  H_{n}\left(  g_{1}\right)  }{\left(
		P_{n}+t_{n}H_{n}\right)  \left(  g_{1}\right)  P_{n}\left(  g_{1}\right)  }.
	\end{align*}
	Thus it is easy to show that
	\begin{align*}
	  \lim_{n\rightarrow\infty}\sup_{\left(  h,g\right)  \in  \bar{\mathcal{H}}\times\mathcal{G}  }\left\vert
	\frac{\mathcal{L}\left(  P_{n}+t_{n}H_{n}\right)  \left(  h,g\right)
		-\mathcal{L}\left(  P_{n}\right)  \left(  h,g\right)  }{t_{n}}-\mathcal{L}_{P}^{\prime}\left(  H\right)  \left(  h,g\right)
	\right\vert =0,
	\end{align*}
	where $\mathcal{L}_{P}^{\prime}$ is defined as in \eqref{eq.L HDD}.
	This implies that $\mathcal{L}$ is uniformly differentiable and verifies the derivative in \eqref{eq.L HDD}.
\end{proof}

\begin{lemma}\label{lemma.weak convergence Pn_hat and convergence Pn}
	Under Assumptions \ref{ass.independent data} and \ref{ass.probability path} with $P_n,P\in\ell^{\infty}(\tilde{\mathcal{V}})$, we have  $\sup_{v\in\tilde{\mathcal{V}}}\vert\sqrt{n}( P_n -P)(v)-Q_0(v)\vert\to 0$, where  $Q_0(v)=P(vv_0)$ for all $v\in\tilde{\mathcal{V}}$ and $v_0$ is  as in Assumption \ref{ass.probability path}, and that $\sqrt{n}(  \hat{P}_{n}-P)$ converges under $P_n$ in distribution to the process $\mathbb{G}_{P}+Q_0$
	for a tight $P$-Brownian bridge $\mathbb{G}_{P}$ with $E[\mathbb{ G}_P(v_1)\mathbb{ G}_P(v_2)]=P(v_1v_2)-P(v_1)P(v_2)$ for all $v_1,v_2\in\tilde{\mathcal{V}}$.
\end{lemma}

\begin{proof}[Proof of Lemma \ref{lemma.weak convergence Pn_hat and convergence Pn}]
	The Lemma holds by Assumptions \ref{ass.independent data} and \ref{ass.probability path}, the facts that $\sup_{v\in \tilde{\mathcal{V}}}\left\vert P(v)\right\vert\le1$ and $\sup_{v\in\tilde{\mathcal{V}}}\vert P_n(v^2)\vert\le1$ for all $n$, Lemma \ref{lemma.V Donsker} in this paper, and Theorem 3.10.12 of \citet{van1996weak}.
\end{proof}

\begin{lemma}\label{lemma.almost uniform convergence phi sigma}
	Under Assumptions \ref{ass.independent data} and \ref{ass.probability path} with $P_n,P\in\ell^{\infty}(\tilde{\mathcal{V}})$, we have that $P_n\to P$ and that  $\hat{P}_n\to P$, $\hat{\phi}_{P_n}\to\phi_P$, $T_n/n \to \Lambda(P)$,  and $\hat{\sigma}_{P_n}\to \sigma_P$ almost uniformly.
\end{lemma}

\begin{proof}[Proof of Lemma \ref{lemma.almost uniform convergence phi sigma}]
By Lemma \ref{lemma.weak convergence Pn_hat and convergence Pn} in this paper, H\"older's inequality, and Lemma 3.10.11 of \citet{van1996weak}, we have that 
\begin{align*}
\Vert {P}_n - P\Vert_{\infty}\le& \Vert P_n -P -n^{-1/2}Q_0 \Vert_{\infty}+\Vert n^{-1/2}Q_0 \Vert_{\infty}\\
 \le &\Vert P_n -P -n^{-1/2}Q_0 \Vert_{\infty}+n^{-1/2}\sup_{v\in\tilde{\mathcal{V}}}\vert P(v^2)P(v_0^2) \vert^{1/2} \to 0,
\end{align*} 
where $Q_0$ is the function defined in Lemma \ref{lemma.weak convergence Pn_hat and convergence Pn}.
By Lemma \ref{lemma.uniform Glivenko-Cantelli} in this paper and Lemma 1.9.3 of \citet{van1996weak}, $\Vert \hat{P}_n - P_n\Vert_{\infty}\to 0$ almost uniformly. Then we have that $\Vert \hat{P}_n - P\Vert_{\infty}\to 0$ almost uniformly. The rest of the results follow from the constructions of $\hat{\phi}_{P_n}$, $T_n/n$, and $\hat{\sigma}_{P_n}$. By the construction of $\bar{\mathcal{H}}$, the $\sigma_{Q}^2(h,g)$ in \eqref{eq.stat variance multi} can also be written as
\begin{align}\label{eq.stat variance multi 2}
\sigma_{Q}^2  (h,g)=\Lambda(Q) \cdot \left\{\frac{ |Q\left(  h\cdot g_{2}\right)|   }{Q^{2}\left(
	g_{2}\right)  }  -\frac{ Q^2\left(  h\cdot g_{2}\right)   }{Q^{3}\left(
	g_{2}\right)  }  
+\frac{ |Q\left(	h\cdot g_{1}\right)|   }{Q^{2}\left(  g_{1}\right)  }
-\frac{ Q^2\left(	h\cdot g_{1}\right)   }{Q^{3}\left(  g_{1}\right)  }\right\}.
\end{align}
Similar to \eqref{eq.stat variance multi 2}, we can write the $\hat{\sigma}_{P_n}^{2}\left(  h,g\right)$ in \eqref{eq.estimated stat variance multi} as
\begin{align}\label{eq.estimated stat variance multi 2}
\hat{\sigma}_{P_n}^{2}\left(  h,g\right)  =\frac{T_n}{n}\cdot\left\{\frac{ |\hat{P}_n\left(  h\cdot g_{2}\right)|   }{\hat{P}_n^{2}\left(
	g_{2}\right)  }  -\frac{ \hat{P}_n^2\left(  h\cdot g_{2}\right)   }{\hat{P}_n^{3}\left(
	g_{2}\right)  }  
+\frac{ |\hat{P}_n\left(	h\cdot g_{1}\right) |  }{\hat{P}_n^{2}\left(  g_{1}\right)  }
-\frac{ \hat{P}_n^2\left(	h\cdot g_{1}\right)   }{\hat{P}_n^{3}\left(  g_{1}\right)  }\right\}.
\end{align}
Then the almost uniform convergence of $\hat{P}_n$ to $P$ in $\ell^{\infty}(\tilde{\mathcal{V}})$ implies the almost uniform convergence of the   $\hat{\sigma}_{P_n}^{2}$ in \eqref{eq.estimated stat variance multi 2} to the $\sigma_{P}^2$ as in \eqref{eq.stat variance multi 2}.
\end{proof}
%\begin{lemma}\label{lemma.convergence phi Pn}
%Define $Q_0\in\ell^{\infty}(\tilde{\mathcal{V}})$ such that $Q_0(v)=P(vv_0)$ for all $v\in\tilde{\mathcal{V}}$.	Under Assumption \ref{ass.probability path}, $\sqrt{n}(\phi_{P_n}-\phi_{P})\to \mathcal{L}_{P}^{\prime}(Q_0)$.
%\end{lemma}

%\begin{proof}[Proof of Lemma \ref{lemma.convergence phi Pn}]
%	Let $t_n=n^{-1/2}$ and $H_n=t_n^{-1}(P_n-P)$. By Lemma \ref{lemma.weak convergence Pn_hat and convergence Pn}, $H_n\to Q_0$. We can write $P_n=P+(P_n-P)=P+t_n H_n$. Remember $\phi_{P_n}=\mathcal{L}(P_n)$ and $\phi_{P}=\mathcal{L}(P)$. Then the lemma holds by Lemma \ref{lemma.L HD}. 
%\end{proof}

\begin{proof}[Proof of Lemma \ref{lemma.weak convergence phi_K}]
	By the Hadamard derivative of $\mathcal{L}$ in \eqref{eq.L HDD}, together with Lemma \ref{lemma.weak convergence Pn_hat and convergence Pn} in this paper and Theorem 3.9.4 (delta method) of \citet{van1996weak}, we have that under $P_n$,
	\begin{align}\label{eq.convergence phi}
		\sqrt{n}(  \hat{\phi}_{P_n}-\phi_{P})  =\sqrt{n}\{  \mathcal{L}(
		\hat{P}_{n})  -\mathcal{L}\left(  P\right)  \}  \leadsto\mathcal{L}_{P}^{\prime}\left(  \mathbb{G}_{P}+Q_0\right)  .
	\end{align}
	By Lemma \ref{lemma.almost uniform convergence phi sigma}, $T_n/n\to \Lambda(P)$ almost uniformly. Thus by Lemmas 1.9.3(ii)  and  1.10.2(iii), Example 1.4.7 (Slutsky's lemma), and Theorem 1.3.6 (continuous mapping) of \citet{van1996weak}, 
	\begin{align}\label{eq.weak convergence phi}
		\sqrt{T_n}(  \hat{\phi}_{P_n}-\phi_{P})=\sqrt{{T_n}/{n}}\cdot\sqrt{n}(  \hat{\phi}_{P_n}-\phi_{P})   \leadsto\Lambda(P)^{1/2}\mathcal{L}_{P}^{\prime}\left(  \mathbb{G}_{P}+Q_0\right)  .
	\end{align}
	Let $\mathbb{ G}=\Lambda(P)^{1/2}\mathcal{L}_{P}^{\prime}\left(  \mathbb{G}_{P}+Q_0\right)$. Then $\mathbb{ G}$ is tight, because $\mathbb{ G}_P$ is tight and $\mathcal{L}_{P}^{\prime}$ is a continuous map. Thus \eqref{eq.weak convergence phi} verifies the first claim of Lemma \ref{lemma.weak convergence phi_K}. Now we show the continuity of $\mathbb{ G}$ under $\rho_P$. Define a semimetric on $\tilde{\mathcal{V}}$ by 
	\begin{align*}
		\rho_2(v,v^{\prime})=E\left[ \vert \mathbb{G}_P(v)-\mathbb{G}_P(v^{\prime})\vert^2 \right]^{1/2}
	\end{align*}
	for all $v,v^{\prime}\in\tilde{\mathcal{V}}$. This semimetric is the one defined in \citet[p.~39]{van1996weak} with $p=2$. Since $\mathbb{G}_{P}$ is tight, it follows from the discussion in Example 1.5.10 of \citet{van1996weak} that $\mathbb{G}_{P}$ almost surely has a uniformly $\rho_2$-continuous path. Since $\mathbb{G}_P$ is a $P$-Brownian bridge,
	\begin{align}\label{eq.rho_V}
		\rho_2^2(v,v^{\prime})
		=P((v-v^{\prime})^2)-P^2(v-v^{\prime})\le \Vert v-v^{\prime}\Vert_{L^2(P)}^2
	\end{align}
	for all $v,v^{\prime}\in\tilde{\mathcal{V}}$. Therefore, $\mathbb{G}_{P}$ almost surely has a uniformly continuous path under $\Vert\cdot\Vert_{L^2(P)}$. By Lemma 3.10.11 of \citet{van1996weak}, $P(v_0)=0$ and $P(v_0^2)<\infty$, where $v_0$ is as in Assumption \ref{ass.probability path}. H\"older's inequality implies that for every $v\in L^2(P)$, $\left\Vert v\cdot 1\right\Vert_{L^1(P)}\le 1\cdot \left\Vert v\right\Vert_{L^2(P)} $. By H\"older's inequality, $P$ and $Q_0$ are both continuous on $\tilde{\mathcal{V}}$ under $\Vert\cdot\Vert_{L^2(P)}$, where $Q_0$ is as in Lemma \ref{lemma.weak convergence Pn_hat and convergence Pn}. Suppose that there are $(h,g),(h^{\prime},g^{\prime})\in{\bar{\mathcal{H}}	\times\mathcal{G}}$ with $g=(g_1,g_2)$ and $g^{\prime}=(g^{\prime}_1,g^{\prime}_2)$. Then for $j\in\{1,2\}$ we have
	\begin{align}\label{eq.semimetric}
		&\left\Vert g_{j}-g_{j}^{\prime}\right\Vert _{L^2\left(  P\right)  }  \leq\rho_{P}\left(  \left(  h,g\right)  ,\left(  h^{\prime},g^{\prime}\right)  \right) \text{ and } \notag\\
		&\left\Vert h\cdot g_{j}-h^{\prime}\cdot g_{j}^{\prime}\right\Vert
		_{L^2\left(  P\right)  }  \leq\left\Vert h - h^{\prime}\right\Vert _{L^2\left(  P\right)  }+\left\Vert g_{j}-g_{j}^{\prime}\right\Vert _{L^2\left(  P\right)  }  \leq\rho_{P}\left(  \left(  h,g\right)  ,\left(  h^{\prime},g^{\prime}\right)  \right).
	\end{align}
	By \eqref{eq.L HDD} and \eqref{eq.semimetric}, together with the continuity of $\mathbb{ G}_P$, $P$, and $Q_0$ under $\Vert\cdot\Vert_{L^2(P)}$, we conclude that $  \mathbb{G}$ almost surely has a continuous path under $\rho_P$. 
	
	Next, we show the variance of $\mathbb{G}(h,g)$ for each $\left(  h,g\right)  \in
	\bar{\mathcal{H}}\times\mathcal{G}$ with $g=(g_1,g_2)$. Since $\mathcal{L}_{P}^{\prime}(H)$ is linear in $H$, $Var(\mathcal{\mathbb{ G}}(h,g))={\Lambda(P)} \cdot Var(\mathcal{L}_{P}^{\prime}(\mathbb{ G}_P)(h,g))$. First, we have that
	\begin{align}\label{eq.L variance}
		&Var( \mathcal{L}_{P}^{\prime}\left(  \mathbb{G}_{P}\right)  \left(
		h,g\right)  ) \notag\\
		=&\, E\left[  \left(  \frac{\mathbb{G}_{P}\left(  h\cdot
			g_{2}\right)  P\left(  g_{2}\right)  -P\left(  h\cdot g_{2}\right)
			\mathbb{G}_{P}\left(  g_{2}\right)  }{P^{2}\left(  g_{2}\right)  }%
		-\frac{\mathbb{G}_{P}\left(  h\cdot g_{1}\right)  P\left(  g_{1}\right)
			-P\left(  h\cdot g_{1}\right)  \mathbb{G}_{P}\left(  g_{1}\right)  }%
		{P^{2}\left(  g_{1}\right)  }\right)  ^{2}\right] .
	\end{align}
	Since $\mathbb{G}_{P}$ is a Brownian bridge with $E[\mathbb{ G}_P(v_1)\mathbb{ G}_P(v_2)]=P(v_1v_2)-P(v_1)P(v_2)$ for all $v_1,v_2\in\tilde{\mathcal{V}}$, we have
	\begin{align}\label{eq.L variance 2}
		&  E\left[  \left(  \frac{\mathbb{G}_{P}\left(  h\cdot g_{2}\right)  P\left(
			g_{2}\right)  -P\left(  h\cdot g_{2}\right)  \mathbb{G}_{P}\left(
			g_{2}\right)  }{P^{2}\left(  g_{2}\right)  }\right)  ^{2}\right] \notag \\
		=&\,\frac{P\left(  h^{2}\cdot g_{2}\right)  -P^{2}\left(  h\cdot g_{2}\right)
		}{P^{2}\left(  g_{2}\right)  }+\frac{P^{2}\left(  h\cdot g_{2}\right)  }%
		{P^{3}\left(  g_{2}\right)  }-\frac{P^{2}\left(  h\cdot g_{2}\right)  }%
		{P^{2}\left(  g_{2}\right)  }-\frac{2P^{2}\left(  h\cdot g_{2}\right)  }%
		{P^{3}\left(  g_{2}\right)  }+\frac{2P^{2}\left(  h\cdot g_{2}\right)  }%
		{P^{2}\left(  g_{2}\right)  }\notag\\
		=&\,\frac{P\left(  h^{2}\cdot g_{2}\right)  }{P^{2}\left(  g_{2}\right)
		}-\frac{P^{2}\left(  h\cdot g_{2}\right)  }{P^{3}\left(  g_{2}\right)  }.
	\end{align}
	Similarly,
	\begin{align}\label{eq.L variance 1}
		E\left[  \left(  \frac{\mathbb{G}_{P}\left(  h\cdot g_{1}\right)  P\left(
			g_{1}\right)  -P\left(  h\cdot g_{1}\right)  \mathbb{G}_{P}\left(
			g_{1}\right)  }{P^{2}\left(  g_{1}\right)  }\right)  ^{2}\right]
		=\frac{P\left(  h^{2}\cdot g_{1}\right)  }{P^{2}\left(  g_{1}\right)  }%
		-\frac{P^{2}\left(  h\cdot g_{1}\right)  }{P^{3}\left(  g_{1}\right)  }.
	\end{align}
	Also, we have that 
	\begin{align}\label{eq.L covariance}
		&  E\left[  \left(  \mathbb{G}_{P}\left(  h\cdot g_{2}\right)  P\left(
		g_{2}\right)  -P\left(  h\cdot g_{2}\right)  \mathbb{G}_{P}\left(
		g_{2}\right)  \right)  \left(  \mathbb{G}_{P}\left(  h\cdot g_{1}\right)
		P\left(  g_{1}\right)  -P\left(  h\cdot g_{1}\right)  \mathbb{G}_{P}\left(
		g_{1}\right)  \right)  \right] \notag \\
		=&\,P\left(  g_{2}\right)  P\left(  g_{1}\right)  P\left(  h^{2}g_{2}%
		g_{1}\right)  -P\left(  g_{2}\right)  P\left(  hg_{1}\right)  P\left(
		hg_{2}g_{1}\right)  -P\left(  hg_{2}\right)  P\left(  g_{1}\right)  P\left(
		hg_{2}g_{1}\right) \notag \\
		&+P\left(  hg_{2}\right)  P\left(  hg_{1}\right)  P\left(
		g_{2}g_{1}\right) =0,
	\end{align}
	where we use the fact that $g_1g_2=0$ by the construction of $\mathcal{ G}$.
	By \eqref{eq.L covariance}, the expectation on the right-hand side of \eqref{eq.L variance} is equal to the sum of the expectations in \eqref{eq.L variance 2} and  \eqref{eq.L variance 1}. Thus we now have that
	\begin{align*}
		Var( \mathcal{L}_{P}^{\prime}\left(  \mathbb{G}_{P}\right)  \left(
		h,g\right)  )=\frac{P\left(  h^{2}\cdot g_{2}\right)  }{P^{2}\left(  g_{2}\right)
		}-\frac{P^{2}\left(  h\cdot g_{2}\right)  }{P^{3}\left(  g_{2}\right)  }+\frac{P\left(  h^{2}\cdot g_{1}\right)  }{P^{2}\left(  g_{1}\right)  }%
		-\frac{P^{2}\left(  h\cdot g_{1}\right)  }{P^{3}\left(  g_{1}\right)  },
	\end{align*}
	which, together with $Var(\mathcal{\mathbb{ G}}(h,g))={\Lambda(P)} \cdot Var(\mathcal{L}_{P}^{\prime}(\mathbb{ G}_P)(h,g))$,  verifies the equality that $Var\left(  \mathbb{G}\left(  h,g\right)  \right)=\sigma_{P}^2  (h,g)$ for the $\sigma_{P}^2 $ in \eqref{eq.stat variance multi}.
	For every $\left(  h,g\right)  \in\mathcal{\bar{H}}\times\mathcal{G}$ with $g=(g_1,g_2)$,
	\begin{align*}
		\sigma_P^{2}\left(  h,g\right)    & =\Lambda\left(  P\right)  \left\{
		\frac{P\left(  h^{2}\cdot g_{2}\right)  }{P^{2}\left(  g_{2}\right)  }-\frac
		{P^{2}\left(  h\cdot g_{2}\right)  }{P^{3}\left(  g_{2}\right)  }+\frac{P\left(
			h^{2}\cdot g_{1}\right)  }{P^{2}\left(  g_{1}\right)  }-\frac{P^{2}\left(
			h\cdot g_{1}\right)  }{P^{3}\left(  g_{1}\right)  }\right\}  \\
		& =\frac{\Lambda\left(  P\right)  }{P\left(  g_{2}\right)  }\frac{\left\vert
			P\left(  h\cdot g_{2}\right)  \right\vert }{P\left(  g_{2}\right)  }\left[
		1-\frac{\left\vert P\left(  h\cdot g_{2}\right)  \right\vert }{P\left(
			g_{2}\right)  }\right]  +\frac{\Lambda\left(  P\right)  }{P\left(
			g_{1}\right)  }\frac{\left\vert P\left(  h\cdot g_{1}\right)  \right\vert }{P\left(
			g_{1}\right)  }\left[  1-\frac{\left\vert P\left(  h\cdot g_{1}\right)  \right\vert
		}{P\left(  g_{1}\right)  }\right]  .
	\end{align*}
	Then $\sigma_P^{2}\left(  h,g\right)  \leq1/4\cdot\left\{
	\Lambda\left(  P\right)  /P\left(  g_{2}\right)  +\Lambda\left(  P\right)
	/P\left(  g_{1}\right)  \right\}  $, since $0\leq\left\vert P\left(  hg_{j}\right)  \right\vert /P\left(
	g_{j}\right)  \leq1$ for $j\in\{1,2\}$. Recall that $K$ is the number of elements
	in $\mathcal{Z}$. We have that for each $j\in\left\{  1,2\right\}  $,
	\[
	\frac{\Lambda\left(  P\right)  }{P\left(  g_{j}\right)  }\leq\max_{1\leq
		k^{\prime}\leq K}\frac{\prod_{k=1}^{K}P\left(  1_{\mathbb{R\times R\times
			}\left\{  z_{k}\right\}  }\right)  }{P\left(  1_{\mathbb{R\times R\times}\left\{
			z_{k^{\prime}}\right\}  }\right)  }\leq\left(  \frac{1}{K-1}\right)  ^{K-1},
	\]
	which implies that
	\begin{align*}
		\sigma_P^{2}\left(  h,g\right)  \leq 1/4\cdot\max_{(g_1^{\prime},g_2^{\prime})\in\mathcal{G}}\left\{
		\Lambda\left(  P\right)  /P\left(  g_{2}^{\prime}\right)  +\Lambda\left(  P\right)
		/P\left(  g_{1}^{\prime}\right)  \right\}\le 1/2\cdot (K-1)^{-(K-1)}.
	\end{align*}
	When $K=2$, $\sigma_P^{2}\left(  h,g\right)  \leq 1/4$ by the construction of $\Lambda(P)$.
\end{proof}

\begin{lemma}\label{lemma.phi is continuous}
Under $\rho_{P}$, $\phi_P$ and $\sigma_P$ are continuous on $\bar{\mathcal{H}}\times\mathcal{G}$.
\end{lemma}

\begin{proof}[Proof of Lemma \ref{lemma.phi is continuous}]
Suppose there are $\left(  h,g\right)  ,(  h^{k},g^{k})
\in\bar{\mathcal{H}}\times\mathcal{G}$ with $g=\left(  g_{1},g_{2}\right)  $, $g^{k
}=(  g_{1}^{k},g_{2}^{k})  $, and $(h^k,g^k)\to (h,g)$ under $\rho_P$. Since $\mathcal{ G}_K$ is finite, $(h^k,g^k)\to (h,g)$ under $\rho_P$ implies that $P(g_j^k)=P(g_j)$ for each $j\in\{1,2\}$ when $k$ is sufficiently large. 
If $P(g_j)=0$,\footnote{If $P(g_j)=0$ for some $g_j\in\mathcal{G}_K$, then $\Lambda(P)=0$, which is a trivial case. We consider this case only for the sake of completeness.} then by \eqref{eq.0timesinfinity} $P(h\cdot g_j)/P(g_j)=0$, $P(h^k\cdot g_j^k)/P(g_j^k)=0$ when $k$ is large, and  
\begin{align*}
\left\vert \frac{P\left(  h\cdot g_{j}\right)  }{P\left(  g_{j}\right)
}-\frac{P(  h^{k}\cdot g_{j}^{k})  }{P(
	g_{j}^k)  }\right\vert=0.
\end{align*}
If $P(g_j)\neq0$, then for each $j\in\{1,2\}$ and large $k$, $P(g_j^k)=P(g_j)\neq0$ and 
\begin{align*}
\left\vert \frac{P\left(  h\cdot g_{j}\right)  }{P\left(  g_{j}\right)
}-\frac{P(  h^{k}\cdot g_{j}^{k})  }{P(
	g_{j}^k)  }\right\vert\le \frac{\left\Vert h\cdot g_{j}-h^{k
	}\cdot g_{j}^{k}\right\Vert _{L^2\left(  P\right)  }}{P\left(
	g_{j}\right)  } \leq\frac{  \rho_{P}(  \left(  h,g\right)  ,(  h^{k
	},g^{k})  )   }{P\left(  g_{j}\right)  }
\end{align*}
by H\"older's inequality and \eqref{eq.semimetric}.
Thus we can conclude that
\begin{align*}
\left\vert \phi_P\left(  h,g\right)  -\phi_P(  h^{k},g^{k})
\right\vert    =&\left\vert \left(  \frac{P\left(  h\cdot g_{2}\right)
}{P\left(  g_{2}\right)  }-\frac{P\left(  h\cdot g_{1}\right)  }{P\left(
	g_{1}\right)  }\right)  -\left(  \frac{P(  h^{k}\cdot g_{2}^{k
	})  }{P(  g_{2}^{k})  }-\frac{P(  h^{k}\cdot
	g_{1}^{k})  }{P(  g_{1}^{k})  }\right)  \right\vert \to 0
\end{align*}
if  $(h^k,g^k)\to (h,g)$ under $\rho_P$. Similarly, we can show that $\sigma_P$ is continuous on $\bar{\mathcal{H}}\times\mathcal{G}$ under $\rho_P$.
\end{proof}

We define some new notation which will be used in the following results. 
Define a random element $\hat{\varphi}_{P}
:\Omega\rightarrow\ell^{\infty}\left(  \Xi\times\mathcal{\bar{H}}
\times\mathcal{G}\right)  $ such that for each $\omega\in\Omega$ and each $(\xi,h,g)\in\Xi\times
\mathcal{\bar{H}}\times\mathcal{G}$,
\begin{align}\label{eq.varphi_P_hat}
\hat{\varphi}_{P}(\omega)(\xi,h,g)=\frac{{\phi_{P}}(h,g)}{\mathcal{M}\left(  {\hat{\sigma}_{P_n}(\omega)}\right)
\left(  \xi,h,g\right) } ,
\end{align}
and let $\varphi_P \in \ell^{\infty}\left(  \Xi\times\mathcal{\bar{H}}%
\times\mathcal{G}\right)$ be such that for each $(\xi,h,g)\in\Xi\times
\mathcal{\bar{H}}\times\mathcal{G}$,
\[
{\varphi_P}(\xi,h,g)=\frac{{\phi_P}(h,g)}{\mathcal{M}\left(  \sigma_P\right)
	\left(  \xi,h,g\right) }.
\]
Here, $\hat{\sigma}_{P_n}$ is estimated from data, hence it depends on $\omega$, and so does $\hat{\varphi}_{P}$. When there is no danger of confusion, we omit the $\omega$ from $\hat{\sigma}_{P_n}$ and $\hat{\varphi}_{P}$ for brevity. Given each
sequence $r_{n}\rightarrow\infty$ and each $\nu$ which satisfies Assumption \ref{ass.nu}, define
\begin{align}\label{eq.Dn}
\mathbb{D}_{n} (\omega) =\left\{  \psi\in\ell^{\infty}\left(
\Xi\times\mathcal{\bar{H}}\times\mathcal{G}\right)  :{\mathcal{S}}\left(  \hat{\varphi}_{P}(\omega)+r_{n}^{-1}
\psi\right)  \in L^{1}\left(  \nu\right) \right\}
\end{align}
for all $\omega\in\Omega$, and 
\begin{align}\label{eq.gn}
g_{n} (\omega) \left(  \psi\right)  =r_{n} {\mathcal{I}\circ\mathcal{S}}\left(  \hat{\varphi}_{P}(\omega)+r_{n}^{-1}\psi\right)  
\end{align}
for all $\omega\in\Omega$ and  all $\psi\in\mathbb{D}_{n}(\omega)$. Here, $g_n$ also depends on $\omega$; for brevity, we omit $\omega$ from $g_{n}$ as well. If the $H_0$ in \eqref{eq.null 1} is true with $Q=P_n$ for all $n$, then $\mathcal{S}(\hat{\varphi}_P)=0$ by Lemma \ref{lemma.Sphi under null}, and so $g_n(\psi)=r_{n}\left\{ {\mathcal{I}\circ\mathcal{S}}\left(  \hat{\varphi}_{P}+r_{n}^{-1}\psi\right)  -{\mathcal{I}\circ\mathcal{S}}\left(  \hat{\varphi}_{P}\right)  \right\}$. Define a
correspondence $\Psi:\Xi\times\ell^{\infty}\left(  \Xi\times\mathcal{\bar{H}
}\times\mathcal{G}\right)  \twoheadrightarrow\mathcal{\bar{H}}\times
\mathcal{G}$ by
\begin{align}\label{eq.Psi}
\Psi\left(  \xi,\psi\right)  =\left\{  \left(  h,g\right)  \in\mathcal{\bar
	{H}}\times\mathcal{G}:\psi\left(  \xi,h,g\right)  ={\mathcal{S}}\left(
\psi\right)  \left(  \xi\right)  \right\}
\end{align}
for all $\xi\in\Xi$ and all $\psi\in\ell^{\infty}\left( \Xi\times \mathcal{\bar{H}}%
\times\mathcal{G}\right)  $, and define a metric $\rho_{\xi \psi}$ on $\Xi\times\ell^{\infty}\left(  \Xi\times\mathcal{\bar{H}}\times\mathcal{G}\right)$ by 
\begin{align}\label{eq.rho_xipsi}
\rho_{\xi \psi}((\xi_1,\psi_1),(\xi_2,\psi_2))=|\xi_1-\xi_2|+\Vert \psi_1-\psi_2\Vert_{\infty}
\end{align}
for all $(\xi_1,\psi_1),(\xi_2,\psi_2)\in \Xi\times\ell^{\infty}\left(  \Xi\times\bar{\mathcal{H}}\times\mathcal{G}\right)$. Also, define a metric on $\Xi\times\bar{\mathcal{H}}\times\mathcal{G}$ by 
\begin{align}\label{eq.rho_xihg}
\rho_{\xi h g}((\xi_1,h_1,g_1),(\xi_2,h_2,g_2))=|\xi_1-\xi_2|+\rho_{P}((h_1,g_1),(h_2,g_2))
\end{align}
for all $(\xi_1,h_1,g_1),(\xi_2,h_2,g_2)\in\Xi\times\bar{\mathcal{H}}\times\mathcal{G}$.
For every set $A\subset\mathcal{\bar{H}
}\times\mathcal{G}$ and every $\delta>0$, define
\begin{align}\label{eq.A^delta}
A^{\delta}=\left\{  \left(  h,g\right)  \in\mathcal{\bar{H}}\times
\mathcal{G}:\inf_{\left(  h^{\prime},g^{\prime}\right)  \in A}\rho_{P}\left(
\left(  h,g\right)  ,\left(  h^{\prime},g^{\prime}\right)  \right)  \leq
\delta\right\}  .
\end{align}

\begin{lemma}\label{lemma.Sphi under null}
	Suppose Assumption \ref{ass.probability path} holds and the $H_0$ in \eqref{eq.null 1} is true with $Q=P_n$ for all $n$. Then the $H_0$ in \eqref{eq.null 1} is true with $Q=P$. This implies that $ \sup_{\left(  h,g\right)  \in		{\bar{\mathcal{H}}\times\mathcal{ G}}}\phi_P\left( h,g\right)= 0$, and hence that $\mathcal{S}\left(  \varphi_{P}\right)=0$ and $\mathcal{S}\left(  \hat{\varphi}_{P}\right)=0$ for all $\omega\in\Omega$.
\end{lemma}

\begin{proof}[Proof of Lemma \ref{lemma.Sphi under null}]
	By Lemma \ref{lemma.almost uniform convergence phi sigma}, we have $\Vert P_n-P\Vert_{\infty}\to 0$. Thus $\phi_{P_n}\to\phi_{P}$ in $\ell^{\infty}(\bar{\mathcal{H}}\times\mathcal{G})$, and by the assumption that $ \sup_{\left(  h,g\right)  \in{{\mathcal{H}}\times\mathcal{ G}}}\phi_{P_n}\left( h,g\right)\le 0$ for all $n$, we have that $\sup_{\left(  h,g\right)  \in	{{\mathcal{H}}\times\mathcal{ G}}}\phi_P\left( h,g\right)\le0$. This implies that $\sup_{\left(  h,g\right)  \in	{\bar{\mathcal{H}}\times\mathcal{ G}}}\phi_P\left( h,g\right)\le0$ by the constructions of $\phi_P$ and $\bar{\mathcal{H}}$.
	By the construction of $\bar{\mathcal{H}}\times\mathcal{ G}$, there is some $(h,g)\in \bar{\mathcal{H}}\times\mathcal{ G}$, such as $h=1_{\{a\}\times\{0\}\times\mathbb{R}}$ for some $a\in\mathbb{R}$, for which $\phi_{P}(h,g)=0$. Therefore, $ \sup_{\left(  h,g\right)  \in		{\bar{\mathcal{H}}\times\mathcal{ G}}}\phi_P\left( h,g\right)= 0$ under the assumptions. Because $\xi\in\Xi$ is always positive by the construction of $\Xi$, we have that $\mathcal{S}\left(  \varphi_{P}\right)(\xi)=0$ for all $\xi\in\Xi$. For the same reason, $\mathcal{S}\left(  \hat{\varphi}_{P}\right)(\xi)=0$ for all $\xi\in\Xi$ and all $\omega\in\Omega$.
\end{proof}

\begin{lemma}\label{lemma.upper hemicontinuity of Psi}
	The correspondence $\Psi$ defined in \eqref{eq.Psi} is upper hemicontinuous\footnote{See Definition 17.2 of
		upper hemicontinuity in \citet{aliprantis2006infinite}.} at $(\xi,\varphi_{P})$ for all $\xi\in\Xi$. In addition, suppose the $H_0$ in \eqref{eq.null 1} is true with $Q=P$. Then for every $\delta>0$ there is an $\varepsilon>0$ such
	that $\Psi\left(  \xi^{\prime},\psi\right)  \subset\Psi\left(  \xi
	,\varphi_{P}\right)  ^{\delta}$ (where the latter is defined as in \eqref{eq.A^delta})  for all $\xi,\xi^{\prime}\in\Xi$ and all
	$\psi\in\ell^{\infty}\left(  \Xi\times\mathcal{\bar{H}}\times
	\mathcal{G}\right)  $ with $\left\Vert \psi-\varphi_{P}\right\Vert
	_{\infty}<\varepsilon$.
\end{lemma}

\begin{proof}[Proof of Lemma \ref{lemma.upper hemicontinuity of Psi}]
	We first show that $\Psi$ is upper hemicontinuous at $\left(  \xi,\varphi
	_{P}\right)  $ for all $\xi\in\Xi$. We do this \textbf{in three steps}. 
	\textbf{First}, we
	show that $\Psi\left(  \xi,\varphi_{P}\right)  $ is compact for each $\xi\in\Xi$
	under $\rho_{P}$. Clearly, given an arbitrary $\xi\in\Xi$, $\varphi_{P}\left(  \xi,\cdot
	,\cdot\right)  $ is continuous on $\mathcal{\bar{H}\times G}$ under $\rho_P$ by Lemma
	\ref{lemma.phi is continuous}. Because $\mathcal{\bar{H}\times G}$ is compact by Lemma \ref{lemma.complete HG}, $\Psi\left(  \xi,\varphi_{P}\right)  $ is not empty. Since $\Psi\left(  \xi,\varphi_{P}\right)  \subset\mathcal{\bar
		{H}\times G}$, it suffices to show that $\Psi\left(  \xi,\varphi_{P}\right)  $
	is closed in $\mathcal{\bar
		{H}\times G}$. Fix $\xi\in\Xi$. Suppose there is a sequence $\left\{  \left(
	h_{n},g_{n}\right)  \right\}  _{n}\subset\Psi\left(  \xi,\varphi_{P}\right)$ such that $\left(  h_{n},g_{n}\right)	\rightarrow\left(  h,g\right) \in\bar{\mathcal{H}}\times\mathcal{ G} $ under $\rho_{P}$. Then for all $n$,
	$	\varphi_{P}\left(  \xi,h_{n},g_{n}\right)  =\mathcal{S}\left(
	\varphi_{P}\right)  \left(  \xi\right)$.
	Since $\varphi_{P}\left(  \xi,\cdot,\cdot\right)  $ is continuous by Lemma
	\ref{lemma.phi is continuous},
	$\varphi_{P}\left(  \xi,h_{n},g_{n}\right)  \rightarrow\varphi_{P}\left(
	\xi,h,g\right)  $ as $\left(  h_{n},g_{n}\right)  \rightarrow\left(
	h,g\right)  $. Thus $\varphi_{P}\left(  \xi,h,g\right)  =\mathcal{S}\left(  \varphi_{P}\right)  \left(  \xi\right)  $, which implies that $\Psi\left(\xi,\varphi_{P}\right)  $ is closed in $\mathcal{\bar{H}}\times\mathcal{G}$
	and therefore compact. \textbf{Second}, we show that if there is a sequence
	$\left\{  \left(  \xi_{n},\psi_{n}\right)  ,(h_{n},g_{n})\right\}  $ such that
	$(h_{n},g_{n})\in\Psi\left(  \xi_{n},\psi_{n}\right)  $ and $\rho_{\xi \psi}( (\xi_{n},\psi_{n}),(\xi,\varphi_{P}))\to 0$, where $\rho_{\xi \psi}$ is defined in \eqref{eq.rho_xipsi}, then $\left(
	h_{n},g_{n}\right)  $ has a limit point\footnote{See the definition of limit point in \citet[p.~31]{aliprantis2006infinite}.} in $\Psi\left(  \xi,\varphi
	_{P}\right)  $. Notice that by the constructions of $\Xi$ and $\bar{\mathcal{H}}\times\mathcal{G}$, $\Xi\times\bar{\mathcal{H}}\times\mathcal{G}$ is compact under the metric $\rho_{\xi h g}$ defined in \eqref{eq.rho_xihg}.
	It is easy to show, by Lemma \ref{lemma.phi is continuous}, that $\varphi_{P}$ is continuous on $\Xi\times\bar{\mathcal{H}}\times\mathcal{G}$ under $\rho_{\xi h g}$, and hence that it is uniformly continuous.
	Thus $\rho_{\xi \psi}( (\xi_{n},\psi_{n}),(\xi,\varphi_{P}))\to 0$ implies that
	\begin{align*}
	\left\vert {\mathcal{S}}\left(  \psi_{n}\right)  \left(  \xi_{n}\right)
	-{\mathcal{S}}\left(  \varphi_{P}\right)  \left(  \xi\right)  \right\vert
	  \leq & \sup_{(h,g)\in\mathcal{\bar{H}}\times\mathcal{G}}\left\vert \psi
	_{n}(\xi_{n},h,g)-\varphi_{P}(\xi_{n},h,g)\right\vert \\
	&  +\sup_{(h,g)\in\mathcal{\bar{H}}\times\mathcal{G}}\left\vert \varphi
	_{P}(\xi_{n},h,g)-\varphi_{P}(\xi,h,g)\right\vert \rightarrow 0,
	\end{align*}
	where $\sup_{(h,g)\in\mathcal{\bar{H}}\times\mathcal{G}}\left\vert \varphi
	_{P}(\xi_{n},h,g)-\varphi_{P}(\xi,h,g)\right\vert $ converges to $0$ because
	$\varphi_{P}$ is uniformly continuous on $\Xi\times\mathcal{\bar{H}}%
	\times\mathcal{G}$ under $\rho_{\xi h g}$. This implies that
	$\psi_{n}\left(  \xi_{n},h_{n},g_{n}\right)  \rightarrow {\mathcal{S}}\left(  \varphi_{P}\right)  \left(  \xi\right)  $. Suppose, by way of
	contradiction, that $\left(  h_{n},g_{n}\right)  $ has no limit point in
	$\Psi\left(  \xi,\varphi_{P}\right)  $. This implies that for each $(h,g)\in
	\Psi\left(  \xi,\varphi_{P}\right)  $ there exists an open neighborhood $V_{h,g}$
	of $(h,g)$ and an $n_{h,g}$ such that $(h_{n},g_{n})\not \in V_{h,g}$ when $n\geq
	n_{h,g}$. Because we have shown that $\Psi\left(  \xi,\varphi_{P}\right)  $ is
	compact in $\mathcal{\bar{H}}\times\mathcal{G}$, there is a finite open cover $V$
	such that $\Psi\left(  \xi,\varphi_{P}\right)  \subset V=V_{h^{1},g^{1}}%
	\cup\cdots\cup V_{h^{M},g^{M}}$. Let $n_{0}=\max_{m\leq M}n_{h^{m},g^{m}}$.
	Thus if $n>n_{0}$, then $(h_{n},g_{n})\not \in V$, and hence $(h_{n}%
	,g_{n})\not \in \Psi\left(  \xi,\varphi_{P}\right)  $. Since $\mathcal{\bar
		{H}}\times\mathcal{G}$ is compact and $V^{c}$ is closed in $\mathcal{\bar{H}%
	}\times\mathcal{G}$, $V^{c}$ is compact. Notice that $V^c\cap \Psi(\xi,\varphi_P)=\varnothing$. Thus
	\[
	\sup_{(h,g)\in V^{c}}\varphi_{P}\left(  \xi,h,g\right)  <\sup_{(h,g)\in
		\mathcal{\bar{H}}\times\mathcal{G}}\varphi_{P}\left(  \xi,h,g\right)
	=\sup_{(h,g)\in\Psi\left(  \xi,\varphi_{P}\right)  }\varphi_{P}\left(
	\xi,h,g\right)  .
	\]
	Let $\delta=\sup_{(h,g)\in\mathcal{\bar{H}}\times\mathcal{G}}\varphi
	_{P}\left(\xi,  h,g\right)  -\sup_{(h,g)\in V^{c}}\varphi_{P}\left(\xi,  h,g\right)
	$. Recall that $(h_{n},g_{n})\in V^{c}$ for all $n>n_{0}$. Thus
	$
	\psi_{n}\left(  \xi_{n},h_{n},g_{n}\right)  =\sup_{(h,g)\in\mathcal{\bar{H}
		}\times\mathcal{G}}\psi_{n}\left(  \xi_{n},h,g\right)  =\sup_{(h,g)\in V^{c}
	}\psi_{n}\left(  \xi_{n},h,g\right)
	$, so
	\begin{align*}
	\left\vert \psi_{n}\left(  \xi_{n},h_{n},g_{n}\right)  -\sup_{(h,g)\in V^{c}%
	}\varphi_{P}\left(  \xi,h,g\right)  \right\vert 
    \leq&\sup_{(h,g)\in
		\mathcal{\bar{H}}\times\mathcal{G}}\left\vert \psi_{n}(\xi_{n},h,g)-\varphi
	_{P}(\xi_{n},h,g)\right\vert \\
	&  +\sup_{(h,g)\in\mathcal{\bar{H}}\times\mathcal{G}}\left\vert \varphi
	_{P}(\xi_{n},h,g)-\varphi_{P}(\xi,h,g)\right\vert \rightarrow0.
	\end{align*}
	This implies that for sufficiently large $n$,
	\[
	\psi_{n}\left(  \xi_{n},h_{n},g_{n}\right)  \leq\sup_{(h,g)\in V^{c}}
	\varphi_{P}\left(  \xi,h,g\right)  +\frac{\delta}{2}=\sup_{(h,g)\in
		\mathcal{\bar{H}}\times\mathcal{G}}\varphi_{P}\left(  \xi,h,g\right) -  \frac{\delta}{2}.
	\]
	This contradicts $\psi_{n}\left(  \xi_{n},h_{n},g_{n}\right)  \rightarrow
	\mathcal{S}\left(  \varphi_{P}\right)  \left(  \xi\right)  $. Thus
	$\left(  h_{n},g_{n}\right)  $ has a limit point in $\Psi\left(  \xi
	,\varphi_{P}\right)  $. \textbf{Third}, by Theorem 17.20(ii) of
	\citet{aliprantis2006infinite}, together with the fact that $\Xi\times\ell^{\infty}\left(  \Xi\times\mathcal{\bar{H}
	}\times\mathcal{G}\right)$ is first countable under the metric $\rho_{\xi \psi}$ defined in \eqref{eq.rho_xipsi} (every metric space is first countable), $\Psi$ is upper hemicontinuous at $\left(
	\xi,\varphi_{P}\right)  $. 
	
	Now we prove the second claim in the Lemma. Fix $\delta>0$. Since $\Psi$ is upper hemicontinuous at $\left(
	\xi,\varphi_{P}\right)  $ for all $\xi\in\Xi$, we have that for each $\xi$ there is an open ball
	$B_{\varepsilon_{\xi}}\left(  \xi,\varphi_{P}\right)  $  under $\rho_{\xi \psi}$ with center  $(\xi,\varphi_{P})$ and radius $\varepsilon_{\xi}$ such that $\Psi\left(
	\xi^{\prime},\varphi^{\prime}\right)  \subset\Psi\left(  \xi,\varphi
	_{P}\right)  ^{\delta}$ for all $\left(  \xi^{\prime},\varphi^{\prime}\right)
	\in B_{\varepsilon_{\xi}}\left(  \xi,\varphi_{P}\right)  $, where $\Psi\left(  \xi,\varphi
	_{P}\right)  ^{\delta}$ is defined as in \eqref{eq.A^delta}. Notice that $\{B_{\varepsilon_{\xi}/2}\left(  \xi \right) \}_{\xi\in\Xi} $ is an open cover of $\Xi$, where each $B_{\varepsilon_{\xi}/2}\left(  \xi\right)$ is an open ball in $\mathbb{R}$ with center $\xi$ and radius $\varepsilon_{\xi}/2$. Since $\Xi$ is
	compact by construction, there is a finite open cover $\{B_{\varepsilon_{i}	}\left(  \xi_{i}\right) \}_{i=1}^{M} $ of $\Xi$ with $\varepsilon_{i}=\varepsilon_{\xi_{i}}/2
	$. Let $\varepsilon=\min_{i\leq M}\varepsilon_{i}$. Then for every $\xi^{\prime}\in\Xi$ and every $\psi\in\ell^{\infty}\left(
	\Xi\times\mathcal{\bar{H}}\times\mathcal{G}\right)  $ with $\left\Vert
	\psi-\varphi_{P}\right\Vert _{\infty}<\varepsilon$, there is an open ball
	$B_{\varepsilon_{\xi_{i}}}\left(  \xi_{i},\varphi_{P}\right)  $ such that
	$\left(  \xi^{\prime},\psi\right)  \subset B_{\varepsilon_{\xi_{i}}}\left(
	\xi_{i},\varphi_{P}\right)  $. This implies that $\Psi\left(  \xi^{\prime}
	,\psi\right)  \subset\Psi\left(  \xi_{i},\varphi_{P}\right)  ^{\delta}$.
	Suppose the $H_0$ in \eqref{eq.null 1} is true with $Q=P$. By Lemma \ref{lemma.Sphi under null}, we have that $\mathcal{S}\left(  \varphi_{P}\right)=0$ and
	\begin{align*}
	\Psi\left(  \xi,\varphi_{P}\right)  =\Psi(	\tilde{\xi},\varphi_{P})=\left\{  \left(  h,g\right)
	\in{\bar{\mathcal{H}}\times\mathcal{G}}:\phi_P\left(  h,g\right)  =0\right\}
	\end{align*} 
	for all $\xi,\tilde{\xi}\in\Xi$. Thus
	$\Psi\left(  \xi^{\prime},\psi\right)  \subset\Psi\left(  \xi,\varphi
	_{P}\right)  ^{\delta}$ for all $\xi\in\Xi$, that is, the second claim holds. 
\end{proof}

\begin{lemma}\label{lemma.random differentiability}
	Suppose Assumptions \ref{ass.independent data}, \ref{ass.probability path}, and \ref{ass.nu} hold and the $H_0$ in \eqref{eq.null 1} is true with $Q=P_n$ for all $n$. For every $\varepsilon>0$, there is a measurable set $\Omega_0\subset \Omega$ with $\mathbb{P}(\Omega_0)\ge 1-\varepsilon$ such that 
	for every subsequence $\{\psi_{n_m}\}$ with $\psi_{n_m}
		\in\mathbb{D}_{n_m}(\omega_{n_m})$, $\omega_{n_m}\in \Omega_0$, where $\mathbb{D}_{n_m}(\omega_{n_m})$ is defined in \eqref{eq.Dn}, and $\psi_{n_m}\rightarrow\psi$ for
		some $\psi\in C\left(  \Xi\times\mathcal{\bar{H}}\times\mathcal{G}\right)  $ under the $\rho_{\xi h g}$ defined in \eqref{eq.rho_xihg}, we have that
		\[
		g_{n_m}(\omega_{n_m})  \left(  \psi_{n_m}\right) \rightarrow \mathcal{I}\circ \mathcal{S}_{\Psi\left( \xi, \varphi_{P}\right)  }(\psi),
		\]
		where $g_{n_m}$ is defined in \eqref{eq.gn}.\footnote{Lemma \ref{lemma.random differentiability} implies that under $H_0$, $\mathcal{I}\circ\mathcal{S}
			_{\Psi_{\mathcal{\bar{H}}\times\mathcal{G}}}$ is the Hadamard
			directional derivative of $\mathcal{I}\circ\mathcal{S}$ at $\varphi_{P}$. See the definition of Hadamard directional differentiability in \citet{shapiro1990concepts}.}

\end{lemma}

\begin{proof}[Proof of Lemma \ref{lemma.random differentiability}]
	For simplicity of notation, we replace $n_m$ with $n$. Note that all the following results hold for every subsequence indexed by $n_m$.
	By Lemma \ref{lemma.complete HG}, $\mathcal{\bar{H}}\times\mathcal{G}$ is
	compact under $\rho_{P}$. By Lemma \ref{lemma.almost uniform convergence phi sigma}, we have $\hat{\sigma}_{P_n}\rightarrow\sigma_P$ almost uniformly. Then by construction,  $\hat{\varphi}_{P}\rightarrow\varphi_{P}$ almost uniformly, where $\hat{\varphi}_{P}$ is defined in \eqref{eq.varphi_P_hat}.
By Lemma \ref{lemma.Sphi under null}, $\mathcal{S}(\varphi_{P})=0$ and  $\mathcal{S}(\hat{\varphi}_{P})=0$ for all $\omega\in\Omega$. For every $\psi\in C\left(  \Xi\times\mathcal{\bar{H}}\times\mathcal{G}\right)  $, since  $\hat{\varphi}_{P}\left(  \xi,\cdot,\cdot\right)
	+r_{n}^{-1}\psi\left(  \xi,\cdot,\cdot\right)  $ may not be continuous on
	$\mathcal{\bar{H}}\times\mathcal{G}$, $\Psi\left(  \xi,\hat{\varphi}_{P}
	+r_{n}^{-1}\psi\right)  $ may be empty. Here, we construct a
	modified version of $\hat{\varphi}_{P}$, denoted by $\tilde{\varphi}_{P}$, such that
	
	\begin{enumerate}
		[label=(\roman*)]
		
		\item $\tilde{\varphi}_{P}\left(  \xi,\cdot,\cdot\right)  $ is upper semicontinuous for every $\omega\in\Omega$, every $n$, and every $\xi\in\Xi$;
		
		\item $\sup_{\left(  h,g\right)  \in\mathcal{\bar{H}}\times\mathcal{G}}
		\hat{\varphi}_{P}\left(  \xi,h,g\right)  =\sup_{\left(  h,g\right)  \in
			\mathcal{\bar{H}}\times\mathcal{G}}\tilde{\varphi}_{P}\left(  \xi,h,g\right)  $ for every $\omega\in\Omega$, every $n$, and every $\xi\in\Xi$;
		
		\item $\sup_{\left(  h,g\right)  \in\mathcal{\bar{H}}\times\mathcal{G}}\left(
		\hat{\varphi}_{P}+r_{n}^{-1}\psi\right)  \left(  \xi,h,g\right)  =\sup_{\left(
			h,g\right)  \in\mathcal{\bar{H}}\times\mathcal{G}}\left(  \tilde{\varphi}_{P}+r_{n}^{-1}\psi\right)  \left(  \xi,h,g\right)  $ for every function $\psi\in C\left(  \Xi\times\mathcal{\bar{H}}\times\mathcal{G}\right)  $, every $\omega\in\Omega$, every $n$, and every $\xi\in\Xi$;
		
		\item for every $\varepsilon>0$ there is a measurable set $A\subset\Omega$ with $\mathbb{P}(A)\ge 1-\varepsilon$ such that for all $\varphi\in \ell^{\infty}\left(  \Xi\times\mathcal{\bar{H}}\times\mathcal{G}\right)  $, $\tilde{\varphi}_{P}+r_{n}^{-1}\varphi\rightarrow\varphi_{P}$ uniformly on $A$.
	\end{enumerate}
	Specifically, for all $\omega\in\Omega$, all $\left(  \xi, h,g\right)  \in\Xi\times\mathcal{\bar{H}\times
		G}$, and all $n$, we define $\tilde{\varphi}_{P}\left(  \xi,h,g\right) $ by
	\begin{align}\label{def.phi tilde}
	\tilde{\varphi}_{P}\left(  \xi,h,g\right)  =\lim_{\delta\downarrow0} \sup_{(h^{\prime},g^{\prime})\in B_{\delta}(h,g)} \hat{\varphi}_P(\xi,h^{\prime},g^{\prime})   ,
	\end{align}
	where $B_{\delta}(h,g)$ is an open ball in $\bar{\mathcal{H}}\times\mathcal{ G}$ under $\rho_P$ with center $(h,g)$ and radius $\delta$.
	
	Fix $\omega\in\Omega$, $n$, and $\xi\in\Xi$. \textbf{First},  we prove (i), that is, $\tilde{\varphi}_{P}(\xi,\cdot,\cdot)$ is upper semicontinuous at every $(h,g)\in\bar{\mathcal{H}}\times\mathcal{G}$. Fix $(h,g)\in\bar{\mathcal{H}}\times\mathcal{G}$. By
	\eqref{def.phi tilde}, for each $\varepsilon>0$, there is a $\delta
	_{\varepsilon}>0$ such that
	\begin{align}\label{ineq.phihat}
	\hat{\varphi}_{P}\left(  \xi,h^{\prime},g^{\prime}\right)  \le \tilde{\varphi}_{P}\left(
	\xi,h,g\right)  +\frac{\varepsilon}{2}%
	\end{align}
	for all $\left(  h^{\prime},g^{\prime}\right)  \in B_{\delta_{\varepsilon}%
	}\left(  h,g\right)  $, where $B_{\delta
		_{\varepsilon}}\left(  h,g\right)  $ denotes the open ball in $\bar{\mathcal{H}}\times\mathcal{G}$ under $\rho_{P}$ with center $\left(
	h,g\right)  $ and radius $\delta_{\varepsilon}$. 
	Fix $\left(  h_{1},g_{1}\right)  \in B_{\delta_{\varepsilon}%
		/2}\left(  h,g\right)  $. By definition,
	there is a $\delta_{2}>0$ such that for all $\delta^{\prime}$ with $0<\delta^{\prime}\le\delta_{2}$,
	\[
	\tilde{\varphi}_{P}\left(  \xi,h_{1},g_{1}\right)  \leq\sup_{\left(  h_{2},g_{2}\right)  \in
		B_{\delta^{\prime}}\left(  h_{1},g_{1}\right)  }  \hat{\varphi
	}_{P}\left(  \xi,h_{2},g_{2}\right) 
	  +\frac{\varepsilon}{2}.
	\]
	Let $\delta=\min\left\{  {\delta_{\varepsilon}}/{2},\delta_{2}\right\}
	$. Then for this $\left(  h_{1},g_{1}\right)$, we have that
	\[
	\tilde{\varphi}_{P}\left(  \xi,h_{1},g_{1}\right)  \leq\sup_{\left(  h_{2},g_{2}\right)  \in B_{\delta
		}\left(  h_{1},g_{1}\right) }  \hat{\varphi
	}_{P}\left(  \xi,h_{2},g_{2}\right)  
	+\frac{\varepsilon}{2}.
	\]
	Notice that if $\left(  h_{2},g_{2}\right)  \in B_{\delta}\left(  h_{1}%
	,g_{1}\right)  $, then $\left(  h_{2},g_{2}\right)  \in B_{\delta
		_{\varepsilon}}\left(  h,g\right)  $, and hence
	$	\hat{\varphi}_{P}\left(  \xi,h_{2},g_{2}\right)  \leq\tilde{\varphi}_{P}\left(
	\xi,h,g\right)  +{\varepsilon}/{2}$.
	This implies that
	$
	\sup_{\left(  h_{2}
		,g_{2}\right)  \in B_{\delta}\left(  h_{1},g_{1}\right) }  \hat{\varphi}_{P}\left(  \xi,h_{2},g_{2}\right)   \leq\tilde{\varphi}_P\left(\xi,  h,g\right)
	+{\varepsilon}/{2},
	$
	and hence
	$	\tilde{\varphi}_{P}\left(  \xi,h_{1},g_{1}\right)  \leq\tilde{\varphi}_{P}\left(
	\xi,h,g\right)  +\varepsilon$. This shows that for each
	$\varepsilon>0$, there is a $\delta_{\varepsilon}>0$ such that for all $\left(  h_{1}%
	,g_{1}\right)  \in B_{\delta_{\varepsilon}/2}\left(  h,g\right)  $,
	$	\tilde{\varphi}_{P}\left(  \xi,h_{1},g_{1}\right)  \leq\tilde{\varphi}_{P}\left(
	\xi,h,g\right)  +\varepsilon$.
	\textbf{Second}, we prove (ii), that is, 
	\begin{align}\label{ineq.suphatphi suptildephi}
	\sup_{\left(  h,g\right)  \in\mathcal{\bar{H}}\times\mathcal{G}}\hat{\varphi
	}_{P}\left(  \xi,h,g\right)  =\sup_{\left(  h,g\right)  \in\mathcal{\bar{H}}
		\times\mathcal{G}}\tilde{\varphi}_{P}\left(  \xi,h,g\right).
	\end{align}
    By the definition of $\tilde{\varphi}_P$, we have $\hat{\varphi
    }_{P}\left(  \xi,h,g\right)  \leq\tilde{\varphi}_{P}\left(  \xi,h,g\right)$ for all $(h,g)\in\bar{\mathcal{H}}\times\mathcal{G}$, and hence $\sup_{\left(  h,g\right)  \in\mathcal{\bar{H}}\times\mathcal{G}}\hat{\varphi
	}_{P}\left(  \xi,h,g\right)  \leq\sup_{\left(  h,g\right)  \in\mathcal{\bar{H}		}\times\mathcal{G}}\tilde{\varphi}_{P}\left(  \xi,h,g\right)$.
	Also, by the definition of $\tilde{\varphi}_{P}$,
	$	\tilde{\varphi}_{P}\left(  \xi,h,g\right)  \leq\sup_{\left(  h^{\prime},g^{\prime
		}\right)  \in\mathcal{\bar{H}}\times\mathcal{G}}\hat{\varphi}_{P}\left(
	\xi,h^{\prime},g^{\prime}\right)$
	for all $(h,g)$. Thus 
	$
	\sup_{\left(  h,g\right)  \in\mathcal{\bar{H}}
		\times\mathcal{G}}\tilde{\varphi}_{P}\left(  \xi,h,g\right)  \leq\sup_{\left(
		h,g\right)  \in\mathcal{\bar{H}}\times\mathcal{G}}
	\hat{\varphi}_{P}\left(  \xi,h,g\right)
	$,
	 and \eqref{ineq.suphatphi suptildephi} holds. \textbf{Similarly},
    by the definition of $\tilde{\varphi}_P$, we have that
	$
	\hat{\varphi}_{P}\left(  \xi,h,g\right)  +r_{n}^{-1}\psi\left(  \xi,h,g\right)
	\leq\tilde{\varphi}_{P}\left(  \xi,h,g\right)  +r_{n}^{-1}\psi\left(
	\xi,h,g\right) $
	for all $(h,g)\in\bar{\mathcal{H}}\times\mathcal{G}$, and hence
	\[
	\sup_{\left(  h,g\right)  \in\mathcal{\bar{H}}\times\mathcal{G}}\{\hat
	{\varphi}_{P}\left(  \xi,h,g\right)  +r_{n}^{-1}\psi\left(  \xi,h,g\right)
	\}\leq\sup_{\left(  h,g\right)  \in\mathcal{\bar{H}}\times\mathcal{G}}
	\{\tilde{\varphi}_{P}\left(  \xi,h,g\right)  +r_{n}^{-1}\psi\left(  \xi
	,h,g\right)  \}.
	\]
	Fix $(h,g)\in\bar{\mathcal{H}}\times\mathcal{ G}$. Since $\psi(\xi,\cdot,\cdot)$ is continuous under $\rho_{P}$, for every $\varepsilon>0$ there is a $\bar{\delta}>0$ such that
	\begin{align*}
	\sup_{(h^{\prime},g^{\prime})\in B_{\delta}(h,g)} \{ \hat{\varphi}_P(\xi,h^{\prime},g^{\prime})+r_{n}^{-1}\psi(\xi,h,g)-\varepsilon \}  \le &  \sup_{(h^{\prime},g^{\prime})\in B_{\delta}(h,g)} \{ \hat{\varphi}_P(\xi,h^{\prime},g^{\prime})+r_{n}^{-1}\psi(\xi,h^{\prime},g^{\prime}) \}
	\end{align*}
	for all $\delta\le\bar{\delta}$. By the definition of $\tilde{\varphi}_P$, this implies that
	\begin{align*}
	\tilde{\varphi}_{P}\left(  \xi,h,g\right)+r_{n}^{-1}\psi(\xi,h,g)-\varepsilon \le & \lim_{\delta\downarrow0} \sup_{(h^{\prime},g^{\prime})\in B_{\delta}(h,g)} \{ \hat{\varphi}_P(\xi,h^{\prime},g^{\prime})+r_{n}^{-1}\psi(\xi,h^{\prime},g^{\prime}) \} \\
	\le& \sup_{\left(  h,g\right)  \in\mathcal{\bar{H}}\times\mathcal{G}}\{\hat
	{\varphi}_{P}\left(  \xi,h,g\right)  +r_{n}^{-1}\psi\left(  \xi,h,g\right)
	\}.
	\end{align*}
	Since $\varepsilon$ is arbitrary, we have 
	\begin{align*}
	\tilde{\varphi}_{P}\left(  \xi,h,g\right)+r_{n}^{-1}\psi(\xi,h,g) \le \sup_{\left(  h,g\right)  \in\mathcal{\bar{H}}\times\mathcal{G}}\{\hat
	{\varphi}_{P}\left(  \xi,h,g\right)  +r_{n}^{-1}\psi\left(  \xi,h,g\right)
	\}.
	\end{align*}
	This holds for all $(h,g)\in\bar{\mathcal{H}}\times\mathcal{ G}$, which
	 implies that
	\[
	\sup_{\left(  h,g\right)  \in\mathcal{\bar{H}}\times\mathcal{G}}\{\hat
	{\varphi}_{P}\left(  \xi,h,g\right)  +r_{n}^{-1}\psi\left(  \xi,h,g\right)
	\}\geq\sup_{\left(  h,g\right)  \in\mathcal{\bar{H}}\times\mathcal{G}}%
	\{\tilde{\varphi}_{P}\left(  \xi,h,g\right)  +r_{n}^{-1}\psi\left(  \xi
	,h,g\right)  \}.
	\]
	Thus (iii) is proved.
	
	\textbf{Last}, we prove (iv). Since $\varphi_P(\xi,\cdot,\cdot)$ is continuous, we have that
	\begin{align*}
	&\sup_{\left(  \xi,h,g\right)  \in\Xi\times\mathcal{\bar{H}}\times\mathcal{G}}
	  \left\vert \tilde{\varphi}_{P}(\xi,h,g) + r_n^{-1}\varphi(\xi,h,g)-\varphi_{P}(\xi,h,g)\right\vert\\
	\leq & \sup_{\left(  \xi,h,g\right)  \in\Xi\times\mathcal{\bar{H}}\times
		\mathcal{G}}\left\vert \hat{\varphi}_{P}(\xi,h,g)-\varphi_{P}(\xi,h,g)\right\vert + r_n^{-1}\Vert\varphi\Vert_{\infty}.
	\end{align*}
	(iv) follows from the facts that $\hat{\varphi}_{P}\rightarrow
	\varphi_{P}$ almost uniformly, as mentioned at the beginning of the proof, and $\Vert\varphi\Vert_{\infty}<\infty$. 
	
	Fix $\varepsilon>0$. By property (iv), let $\Omega_0\subset \Omega$ be a measurable set such that $\mathbb{P}\left(  \Omega_{0}\right)  \ge 1-\varepsilon$ and $\tilde{\varphi}_{P}+r_n^{-1}\varphi\rightarrow\varphi_{P}$ uniformly on $\Omega_0$ for all $\varphi\in \ell^{\infty}\left(  \Xi\times\mathcal{\bar{H}}\times\mathcal{G}\right)  $. Let $\psi_n\in\mathbb{D}_n(\omega_{n})$, $\omega_{n}\in\Omega_0$, and $\psi\in C\left(  \Xi\times\mathcal{\bar{H}}\times\mathcal{G}\right)  $ be arbitrary maps such that $\psi_n\to\psi$. By property (i) that we proved above,
	we have that $\Psi\left(  \xi,\tilde{\varphi}_{P}+r_{n}^{-1}\psi\right)
	\neq\varnothing$ for all $\omega\in \Omega_0$, all $n$, and all $\xi\in\Xi$. It is easy to show that because $\psi_{n}\rightarrow\psi$ in
	$\ell^{\infty}\left(  \Xi\times\mathcal{\bar{H}}\times\mathcal{G}\right)  $,
\begin{align*}
	&  \sup_{\xi\in\Xi}\left\vert
	\begin{array}
		[c]{c}%
		\sup_{\left(  h,g\right)  \in\mathcal{\bar{H}}\times\mathcal{G}}\left\{
		\hat{\varphi}_{P}(\omega_{n})\left(  \xi,h,g\right)  +r_{n}^{-1}\psi
		_{n}\left(  \xi,h,g\right)  \right\}  \\
		-\sup_{\left(  h,g\right)  \in\mathcal{\bar{H}}\times\mathcal{G}}\left\{
		\hat{\varphi}_{P}(\omega_{n})\left(  \xi,h,g\right)  +r_{n}^{-1}\psi\left(
		\xi,h,g\right)  \right\}
	\end{array}
	\right\vert \\
	\leq &  \,r_{n}^{-1}\sup_{\left(  \xi,h,g\right)  \in\Xi\times\mathcal{\bar
			{H}}\times\mathcal{G}}\left\vert \psi_{n}(\xi,h,g)-\psi(\xi,h,g)\right\vert
	=o\left(  r_{n}^{-1}\right)  .
\end{align*}
    Since $\tilde{\varphi}_{P}+r_{n}^{-1}\psi$ converges to $\varphi_{P}$ uniformly on $\Omega_0$, by Lemma \ref{lemma.upper hemicontinuity of Psi} there is a sequence $\delta_{n}\downarrow0$ such that
	$
	\Psi\left(  \xi,\tilde{\varphi}_{P}(\omega)+r_{n}^{-1}\psi\right)  \subset\Psi\left(
	\xi,\varphi_{P}\right)  ^{\delta_{n}}
	$
	for all
	$\xi\in\Xi$ and all $\omega\in\Omega_0$. (By Lemma \ref{lemma.upper hemicontinuity of Psi}, $\delta_n$ does not depend on $\xi\in\Xi$ or on $\omega\in\Omega_0$.)	
	Since $\mathcal{S}(\varphi_{P})=0$ by Lemma \ref{lemma.Sphi under null}, we have that for all $\xi\in\Xi$,
	\begin{align}\label{eq.Psi equivalence}
	\Psi\left(  \xi,\varphi_{P}\right)  =\{  \left(  h,g\right)
	\in{\bar{\mathcal{H}}\times\mathcal{G}}:\phi_{P}\left(  h,g\right)  =0\}.
	\end{align}
	By Lemma \ref{lemma.Sphi under null} and the constructions of $\hat{\varphi}_P$ and $\tilde{\varphi}_P$, we also have that for all $\omega$, $\hat{\varphi}_{P} \le0$ and $\tilde{\varphi}_{P} \le0$ on
	$\Xi\times\bar{\mathcal{H}}\times\mathcal{G}$, and $\hat{\varphi}_{P}\left(  \xi,\cdot,\cdot\right) = 0$ on $\Psi\left(  \xi,\varphi_{P}\right) $. 
	Thus for every $\xi\in\Xi$,
	\begin{align*}
	 &\sup_{\left(  h,g\right)  \in\mathcal{\bar{H}%
		}\times\mathcal{G}}\left\{  \hat{\varphi}_{P}(\omega_{n})\left(  \xi,h,g\right)  +r_{n}
	^{-1}\psi\left(  \xi,h,g\right)  \right\}\\
	\ge & \sup_{\left(  h,g\right)  \in \Psi\left(  \xi,\varphi_{P}\right) }\left\{  \hat{\varphi}_{P}(\omega_{n})\left(  \xi,h,g\right)  +r_{n}%
	^{-1}\psi\left(  \xi,h,g\right)  \right\}=\sup_{\left(  h,g\right)  \in
		\Psi\left(  \xi,\varphi_{P}\right)  } r_{n}^{-1}\psi\left(  \xi,h,g\right).
	\end{align*}
	By property (iii) of $\tilde{\varphi}_P$, together with the results shown above, we have that
	\begin{align*}
		&  \sup_{\xi\in\Xi}\left\vert \sup_{\left(  h,g\right)  \in\mathcal{\bar{H}%
			}\times\mathcal{G}}\left\{  \hat{\varphi}_{P}(\omega_{n})\left(
		\xi,h,g\right)  +r_{n}^{-1}\psi\left(  \xi,h,g\right)  \right\}
		-\sup_{\left(  h,g\right)  \in\Psi\left(  \xi,\varphi_{P}\right)  }r_{n}%
		^{-1}\psi\left(  \xi,h,g\right)  \right\vert \\
		= &  \sup_{\xi\in\Xi}\left\{
		\begin{array}
			[c]{c}%
			\sup_{\left(  h,g\right)  \in\Psi\left(  \xi,\tilde{\varphi}_{P}(\omega
				_{n})+r_{n}^{-1}\psi\right)  }\left\{  \tilde{\varphi}_{P}(\omega_{n})\left(
			\xi,h,g\right)  +r_{n}^{-1}\psi\left(  \xi,h,g\right)  \right\}  \\
			-\sup_{\left(  h,g\right)  \in\Psi\left(  \xi,\varphi_{P}\right)  }r_{n}%
			^{-1}\psi\left(  \xi,h,g\right)
		\end{array}
		\right\}  \\
		\leq &  \sup_{\xi\in\Xi}\left\{  \sup_{\left(  h,g\right)  \in\Psi\left(
			\xi,\varphi_{P}\right)  ^{\delta_{n}}}\left\{  \tilde{\varphi}_{P}(\omega
		_{n})\left(  \xi,h,g\right)  +r_{n}^{-1}\psi\left(  \xi,h,g\right)  \right\}
		-\sup_{\left(  h,g\right)  \in\Psi\left(  \xi,\varphi_{P}\right)  }r_{n}%
		^{-1}\psi\left(  \xi,h,g\right)  \right\}  .
	\end{align*}
	Then by the definition of $\Psi(\xi,\varphi_P)^{\delta_n}$,
	\begin{align*}
	& \sup_{\xi\in\Xi}\left\{  \sup_{\left(  h,g\right)  \in\Psi\left(
		\xi,\varphi_{P}\right)  ^{\delta_{n}}}\left\{  \tilde{\varphi}_P(\omega_{n})\left(
	\xi,h,g\right)  +r_{n}^{-1}\psi\left(  \xi,h,g\right)  \right\}
	-\sup_{\left(  h,g\right)  \in\Psi\left(  \xi,\varphi_{P}\right)  }r_{n}%
	^{-1}\psi\left(  \xi,h,g\right)  \right\}  \\
	\leq &  \sup_{\xi\in\Xi}\left\{\sup_{\rho_{P}\left(  \left(  h_{1},g_{1}\right)
		,\left(  h_{2},g_{2}\right)  \right)  \leq\delta_{n}}r_{n}^{-1}\left\vert
	\psi\left(  \xi,h_{1},g_{1}\right)  -\psi\left(  \xi,h_{2},g_{2}\right)
	\right\vert\right\} =o(  r_{n}^{-1}).
	\end{align*}
	Finally, combining all the results above, we can conclude that
	\begin{align*}
	\sup_{\xi\in\Xi}\left\vert
	\mathcal{S}\left(  \hat
	{\varphi}_P(\omega_{n})+r_{n}^{-1}\psi_{n}\right)  \left(  \xi\right) 
	-r_{n}^{-1}\sup_{\left(  h,g\right)  \in {\Psi}\left(  \xi,\varphi_{P}\right)  }\psi\left(  \xi,h,g\right)
	\right\vert 
	= o\left(  r_{n}^{-1}\right)  .
	\end{align*}
	This implies that
	\begin{align*}
	&  \left\vert g_{n}(\omega_{n})  \left(  \psi
	_{n}\right)    -\int_{\Xi}\sup_{\left(  h,g\right)
		\in\Psi\left(  \xi,\varphi_{P}\right)  }\psi\left(  \xi,h,g\right)
	\,\mathrm{d}\nu\left(  \xi\right)  \right\vert \\
	\leq & \int_{\Xi}\left\vert r_{n}  \mathcal{S}\left(  \hat
	{\varphi}_P(\omega_{n})+r_{n}^{-1}\psi_{n}\right)  \left(  \xi\right)    -\sup_{\left(  h,g\right)
		\in\Psi\left(  \xi,\varphi_{P}\right)  }\psi\left(  \xi,h,g\right)
	\right\vert \,\mathrm{d}\nu\left(  \xi\right)  =o\left(  1\right)  .
	\end{align*} 
\end{proof} 

\begin{proof}[Proof of Theorem \ref{thm.weak convergence SK}] 
	By \eqref{eq.convergence phi}, $\sqrt{n}(\hat{\phi}_{P_n}-\phi_P)\leadsto \mathcal{L}^{\prime}_P(\mathbb{ G}_P+Q_0)$, where $\mathcal{L}^{\prime}_P(\mathbb{ G}_P+Q_0)$ is tight as shown in the proof of Lemma \ref{lemma.weak convergence phi_K}. By Lemma \ref{lemma.almost uniform convergence phi sigma}, $\mathcal{M}(\hat{\sigma}_{P_n})\to \mathcal{M}({\sigma}_P)$ almost uniformly, and hence this convergence is also in outer probability by Lemma 1.9.3(ii) of \citet{van1996weak}. By Lemma 1.10.2(iii) of \citet{van1996weak}, $\mathcal{M}(\hat{\sigma}_{P_n})\leadsto \mathcal{M}({\sigma}_P)$. By Example 1.4.7 (Slutsky's lemma) of \citet{van1996weak}, we have that  $(\sqrt{n}(\hat{\phi}_{P_n}-\phi_P),\mathcal{M}(\hat{\sigma}_{P_n}))\leadsto(\mathcal{L}^{\prime}_P(\mathbb{ G}_P+Q_0),\mathcal{M}({\sigma}_P))$. Let $\ell^{\infty}(\Xi\times\bar{\mathcal{H}}\times\mathcal{G})^+=\{\psi\in\ell^{\infty}(\Xi\times\bar{\mathcal{H}}\times\mathcal{G}):\Vert 1/\psi \Vert_{\infty}<\infty\}$. Define a map $f:\ell^{\infty}(\bar{\mathcal{H}}\times\mathcal{G})\times\ell^{\infty}(\Xi\times\bar{\mathcal{H}}\times\mathcal{G})^+\to \ell^{\infty}(\Xi\times\bar{\mathcal{H}}\times\mathcal{G})$ by $f(\varphi,\psi)=\varphi/\psi$ for all $(\varphi,\psi)\in\ell^{\infty}(\bar{\mathcal{H}}\times\mathcal{G})\times\ell^{\infty}(\Xi\times\bar{\mathcal{H}}\times\mathcal{G})^+$. Clearly, $(\mathcal{L}^{\prime}_P(\mathbb{ G}_P+Q_0),\mathcal{M}({\sigma}_P))$ takes its values in $\ell^{\infty}(\bar{\mathcal{H}}\times\mathcal{G})\times\ell^{\infty}(\Xi\times\bar{\mathcal{H}}\times\mathcal{G})^+$. It is easy to show that $f$ is continuous under the metric $\Vert (\varphi,\psi)-(\varphi^{\prime},\psi^{\prime})\Vert=\Vert \varphi-\varphi^{\prime}\Vert_{\infty}+\Vert \psi-\psi^{\prime}\Vert_{\infty}$. By Theorem 1.3.6 (continuous mapping) of \citet{van1996weak}, 
	\begin{align*}
		f(\sqrt{n}(\hat{\phi}_{P_n}-\phi_P),\mathcal{M}(\hat{\sigma}_{P_n}))=\frac{\sqrt{n}(\hat{\phi}_{P_n}-\phi_P)}{\mathcal{M}(\hat{\sigma}_{P_n})}\leadsto \frac{\mathcal{L}^{\prime}_P(\mathbb{ G}_P+Q_0)}{\mathcal{M}({\sigma}_{P})}.
	\end{align*}
	By Lemma \ref{lemma.Sphi under null}, we have that $\mathcal{I}\circ\mathcal{S}\left(  \phi_P/\mathcal{M}\left(  \hat{\sigma}_{P_n}  \right) \right)=0$. Then by Theorem \ref{lemma.random delta method}(ii) and Lemma \ref{lemma.random differentiability}, together with the continuity of $\mathcal{I}\circ\mathcal{S}_{\Psi\left(  \xi,\varphi_{P}\right)}$ under $\Vert\cdot\Vert_{\infty}$,  we have
	\begin{align}\label{eq.weak convergence IS1}
		\sqrt{n}\left\{  \mathcal{I}\circ{\mathcal{S}}\left(  \frac{\hat{\phi}_{P_n}
		}{\mathcal{M}\left(  {\hat{\sigma}_{P_n}}\right)  }\right)  -\mathcal{I}\circ
		{\mathcal{S}}\left(  \frac{{\phi}_{P}}{\mathcal{M}\left(  {\hat{\sigma}_{P_n}
			}\right)  }\right)  \right\}  \leadsto\mathcal{I}\circ\mathcal{S}_{\Psi\left(  \xi,\varphi_{P}\right)}\left(  \frac{\mathcal{L}^{\prime}_P(\mathbb{ G}_P+Q_0)}
		{\mathcal{M}\left(  {\sigma}_P\right)    }\right).
	\end{align}
	By Lemma \ref{lemma.almost uniform convergence phi sigma}, $T_n/n\to\Lambda(P)$ almost uniformly. Then by Lemmas 1.9.3(ii) and 1.10.2(iii), Example 1.4.7 (Slutsky's lemma), and Theorem 1.3.6 (continuous mapping) of \citet{van1996weak}, together with \eqref{eq.weak convergence IS1}, we have that 
	\begin{align*}
		\sqrt{\frac{T_n}{n}}\cdot\sqrt{n}\left\{  \mathcal{I}\circ{\mathcal{S}}\left(  \frac{\hat{\phi}_{P_n}
		}{\mathcal{M}\left(  {\hat{\sigma}_{P_n}}\right)  }\right)    \right\}  \leadsto\mathcal{I}\circ\mathcal{S}_{\Psi\left(  \xi,\varphi_{P}\right)}\left( \frac{\mathbb{G}}{\mathcal{M}(\sigma_P)}  \right),
	\end{align*}
	where $\mathbb{G}=\sqrt{\Lambda(P)}\mathcal{L}^{\prime}_P(\mathbb{ G}_P+Q_0)$ as in Lemma \ref{lemma.weak convergence phi_K}.
	By Lemma \ref{lemma.Sphi under null}, we have that $\Psi(\xi,\varphi_P)=\Psi_{\bar{\mathcal{H}}\times\mathcal{G}}$ defined by \eqref{eq.Psi_HG} for all $\xi\in\Xi$ under the assumptions.
	
	If $\mathcal{D}$ is a finite set with $\mathcal{D}=\{d_1,\ldots,d_J\}$, then under null  $\mathcal{I}\circ\mathcal{S}_{\Psi_{\bar{\mathcal{H}}\times\mathcal{G}}}\left(  {\mathbb{G}}/{\mathcal{M}(\sigma_{P})}\right)=\mathcal{I}\circ\mathcal{S}_{\Psi_{{\mathcal{H}}
	\times\mathcal{G}}}\left(  {\mathbb{G}}/{\mathcal{M}(\sigma_{P})}\right)$ almost surely, because it can be shown that in this special case $\Psi_{\bar{\mathcal{H}}\times\mathcal{G}}$ is equal to the closure of $\Psi_{{\mathcal{H}}	\times\mathcal{G}}$ in $\bar{\mathcal{H}}\times\mathcal{ G}$ under $\rho_P$ and $\mathbb{G}/\mathcal{M}(\sigma_P)$ is continuous under $\rho_P$ almost surely for every fixed $\xi$. We summarize this in the following. 
	
	By Lemma \ref{lemma.Sphi under null}, $\phi_{P}\left(  h,g\right)  \le 0$ for all $\left(  h,g\right)  \in\mathcal{\bar{H}}\times\mathcal{G}$, and there exists
	$(  h^0,g^0)  \in\mathcal{\bar{H}}\times\mathcal{G}$ with $g^0=(
	g^0_{1},g^0_{2})  $ such that $\phi_{P}(  h^0,g^0)  =0$. \textbf{First}, we show that if
	$h^0=\left(  -1\right)  ^{d}\cdot1_{A\times\left\{  d\right\}  \times\mathbb{R}}$, where $d\in\left\{  0,1\right\}  $ and $A$ is a half-closed interval or an
	open interval, then for every closed interval $B$ such that $B\subset A$, we
	have that $\phi_{P}(\tilde{h},g^0)=0$ with $\tilde{h}=\left(  -1\right)  ^{d}
	\cdot 1_{B\times\left\{  d\right\}  \times\mathbb{R}}$. 
	Suppose, by way of contradiction, that $A=(a_{1},a_{2})$ and $B=\left[  b_{1},b_{2}\right]  $ with
	$a_{1}<b_{1}$, $a_{2}>b_{2}$, and $\phi_{P}(\tilde{h},g^0)<0$ with
	$\tilde{h}=\left(  -1\right)  ^{d}\cdot1_{B\times\left\{  d\right\}  \times
		\mathbb{R}}$. Let $h_{L}=\left(  -1\right)  ^{d}\cdot 1_{(a_{1},b_{1})\times\left\{
		d\right\}  \times\mathbb{R}}$ and $h_{R}=\left(  -1\right)  ^{d}\cdot 1_{(b_{2},a_{2})\times\left\{
		d\right\}  \times\mathbb{R}}$. Then by the definition of $\phi_{P}$,
	\begin{align*}
		\phi_{P}(h^0,g^0)  & =\frac{P(  h^0\cdot g_{2}^0)  }{P(  g_{2}^0)
		}-\frac{P(  h^0\cdot g_{1}^0)  }{P(  g_{1}^0)  }%
		=\frac{P((h_{L}+\tilde{h}+h_{R})\cdot g_{2}^0)}{P(  g_{2}^0)  }%
		-\frac{P((h_{L}+\tilde{h}+h_{R})\cdot g_{1}^0)}{P(  g_{1}^0)  }\\
		& =\phi_{P}(\tilde{h},g^0)+\phi_{P}(h_{L},g^0)+\phi_{P}(h_{R},g^0).
	\end{align*}
	Since $\phi_{P}(h^0,g^0)=0$ but  $\phi_{P}(\tilde{h},g^0)<0$, we have $\phi_{P}
	(h_{L},g^0)+\phi_{P}(h_{R},g^0)>0$. This implies that either $\phi_{P}(h_{L},g^0)>0$ or
	$\phi_{P}(h_{R},g^0)>0$. However, since $(h_{L},g^0),(h_{R},g^0)\in\mathcal{\bar{H}}\times\mathcal{G}$, Lemma \ref{lemma.Sphi under null} shows that both $\phi_{P}(h_{L},g^0)$ and
	$\phi_{P}(h_{R},g^0)$ are nonpositive. This is a
	contradiction. When $A$ is a half-closed interval, we can show analogously that the
	claim is true. 
	\textbf{Second}, we show that if $h^0=1_{
		\mathbb{R}\times C\times\mathbb{R}}$ with $C=\left(  -\infty,c\right)  $ for some $c\in\mathbb{R}$, then there is a sequence of sets $C_k=(-\infty,c_k]$ with $c_k\uparrow c$ such that $\phi_{P}
	(h^k,g^0)=0$ with $h^k=1_{\mathbb{R}\times C_k\times\mathbb{R}}$. 
	By assumption, $\mathcal{D}$ is a finite set. Under Assumption \ref{ass.independent data}, $D$ is a discrete random variable with $D\in \mathcal{D}$ under $P_{n}$. Then $D\in\mathcal{D}$ under $P$ by Lemma \ref{lemma.almost uniform convergence phi sigma}, and the claim holds.
	
	The above results imply that $\Psi_{\bar{\mathcal{H}}\times\mathcal{G}}\subset \overline{\Psi_{{\mathcal{H}}	\times\mathcal{G}}}$, where  $\overline{\Psi_{{\mathcal{H}}	\times\mathcal{G}}}$ is the closure of ${\Psi_{{\mathcal{H}}	\times\mathcal{G}}}$ in $\bar{\mathcal{H}}	\times\mathcal{G}$ under $\rho_P$. By \eqref{eq.Psi_HG} and Lemma \ref{lemma.phi is continuous},  $\Psi_{\bar{\mathcal{H}}\times\mathcal{G}}= \overline{\Psi_{{\mathcal{H}}	\times\mathcal{G}}}$. By Lemma \ref{lemma.weak convergence phi_K}, $\mathbb{G}$ almost surely has a continuous path under $\rho_P$. By Lemma \ref{lemma.phi is continuous}, $\sigma_P$ is continuous under $\rho_P$. The result follows from the continuity of $\mathbb{G}/\mathcal{M}(\sigma_P)$ under $\rho_P$ for every fixed $\xi\in\Xi$.
\end{proof}

\begin{remark}\label{remark.application of extended delta method}
	If the $H_{0}$ in \eqref{eq.null 1} is true with $Q=P_n$ for all $n$, we have that $\mathcal{S}\left({\phi}_{P}/\mathcal{M}( {\sigma}_{P})\right)=0$ (see Lemma \ref{lemma.Sphi under null}). Thus it suffices to find the asymptotic distribution of 
	\begin{align}\label{eq.asymptotic distribution S1}
		\sqrt{n}\mathcal{I}\circ \mathcal{S}_{\mathcal{H}\times\mathcal{G}}\left(  \frac{\hat{\phi}_{P_n}}{\mathcal{M}(
			\hat{\sigma}_{P_n})}\right)=\sqrt{n} \left\{\mathcal{I}\circ \mathcal{S}\left(  \frac{\hat{\phi}_{P_n}}{\mathcal{M}(
			\hat{\sigma}_{P_n})}\right) -\mathcal{I}\circ\mathcal{S}\left(  \frac{{\phi}_P}{\mathcal{M}
			({\sigma}_P)}\right) \right\}.
	\end{align} 
	If we can find the asymptotic distribution of $\sqrt{n} (   \hat{\phi}_{P_n} / \mathcal{M}(\hat{\sigma}_{P_n}) -  \phi_P /\mathcal{M}( {\sigma}_P) )$ and the ``derivative'' of $\mathcal{I}\circ\mathcal{S}$ (see, for example, the definition of Hadamard directional derivative in \citet{shapiro1990concepts} and  \citet{fang2014inference}), then by the delta method of \citet{fang2014inference}, it is straightforward to obtain the asymptotic distribution of \eqref{eq.asymptotic distribution S1}. However, establishing the limiting distribution of   $\sqrt{n} (   \hat{\phi}_{P_n} / \mathcal{M}(\hat{\sigma}_{P_n}) -  \phi_P /\mathcal{M}( {\sigma}_P) )$ is technically tricky. By the constructions of $\phi_P$ and $\sigma_P$, we can view $\phi_P /\mathcal{M}( {\sigma}_P) $ as a map of $P$. Specifically, let $\mathcal{V}_0=\{v:v=h\cdot g_l \text{ or } v=h^2\cdot g_l \text{  for some } h\in\bar{\mathcal{H}} \text{ and } g_l\in\mathcal{ G}_K\}$ and $\mathbb{D}_{Q}=\{Q\in\ell^{\infty}({\mathcal{V}_0}\cup \mathcal{ G}_K ): Q(h\cdot g_l)/Q(g_l) \text{  and } Q(h^2\cdot g_l)/Q(g_l) \text{ exist for all }  h\in\bar{\mathcal{H}} \text{ and } g_l\in\mathcal{ G}_K\}$. Then we extend the definitions of $\phi_Q$ and $\sigma_Q$ for all $Q\in\mathcal{P}$, that is, the $\phi_Q$ defined in \eqref{eq.phi_Q} and the $\sigma_Q$ defined in \eqref{eq.stat variance multi}, to all $Q\in\mathbb{ D}_Q$. Clearly, $\mathcal{P}\subset\mathbb{ D}_Q$ by \eqref{eq.0timesinfinity}. Define a map $\mathcal{T}:\mathbb{D}_{Q}\to\ell^{\infty}(\Xi\times\bar{\mathcal{H}}\times\mathcal{G})$ by
	\begin{align*}
		\mathcal{T}(Q)(\xi,h,g)=\frac{\phi_Q(h,g)} {\mathcal{M}(\sigma_Q)(\xi,h,g)}
	\end{align*}
	for all $Q\in\mathbb{D}_Q$ and $(\xi,h,g)\in\Xi\times\bar{\mathcal{H}}\times\mathcal{G}$. Now we have that $\mathcal{T}(P)=\phi_P /\mathcal{M}( {\sigma}_P)$ and $\mathcal{T}(\hat{P}_n)=\hat{\phi}_{P_n} /\mathcal{M}( \hat{\sigma}_{P_n})$. Suppose we have weak convergence of $\sqrt{n}(\hat{P}_n - P)$ in some suitable space. Then if $\mathcal{T}$ is Hadamard (directionally) differentiable, by delta method we can establish weak convergence of \begin{align}\label{eq.asymptotic distribution T}
		\sqrt{n} \left(   \frac{\hat{\phi}_{P_n}}{\mathcal{M}( \hat{\sigma}_{P_n})} -  \frac{{\phi}_P}{\mathcal{M}({\sigma}_P)} \right)= \sqrt{n}\left(\mathcal{T}(\hat{P}_n)-\mathcal{T}(P)\right).
	\end{align}
	Unfortunately, however, $\mathcal{T}$ is nondifferentiable, because of the nondifferentiability of the $\mathcal{M}$ defined in \eqref{eq.M} (to the best of our knowledge, the directional derivative of $\mathcal{M}$ may not exist even when $\Xi$ is a singleton), and hence it is not straightforward to show the convergence of $\sqrt{n}(
	\mathcal{T}(  \hat{P}_{n})  -\mathcal{T}\left(  P\right)  )  $. 
	The random denominator problem also arises in other testing issues. See, for example, \citet{bugni2017inference}. If the random denominator does not need to be bounded away from $0$ in the test statistic as in \citet{bugni2017inference}, we will not have the undifferentiability issue caused by $\mathcal{M}$.
	Inspired by \citet{kitagawa2015test},
	with the asymptotic distribution of $\sqrt{n} (   \hat{\phi}_{P_n} / \mathcal{M}(\hat{\sigma}_{P_n}) -  \phi_P /\mathcal{M}( \hat{\sigma}_{P_n}) )$ (which can be obtained by using Slutsky's theorem), we can instead establish the asymptotic distribution of
	\begin{align}\label{eq.asymptotic distribution S}
		\sqrt{n} \left\{\mathcal{I}\circ \mathcal{S}\left(  \frac{\hat{\phi}_{P_n}}{\mathcal{M}(
			\hat{\sigma}_{P_n})}\right) -\mathcal{I}\circ\mathcal{S}\left(  \frac{{\phi}_P}{\mathcal{M}
			(\hat{\sigma}_{P_n})}\right) \right\},
	\end{align}
	where $\mathcal{S}\left({\phi}_{P}/\mathcal{M}( \hat{\sigma}_{P_n})\right)=0$ by Lemma \ref{lemma.Sphi under null} if the $H_{0}$ in \eqref{eq.null 1} is true with $Q=P_n$ for all $n$.
	However, existing delta methods cannot be used to establish the asymptotic distribution of \eqref{eq.asymptotic distribution S} either. Since $ {\phi}_{P}/\mathcal{M}( \hat{\sigma}_{P_n})$ is a random element, delta methods such as Theorem 3.9.4 or Theorem 3.9.5 of \citet{van1996weak}, or Theorem 2.1 of \citet{fang2014inference}, do not work in this case. To overcome the technical complications due to the random element $ {\phi}_{P}/\mathcal{M}( \hat{\sigma}_{P_n})$, we provide the extended continuous mapping theorem and the extended delta method elaborated by Theorems \ref{lemma.random continuous mapping} and \ref{lemma.random delta method}, respectively.
\end{remark}

We now introduce the notation for the bootstrap elements. 
Let $(W_{n1},\ldots,W_{nn})$ be a vector of random multinomial weights independent of $\left\{
\left(  Y_{i},D_{i},Z_{i}\right)  \right\}  _{i=1}^{n}$ for all $n$. As defined
in \eqref{eq.defPn}, $\hat{P}_{n}$ is the empirical measure of an i.i.d.\ sample $\left\{
\left(  Y_{i},D_{i},Z_{i}\right)  \right\}  _{i=1}^{n}$ from probability distribution $P_n$. Given the sample values, the $\{(  \hat{Y}_{i},\hat{D}_{i},\hat{Z}_{i})  \}  _{i=1}^{n}$ introduced in Section \ref{subsbusec.test procedure} is
an i.i.d.\ sample from $\hat{P}_{n}$. We can write the empirical measure of $\{( \hat{Y}_{i},\hat{D}_{i},\hat{Z}_{i})  \}   _{i=1}^{n}$, given sample $\left\{
\left(  Y_{i},D_{i},Z_{i}\right)  \right\}  _{i=1}^{n}$, as $\hat{P}_{n}^{B}=n^{-1}
\sum_{i=1}^{n}W_{ni}\delta_{\left(  Y_{i},D_{i},Z_{i}\right)}$, where $\delta_{\left(  Y_{i},D_{i},Z_{i}\right)}$ is a Dirac measure centered at $\left(  Y_{i},D_{i},Z_{i}\right)$. Given the $\hat{\phi}_{P_n}^{B}$, $T_n^B$, and $\hat{\sigma}_{P_n}^{B}$ defined in Section \ref{subsbusec.test procedure},
$\hat{\phi}_{P_n}^{B}/\mathcal{M}(\hat{\sigma}_{P_n}^{B})$ is a map of
$\left\{  \left(  Y_{i},D_{i},Z_{i},W_{ni}\right)  \right\}  _{i=1}^{n}$ to the space
$\ell^{\infty}\left(  \Xi\times\bar{\mathcal{H}}\times\mathcal{G}\right)  $.  

We follow Section 3.6 of \citet{van1996weak} and \eqref{eq.conditional expectation} to define the conditional outer expectations. When we compute the outer expectations as in \eqref{eq.conditional expectation}, independence is understood in terms of a product space. Under Assumptions \ref{ass.independent data} and \ref{ass.probability path}, each term $(Y_i,D_i,Z_i)$ of the sequence $\left\{\left(  Y_{i},D_{i},Z_{i}\right)  \right\}  _{i=1}^{\infty}$ has probability distribution $P$. Let $\left\{\left(  Y_{i},D_{i},Z_{i}\right)  \right\}  _{i=1}^{\infty}$ be the coordinate projections on the first $\infty$ coordinates of the product space $((\mathbb{R}^3)^{\infty},\mathcal{B}_{\mathbb{R}^3}^{\infty},P^{\infty})\times(\mathcal{W},\mathcal{C},P_W)$, and let the multinomial vectors $W$ depend on the last factor only. For each real-valued map $T$ on $((\mathbb{R}^3)^{\infty},\mathcal{B}_{\mathbb{R}^3}^{\infty},P^{\infty})\times(\mathcal{W},\mathcal{C},P_W)$, we can take $(\Omega_1,\mathcal{A}_1,\mathbb{P}_1)=((\mathbb{R}^3)^{\infty},\mathcal{B}_{\mathbb{R}^3}^{\infty},P^{\infty})$ and $(\Omega_2,\mathcal{A}_2,\mathbb{P}_2)=(\mathcal{W},\mathcal{C},P_W)$ and define a real-valued map $E_W^{\ast}[T]$ on $((\mathbb{R}^3)^{\infty},\mathcal{B}_{\mathbb{R}^3}^{\infty},P^{\infty})$ by
\begin{align}\label{eq.conditional expectation 2}
E^{\ast}_W[T](\left\{\left(  Y_{i},D_{i},Z_{i}\right)  \right\}  _{i=1}^{\infty})=E^{\ast}_2[T ](\left\{\left(  Y_{i},D_{i},Z_{i}\right)  \right\}  _{i=1}^{\infty})
\end{align}
for each sequence $\left\{\left(  Y_{i},D_{i},Z_{i}\right)  \right\}  _{i=1}^{\infty} \in (\mathbb{R}^3)^{\infty}$, where $E^{\ast}_2[T]$ is defined as in \eqref{eq.conditional expectation}. We call the left-hand side of \eqref{eq.conditional expectation 2} the conditional outer expectation of $T$ given the sequence $\left\{\left(  Y_{i},D_{i},Z_{i}\right)  \right\}  _{i=1}^{\infty}$. Since $E^{\ast}_W[T]$ is a real-valued map on $((\mathbb{R}^3)^{\infty},\mathcal{B}_{\mathbb{R}^3}^{\infty},P^{\infty})$, we can compute its outer and inner integrals (expectations) with respect to $((\mathbb{R}^3)^{\infty},\mathcal{B}_{\mathbb{R}^3}^{\infty},P^{\infty})$. For simplicity of notation, we write them as $E^{\ast}[E^{\ast}_W[T]]$ and $E_{\ast}[E^{\ast}_W[T]]$, respectively. 

If $T(\left\{\left(  Y_{i},D_{i},Z_{i}\right)  \right\}  _{i=1}^{\infty},\cdot)$ is a measurable integrable map on $(\mathcal{W},\mathcal{C},P_W)$ for every given sequence $\left\{\left(  Y_{i},D_{i},Z_{i}\right)  \right\}  _{i=1}^{\infty}$, we write $E_W[T]$ for $E^{\ast}_W[T]$ and call $E_W[T](\left\{\left(  Y_{i},D_{i},Z_{i}\right)  \right\}  _{i=1}^{\infty})$ the conditional expectation  of $T$ given the sequence $\left\{\left(  Y_{i},D_{i},Z_{i}\right)  \right\}  _{i=1}^{\infty}$. The conditional inner expectation is defined analogously.
If
$\mathbb{D}$ is a metric space with metric $d$, we define
\[
\mathrm{BL}_{1}\left(  \mathbb{D}\right)  =\left\{  f:\mathbb{D}\rightarrow
\mathbb{R}
:\left\Vert f\right\Vert _{\infty}\le 1,\left\vert f\left(  x_1\right)
-f\left( x_2\right)  \right\vert \leq d(x_1,x_2)\text{
	for all }x_1,x_2\in\mathbb{D}\right\}  .
\] 

\begin{lemma} \label{lemma.bootstrap properties}
	Suppose Assumptions \ref{ass.independent data} and \ref{ass.probability path} hold.
	\begin{enumerate}[label=(\roman{*})]
		\item $\sqrt{T_n^B}(  \hat{\phi}_{P_n}^{B}-\hat{\phi}_{P_n})  /  \mathcal{M}(\hat{\sigma
		}_{P_n}^{B})  $ satisfies
		\begin{align}\label{eq.bootstrap property 1}
		\sup_{f\in \mathrm{BL}_{1}(  \ell^{\infty}( \Xi\times \bar{\mathcal{H}}\times\mathcal{ G})
			)  }\left\vert E_W\left[  f\left( \frac{ \sqrt{T_n^B}(  \hat{\phi}_{P_n}^{B}
			-\hat{\phi}_{P_n})}{  \mathcal{M}(\hat{\sigma}_{P_n}^{B}) } \right) \right]
		-E\left[  f\left(  \frac{\mathbb{G}_0}{\mathcal{M}(\sigma_P)}\right)  \right]  \right\vert \to 0
		\end{align}
		in outer probability, where $\mathbb{ G}_0=\sqrt{\Lambda(P)}\cdot\mathcal{L}_{P}^{\prime}(  \mathbb{G}_{P})$ is tight and $\mathbb{G}_P$ is as in Lemma \ref{lemma.weak convergence Pn_hat and convergence Pn};
		\item $\sqrt{T_n^B}(  \hat{\phi}_{P_n}^{B}-\hat{\phi}_{P_n})  /\mathcal{M}(  \hat{\sigma}_{P_n}^{B})\leadsto \mathbb{G}_0/\mathcal{M}(\sigma_P)  $;\footnote{This implies that $\sqrt{T_n^B}(  \hat{\phi}_{P_n}^{B}-\hat{\phi}_{P_n})  /\mathcal{M}(  \hat{\sigma}_{P_n}^{B})$ is asymptotically measurable jointly in $\left\{
			\left(  Y_{i},D_{i},Z_{i}\right)  \right\}  _{i=1}^{\infty}$ and $W$ by Lemma 1.3.8 of \citet{van1996weak}.}
		
		\item For each continuous,
		bounded $f:\ell^{\infty}(  \Xi\times\bar{\mathcal{H}}\times \mathcal{ G})
		\to\mathbb{R}$, $f(  \sqrt{T_n^B}(  \hat{\phi}_{P_n}^{B}-\hat{\phi}_{P_n})  /\mathcal{M}(
		\hat{\sigma}_{P_n}^{B})  )  $ is a measurable function of
		$\{  W_{ni}\}  _{i=1}^{n}$ for every given sequence $\{  (
		Y_{i},D_{i},Z_{i})  \}  _{i=1}^{\infty}$.
	\end{enumerate}

\end{lemma}

\begin{proof}[Proof of Lemma \ref{lemma.bootstrap properties}]
	(i). To explore the conditional property of the bootstrap element $\sqrt{T_n^B}(  \hat{\phi}_{P_n}^{B}-\hat{\phi}_{P_n})  /  \mathcal{M}(\hat{\sigma
	}_{P_n}^{B})  $,  we consider the entire sequence $\left\{\left(  Y_{i},D_{i},Z_{i}\right)  \right\}  _{i=1}^{\infty}$.\footnote{We follow Section 3.6 of \citet{van1996weak} to obtain the conditional property of the bootstrap element $\sqrt{T_n^B}(  \hat{\phi}_{P_n}^{B}-\hat{\phi}_{P_n})  /  \mathcal{M}(\hat{\sigma
	}_{P_n}^{B})  $ given the entire sequence $\left\{\left(  Y_{i},D_{i},Z_{i}\right)  \right\}  _{i=1}^{\infty}$.} Each term $(Y_i,D_i,Z_i)$ in $\left\{\left(  Y_{i},D_{i},Z_{i}\right)  \right\}  _{i=1}^{\infty}$ has probability distribution $P$ under Assumptions \ref{ass.independent data} and \ref{ass.probability path}. Now the $\hat{P}_n$ defined in \eqref{eq.defPn} can be viewed as being computed with the first $n$ elements of $\left\{\left(  Y_{i},D_{i},Z_{i}\right)  \right\}  _{i=1}^{\infty}$ that are distributed according to $P$. By Lemma \ref{lemma.V Donsker}, $\sqrt{n}(\hat{P}_n-P)\leadsto\mathbb{G}_P$ under $P$, where $\mathbb{G}_P$ is the limit shown in Lemma \ref{lemma.weak convergence Pn_hat and convergence Pn}.
	By the construction of $\tilde{\mathcal{V}}$ in \eqref{eq.tilde V},
	$F=1$ is an envelope function of $\tilde{\mathcal{V}}$ and $P^{\ast}(
	\sup_{v\in\tilde{\mathcal{V}}}\vert v-P(  v)  \vert^{2})  <\infty$, where $P^{\ast}$ is the outer probability measure of $P$. By Lemma \ref{lemma.V Donsker}, $\tilde{\mathcal{V}}$ is Donsker. By
	Theorem 3.6.2 of \citet{van1996weak}, we have that
	\begin{align}\label{eq.conditional sup EP convergence}
	\sup_{f\in \mathrm{BL}_{1}(\ell^{\infty}(\tilde{\mathcal{V}}))}\vert E_W[  f\{  \sqrt{n}(  \hat{P}_{n}^{B}-\hat{P}_{n})  \}  ]  -E[  f(  \mathbb{G}_{P})  ]
	\vert \to 0
	\end{align}
	outer almost surely\footnote{As discussed in \citet[p.~183]{van1996weak}, $f\{  \sqrt{n}(  \hat{P}_{n}^{B}-\hat{P}_{n})  \}$ is measurable as a function of the random weights given the values of the sample. Thus we use the conditional expectation $E_W[  f\{  \sqrt{n}(  \hat{P}_{n}^{B}-\hat{P}_{n})  \}  ] $ in \eqref{eq.conditional sup EP convergence}. Similarly, we use the conditional expectation $E_W[  f\{  \sqrt{n}(  \mathcal{L}	(  \hat{P}_{n}^{B})  -\mathcal{L}(  \hat{P}_{n})  )
		\}  ] $ in \eqref{eq.conditional sup EL convergence}.} and
	\begin{align}\label{eq.conditional EP convergence}
	E_W[  f\{  \sqrt{n}(  \hat{P}_{n}^{B}-\hat{P}_{n})  \}
	^{\ast}]  -E_W[f\{  \sqrt{n}(  \hat{P}_{n}^{B}-\hat{P}_{n})  \}  _{\ast}]  \to0
	\end{align}
	almost surely for every $f\in \mathrm{BL}_{1}(\ell^{\infty}(\tilde{\mathcal{V}}))$. Here, the asterisks denote the measurable cover functions with respect to $\{(Y_i,D_i,Z_i)\}_{i=1}^{\infty}$ and $W$ jointly. Then by Lemmas \ref{lemma.L HD}, \ref{lemma.V Donsker}, and \ref{lemma.uniform Glivenko-Cantelli} in this paper, and Theorem 3.9.13 of \citet{van1996weak}, we have
	\begin{align}\label{eq.conditional sup EL convergence}
	\sup_{f\in \mathrm{BL}_{1}(\ell^{\infty}(\bar{\mathcal{H}}\times\mathcal{ G} ))}\vert E_W[  f\{  \sqrt{n}(  \mathcal{L}	(  \hat{P}_{n}^{B})  -\mathcal{L}(  \hat{P}_{n})  )
	\}  ]
	-E[  f(  \mathcal{L}_{P}^{\prime}(  \mathbb{G}_{P})
	)  ]  \vert \to 0
	\end{align}
	outer almost surely	and
	\begin{align}\label{eq.conditional EL convergence}
	&E_W[  f\{  \sqrt{n}  (  \mathcal{L}(  \hat{P}_{n}^{B
	})  -\mathcal{L}(  \hat{P}_{n})  )    \}
	^{\ast}] -E_W[f\{  \sqrt{n}  (  \mathcal{L}(  \hat{P}_{n}^{B})
	-\mathcal{L}(  \hat{P}_{n})  )    \}  _{\ast}]  \to0
	\end{align}
	almost surely for every $f\in \mathrm{BL}_{1}(\ell^{\infty}(\bar{\mathcal{H}}\times\mathcal{G}))$. The outer almost sure convergence in \eqref{eq.conditional sup EL convergence} implies that the weak convergence $\sqrt{n}(  \mathcal{L}	(  \hat{P}_{n}^{B})  -\mathcal{L}(  \hat{P}_{n})) \leadsto \mathcal{L}_{P}^{\prime}(  \mathbb{G}_{P})$ holds for almost every given sequence $\{
	(  Y_{i},D_{i},Z_{i})  \}  _{i=1}^{\infty}$.
	By Lemma \ref{lemma.uniform Glivenko-Cantelli} in this paper, and Lemmas 1.9.2 and 1.9.3 of \citet{van1996weak}, we have that $\Vert \hat{P}_{n}^{B}-\hat{P}_{n}\Vert_{\infty}\to0$ outer almost surely for almost every given sequence $\{  (
	Y_{i},D_{i},Z_{i})  \}  _{i=1}^{\infty}$. By Lemma \ref{lemma.uniform Glivenko-Cantelli} again, $\Vert
	\hat{P}_{n}-P\Vert _{\infty}\to 0$ for almost every sequence $\{  (
	Y_{i},D_{i},Z_{i})  \}  _{i=1}^{\infty}$. Thus now we have that
	$\Vert \hat{P}_{n}^{B}-P\Vert _{\infty}\leq\Vert
	\hat{P}_{n}^{B}-\hat{P}_{n}\Vert _{\infty}+\Vert
	\hat{P}_{n}-P\Vert _{\infty}\to0$ outer almost surely for
	almost every given sequence $\{  (  Y_{i},D_{i},Z_{i})  \}
	_{i=1}^{\infty}$. This implies that $\Vert
	\hat{\sigma}_{P_n}^{B}-\sigma_P\Vert _{\infty}\to0$ and $T_n^B/n\to \Lambda(P)$
	outer almost surely for almost every given sequence $\{  (  Y_{i},D_{i},Z_{i})
	\}  _{i=1}^{\infty}$. This, together with \eqref{eq.conditional sup EL convergence}, and Lemmas 1.9.2(i) and 1.10.2(iii), Example 1.4.7 (Slutsky's lemma), and Theorem 1.3.6 (continuous mapping) of \citet{van1996weak}, implies that $\sqrt{T_n^B}(  \mathcal{L}	(  \hat{P}_{n}^{B})  -\mathcal{L}(  \hat{P}_{n}) )/\mathcal{M}(\hat{\sigma}_{P_n}^{B}) \leadsto \mathbb{ G}_0/\mathcal{M}(\sigma_P)$ for
	almost every given sequence $\{  (  Y_{i},D_{i},Z_{i})  \}
	_{i=1}^{\infty}$. Since $\mathbb{ G}_P$ is tight, $\mathbb{ G}_0$ is tight by \eqref{eq.L HDD}.
	
	(ii). By \eqref{eq.conditional EL convergence} and Theorem 2.37 of \citet{folland2013real} (Fubini), together with the dominated convergence theorem and Lemma 1.2.1 of \citet{van1996weak},
	\begin{align}\label{eq.unconditional EL convergence}
	E^{\ast}[  f\{  \sqrt{n}  (  \mathcal{L}(  \hat{P}_{n}^{B
	})  -\mathcal{L}(  \hat{P}_{n})  )    \}
	]  -E_{\ast}[  f\{  \sqrt{n}  (  \mathcal{L}%
	(  \hat{P}_{n}^{B})  -\mathcal{L}(  \hat{P}_{n})  )
	\}  ]  \to0
	\end{align}
	for every $f\in \mathrm{BL}_{1}(\ell^{\infty}(\bar{\mathcal{H}}\times\mathcal{G}))$. By \eqref{eq.conditional sup EL convergence}, together with the definition of outer almost sure convergence (Definition 1.9.1(iii) of \citet{van1996weak}), we have that for every function $f\in \mathrm{BL}_{1}(\ell^{\infty}(\bar{\mathcal{H}}\times\mathcal{G}))$, 
	\begin{align}\label{eq.outer conditional EL convergence}
	\vert E_W[  f\{  \sqrt{n}(  \mathcal{L}	(  \hat{P}_{n}^{B})  -\mathcal{L}(  \hat{P}_{n})  )\} ]
	-E[  f(  \mathcal{L}_{P}^{\prime}(  \mathbb{G}_{P})
	)  ]  \vert^{\ast} \to 0
	\end{align}
	almost surely. 
	Thus by \eqref{eq.outer conditional EL convergence}, together with Lemma 1.2.2(iii) of \citet{van1996weak}, we have that 
	\begin{align}\label{eq.outer conditional EL convergence 2}
	\vert (E_W[  f\{  \sqrt{n}(  \mathcal{L}	(  \hat{P}_{n}^{B})  -\mathcal{L}(  \hat{P}_{n})  )\}  ])^{\ast}
	-E[  f(  \mathcal{L}_{P}^{\prime}(  \mathbb{G}_{P})
	)  ]  \vert \to 0
	\end{align}
	almost surely for every $f\in \mathrm{BL}_{1}(\ell^{\infty}(\bar{\mathcal{H}}\times\mathcal{G}))$. 
	By Lemma 1.2.6 (Fubini's theorem) of \citet{van1996weak},
	\begin{align}\label{eq.one hand}
	E^{\ast}[  f\{  \sqrt{n}(  \mathcal{L}(  \hat{P}_{n}^{B
	})  -\mathcal{L}(  \hat{P}_{n})  )  \}  ]    
	 &\geq E^{\ast}[  E_W[  f\{  \sqrt{n}(  \mathcal{L}
	(  \hat{P}_{n}^{B})  -\mathcal{L}(  \hat{P}_{n})  )\}  ]  ]\notag\\
	& \ge E_{\ast}[  f\{  \sqrt{n}(  \mathcal{L}(
	\hat{P}_{n}^{B})  -\mathcal{L}(  \hat{P}_{n})  )  \}].
	\end{align}
	Then by Lemma 1.2.1 of \citet{van1996weak}
	and \eqref{eq.unconditional EL convergence}, we have that
	\begin{align}\label{eq.iterated EL}
	E^{\ast}[  f\{  \sqrt{n}(  \mathcal{L}(  \hat{P}_{n}^{B
	})  -\mathcal{L}(  \hat{P}_{n})  )  \}  ] =E[  (  E_W[  f\{  \sqrt{n}(  \mathcal{L}(
	\hat{P}_{n}^{B})  -\mathcal{L}(  \hat{P}_{n})  )  \}]  )  ^{\ast
	}] +o(  1).
	\end{align}
	Now with \eqref{eq.outer conditional EL convergence 2} we can conclude that
	\begin{align*}
	& \vert E^{\ast}[  f\{  \sqrt{n}(  \mathcal{L}(
	\hat{P}_{n}^{B})  -\mathcal{L}(  \hat{P}_{n})  )  \}
	]  -E[  f(  \mathcal{L}_{P}^{\prime}(  \mathbb{G}_{P})  )  ]  \vert \\
	=&\,\vert E[  (  E_W[  f\{  \sqrt{n}(
	\mathcal{L}(  \hat{P}_{n}^{B})  -\mathcal{L}(  \hat{P}_{n}))  \} ])  ^{\ast}]  +o( 1)  -E[  f(
	\mathcal{L}_{P}^{\prime}(  \mathbb{G}_{P})  )  ]
	\vert \\
	\leq&\, E[  \vert (  E_W[  f\{  \sqrt{n}(
	\mathcal{L}(  \hat{P}_{n}^{B})  -\mathcal{L}(  \hat{P}_{n})
	)  \} ]
	)  ^{\ast}-E[  f(  \mathcal{L}_{P}^{\prime}(
	\mathbb{G}_{P})  )  ]  \vert ]  +o(1)  \to  0
	\end{align*}
	for every $f\in \mathrm{BL}_{1}(\ell^{\infty}(\bar{\mathcal{H}}\times\mathcal{G}))$, where the equality is from \eqref{eq.iterated EL} and the convergence is by the dominated convergence theorem together with the almost sure convergence in \eqref{eq.outer conditional EL convergence 2}. 
	This implies that  $\sqrt{n}(  \mathcal{L}(  \hat{P}_{n}^{B})
	-\mathcal{L}(  \hat{P}_{n})  )  \leadsto {\mathcal{L}_{P}^{\prime}(\mathbb{G}_P)}$ unconditionally. Similarly, by \eqref{eq.conditional sup EP convergence} and \eqref{eq.conditional EP convergence} we can easily show that $\sqrt{n}(\hat{P}_n^{B}-\hat{P}_n)\leadsto\mathbb{G}_P$ unconditionally. Thus we can conclude that $\hat{P}_{n}^{B} - \hat{P}_n \to 0$ in outer probability by Lemma 1.10.2(iii) of \citet{van1996weak}. By Lemma \ref{lemma.uniform Glivenko-Cantelli} in this paper and Lemmas 1.9.3 and 1.2.2(i) of \citet{van1996weak}, we have that $\hat{P}_{n}^{B} \to P$ in outer probability, and hence $T_n^B/n\to \Lambda(P)$ and $\mathcal{M}(\hat{\sigma}_{P_n}^{B})\to \mathcal{M}({\sigma}_{P})$ in outer probability by Theorem 1.9.5 (continuous mapping) of \citet{van1996weak}. 
	By Lemma 1.10.2(iii), Example 1.4.7 (Slutsky's lemma), and Theorem 1.3.6 (continuous mapping) of \citet{van1996weak},
	$\sqrt{T_n^B}(  \mathcal{L}	(  \hat{P}_{n}^{B})  -\mathcal{L}(  \hat{P}_{n}) )/\mathcal{M}(\hat{\sigma}_{P_n}^{B})  \leadsto\mathbb{G}_0/\mathcal{M}(\sigma_P) $
	unconditionally. This verifies (ii) of the Lemma. 
	
	(iii). This claim holds naturally under our constructions.
\end{proof}

To explore the property of the bootstrap test statistic, we introduce the following notation. For all sets $A_{1},A_{2}\subset\bar{\mathcal{H}}\times\mathcal{G}$, define
$\overrightarrow{d_{H}}\left(  A_{1},A_{2}\right)  =\sup_{a\in A_{1}}\inf_{b\in
A_{2}}\rho_{P}\left(  a,b\right)$
and
\[
d_{H}\left(  A_{1},A_{2}\right)  =\max\left\{  \overrightarrow{d_{H}}\left(
A_{1},A_{2}\right)  ,\overrightarrow{d_{H}}\left(  A_{2},A_{1}\right)
\right\}  .
\]
Also,  define 
\begin{align}\label{eq.new Psi_HG}
\widehat{\Psi_{\bar{\mathcal{H}}\times\mathcal{G}}}=\left\{  \left(  h,g\right)
\in\bar{\mathcal{H}}\times\mathcal{G}:\sqrt{T_n}\left\vert \frac{\hat{\phi}_{P_n}\left(
	h,g\right)  }{\mathcal{M}(\hat{\sigma}_{P_n})\left( \xi_0, h,g\right)  }\right\vert  \leq\tau
_{n}\right\},
\end{align}
where $\xi_0$ and $\tau_n$ are as in \eqref{eq.Psi_hat}.
Notice the difference between $\widehat{\Psi 		_{{\mathcal{H}}\times\mathcal{G}}}$ in \eqref{eq.Psi_hat} and $\widehat{\Psi 		_{\bar{\mathcal{H}}\times\mathcal{G}}}$ in \eqref{eq.new Psi_HG}. Clearly, $\widehat{\Psi 		_{{\mathcal{H}}\times\mathcal{G}}}\subset\widehat{\Psi 		_{\bar{\mathcal{H}}\times\mathcal{G}}}$.

\begin{lemma}
\label{lemma.consistent hat Psi multi}Under Assumptions
\ref{ass.independent data} and \ref{ass.probability path}, if the $H_{0}$ in \eqref{eq.null 1} is true with $Q=P_n$ for all $n$, then $d_{H}(  \widehat{\Psi
_{\bar{\mathcal{H}}\times\mathcal{G}}},\Psi_{\bar{\mathcal{H}}\times\mathcal{G}})  \rightarrow 0$ in outer probability, where  $\Psi_{\bar{\mathcal{H}}\times\mathcal{G}}$ is defined as in \eqref{eq.Psi_HG}. 
\end{lemma}

\begin{proof}[Proof of Lemma \ref{lemma.consistent hat Psi multi}]
First, under the assumptions, we have that for all $\varepsilon>0$,
\begin{align*}
\lim_{n\rightarrow\infty}\mathbb{P}^{\ast}  &  \left(  \overrightarrow{d_{H}}\left(
\Psi_{\bar{\mathcal{H}}\times\mathcal{G}},\widehat{\Psi_{\bar{\mathcal{H}}\times
		\mathcal{G}}}\right)  >\varepsilon\right)  \leq\lim_{n\rightarrow\infty
}\mathbb{P}^{\ast}\left(  \Psi_{\bar{\mathcal{H}}\times\mathcal{G}}\backslash\widehat
{\Psi_{\bar{\mathcal{H}}\times\mathcal{G}}}\neq\varnothing\right) \\
&\leq   \lim_{n\rightarrow\infty}\mathbb{P}^{\ast}\left(  \sup_{\left(  h,g\right)
	\in\bar{\mathcal{H}}\times\mathcal{G}}\sqrt{T_n}\left\vert \frac{\hat{\phi
	}_{P_n}\left(  h,g\right)  -\phi_P\left(  h,g\right)  }{\xi_0\vee\hat{\sigma}_{P_n}\left(
	h,g\right)  }\right\vert   >\tau_{n}\right)  .
\end{align*}
By Lemma \ref{lemma.weak convergence phi_K}, $\sqrt{T_n}(  \hat{\phi
}_{P_n}-\phi_P )  \leadsto\mathbb{G}$. By Lemma \ref{lemma.almost uniform convergence phi sigma}, $\hat{\sigma}_{P_n}\to \sigma_P$ almost uniformly, which implies that $\hat{\sigma}_{P_n}\leadsto \sigma_P$ by Lemmas 1.9.3(ii) and 1.10.2(iii) of \citet{van1996weak}. Thus by Example 1.4.7 (Slutsky's lemma) and Theorem 1.3.6 (continuous mapping) of
\citet{van1996weak},
\[
\sup_{\left(  h,g\right)  \in\bar{\mathcal{H}}\times\mathcal{G}}\sqrt{T_n}\left\vert \frac{\hat{\phi}_{P_n}\left(  h,g\right)  -\phi_P\left(  h,g\right)  }%
{\xi_0\vee\hat{\sigma}_{P_n}\left(  h,g\right)  }\right\vert 
\leadsto\sup_{\left(  h,g\right)  \in\bar{\mathcal{H}}\times\mathcal{G}}\left\vert \frac{\mathbb{G}\left(  h,g\right)  }{\xi_0\vee\sigma_P\left(
	h,g\right)  }\right\vert.
\]
Since $\tau_{n}\rightarrow\infty$, we have that $\lim_{n\rightarrow\infty}
\mathbb{P}^{\ast}(  \overrightarrow{d_{H}}(  \Psi_{\bar{\mathcal{H}}
	\times\mathcal{G}},\widehat{\Psi_{\bar{\mathcal{H}}\times\mathcal{G}}})
>\varepsilon)  =0$.

Next, consider $\overrightarrow{d_{H}}(  \widehat{\Psi_{\bar{\mathcal{H}}
		\times\mathcal{G}}},\Psi_{\bar{\mathcal{H}}\times\mathcal{G}})  $. Define 
\[
d\left(  \left(  h,g\right)  ,A\right)  =\inf_{\left(  h^{\prime},g^{\prime
	}\right)  \in A}\rho_{P}\left(  \left(  h,g\right)  ,\left(  h^{\prime
},g^{\prime}\right)  \right)
\]
for all $\left(  h,g\right)  \in\mathcal{\bar{H}}\times\mathcal{G}$ and all
subsets $A\subset\mathcal{\bar{H}}\times\mathcal{G}$. For each $\varepsilon
>0$, define
\begin{align*}
\tilde{D}_{\varepsilon}  =\left\{  \left(  h,g\right)  \in
\bar{\mathcal{H}}\times\mathcal{G}:d\left(  \left(  h,g\right)  ,{\Psi
_{\bar{\mathcal{H}}\times\mathcal{G}}}\right)  \geq\varepsilon\right\}.
\end{align*}
The product space $\mathcal{\bar{H}}\times\mathcal{G}$ is compact under $\rho_{P}$ by
Lemma \ref{lemma.complete HG}. Suppose $\left\{  \left(  h_{n}%
,g_{n}\right)  \right\}  _{n}\subset \tilde{D}_{\varepsilon}$ such that $\left(
h_{n},g_{n}\right)  \rightarrow\left(  h,g\right)  $ for some $\left(
h,g\right)  \in\mathcal{\bar{H}}\times\mathcal{G}$. Then
\begin{align*}
d &  \left(  \left(  h,g\right)  ,\Psi_{\bar{\mathcal{H}}\times\mathcal{G}}\right)
=\inf_{\left(  h^{\prime},g^{\prime}\right)  \in\Psi_{\bar{\mathcal{H}}
		\times\mathcal{G}}}\rho_{P}\left(  \left(  h,g\right)  ,\left(  h^{\prime
},g^{\prime}\right)  \right)  \\
&  \geq\inf_{\left(  h^{\prime},g^{\prime}\right)  \in\Psi_{\bar{\mathcal{H}}
		\times\mathcal{G}}}\rho_{P}\left(  \left(  h_{n},g_{n}\right)  ,\left(
h^{\prime},g^{\prime}\right)  \right)  -\rho_{P}\left(  \left(  h,g\right)
,\left(  h_{n},g_{n}\right)  \right)  \geq\varepsilon-\rho_{P}\left(  \left(
h,g\right)  ,\left(  h_{n},g_{n}\right)  \right)  ,
\end{align*}
which is true for all $n$. Letting $n\rightarrow\infty$ gives $d(
\left(  h,g\right)  ,\Psi_{\bar{\mathcal{H}}\times\mathcal{G}})
\geq\varepsilon$. This implies that $\tilde{D}_{\varepsilon}$ is closed in
$\mathcal{\bar{H}}\times\mathcal{G}$, which is compact, and thus $
\tilde{D}_{\varepsilon}$ is compact. If $\tilde{D}_{\varepsilon}=\varnothing$,
then clearly
\begin{align*}
 &\lim_{n\rightarrow\infty}\mathbb{P}^{\ast}\left(  \overrightarrow{d_{H}}\left(
\widehat{\Psi_{\bar{\mathcal{H}}\times\mathcal{G}}},\Psi_{\bar{\mathcal{H}}\times
	\mathcal{G}}\right)  >\varepsilon\right)\notag\\
 = &  \lim_{n\rightarrow\infty}\mathbb{P}^{\ast}\left(  \sup_{\left(  h,g\right)
	\in\widehat{\Psi_{\bar{\mathcal{H}}\times\mathcal{G}}}}\inf_{\left(  h^{\prime
	},g^{\prime}\right)  \in\Psi_{\bar{\mathcal{H}}\times\mathcal{G}}}\rho_{P}\left(
\left(  h,g\right)  ,\left(  h^{\prime},g^{\prime}\right)  \right)
>\varepsilon\right)  =0.
\end{align*}
If ${\tilde{D}_{\varepsilon}}\neq\varnothing$, then there is a $\delta
_{\varepsilon}>0$ such that $\inf_{\left(  h,g\right)  \in
	{\tilde{D}_{\varepsilon}}}\left\vert \phi_P\left(  h,g\right)  \right\vert
>\delta_{\varepsilon}$, since $\phi_P$ is continuous by Lemma
\ref{lemma.phi is continuous}. Also, $\hat{\sigma}_{P_n}$ is uniformly bounded in
$\left(  h,g\right)  $ and $\omega$, so there is a $\delta_{\varepsilon}^{\prime}>0$
such that for all $\omega\in\Omega$, 
$
\inf_{\left(  h,g\right)  \in \tilde{D}_{\varepsilon}}\left\vert
\phi_P\left(  h,g\right)  /\left(  \xi_0\vee \hat{\sigma}_{P_n}\left(  h,g\right)
\right)  \right\vert   >\delta_{\varepsilon}^{\prime}.
$
Thus if ${\tilde{D}_{\varepsilon}}\neq\varnothing$, we have
\begin{align*}
	&  \lim_{n\rightarrow\infty}\mathbb{P}^{\ast}\left(  \overrightarrow{d_{H}%
	}\left(  \widehat{\Psi_{\bar{\mathcal{H}}\times\mathcal{G}}},\Psi
	_{\bar{\mathcal{H}}\times\mathcal{G}}\right)  >\varepsilon\right)  \\
	=&\lim_{n\rightarrow\infty}\mathbb{P}^{\ast}\left(  \sup_{\left(
		h,g\right)  \in\widehat{\Psi_{\bar{\mathcal{H}}\times\mathcal{G}}}}%
	\inf_{\left(  h^{\prime},g^{\prime}\right)  \in\Psi_{\bar{\mathcal{H}}%
			\times\mathcal{G}}}\rho_{P}\left(  \left(  h,g\right)  ,\left(  h^{\prime
	},g^{\prime}\right)  \right)  >\varepsilon\right)  \\
	\leq &  \lim_{n\rightarrow\infty}\mathbb{P}^{\ast}\left(
	\begin{array}
		[c]{c}%
		\sup_{\left(  h,g\right)  \in\widehat{\Psi_{\bar{\mathcal{H}}\times
					\mathcal{G}}}\backslash\Psi_{\bar{\mathcal{H}}\times\mathcal{G}}}\left\vert
		\frac{\phi_{P}\left(  h,g\right)  }{\xi_{0}\vee\hat{\sigma}_{P_{n}}\left(
			h,g\right)  }\right\vert >\delta_{\varepsilon}^{\prime},\\
		\sup_{\left(  h,g\right)  \in\widehat{\Psi_{\bar{\mathcal{H}}\times
					\mathcal{G}}}\backslash\Psi_{\bar{\mathcal{H}}\times\mathcal{G}}}\sqrt{T_{n}%
		}\left\vert \frac{\hat{\phi}_{P_{n}}\left(  h,g\right)  }{\xi_{0}\vee
			\hat{\sigma}_{P_{n}}\left(  h,g\right)  }\right\vert \leq\tau_{n}%
	\end{array}
	\right)  .
\end{align*}
By Lemma \ref{lemma.almost uniform convergence phi sigma}, we have that $\hat{\phi}_{P_n}\to \phi_P$ almost uniformly. Thus there is a measurable set $A$ with $\mathbb{P}(A)\ge 1-\varepsilon$ such that for sufficiently large $n$, 
\begin{align*}
\sup_{\left(  h,g\right)  \in\widehat{\Psi
		_{\bar{\mathcal{H}}\times\mathcal{G}}}\backslash\Psi_{\bar{\mathcal{H}}\times\mathcal{G}}
}\left\vert \frac{\hat{\phi}_{P_n}\left(  h,g\right)  }{\xi_0\vee\hat
	{\sigma}_{P_n}\left(  h,g\right)  }\right\vert \ge \sup_{\left(  h,g\right)  \in\widehat{\Psi
		_{\bar{\mathcal{H}}\times\mathcal{G}}}\backslash\Psi_{\bar{\mathcal{H}}\times\mathcal{G}}%
}\left\vert \frac{{\phi}_{P}\left(  h,g\right)  }{\xi_0\vee\hat
{\sigma}_{P_n}\left(  h,g\right)  }\right\vert  - \frac{\delta_{\varepsilon}^{\prime}}{2}
\end{align*}
uniformly on $A$.
Thus we now have that
\begin{align*}
&  \lim_{n\rightarrow\infty}\mathbb{P}^{\ast}\left(  \overrightarrow{d_{H}}\left(
\widehat{\Psi_{\bar{\mathcal{H}}\times\mathcal{G}}},\Psi_{\bar{\mathcal{H}}\times
	\mathcal{G}}\right)  >\varepsilon\right)  \\
\leq &  \lim_{n\rightarrow\infty}\mathbb{P}^{\ast}\left(
\begin{array}
[c]{c}%
\left\{\sup_{\left(  h,g\right)  \in\widehat{\Psi_{\bar{\mathcal{H}}\times\mathcal{G}}
	}\backslash\Psi_{\bar{\mathcal{H}}\times\mathcal{G}}}\left\vert \frac
{\phi_P\left(  h,g\right)  }{\xi_0\vee\hat{\sigma}_{P_n}\left(  h,g\right)  }\right\vert
  >\delta_{\varepsilon}^{\prime}\right\}\\
\cap\left\{\sup_{\left(  h,g\right)  \in\widehat{\Psi_{\bar{\mathcal{H}}\times\mathcal{G}}
	}\backslash\Psi_{\bar{\mathcal{H}}\times\mathcal{G}}}\sqrt{T_n}\left\vert
\frac{\hat{\phi}_{P_n}\left(  h,g\right)  }{\xi_0\vee\hat{\sigma}_{P_n}\left(  h,g\right)
}\right\vert   \leq\tau_{n}\right\}\cap A
\end{array}
\right) +\mathbb{P}(A^c) \\
\leq &  \lim_{n\rightarrow\infty}\mathbb{P}^{\ast}\left(\sqrt{\frac{T_n}{n}}\frac{\delta_{\varepsilon}^{\prime}}{2}< 
\sup_{\left(  h,g\right)  \in\widehat{\Psi_{\bar{\mathcal{H}}\times\mathcal{G}}
	}\backslash\Psi_{\bar{\mathcal{H}}\times\mathcal{G}}}\sqrt{\frac{{T_n}}{{n}}}\left\vert \frac
{\hat{\phi}_{P_n}\left(  h,g\right)  }{\xi_0\vee\hat{\sigma}_{P_n}\left(  h,g\right)
}\right\vert   \leq\frac{\tau_{n}}{\sqrt{n}}
\right)  +\varepsilon = \varepsilon,
\end{align*}
because $\tau_{n}/\sqrt{n}\rightarrow0$ as $n\rightarrow\infty$. Here, $\varepsilon$ can be arbitrarily small. 
\end{proof}

\begin{proof}[Proof of Theorem \ref{thm.test multi}]
	(i). Fix
	$\psi\in C\left(  \Xi\times\mathcal{\bar{H}}\times\mathcal{G}\right)  $  under the $\rho_{\xi hg}$ defined in \eqref{eq.rho_xihg}. It is easy to show that 
	$\Xi\times\mathcal{\bar{H}}\times\mathcal{G}$ is compact under $\rho_{\xi hg}$, and thus $\psi$ is uniformly continuous on $\Xi\times\mathcal{\bar{H}}%
	\times\mathcal{G}$. This implies that for every $\varepsilon>0$, there is a 
	$\delta>0$ such that $\left\vert \psi(\xi^{\prime},h^{\prime},g^{\prime}
	)-\psi\left(  \xi,h,g\right)  \right\vert \leq\varepsilon/\nu\left(
	\Xi\right)  $ for all $\left(  \xi,h,g\right)  ,\left(  \xi^{\prime}%
	,h^{\prime},g^{\prime}\right)  \in\Xi\times\mathcal{\bar{H}}\times\mathcal{G}$
	with $\rho_{\xi hg}((\xi^{\prime},h^{\prime},g^{\prime}),\left(
	\xi,h,g\right)  )\leq\delta$. Also, by the constructions of $\Psi
	_{\mathcal{\bar{H}}\times\mathcal{G}}$ in \eqref{eq.Psi_HG} and $\widehat{\Psi_{\mathcal{\bar{H}%
			}\times\mathcal{G}}}$ in \eqref{eq.new Psi_HG}, we have that
	\begin{align*}
		& \left\vert \mathcal{I}\circ\mathcal{S}_{\widehat{\Psi_{\mathcal{\bar{H}%
					}\times\mathcal{G}}}}\left(  \psi\right)  -\mathcal{I}\circ\mathcal{S}%
		_{\Psi_{\mathcal{\bar{H}}\times\mathcal{G}}}\left(  \psi\right)  \right\vert
		\\
		\leq&\,\nu\left(  \Xi\right)  \sup_{\rho_{\xi hg}((\xi^{\prime},h^{\prime
			},g^{\prime}),\left(  \xi,h,g\right)  )\leq d_{H}\left(  \widehat
			{\Psi_{\mathcal{\bar{H}}\times\mathcal{G}}},\Psi_{\mathcal{\bar{H}}%
				\times\mathcal{G}}\right)  }\left\vert \psi\left(  \xi^{\prime},h^{\prime
		},g^{\prime}\right)  -\psi\left(  \xi,h,g\right)  \right\vert .
	\end{align*}
	By Lemma \ref{lemma.consistent hat Psi multi}, this implies that
	\[
	\mathbb{P}^{\ast}\left(  \left\vert \mathcal{I}\circ\mathcal{S}_{\widehat
		{\Psi_{\mathcal{\bar{H}}\times\mathcal{G}}}}\left(  \psi\right)
	-\mathcal{I}\circ\mathcal{S}_{\Psi_{\mathcal{\bar{H}}\times\mathcal{G}}%
	}\left(  \psi\right)  \right\vert >\varepsilon\right)  \leq\mathbb{P}^{\ast
	}\left(  d_{H}\left(  \widehat{\Psi_{\mathcal{\bar{H}}\times\mathcal{G}}}%
	,\Psi_{\mathcal{\bar{H}}\times\mathcal{G}}\right)  >\delta\right)
	\rightarrow0.
	\]
	Notice that 
	\begin{align*}
		\vert \mathcal{I}\circ\mathcal{S}_{\widehat{\Psi
				_{\mathcal{\bar{H}}\times\mathcal{G}}}}\left(  \psi_{1}\right)  -\mathcal{I}
		\circ\mathcal{S}_{\widehat{\Psi_{\mathcal{\bar{H}}\times\mathcal{G}}}}\left(
		\psi_{2}\right)  \vert \leq\nu\left(  \Xi\right)  \left\Vert \psi
		_{1}-\psi_{2}\right\Vert _{\infty}	
	\end{align*}
	for all $\psi_{1},\psi_{2}\in\ell^{\infty
	}\left(  \Xi\times\mathcal{\bar{H}}\times\mathcal{G}\right)  $.
	By Lemma S.3.6 of \citet{fang2014inference}, $\mathcal{I}\circ\mathcal{S}_{\widehat{\Psi
			_{\mathcal{\bar{H}}\times\mathcal{G}}}}$ satisfies Assumption 4 of \citet{fang2014inference}.
	Together with Lemma \ref{lemma.bootstrap properties}, by repeating the proof of Theorem 3.2 of \citet{fang2014inference} with $\mathbb{ G}_{n}^{B}=\sqrt{T_n^B}(  \hat{\phi}_{P_n}^{B}-\hat{\phi}_{P_n})  /\mathcal{M}(  \hat{\sigma}_{P_n}^{B})$, where $\mathbb{ G}_{n}^{B}$ replaces $\mathbb{ G}_{n}^{\ast}$ in their notation, we can show that 
	\begin{align}\label{eq.bootstrap consistency}
		\sup_{f\in\mathrm{BL}_{1}(\mathbb{R})}\left\vert
		\begin{array}
			[c]{c}%
			E_{W}\left[  f\left\{  \mathcal{I}\circ\mathcal{S}_{\widehat{\Psi
					_{\bar{\mathcal{H}}\times\mathcal{G}}}}\left(  \frac{\sqrt{T_{n}^{B}}\left(
				\hat{\phi}_{P_{n}}^{B}-\hat{\phi}_{P_{n}}\right)  }{\mathcal{M}\left(
				\hat{\sigma}_{P_{n}}^{B}\right)  }\right)  \right\}  \right]  \\
			-E\left[  f\left\{  \mathcal{I}\circ\mathcal{S}_{{\Psi_{\bar{\mathcal{H}%
						}\times\mathcal{G}}}}\left(  \frac{\mathbb{G}_{0}}{\mathcal{M}(\sigma_{P}%
				)}\right)  \right\}  \right]
		\end{array}
		\right\vert \rightarrow0
	\end{align}
	in outer probability, where $\mathbb{G}_0$ is the limit obtained in Lemma \ref{lemma.bootstrap properties} and $\mathbb{G}_0/\mathcal{M}(\sigma_P)$ is tight by Lemma \ref{lemma.bootstrap properties}(i). Since the sample is finite, that is, we have only finitely many observations $\{(Y_i,D_i,Z_i)\}_{i=1}^n$ in the data set, by the constructions of $\widehat{\Psi_{{\mathcal{H}}\times\mathcal{G}}}$ in \eqref{eq.Psi_hat} and $\widehat{\Psi_{\bar{\mathcal{H}}\times\mathcal{G}}}$ in \eqref{eq.new Psi_HG} we have that
	\begin{align}\label{eq.equi bootstrap test stat}
		\mathcal{I}\circ
		\mathcal{S}_{\widehat{\Psi_{{\mathcal{H}}\times\mathcal{G}}}}\left(  \frac{\sqrt{T_n^B}\left(  \hat{\phi}_{P_n}^{B}-\hat{\phi}_{P_n}\right)  }{\mathcal{M}\left(
			\hat{\sigma}_{P_n}^{B}\right)  }\right)=\mathcal{I}\circ
		\mathcal{S}_{\widehat{\Psi_{\bar{\mathcal{H}}\times\mathcal{G}}}}\left(  \frac{\sqrt{T_n^B}\left(  \hat{\phi}_{P_n}^{B}-\hat{\phi}_{P_n}\right)  }{\mathcal{M}\left(
			\hat{\sigma}_{P_n}^{B}\right)  }\right).
	\end{align}
	Then  \eqref{eq.bootstrap consistency} and \eqref{eq.equi bootstrap test stat} imply that
	\begin{align}\label{eq.bootstrap consistency 2}
		\sup_{f\in\mathrm{BL}_{1}(\mathbb{R})}\left\vert
		\begin{array}
			[c]{c}%
			E_{W}\left[  f\left\{  \mathcal{I}\circ\mathcal{S}_{\widehat{\Psi
					_{{\mathcal{H}}\times\mathcal{G}}}}\left(  \frac{\sqrt{T_{n}^{B}}\left(
				\hat{\phi}_{P_{n}}^{B}-\hat{\phi}_{P_{n}}\right)  }{\mathcal{M}\left(
				\hat{\sigma}_{P_{n}}^{B}\right)  }\right)  \right\}  \right]  \\
			-E\left[  f\left\{  \mathcal{I}\circ\mathcal{S}_{{\Psi_{\bar{\mathcal{H}%
						}\times\mathcal{G}}}}\left(  \frac{\mathbb{G}_{0}}{\mathcal{M}(\sigma_{P}%
				)}\right)  \right\}  \right]
		\end{array}
		\right\vert \rightarrow0
	\end{align}
	in outer probability.
	Let
	$F$ denote the CDF of $\mathcal{I}\circ\mathcal{S}_{\Psi_{\bar{\mathcal{H}}
			\times\mathcal{G}}}\left(  \mathbb{G}_0/\mathcal{M}\left(  \sigma_P\right)
	\right)  $, and define $\hat{F}_{n}$ by
	\[
	\hat{F}_{n}\left(  c\right)  =\mathbb{P}\left(  \mathcal{I}\circ
	\mathcal{S}_{\widehat{\Psi_{{\mathcal{H}}\times\mathcal{G}}}}\left(  \frac{\sqrt{T_n^B}\left(  \hat{\phi}_{P_n}^{B}-\hat{\phi}_{P_n}\right)  }{\mathcal{M}\left(
		\hat{\sigma}_{P_n}^{B}\right)  }\right)  \leq c \bigg|\left\{  \left(  Y_{i},D_{i}
	,Z_{i}\right)  \right\}  _{i=1}^{\infty}\right).\footnote{This conditional probability given $\left\{  \left(  Y_{i},D_{i},Z_{i}\right)  \right\}  _{i=1}^{\infty}$ is numerically equal to that given $\left\{  \left(  Y_{i},D_{i},Z_{i}\right)  \right\}  _{i=1}^{n}$ in \eqref{eq.c hat}.}
	\]
	Since by assumption $F$ is continuous and increasing at $c_{1-\alpha}$, by a proof similar to that of Theorem S.1.1 of \citet{fang2014inference} together with  \eqref{eq.bootstrap consistency 2} in this paper, we can conclude that for each $\varepsilon>0$,
	\begin{align}\label{eq.critical value consistency}
		\mathbb{P}^{\ast}(|\hat{c}_{1-\alpha}-c_{1-\alpha}|>\varepsilon)\to 0.
	\end{align}
	
	By the definitions of $\mathbb{G}$ (in the proof of Lemma \ref{lemma.weak convergence phi_K}) and $\mathbb{G}_0$ (in Lemma \ref{lemma.bootstrap properties}), together with the linearity of $\mathcal{ L}_P^{\prime}$, we have that $\mathbb{G}=\mathbb{G}_0+\Lambda(P)^{1/2}\mathcal{L}_{P}^{\prime}\left(  Q_0\right)$. Let $H_n=\sqrt{n}(P_n-P)$. By Lemma \ref{lemma.weak convergence Pn_hat and convergence Pn}, $\Vert H_n-Q_0\Vert_{\infty}\to 0$ as $n\to \infty$. Notice that $P_n=P+n^{-1/2}H_n$. By Lemma \ref{lemma.L HD}, we have that 
	\begin{align}\label{eq.L derivative P}
		&\lim_{n\rightarrow\infty}\sup_{\left(  h,g\right)  \in  \Psi_{\bar{\mathcal{H}}\times\mathcal{G}}  }\left\vert
		\frac{\mathcal{L}\left(  P_{n}\right)  \left(  h,g\right)
			-\mathcal{L}\left(  P\right)  \left(  h,g\right)  }{{n}^{-1/2}}-\mathcal{L}_{P}^{\prime}\left(  Q_0\right)  \left(  h,g\right)
		\right\vert \notag \\
		\le&\lim_{n\rightarrow\infty}\sup_{\left(  h,g\right)  \in  \bar{\mathcal{H}}\times\mathcal{G}  }\left\vert
		\frac{\mathcal{L}\left(  P+n^{-1/2}H_n\right)  \left(  h,g\right)
			-\mathcal{L}\left(  P\right)  \left(  h,g\right)  }{{n}^{-1/2}}-\mathcal{L}_{P}^{\prime}\left(  Q_0\right)  \left(  h,g\right)
		\right\vert =0.
	\end{align}
	By construction, $\mathcal{ L}(P)=0$ on $\Psi_{\bar{\mathcal{H}}\times\mathcal{G}}$ because $\mathcal{ L}(P)=\phi_P$. By assumption, we have that $\mathcal{ L}(P_n)=\phi_{P_n}\le0$ on $\Psi_{\bar{\mathcal{H}}\times\mathcal{G}}$ and \eqref{eq.L derivative P} implies that $\mathcal{ L}_P^{\prime}(Q_0)\le 0$ on $\Psi_{\bar{\mathcal{H}}\times\mathcal{G}}$.
	Thus we have that $\mathbb{G}\le \mathbb{G}_0$ and
	$
	\mathcal{I}\circ\mathcal{S}_{{\Psi_{\bar{\mathcal{H}}\times\mathcal{G}}}}\left( {\mathbb{G}}/{\mathcal{M}(\sigma_P)} \right) \le \mathcal{I}\circ\mathcal{S}_{{\Psi_{\bar{\mathcal{H}}\times\mathcal{G}}}}\left( {\mathbb{G}_0}/{\mathcal{M}(\sigma_P)} \right)
	$.
	Since $\mathbb{G}/\mathcal{M}\left(  \sigma_{P}\right)  \in\ell^{\infty
	}\left(  \Xi\times\mathcal{\bar{H}}\times\mathcal{G}\right)  $, where
	$\ell^{\infty}\left(  \Xi\times\mathcal{\bar{H}}\times\mathcal{G}\right)  $ is
	a Banach space under $\left\Vert \cdot\right\Vert _{\infty}$ and $\mathbb{G}$
	is tight by Lemma \ref{lemma.weak convergence phi_K}, we have that $\mathbb{G}/\mathcal{M}\left(  \sigma
	_{P}\right)  $ is tight (hence separable\footnote{See the definition of separability in \citet[p.~17]{van1996weak}. The closure of a separable subset of a metric space is separable.}) and is Radon by Theorem 7.1.7 of \citet{bogachev2007measure}. Since $\mathcal{I}%
	\circ\mathcal{S}_{\Psi_{\mathcal{\bar{H}}\times\mathcal{G}}}$ is continuous
	and convex, Theorem 11.1(i) of \citet{davydov1998local} implies that the CDF of $\mathcal{I}%
	\circ\mathcal{S}_{\Psi_{\mathcal{\bar{H}}\times\mathcal{G}}}\left(
	\mathbb{G}/\mathcal{M}\left(  \sigma_{P}\right)  \right)  $ is everywhere
	continuous except possibly at the point%
	\[
	r_{0}=\inf\left\{  r:\mathbb{P}\left(  \mathcal{I}\circ\mathcal{S}%
	_{\Psi_{\mathcal{\bar{H}}\times\mathcal{G}}}\left(  \mathbb{G}/\mathcal{M}%
	\left(  \sigma_{P}\right)  \right)  \leq r\right)  >0\right\}  \text{. }%
	\]
	Because $\mathcal{I}\circ\mathcal{S}_{\Psi_{\mathcal{\bar{H}}\times
			\mathcal{G}}}\left(  \mathbb{G}/\mathcal{M}\left(  \sigma_{P}\right)  \right)
	\leq\mathcal{I}\circ\mathcal{S}_{\Psi_{\mathcal{\bar{H}}\times\mathcal{G}}%
	}\left(  \mathbb{G}_{0}/\mathcal{M}\left(  \sigma_{P}\right)  \right)  $, we
	have that
	\[
	r_{0}\leq\inf\left\{  r:\mathbb{P}\left(  \mathcal{I}\circ\mathcal{S}%
	_{\Psi_{\mathcal{\bar{H}}\times\mathcal{G}}}\left(  \mathbb{G}_{0}%
	/\mathcal{M}\left(  \sigma_{P}\right)  \right)  \leq r\right)  >0\right\}
	<c_{1-\alpha},
	\]
	where the last inequality follows from that the CDF of $\mathcal{I}\circ
	\mathcal{S}_{\Psi_{\mathcal{\bar{H}}\times\mathcal{G}}}\left(  \mathbb{G}_0
	/\mathcal{M}\left(  \sigma_{P}\right)  \right)  $ is continuous and increasing
	at $c_{1-\alpha}$. This implies that the CDF of $\mathcal{I}\circ\mathcal{S}%
	_{\Psi_{\mathcal{\bar{H}}\times\mathcal{G}}}\left(  \mathbb{G}/\mathcal{M}%
	\left(  \sigma_{P}\right)  \right)  $ is continuous at $c_{1-\alpha}$. 
	Now by \eqref{eq.TS weak convergence} and \eqref{eq.critical value consistency} in this paper, together with Example 1.4.7 (Slutsky's lemma), Theorem 1.3.6 (continuous mapping), and Theorem 1.3.4(vi) of \citet{van1996weak}, we conclude that 
	\begin{align}\label{eq.limit size}
		\lim_{n\to\infty}\mathbb{P}^{\ast}\left( \sqrt{T_n}  \mathcal{I}\circ\mathcal{S}\left(  \frac{\hat{\phi}_{P_n}}{\mathcal{M}(\hat{\sigma}_{P_n})}\right)>\hat{c}_{1-\alpha} \right)= \mathbb{P}\left(\mathcal{I}\circ\mathcal{S}_{{\Psi_{\bar{\mathcal{H}}\times\mathcal{G}}}}\left( \frac{\mathbb{G}}{\mathcal{M}(\sigma_P)} \right)>c_{1-\alpha}\right)\le\alpha,
	\end{align}
	where the inequality follows from that $c_{1-\alpha}$ is the $1-\alpha$ quantile for $\mathcal{I}\circ\mathcal{S}_{{\Psi_{\bar{\mathcal{H}}\times\mathcal{G}}}}\left( {\mathbb{G}_0}/{\mathcal{M}(\sigma_P)} \right)$.
	If, in addition, $P_n=P$ for all $n$, then by Assumption \ref{ass.probability path} we have that $v_0=0$ and hence $Q_0=0$. This implies that $\mathbb{ G}=\mathbb{ G}_0$ and that the inequality in \eqref{eq.limit size} holds with equality. 
	
	(ii). Let $\hat{c}_{1-\alpha}^{\prime}$ be the bootstrap critical value obtained using the bootstrap test statistic 
	$\mathcal{I}\circ\mathcal{S(}\sqrt{T_n^{B}}(\hat{\phi}_{P_n}^{B}-\hat{\phi}_{P_n})/\mathcal{M}(  \hat{\sigma}_{P_n}^{B})  )$ in place of $ \mathcal{I}\circ\mathcal{S}_{\widehat{\Psi_{\mathcal{H}\times\mathcal{G}}}}
	(  {\sqrt{T_n^B}(  \hat{\phi}_{P_n}^{B}-\hat{\phi}_{P_n})  }/\mathcal{M}({\hat{\sigma}_{P_n}^{B}}))$ in the test procedure in Section \ref{subsbusec.test procedure}. By arguments similar to those in the proof of part (i), we can show that
	$\hat{c}_{1-\alpha}^{\prime}\rightarrow c_{1-\alpha}^{\prime}$ in outer probability, where
	$c_{1-\alpha}^{\prime}$ is the $1-\alpha$ quantile for $\mathcal{I}
	\circ\mathcal{S}\left(  \mathbb{G}_0/\mathcal{M}\left(  \sigma_P\right)  \right)
	$.\footnote{Here, we implicitly assume that the CDF of $\mathcal{I}\circ\mathcal{S}\left(  \mathbb{G}_0/\mathcal{M}\left(  \sigma_P\right)  \right)
		$ is continuous and strictly increasing at $c_{1-\alpha}^{\prime}$. Theorem 11.1 of \citet{davydov1998local} implies that the CDF of $\mathcal{I}\circ\mathcal{S}\left(  {\mathbb{G}_0}/{\mathcal{M}(\sigma_{P})}\right)$ is differentiable and has a positive derivative everywhere except at countably many points in its support, provided that $\mathcal{I}\circ\mathcal{S}\left(  {\mathbb{G}_0}/{\mathcal{M}(\sigma_{P})}\right)$ is not a constant. By construction, $\mathcal{I}\circ\mathcal{S}\left(  \mathbb{G}_0/\mathcal{M}\left(  \sigma_P\right)  \right)
		$ is not a constant in general cases.} Clearly, $\hat{c}_{1-\alpha}^{\prime}\geq\hat
	{c}_{1-\alpha}$ by construction. By Lemma \ref{lemma.almost uniform convergence phi sigma}, $\hat{\phi}_{P_n}/\mathcal{M}\left(  \hat{\sigma}_{P_n}\right)  \rightarrow\phi_P/\mathcal{M}\left(  \sigma_P\right)  $ in $\ell
	^{\infty}\left(  \Xi\times\mathcal{\bar{H}}\times\mathcal{G}\right)  $ almost uniformly, and hence almost uniformly
	\[
	\mathcal{I}\circ\mathcal{S}_{{{\mathcal{H}\times\mathcal{G}}}}\left(  \frac{\hat{\phi}_{P_n}}{\mathcal{M}\left(
		\hat{\sigma}_{P_n}\right)  }\right)  \rightarrow\mathcal{I}\circ\mathcal{S}_{{{\mathcal{H}\times\mathcal{G}}}}\left(
	\frac{\phi_P}{\mathcal{M}\left(  \sigma_P\right)  }\right)  >0,
	\]
	where the inequality follows from the assumption that the $H_0$ in \eqref{eq.null 1} is false with $Q=P$. Thus we have that $[\mathcal{I}
	\circ\mathcal{S}_{\mathcal{H}\times\mathcal{G}}(  \sqrt{T_n}\hat{\phi}_{P_n}/\mathcal{M}\left(  \hat{\sigma
	}_{P_n}\right)  )] ^{-1} \rightarrow 0$ almost uniformly ($T_n/n\to\Lambda(P)$ almost uniformly by Lemma \ref{lemma.almost uniform convergence phi sigma}). By Lemmas 1.9.3(ii) and 1.10.2(iii), Example 1.4.7 (Slutsky's lemma), and Theorems 1.3.6 (continuous mapping) and 1.3.4(vi) of \citet{van1996weak}, we now conclude that
	\[
	\mathbb{P}^{\ast}\left(  \mathcal{I}\circ\mathcal{S}_{\mathcal{H}\times\mathcal{G}}\left(  \frac{\sqrt{T_n}\hat{\phi}_{P_n}
	}{\mathcal{M}\left(  \hat{\sigma}_{P_n}\right)  }\right)  >\hat{c}_{1-\alpha
	}\right)  \geq\mathbb{P}^{\ast}\left(  \mathcal{I}\circ\mathcal{S}_{\mathcal{H}\times\mathcal{G}}\left(  \frac
	{\sqrt{T_n}\hat{\phi}_{P_n}}{\mathcal{M}\left(  \hat{\sigma}_{P_n}\right)  }\right)
	>\hat{c}_{1-\alpha}^{\prime}\right)  \rightarrow1.
	\]
\end{proof}
\setcounter{equation}{0}
\section{Monotonicity Condition with Unspecified Directions for Unordered Treatment}\label{sec.unordered treatment monotonicity with no deriction}
In Section \ref{subsec.unordered}, we mentioned that the test can be extended for the monotonicity condition with unspecified directions. We now show details for this extension.
Define $2^{J}$ different $J$-dimensional binary vectors by $v_{1}%
,\ldots,v_{2^{J}}$ with
\[
v_{1}=\left(
\begin{array}
	[c]{c}%
	0\\
	0\\
	\vdots\\
	0
\end{array}
\right)  ,v_{2}=\left(
\begin{array}
	[c]{c}%
	1\\
	0\\
	\vdots\\
	0
\end{array}
\right)  ,\ldots,v_{2^{J}}=\left(
\begin{array}
	[c]{c}%
	1\\
	1\\
	\vdots\\
	1
\end{array}
\right)  .
\]
Let $\mathrm{L}:\mathcal{D}\to \{1,\ldots,J\}$ map $d\in\mathcal{D}$ to $d$'s index in $\mathcal{D}$ such that if $d=d_j$, then $\mathrm{L}(d)=j$. For every $q\in\left\{  1,\ldots,2^{J}\right\}  $, define a function $f_{q}:\mathcal{D}  \rightarrow\left\{  1,-1\right\}  $ by
$f_{q}\left(  d\right)  =\left(  -1\right)  ^{v_{q}\left(  \mathrm{L}(d)\right)  }$, where $v_q(j)$ denotes the $j$th element of $v_q$.
If the instrument $Z$ is valid for the unordered treatment $D$ as defined in Assumption \ref{ass.IV validity for unordered D} with (iii) replaced by Assumption \ref{ass.unordered monotonicity}, then for all $z_j,z_k\in\mathcal{Z}$ with $j<k$, there is a $q\in\{1,\ldots,2^J\}$ such that
\begin{align*}
	f_q(d)\cdot\{\mathbb{P}\left(  Y\in B,D=d|Z=z_j\right)  -\mathbb{P}\left(  Y\in B,D=d|Z=z_k\right)\}\le0
\end{align*} 
for every $d\in\mathcal{D}$ and every closed interval $B$.
Then for every $q\in\left\{  1,\ldots,2^{J}\right\}  $, we define%
\begin{align*}
	\mathcal{H}_{q}  &  =\left\{  f_{q}\left(  d\right)  \cdot1_{B\times\left\{
		d\right\}  \times\mathbb{R}}:B\text{ is a closed interval in }\mathbb{R},d\in\mathcal{D}\right\} \text{ and}\\
	\mathcal{\bar{H}}_{q}  &  =\left\{  f_{q}\left(  d\right)  \cdot1_{B\times
		\left\{  d\right\}  \times\mathbb{R}}:B\text{ is a closed}, \text{open}, \text{or
		half-closed interval in }\mathbb{R},d\in\mathcal{D}\right\}  .
\end{align*}
Also, define function spaces
\begin{align}\label{def.function spaces unordered}
	%	& \mathrm{G}=\left\{ \left( 1_{\mathbb{R}\times \mathbb{R}\times \left\{
	%		z_j\right\} },1_{\mathbb{R}\times \mathbb{R}\times \left\{ z_k\right\} }\right) :j,k\in\{1,\ldots,K\},j\neq k\right\} , \notag\\
	\mathcal{H}=\cup_{q=1}^{2^{J}}\mathcal{H}_{q}, \mathcal{\bar{H}}	=\cup_{q=1}^{2^{J}}\mathcal{\bar{H}}_{q}, \text{ and }
	{\mathcal{G}}=\left\{ \left( 1_{\mathbb{R}\times \mathbb{R}\times \left\{
		z_j\right\} },1_{\mathbb{R}\times \mathbb{R}\times \left\{ z_k\right\} }\right) :j,k\in\{1,\ldots,K\},j<k\right\}.
\end{align}
Let $\phi_{Q}$ be defined as in \eqref{eq.phi_Q} with ${\mathcal{H}}$ and $\mathcal{G}$
defined in \eqref{def.function spaces unordered}. Now we obtain the testable implication for Assumption \ref{ass.IV validity for unordered D} with (iii) replaced by Assumption \ref{ass.unordered monotonicity}: 
\begin{align}\label{eq.testable implication unordered no direction}
	H_0: \max_{g\in\mathcal{G}}\min_{q\in\left\{  1,\ldots,2^{J}\right\}  }\sup_{h\in\mathcal{H}_{q}}
	\phi_Q\left(  h,g\right)  = 0,
\end{align}
if the underlying distribution of the data is $Q$. The test proposed in Section \ref{sec.unordered} can be generalized for the $H_0$ in \eqref{eq.testable implication unordered no direction}.

\section{Additional Monte Carlo Studies}\label{sec.appendix Monte Carlo Comparison}

\subsection{Degenerate Case under Null}
In Section \ref{subsec.bootstrap}, we discussed the case where $\mathbb{T}_0=0$. In this section, we design a DGP such that $\mathbb{T}_0=0$ to show the performance of the test in this case. 
We let this DGP be the same as that designed in Section \ref{sec.size}, except that we let $D_0=2\times 1\{V\le 0.328\}+ 1\{0.328  <V\le 0.658\}$, $D_1=2\times 1\{V\le 0.329\}+ 1\{0.329  <V\le 0.659\}$, and $D_2=2\times 1\{V\le 0.33\}+ 1\{0.33  <V\le 0.66\}$. Then it can be shown that $\mathbb{T}_0=0$ in this setting. We used $\eta=0$ as discussed in Section \ref{subsec.bootstrap}. As suggested in Section \ref{sec.size}, $\tau_n$ could be set to $2$. Table \ref{tab:Rej H0 multi (DGP2)} shows that the rejection rates are well controlled by the nominal significance level $0.05$ in this case.  

\input{Rejection_Rates_under_H0DGP2_Multi}

\subsection{Multivalued Treatment with Covariates}

In Section \ref{sec.empirical}, we used the data set of \citet{NBERw4483} to illustrate the application of the proposed test in practice. We revisited this empirical example and reconducted the test with conditioning covariates added into the model. Due to the limitation on the computation power, we added two conditioning variables (the variables ``south66'' and ``black'') from the data set. When we chose the values of $\xi$, we employed the empirical variance formula in \eqref{eq.estimated stat variance multi} to calculate an empirical bound for $\hat{\sigma}_{P_n}$. Specifically, we let $T_n=n\cdot\prod_{k=1}^{K}\prod_{l=1}^{L}\hat{P}_n(1_{\mathbb{R}\times\mathbb{R}\times\{z_k\}\times\{x_l\}})$ and used the first inequality in \eqref{eq.sigma_square_hat bound} to find the empirical bound. Table \ref{tab:ApplicationCardCovariates} shows the $p$-values obtained from the test. The $p$-values are lower than those in Section \ref{sec.empirical}. One possible reason is that when conditioning covariates are included, the number of observations for each category $(z_k,x_l)$ is small. Thus the $p$-values are different from those in the case where no conditioning covariates are included. But the results are consistent with those in Section \ref{sec.empirical} and show that the validity of the instrument is not rejected.

\input{ApplicationResultsCardCovariates}

\subsection{Unordered Treatment}
In this section, we designed Monte Carlo simulations for the case where $D$ is an unordered random variable with $D\in\{a,b,c\}$. For simplicity, we let $Z\in\{0,1\}$. We also consider the presence of a conditioning covariate $X\in\{0,1\}$. The measure $\nu$ was set to be a Dirac measure $\delta_{\xi}$ centered at one of the following values of $\xi$: $0.01$, $0.02$, $0.03$, $0.04$, $0.05$, $0.06$, $0.07$, $0.08$, $0.09$, and $0.1$, or to be a probability measure $\bar{\nu}_{\xi}$ that assigns equal probabilities (weights) to the values of $\xi$ listed above.  The nominal significance level $\alpha$ was set to $0.05$. To expedite the simulation, we employed the warp-speed
method of \citet{giacomini2013warp}.

\subsubsection{Size Control and Tuning Parameter Selection}
The first set of simulations was designed to investigate the size of the test and the selection of the tuning parameter. For this set of simulations, we set $n$ to $2000$ and $\tau_n$ to $0.1$, $0.5$, $1$, $2$, $3$, $4$, and $\infty$. We compared the rejection rates obtained using each of these values of $\tau_n$ and decided which value would be a good option for sample sizes close to $2000$. The
simulation consisted of $1000$ Monte Carlo iterations and $1000$ bootstrap iterations. 
We let $U\sim\mathrm{Unif}(0,1)$, $U_{X}\sim\mathrm{Unif}(0,1)$, $V\sim\mathrm{Unif}(0,1)$, $N_a\sim \mathrm{N}(0,1)$, $N_b\sim \mathrm{N}(1,1)$, $N_c\sim \mathrm{N}(2,1)$, $Z=1\{U \le 0.5\}$  ($\mathbb{P}(Z=1)=0.5$), $X=1\{U_{X}\le0.5\}$,
\[
D_{z}=\left\{
\begin{array}
	[c]{c}%
	a\\
	b\\
	c
\end{array}
\right\vert
\begin{array}
	[c]{c}%
	V>0.6\\
	0.5< V\leq0.6\\
	0<V\le 0.5
\end{array}
\] for $z\in\{0,1\}$, $D=D_z$ if $Z=z$ with $z\in\{0,1\}$, and $Y=\sum_{d\in\{a,b,c\}}1\{D=d\}\times N_d$. All the variables $U$, $U_{X}$, $V$, $N_a$, $N_b$, and $N_c$ were set to be mutually independent. Assumption \ref{ass.IV validity for unordered D} holds in this case with $\mathcal{C}=\{(a,0,1),(b,1,0),(c,1,0)\}$.

Table \ref{tab:Rej H0 unordered} shows the results of the simulations. The rejection rates were influenced by the values of $\tau_n$ and $\xi$. For each measure $\nu$, a smaller $\tau_n$ yields greater rejection rates by construction. 
For $\tau_n=2$, all the rejection rates were close to those for $\tau_n=\infty$ (the conservative case).  Similar to the pattern of the results shown in \citet{kitagawa2015test} and Section \ref{sec.size}, some rejection rates for $\tau_n=2$ with $\delta_{\xi}$ centered at particular values of $\xi$ were slightly upwardly biased compared to the nominal size. Overall, however, the results showed good performance of the test in terms of size control. When sample sizes are less than or close to $2000$, we suggest using $\tau_n=2$ in practice to achieve good size control without a significant power loss.  When the sample size increases, $\tau_n$ should be increased accordingly. It is also worth noting that when we used the measure $\bar{\nu}_{\xi}$, the rejection rates could be well controlled by the nominal significance level. Thus if we have no additional information about the choice of $\xi$,  $\bar{\nu}_{\xi}$ can be a default choice for us. 

\input{Rejection_Rates_under_H0_Unordered.tex}

\subsubsection{Rejection Rates against Fixed Alternatives}
The second set of simulations was designed to investigate the power of the test. A total of five DGPs were considered. Sample sizes were set to $n=200$, $600$, $1000$, $1100$, and $2000$. The probability $\mathbb{P}(Z=1)=r_n$, with $r_n=1/2$, $1/6$, $1/2$, $1/11$, and $1/2$ for the corresponding sample sizes. We set $\tau_n$ to $2$, as suggested in the preceding set of simulations. Each
simulation consisted of $500$ Monte Carlo iterations and $500$ bootstrap iterations. 
We let $U\sim\mathrm{Unif}(0,1)$, $U_{X}\sim\mathrm{Unif}(0,1)$, $V\sim\mathrm{Unif}(0,1)$, $W\sim\mathrm{Unif}(0,1)$, $Z=1\{U \le r_n\}$, and $X=1\{U_{X}\le0.5\}$.
For DGPs (1)--(4), we let 
\[
D_{z}=\left\{
\begin{array}
	[c]{c}%
	a\\
	b\\
	c
\end{array}
\right\vert
\begin{array}
	[c]{c}%
	V>0.6\\
	0.5< V\leq0.6\\
	0<V\le 0.5
\end{array}
\]
for $z\in\{0,1\}$,   $D=D_z$ if $Z=z$ with $z\in\{0,1\}$, $N_Z\sim \mathrm{N}(0,1)$, $N_{az}=N_Z$ for $z\in\{0,1\}$, $N_{bz}=N_Z$ for $z\in\{0,1\}$, and $N_{c1}=N_Z$.
\begin{enumerate}[label=(\arabic*):]
	
	\item  $N_{c0}\sim \mathrm{N}(-0.7,1)$ and $Y=\sum_{z=0}^{1}1\{Z=z\}\times(\sum_{d\in\{a,b,c\}} 1\{D=d\}\times N_{dz})$.
	
	\item $N_{c0}\sim \mathrm{N}(0,1.675^2)$ and $Y=\sum_{z=0}^{1}1\{Z=z\}\times(\sum_{d\in\{a,b,c\}} 1\{D=d\}\times N_{dz})$.
	
	\item $N_{c0}\sim \mathrm{N}(0,0.515^2)$ and $Y=\sum_{z=0}^{1}1\{Z=z\}\times(\sum_{d\in\{a,b,c\}} 1\{D=d\}\times N_{dz})$.
	
	\item $N_{c0a}\sim \mathrm{N}(-1,0.125^2)$, $N_{c0b}\sim \mathrm{N}(-0.5,0.125^2)$, $N_{c0c}\sim \mathrm{N}(0,0.125^2)$, \\$N_{c0d}\sim \mathrm{N}(0.5,0.125^2)$, $N_{c0e}\sim \mathrm{N}(1,0.125^2)$, $N_{c0}=1\{W\le 0.15\}\times N_{c0a}+1\{0.15<W\le 0.35\}\times N_{c0b}+1\{0.35<W\le 0.65\}\times N_{c0c}+1\{0.65<W\le 0.85\}\times N_{c0d}+1\{W>0.85\}\times N_{c0e}$, and $Y=\sum_{z=0}^{1}1\{Z=z\}\times(\sum_{d\in\{a,b,c\}} 1\{D=d\}\times N_{dz})$.
	
\end{enumerate} 
For DGP (5), we let 
\[
D_{0}=\left\{
\begin{array}
	[c]{c}%
	a\\
	b\\
	c
\end{array}
\right\vert
\begin{array}
	[c]{c}%
	V>0.6\\
	0.5< V\leq0.6\\
	0<V\le 0.5
\end{array},
D_{1}=\left\{
\begin{array}
	[c]{c}%
	a\\
	b\\
	c
\end{array}
\right\vert
\begin{array}
	[c]{c}%
	V>0.3\\
	0.2< V\leq0.3\\
	0<V\le 0.2
\end{array}.
\] 
\begin{enumerate}[resume,label=(\arabic*):]
	
	\item Let $N_a\sim \mathrm{N}(0,1)$, $N_b\sim \mathrm{N}(1,1)$, $N_c\sim \mathrm{N}(2,1)$, $D=D_z$ if $Z=z$ with $z\in\{0,1\}$, and $Y=\sum_{d\in\{a,b,c\}} 1\{D=d\}\times N_{d}$.
	
\end{enumerate}
All the variables $U$, $U_X$, $V$, $N_{Z}$, $N_{c0}$, $N_a$, $N_b$, and $N_c$ were set to be mutually independent.

Table \ref{tab:Rej H1 unordered} shows the rejection rates under DGPs (1)--(5), that is, the power of the test. For each DGP and each measure $\nu$, the rejection rate increased as the sample size $n$ was increased. The results for $\nu=\bar{\nu}_{\xi}$ showed that if we have no information about the choice of $\xi$, using the weighted average of the statistics over $\xi$ is a desirable option. When $n>200$, the rejection rates for using $\nu=\bar{\nu}_{\xi}$ were at a relatively high level compared to the results for using a Dirac measure.  

\input{Rejection_Rates_under_H1_Multi_Unordered.tex} 

\subsection{Comparison in Binary Case}\label{sec.comparison binary}
The Monte Carlo experiments discussed in this section followed the design of
\citet{kitagawa2015test}, where the treatment and the instrument were both binary, with $D\in\{0,1\}$ and $Z\in\{0,1\}$, and we compared our results  with theirs.  We simulated the limiting rejection rates using the approach proposed in the present paper and that proposed by
\citet{kitagawa2015test} with the same randomly generated data. In this special case, if the measure $\nu$ is set to be a Dirac measure, the asymptotic distribution of the test statistic under null  can be written as $\sup_{f\in\mathcal{ F}_b^{\ast}}{\mathbb{G}_H(f)}/({\xi\vee\sigma_H(f)})$ in \eqref{eq.upper bound of limiting distribution}. 
%Since the test proposed by \citet{kitagawa2015test} constructed the critical value based on the upper bound $\sup_{f\in\mathcal{ F}_b}{\mathbb{G}_H(f)}/({\xi\vee\sigma_H(f)})$ in \eqref{eq.upper bound of limiting distribution}, to show the power improvement of the proposed test on a finite sample more clearly, we constructed the critical value using $\sup_{f\in\mathcal{ F}_b^{\ast}}{\mathbb{G}_H(f)}/({\xi\vee\sigma_H(f)})$ instead of $\mathbb{T}$, which is equivalent to it in distribution. 
We followed the discussion in Section \ref{sec.binary} to construct the bootstrap critical value.
That is, we approximated $\mathbb{G}_H$ and $\sigma_H$ by $\mathbb{G}_H^B$ and $\sigma_H^B$ following the bootstrap method of \citet{kitagawa2015test}. Then we estimated $\mathcal{F}_b^{\ast}$ by $\widehat{\mathcal{F}_b^{\ast}}$ in a way similar to \eqref{eq.Psi_hat}, which is the key difference between our approach and that of \citet{kitagawa2015test}. Last, we constructed the bootstrap test statistic by $\sup_{f\in\widehat{\mathcal{ F}_b^{\ast} }}{\mathbb{G}_H^B(f)}/({\xi\vee\sigma_H^B(f)})$ and used it to create the critical value. Because of $\widehat{\mathcal{F}_b^{\ast}}$, our bootstrap test statistic can approximate the null distribution consistently and the power of the test can be improved. This new bootstrap test statistic is asymptotically equivalent to that in \eqref{eq.bootstrap test stat} under null, and the new critical value is asymptotically equivalent to $\hat{c}_{1-\alpha}$ in Section \ref{subsbusec.test procedure} under null.

Each simulation consisted of $1000$ Monte Carlo iterations and $1000$ bootstrap iterations. For each DGP, the measure $\nu$ was set to be a Dirac measure centered at $\xi=0.07$, $0.22$, $0.3$, and $1$. The nominal significance level $\alpha$ was set to $0.05$. 

\subsubsection{Size Control and Tuning Parameter Selection}
We first ran simulations to investigate the size of the test and the selection of the tuning parameter. As suggested in Section \ref{sec.simulation}, for sample sizes less than $3000$, we can use $\tau_n=2$ for the tuning parameter. In this set of simulations, we set $n=2000$ and $\tau_n=1,2,3,4,\infty$. For the DGP, we used $U\sim\mathrm{Unif}(0,1)$, $V\sim\mathrm{Unif}(0,1)$, $N_0\sim \mathrm{N}(0,1)$, $N_1\sim \mathrm{N}(1,1)$, $Z=1\{U\le 0.5\}$, $D_0=1\{V\le 0.5\}$, $D_1=1\{V\le 0.5\}$, $D=\sum_{z=0}^1 1\{Z=z\}\times D_z$, and $Y=\sum_{d=0}^{1}1\{D=d\}\times N_d$,
where $U$, $V$, $N_0$, and $N_1$ were mutually independent. This DGP is equivalent to that used by \citet{kitagawa2015test} to show the size control of their test. The results in Table \ref{tab:Rej H0} confirmed the conclusion from Table \ref{tab:Rej H0 multi}: For  $\tau_n=2$, the rejection rates were close to those for $\tau_n=\infty$ and close to the nominal size. Recall that a smaller tuning parameter $\tau_n$ yields greater power for the test. Thus we kept using $\tau_n=2$ in this case. 
\input{Rejection_Rates_under_H0_xi_short.tex}
\subsubsection{Power Comparison}
Four DGPs were considered for the power comparisons. The sample sizes were set to $n=200$, $600$, $1000$, $1100$, and $2000$, and the tuning parameter was set to $\tau_n=2$. The probability $\mathbb{P}(Z=1)=r_n$ with $r_n=1/2$, $1/6$, $1/2$, $1/11$, and $1/2$ for the corresponding sample sizes. We let
$U\sim\mathrm{Unif}(0,1)$, $V\sim\mathrm{Unif}(0,1)$, $W\sim\mathrm{Unif}(0,1)$, $Z=1\{U\le r_n\}$, $D_0=1\{V\le 0.45\}$, $D_1=1\{V\le 0.55\}$,  $D=\sum_{z=0}^1 1\{Z=z\}\times D_z$, $N_{00}\sim \mathrm{N}(0,1)$, $N_{01}\sim \mathrm{N}(0,1)$, and $N_{11}\sim \mathrm{N}(0,1)$.
\begin{enumerate}[label=(\arabic*):]
	
	\item $N_{10}\sim \mathrm{N}(-0.7,1)$ and $Y=\sum_{z=0}^{1}1\{Z=z\}\times(\sum_{d=0}^{1} 1\{D=d\}\times N_{dz})$.
	
	\item $N_{10}\sim \mathrm{N}(0,1.675^2)$ and $Y=\sum_{z=0}^{1}1\{Z=z\}\times(\sum_{d=0}^{1} 1\{D=d\}\times N_{dz})$.
	
	\item $N_{10}\sim \mathrm{N}(0,0.515^2)$ and $Y=\sum_{z=0}^{1}1\{Z=z\}\times(\sum_{d=0}^{1} 1\{D=d\}\times N_{dz})$.
	
	\item  $N_{10a}\sim \mathrm{N}(-1,0.125^2)$, $N_{10b}\sim \mathrm{N}(-0.5,0.125^2)$, $N_{10c}\sim \mathrm{N}(0,0.125^2)$, \\$N_{10d}\sim \mathrm{N}(0.5,0.125^2)$, $N_{10e}\sim \mathrm{N}(1,0.125^2)$, $N_{10}=1\{W\le 0.15\}\times N_{10a}+1\{0.15<W\le 0.35\}\times N_{10b}+1\{0.35<W\le 0.65\}\times N_{10c}+1\{0.65<W\le 0.85\}\times N_{10d}+1\{W>0.85\}\times N_{10e}$, and $Y=\sum_{z=0}^{1}1\{Z=z\}\times(\sum_{d=0}^{1} 1\{D=d\}\times N_{dz})$.
	
\end{enumerate}
All the variables $U$, $V$, $N_{00}$, $N_{10}$, $N_{01}$, and $N_{11}$ were set to be mutually independent for each DGP. Table \ref{tab:Rej H1} shows a comparison of the powers of the two tests. The results suggest that the proposed test achieves a manifest power improvement over that of \citet{kitagawa2015test}.
\input{Rejection_Rates_under_H1_xi.tex}

\subsubsection{Comparison in Empirical Application}
We now revisit the empirical application in Section \ref{sec.empirical} and show the size and power comparisons with the test of \citet{kitagawa2015test} using the data set of \citet{NBERw4483}. We follow \citet{kitagawa2015test} and define $T$ by $T=1\{D\ge 16\}$. As discussed in Section \ref{sec.empirical}, the instrument may not be valid for this coarsened treatment $T$. This has been verified by the empirical study of \citet{kitagawa2015test}. As shown in Table I of \citet{kitagawa2015test}, the sample size was sufficiently large (over $3000$), so the null hypothesis was rejected with $p$-values exactly equal to $0$ when no conditioning covariates were included in the model. To show the comparison of the proposed test and the test of \citet{kitagawa2015test} in this empirical example, we randomly drew relatively small subsamples of sizes $700$, $900$, $1100$, $1300$, $1500$, and $2000$ out of the full data set, and computed the empirical sizes and powers of the two tests using the same subsamples. 

To compare the sizes, given the subsample $\{(Y_i,D_i)\}_{i=1}^{m}$ drawn randomly from the full data set with $m\in\{700, 900,1100,1300,1500,2000\}$, we let $Z_i=0$ for $i=1,\ldots,m/2$ and $Z_i=1$ for $i=m/2+1,\ldots,m$. Then we used the sample $\{(Y_i,D_i,Z_i)\}_{i=1}^{m}$ to compute the sizes of the two tests.  
As shown in Table \ref{tab:ComparisonApplicationSize}, the rejection rates of the proposed test under null are slightly higher than those of the test of \citet{kitagawa2015test}. Both are close to the nominal significance level $0.05$. 
\input{ComparisonApplicationSize}

We used the subsample $\{(Y_i,D_i,Z_i)\}_{i=1}^{m}$ drawn randomly from the full data set to compute the powers of the two tests.  
As shown in Table \ref{tab:ComparisonApplicationPower}, the proposed test achieves a manifest power improvement over that of \citet{kitagawa2015test} when the samples are relatively small. 
\input{ComparisonApplicationPower}

\subsection{Choices of $\Xi$ and $\nu$ in Applications}\label{sec.choice of xi}

As shown in the discussion for \eqref{eq.test stat K} and also in the discussion in \citet{kitagawa2015test}, $\xi$ plays a role of bounding $\hat{\sigma}_{P_n}$ sufficiently away from zero. We set the support of $\xi$, $\Xi$, to be a closed subset of $[0,1]$ such that $0\notin\Xi$. Though our theoretical results show that all such $\Xi$ and measures $\nu$ that satisfy Assumption \ref{ass.nu} yield good asymptotic properties of the test, the choices of $\Xi$ and $\nu$ may affect the finite sample performance of the test. For finite samples, if $\xi$ is small (not far enough from $0$), larger sample sizes would be needed for the test to achieve better size and power properties. 
In this section, we provide more simulation results for the choices of $\Xi$ and $\nu$. We then provide an empirical approach for choosing $\Xi$ and $\nu$ in practice. 

We followed the same constructions of simulations under $H_0$ as those in Section \ref{sec.size}, and we set $\Xi=\{0.01,0.02,0.03,0.04,0.07,0.1,0.13,0.16,0.19,0.22,0.25,1\}$. We set the sample sizes to $1000$, $2000$, and $3000$. In this way, we investigate how small values of $\xi$ would affect the finite sample performance of the test for different sample sizes. The measure $\nu$ was set to be a Dirac measure $\delta_{\xi}$ centered at each value in $\Xi$. The bootstrap iteration was set to $1000$. We focus on the results for $\tau_n=2$ which is chosen in Section \ref{sec.simulation}.
As shown in Table \ref{tab:Rej H0 multi Small xi}, for $\xi\ge0.04$ and $n=3000$, all rejection rates are close to the nominal significance level $\alpha=0.05$. (The rejection rates for some $\xi$ are slightly upward biased. In application-based simulations, these rejection rates get close to $\alpha$.) For small $\xi\in\{0.01,0.02,0.03\}$, most of the rejection rates are lower than $\alpha$, but they increase as $n$ increases.
For example, the rejection rates for $\xi=0.03$ are $0.014$, $0.033$, and $0.070$ for $n=1000$, $2000$, and $3000$, respectively. The rejection rates for $\xi=0.02$ are $0.001$, $0.006$, and $0.011$ for $n=1000$, $2000$, and $3000$, respectively. For $\xi=0.01$, the rejection rates are all $0$. As discussed above, $\xi$ is used to bound $\hat{\sigma}_{P_n}$ away from $0$. When $\xi$ is close to $0$, the test may be conservative in finite samples. 
As the sample size increases, the rejection rates would converge to the nominal significance level, as shown for $\xi=0.02$ and $0.03$. We expect that when $n$ gets larger, the rejection rate for $\xi=0.01$ would converge to $\alpha$. 

\input{Rejection_Rates_under_H0_MultiSmallXi_Rep1000}

\subsubsection{Application-based Simulations for Choosing $\Xi$ and $\nu$}

Since small values of $\xi$ may affect the finite sample performance of the test, we introduce an empirical way of choosing $\Xi$ and $\nu$ in finite samples. In practice, we suggest setting $\Xi$ to be a (large) finite set of values and $\nu$ to be a Dirac measure centered at each value of $\Xi$ or a probability measure that assigns equal weights to each value in $\Xi$. The results in Section \ref{sec.simulation} show that these choices work well in simulations. %We next specify $\Xi$ for every particular finite sample. 
Recall that \eqref{eq.sigma_square bound} and \eqref{eq.sigma_square_hat bound} provide bounds for $\sigma_P$ and $\hat{\sigma}_{P_n}$.
For every finite sample, we can use positive values not larger than $\{1/2\cdot(K-1)^{-(K-1)}\}^{1/2}$ to construct $\Xi$. 
Clearly, $1/2\cdot(K-1)^{-(K-1)}<1$ for all $K$. Thus, we can just include $\xi=1$ in $\Xi$ which leads to the unweighted KS test statistic, and we set the other values of $\xi$ to be smaller than $\{1/2\cdot(K-1)^{-(K-1)}\}^{1/2}$. 
To be more precise, we can calculate the bound $1/2\cdot\max_{(g_1^{\prime},g_2^{\prime})\in\mathcal{G}}\{(T_n/n)  /\hat{P}_n\left(  g_{2}^{\prime}\right)  +(T_n/n)/\hat{P}_n\left(  g_{1}^{\prime}\right)  \}^{1/2}$ and only include values smaller than this bound other than $1$. 

We revisit the application in Section \ref{sec.empirical} and show how to choose $\Xi$ in practice. In this empirical example, $Z\in\{0,1\}$, and it follows that $\hat{\sigma}_{P_n}\le 1/2$, which was also mentioned in \citet{kitagawa2015test}. We first set $\Xi$ to be a large finite set with $$\Xi=\{0.01,0.02,0.03,0.04,0.07,0.1,0.13,0.16,0.19,0.22,0.25,0.28,0.3,1\}.$$ Let $n_0=\sum_{i=1}^n1\{Z_i=0\}$ and $n_1=\sum_{i=1}^n1\{Z_i=1\}$. The following is the procedure for choosing $\Xi$ for the finite sample $\{(Y_i,D_i,Z_i)\}_{i=1}^n$:
\begin{enumerate}[label=(\arabic*)]
    \item Find the subsample of $\{(Y_i,D_i,Z_i)\}_{i=1}^n$ with $Z_i=0$, and denote this subsample by $\{(Y^0_i,D^0_i,Z^0_i)\}_{i=1}^{n_0}$.

    \item Randomly draw  two samples from $\{(Y^0_i,D^0_i)\}_{i=1}^{n_0}$, denoted by $\{(Y^{00}_i,D^{00}_i)\}_{i=1}^{n_0}$ and $\{(Y^{01}_i,D^{01}_i)\}_{i=1}^{n_1}$.

    \item Let $Z^{00}_i=0$ for all $i\in\{1,\ldots,n_0\}$ and $Z^{01}_i=1$ for all $i\in\{1,\ldots,n_1\}$. 
    
    \item Combine the two samples $\{(Y^{00}_i,D^{00}_i,Z^{00}_i)\}_{i=1}^{n_0}$ and $\{(Y^{01}_i,D^{01}_i,Z^{01}_i)\}_{i=1}^{n_1}$. Denote the combined sample by $\{(Y^c_i,D^c_i,Z^c_i)\}_{i=1}^n$.
    
    \item Compute the test statistic and the bootstrap critical value based on the sample \linebreak $\{(Y^c_i,D^c_i,Z^c_i)\}_{i=1}^n$ and record the test results for $\Xi$.

    \item Repeat steps (2)--(5) many times and find the values in $\Xi$ such that the corresponding rejection rates are close to $\alpha$.
    
    \item Repeat steps (2)--(6) using the subsample of $\{(Y_i,D_i,Z_i)\}_{i=1}^n$ with $Z_i=1$. Find the values in $\Xi$ such that the corresponding rejection rates are close to $\alpha$.
    
    \item The intersection of the two sets of values in $\Xi$ from the above steps can be used in the application. 
\end{enumerate}

Table \ref{tab:Rej H0 multi Small xi Application} shows the simulation results following the above procedure. Based on these simulation results, we suggest using $\Xi=\{0.03,0.04,0.07,0.1,0.13,0.16,0.19,0.22,0.3,1\}$ for this application. We reconducted the test in Section \ref{sec.empirical} using this new $\Xi$. Table \ref{tab:ApplicationCard_ChosenXi} shows that the test results are similar to those in Table \ref{tab:ApplicationCard}.
  
\input{Rejection_Rates_under_H0_MultiSmallXi_Application}

\input{ApplicationResultsCard_ChosenXi}

\bibliographystyle{apalike}
\bibliography{reference1}

%% file: Rejection_Rates_under_H0DGP2_Multi.tex
\begin{table}[h]
	
	\centering
	\caption{Rejection Rates under $H_{0}$ ($\mathbb{T}_0=0$) for Multivalued $D$ and Multivalued $Z$}
	\label{tab:Rej H0 multi (DGP2)}
	\scalebox{0.9}{
		\begin{tabular}{   c  c  c  c  c  c  c  c  c  c  c  c  }
			\hline
			\hline
			\multirow{2}{*}{$\tau_n$} & \multicolumn{10}{c}{$\xi$ for $\delta_{\xi}$} &\multirow{2}{*}{$\bar{\nu}_{\xi}$} \\
			\cline{2-11}
			& 0.07 & 0.1 & 0.13 & 0.16 & 0.19 & 0.22 & 0.25 & 0.28 & 0.3 & 1 & \\
			\hline

	$0.1$	&0.117&0.102&0.091&0.104&0.100&0.094&0.090&0.090&0.090&0.090&0.104\\
	$0.5$	&0.080&0.068&0.063&0.072&0.071&0.067&0.071&0.071&0.071&0.071&0.077\\
	$1$	&0.073&0.055&0.048&0.057&0.064&0.055&0.057&0.057&0.057&0.057&0.053\\
	$2$	&0.066&0.045&0.042&0.048&0.052&0.050&0.050&0.050&0.050&0.050&0.045\\
	$3$	&0.066&0.045&0.042&0.048&0.052&0.050&0.050&0.050&0.050&0.050&0.045\\
	$4$	&0.066&0.045&0.042&0.048&0.052&0.050&0.050&0.050&0.050&0.050&0.045\\
	$\infty$	&0.066&0.045&0.042&0.048&0.052&0.050&0.050&0.050&0.050&0.050&0.045\\

			\hline
		\end{tabular}
	}

\end{table}

%% file: ApplicationResultsCardCovariates.tex
\begin{table}[h]
	
	\centering
	\caption{$p$-values Obtained from the Proposed Test for Each Measure $\nu$ with Covariates}
	\scalebox{0.85}{
		\begin{tabular}{   c  c  c  c  c  c  c  c  c  c   }
			\hline
			\hline
			 \multicolumn{9}{c}{$\xi$ for $\delta_{\xi}$} &\multirow{2}{*}{$\bar{\nu}_{\xi}$} \\
			\cline{1-9}
		   0.0001 & 0.00013 & 0.00016 & 0.00019 & 0.00022 & 0.00025 & 0.00028 & 0.00031 & 0.00034 & \\
			\hline
			0.673&0.541&0.519&0.469&0.477&0.489&0.489&0.489&0.489&0.522\\

			\hline
		\end{tabular}
	}
	
	\label{tab:ApplicationCardCovariates}
\end{table}

%% file: Rejection_Rates_under_H0_Unordered.tex
\begin{table}[h]
	
	\centering
	\caption{Rejection Rates under $H_{0}$ for Unordered $D$}
	\scalebox{0.9}{
		\begin{tabular}{   c  c  c  c  c  c  c  c  c  c  c  c  }
			\hline
			\hline
			\multirow{2}{*}{$\tau_n$} & \multicolumn{10}{c}{$\xi$ for $\delta_{\xi}$} &\multirow{2}{*}{$\bar{\nu}_{\xi}$} \\
			\cline{2-11}
		  & 0.01 & 0.02 & 0.03 & 0.04 & 0.05 & 0.06 & 0.07 & 0.08 & 0.09 & 0.1 & \\
			\hline

		 $0.1$   &0.137&0.137&0.118&0.102&0.111&0.104&0.095&0.120&0.116&0.116&0.136\\
		 $0.5$	&0.092&0.093&0.076&0.082&0.061&0.069&0.072&0.084&0.075&0.075&0.082\\
		 $1$	&0.057&0.070&0.065&0.067&0.059&0.065&0.065&0.055&0.052&0.052&0.069\\
		 $2$	&0.009&0.055&0.056&0.061&0.058&0.064&0.058&0.045&0.049&0.049&0.053\\
		 $3$	&0.006&0.050&0.054&0.061&0.058&0.064&0.058&0.045&0.049&0.049&0.053\\
		 $4$	&0.006&0.050&0.054&0.061&0.058&0.064&0.058&0.045&0.049&0.049&0.053\\
		 $\infty$	&0.006&0.050&0.054&0.061&0.058&0.064&0.058&0.045&0.049&0.049&0.053\\

			\hline
		\end{tabular}
	}
	
	\label{tab:Rej H0 unordered}
\end{table}

%% file: Rejection_Rates_under_H1_Multi_Unordered.tex
\begin{table}[h]
	
	\centering
	\caption{Rejection Rates under $H_{1}$ for Unordered $D$}
	\scalebox{0.85}{
		\begin{tabular}{  c  c  c  c  c  c  c  c  c  c  c  c  c  }
			\hline
			\hline
			\multirow{2}{*}{DGP} & \multirow{2}{*}{$n$}& \multicolumn{10}{c}{$\xi$ for $\delta_{\xi}$} &\multirow{2}{*}{$\bar{\nu}_{\xi}$} \\
			\cline{3-12}
			& & 0.01 & 0.02 & 0.03 & 0.04 & 0.05 & 0.06 & 0.07 & 0.08 & 0.09 & 0.1 & \\
			\hline

			\multirow{5}{*}{(1)}	
			&200&0.000&0.090&0.188&0.256&0.324&0.336&0.326&0.290&0.306&0.306&0.222\\
			&600&0.032&0.402&0.528&0.562&0.546&0.502&0.432&0.432&0.432&0.432&0.464\\
			&1000&0.604&0.932&0.954&0.976&0.984&0.984&0.972&0.966&0.962&0.962&0.986\\
			&1100&0.488&0.594&0.626&0.566&0.470&0.448&0.448&0.448&0.448&0.448&0.626\\
			&2000&1.000&1.000&1.000&1.000&1.000&1.000&1.000&1.000&1.000&1.000&1.000\\	
			\hline

			\multirow{5}{*}{(2)}	
			
			&200&0.000&0.006&0.044&0.112&0.134&0.096&0.060&0.050&0.044&0.044&0.034\\
			&600&0.002&0.174&0.092&0.048&0.022&0.028&0.030&0.030&0.030&0.030&0.042\\
			&1000&0.190&0.624&0.772&0.722&0.572&0.358&0.150&0.124&0.108&0.108&0.512\\
			&1100&0.236&0.074&0.048&0.036&0.044&0.042&0.042&0.042&0.042&0.042&0.078\\
			&2000&0.992&0.998&0.998&0.998&0.970&0.898&0.642&0.456&0.398&0.398&0.976\\
			
			\hline
			
			\multirow{5}{*}{(3)}	
			
			&200&0.000&0.160&0.334&0.398&0.452&0.460&0.494&0.484&0.490&0.490&0.364\\
			&600&0.042&0.560&0.666&0.786&0.812&0.798&0.750&0.750&0.750&0.750&0.728\\
			&1000&0.728&0.926&0.948&0.958&0.980&0.986&0.990&0.992&0.990&0.990&0.988\\
			&1100&0.596&0.720&0.824&0.860&0.792&0.764&0.764&0.764&0.764&0.764&0.826\\
			&2000&0.996&1.000&1.000&1.000&1.000&1.000&1.000&1.000&1.000&1.000&1.000\\
			
			\hline
			
			\multirow{5}{*}{(4)}	
			&200&0.000&0.042&0.110&0.150&0.172&0.200&0.214&0.214&0.208&0.208&0.146\\
			&600&0.026&0.326&0.382&0.396&0.428&0.442&0.414&0.404&0.404&0.404&0.436\\
			&1000&0.210&0.472&0.572&0.576&0.618&0.702&0.706&0.746&0.774&0.774&0.704\\
			&1100&0.326&0.444&0.530&0.568&0.504&0.444&0.444&0.444&0.444&0.444&0.580\\
			&2000&0.790&0.930&0.948&0.954&0.962&0.956&0.968&0.978&0.982&0.982&0.986\\
			
			\hline
		
			\multirow{5}{*}{(5)}	
			&200&0.162&0.900&0.958&0.968&0.974&0.974&0.984&0.988&0.988&0.988&0.970\\
			&600&0.688&0.988&1.000&1.000&1.000&1.000&1.000&1.000&1.000&1.000&1.000\\
			&1000&1.000&1.000&1.000&1.000&1.000&1.000&1.000&1.000&1.000&1.000&1.000\\
			&1100&0.974&1.000&1.000&1.000&1.000&0.996&0.996&0.996&0.996&0.996&1.000\\
			&2000&1.000&1.000&1.000&1.000&1.000&1.000&1.000&1.000&1.000&1.000&1.000\\

			\hline
		\end{tabular}
	}
	
	\label{tab:Rej H1 unordered}
\end{table}

%% file: Rejection_Rates_under_H0_xi_short.tex
\begin{table}[h]
	
	\centering
	\caption{Rejection Rates under $H_{0}$ for Binary $D$ and Binary $Z$}
	\scalebox{1}{
		\begin{tabular}{    c  c  c  c  c     }
			\hline
			\hline
		 \multirow{2}{*}{$\tau_n$}	&  \multicolumn{4}{c}{$\xi$}\\
			        \cline{2-5}
			          & 0.07 &  0.22 &  0.3 & 1 \\
        \hline

			$1$&0.077&0.052&0.048&0.069\\
			$2$&0.058&0.048&0.040&0.067\\
			$3$&0.056&0.046&0.040&0.067\\
			$4$&0.056&0.046&0.040&0.067\\
			$\infty$&0.056&0.046&0.040&0.067\\

\hline

		\end{tabular}
	}
	
	\label{tab:Rej H0}
\end{table}

%% file: Rejection_Rates_under_H1_xi.tex
\begin{table}[h]
	
	\centering
	\caption{Rejection Rates under $H_{1}$ for Binary $D$ and Binary $Z$}
		\scalebox{1}{
		\begin{tabular}{  c  c  c  c  c  c  c  c  c  c     }
			\hline
			\hline
			\multirow{4}{*}{DGP} & \multirow{4}{*}{$n$}&\multicolumn{4}{c}{The Proposed Test}&\multicolumn{4}{c}{Test of \citet{kitagawa2015test}}\\
			\cmidrule(lr){3-6 } \cmidrule(lr){7-10 } 
			& &\multicolumn{4}{c}{$\xi$}&\multicolumn{4}{c}{$\xi$}\\
			\cmidrule(lr){3-6 } \cmidrule(lr){7-10 } 
			%&&\multicolumn{4}{c}{$\xi$}&\multicolumn{4}{c}{$\xi$}\\
			   %     \cmidrule(lr){3-6} \cmidrule(lr){7-10}
			        &  & 0.07 & 0.22 & 0.3 & 1 & 0.07 & 0.22 & 0.3 & 1 \\
			\hline

		\multirow{5}{*}{(1)}	
		&200&0.202&0.198&0.186&0.110&0.198&0.193&0.182&0.106\\
		&600&0.300&0.434&0.418&0.180&0.240&0.406&0.375&0.144\\
		&1000&0.874&0.915&0.919&0.804&0.855&0.883&0.894&0.714\\
		&1100&0.309&0.493&0.452&0.163&0.263&0.451&0.423&0.153\\
		&2000&0.997&0.999&1.000&0.997&0.996&0.999&0.999&0.993\\
		
		\hline
		
		\multirow{5}{*}{(2)}
		&200&0.105&0.095&0.059&0.004&0.090&0.084&0.046&0.003\\
		&600&0.261&0.141&0.045&0.000&0.242&0.100&0.026&0.000\\
		&1000&0.907&0.814&0.500&0.105&0.887&0.781&0.421&0.030\\
		&1100&0.255&0.129&0.037&0.001&0.224&0.082&0.022&0.001\\
		&2000&1.000&0.996&0.949&0.674&1.000&0.994&0.909&0.252\\
		
			\hline
		
		\multirow{5}{*}{(3)}	
		&200&0.211&0.209&0.202&0.211&0.185&0.188&0.195&0.205\\
		&600&0.203&0.427&0.473&0.351&0.191&0.377&0.458&0.331\\
		&1000&0.664&0.769&0.816&0.831&0.654&0.739&0.785&0.796\\
		&1100&0.229&0.442&0.487&0.341&0.203&0.399&0.443&0.321\\
		&2000&0.950&0.982&0.992&0.995&0.949&0.971&0.987&0.992\\
        
        \hline
        
        \multirow{5}{*}{(4)}
        &200&0.080&0.082&0.073&0.036&0.079&0.082&0.073&0.036\\
        &600&0.134&0.117&0.103&0.060&0.123&0.111&0.102&0.058\\
        &1000&0.307&0.306&0.224&0.127&0.307&0.281&0.212&0.116\\
        &1100&0.146&0.115&0.112&0.031&0.136&0.115&0.093&0.027\\
        &2000&0.660&0.703&0.556&0.325&0.649&0.673&0.505&0.271\\
        \hline

		\end{tabular}
	}
	
	\label{tab:Rej H1}
\end{table}

%% file: ComparisonApplicationSize.tex
\begin{table}[h]
	
	\centering
	\caption{Empirical Sizes of the Two Tests for Each $\xi$ in Empirical Application}
	\label{tab:ComparisonApplicationSize}
	\scalebox{1}{
		\begin{tabular}{   c  c  c  c  c  c  c  c  c  c  c  }
			\hline
			\hline
	\multirow{2}{*}{$n$}	&\multicolumn{10}{c}{$\xi$ for $\delta_{\xi}$}  \\
			\cline{2-11}
			& 0.07 & 0.1 & 0.13 & 0.16 & 0.19 & 0.22 & 0.25 & 0.28 & 0.3 & 1\\
			\hline
700&0.071&0.071&0.071&0.076&0.068&0.064&0.061&0.055&0.061&0.056\\
900&0.066&0.066&0.065&0.064&0.065&0.068&0.056&0.058&0.051&0.064\\
1100&0.072&0.072&0.070&0.059&0.055&0.059&0.054&0.051&0.053&0.048\\
1300&0.064&0.064&0.061&0.062&0.069&0.074&0.075&0.071&0.065&0.060\\
1500&0.045&0.047&0.046&0.046&0.045&0.039&0.050&0.049&0.066&0.062\\
2000&0.049&0.050&0.049&0.050&0.048&0.051&0.053&0.045&0.046&0.038\\
			\hline
		\end{tabular}
	}
\subcaption{Empirical Sizes of the Test of \citet{kitagawa2015test}}
\bigskip

\scalebox{1}{
	\begin{tabular}{   c  c  c  c  c  c  c  c  c  c  c  }
		\hline
		\hline
		\multirow{2}{*}{$n$}	&\multicolumn{10}{c}{$\xi$ for $\delta_{\xi}$}  \\
		\cline{2-11}
		& 0.07 & 0.1 & 0.13 & 0.16 & 0.19 & 0.22 & 0.25 & 0.28 & 0.3 & 1\\
		\hline
700&0.075&0.075&0.075&0.078&0.068&0.067&0.061&0.059&0.062&0.056\\
900&0.066&0.066&0.066&0.064&0.065&0.068&0.056&0.059&0.051&0.064\\
1100&0.072&0.072&0.070&0.060&0.061&0.060&0.055&0.054&0.053&0.048\\
1300&0.066&0.066&0.062&0.065&0.069&0.080&0.079&0.071&0.065&0.060\\
1500&0.045&0.047&0.046&0.046&0.045&0.040&0.050&0.049&0.066&0.062\\
2000&0.050&0.050&0.051&0.050&0.048&0.051&0.053&0.045&0.046&0.038\\
		\hline
	\end{tabular}
}
\subcaption{Empirical Sizes of the Proposed Test}

\end{table}

%% file: ComparisonApplicationPower.tex
\begin{table}[h]
	
	\centering
	\caption{Empirical Powers of the Two Tests for Each $\xi$ in Empirical Application}
	\label{tab:ComparisonApplicationPower}
	\scalebox{1}{
		\begin{tabular}{   c  c  c  c  c  c  c  c  c  c  c  }
			\hline
			\hline
			\multirow{2}{*}{$n$}	&\multicolumn{10}{c}{$\xi$ for $\delta_{\xi}$}  \\
			\cline{2-11}
			& 0.07 & 0.1 & 0.13 & 0.16 & 0.19 & 0.22 & 0.25 & 0.28 & 0.3 & 1\\
			\hline
		700&0.321&0.321&0.340&0.379&0.390&0.409&0.416&0.408&0.423&0.476\\
		900&0.395&0.390&0.427&0.457&0.478&0.494&0.517&0.538&0.535&0.578\\
		1100&0.534&0.557&0.584&0.633&0.627&0.657&0.659&0.682&0.699&0.739\\
		1300&0.666&0.666&0.706&0.728&0.762&0.768&0.798&0.791&0.793&0.805\\
		1500&0.742&0.779&0.800&0.805&0.808&0.808&0.831&0.844&0.836&0.864\\
		2000&0.902&0.907&0.920&0.904&0.922&0.924&0.929&0.927&0.930&0.952\\
			\hline
		\end{tabular}
	}
	\subcaption{Empirical Powers of the Test of \citet{kitagawa2015test}}
	\bigskip
	
	\scalebox{1}{
		\begin{tabular}{   c  c  c  c  c  c  c  c  c  c  c  }
			\hline
			\hline
			\multirow{2}{*}{$n$}	&\multicolumn{10}{c}{$\xi$ for $\delta_{\xi}$}  \\
			\cline{2-11}
			& 0.07 & 0.1 & 0.13 & 0.16 & 0.19 & 0.22 & 0.25 & 0.28 & 0.3 & 1\\
			\hline
		700&0.336&0.336&0.344&0.408&0.417&0.424&0.446&0.476&0.482&0.511\\
		900&0.404&0.411&0.445&0.512&0.513&0.537&0.546&0.584&0.584&0.654\\
		1100&0.558&0.572&0.615&0.647&0.654&0.668&0.686&0.726&0.751&0.803\\
		1300&0.674&0.725&0.756&0.796&0.803&0.828&0.847&0.854&0.865&0.858\\
		1500&0.781&0.825&0.852&0.857&0.836&0.850&0.860&0.871&0.875&0.905\\
		2000&0.910&0.926&0.942&0.935&0.935&0.940&0.947&0.945&0.951&0.975\\
			\hline
		\end{tabular}
	}
	\subcaption{Empirical Powers of the Proposed Test}

\end{table}

%% file: Rejection_Rates_under_H0_MultiSmallXi_Rep1000.tex
\begin{table}[h]
	
	\centering
	\caption{Rejection Rates under $H_{0}$ for Small $\xi$}
	\label{tab:Rej H0 multi Small xi}
	\scalebox{0.85}{
		\begin{tabular}{   c  c  c  c  c  c  c  c  c  c  c  c  c  }
			\hline
			\hline
			\multirow{2}{*}{$\tau_n$} & \multicolumn{12}{c}{$\xi$ for $\delta_{\xi}$ ($n=1000$)} \\
			\cline{2-13}
		  & 0.01 & 0.02 & 0.03 & 0.04 & 0.07 & 0.1 & 0.13 & 0.16 & 0.19 & 0.22 & 0.25 & 1\\
			\hline
		$0.1$	&0.174&0.174&0.174&0.174&0.200&0.179&0.173&0.137&0.119&0.102&0.104&0.104\\
		$0.5$	&0.097&0.097&0.097&0.105&0.136&0.116&0.115&0.089&0.071&0.069&0.070&0.070\\
		$1$ 	&0.004&0.027&0.037&0.070&0.108&0.084&0.076&0.072&0.065&0.063&0.063&0.063\\
		$2$ 	&0.000&0.001&0.014&0.052&0.105&0.077&0.063&0.069&0.055&0.061&0.058&0.058\\
		$3$ 	&0.000&0.001&0.012&0.052&0.105&0.077&0.062&0.069&0.055&0.061&0.058&0.058\\
		$4$ 	&0.000&0.001&0.011&0.052&0.105&0.077&0.062&0.069&0.055&0.061&0.058&0.058\\
		$\infty$ 	&0.000&0.001&0.011&0.052&0.105&0.077&0.062&0.069&0.055&0.061&0.058&0.058\\
		\hline	
	        \multirow{2}{*}{$\tau_n$} & \multicolumn{12}{c}{$\xi$ for $\delta_{\xi}$ ($n=2000$)} \\
	        \cline{2-13}
	        & 0.01 & 0.02 & 0.03 & 0.04 & 0.07 & 0.1 & 0.13 & 0.16 & 0.19 & 0.22 & 0.25 & 1\\
	        \hline
	    $0.1$   &0.159&0.159&0.159&0.167&0.150&0.121&0.120&0.105&0.103&0.112&0.119&0.119\\
	    $0.5$   &0.086&0.086&0.087&0.084&0.079&0.076&0.089&0.061&0.069&0.064&0.058&0.058\\
	    $1$     &0.014&0.023&0.052&0.066&0.071&0.061&0.075&0.051&0.055&0.047&0.050&0.050\\
	    $2$     &0.000&0.006&0.033&0.053&0.058&0.054&0.065&0.047&0.049&0.036&0.036&0.036\\
	    $3$     &0.000&0.006&0.033&0.053&0.056&0.054&0.064&0.047&0.049&0.036&0.033&0.033\\
	    $4$     &0.000&0.006&0.033&0.053&0.056&0.054&0.064&0.047&0.048&0.035&0.032&0.032\\
	    $\infty$     &0.000&0.006&0.033&0.053&0.056&0.054&0.064&0.047&0.048&0.035&0.032&0.032\\
	       \hline 
	         \multirow{2}{*}{$\tau_n$} & \multicolumn{12}{c}{$\xi$ for $\delta_{\xi}$ ($n=3000$)} \\
	        \cline{2-13}
	        & 0.01 & 0.02 & 0.03 & 0.04 & 0.07 & 0.1 & 0.13 & 0.16 & 0.19 & 0.22 & 0.25 & 1\\
	        \hline
	    $0.1$   &0.198&0.198&0.194&0.185&0.122&0.108&0.096&0.096&0.108&0.092&0.092&0.092\\
	    $0.5$   &0.074&0.074&0.112&0.122&0.092&0.070&0.068&0.074&0.064&0.069&0.069&0.069\\
	    $1$     &0.017&0.023&0.077&0.089&0.079&0.060&0.047&0.068&0.056&0.058&0.061&0.061\\
	    $2$     &0.000&0.011&0.070&0.083&0.073&0.050&0.037&0.050&0.050&0.055&0.048&0.048\\
	    $3$     &0.000&0.011&0.055&0.083&0.073&0.048&0.037&0.050&0.050&0.049&0.048&0.048\\
	    $4$     &0.000&0.011&0.055&0.083&0.073&0.048&0.037&0.050&0.050&0.049&0.048&0.048\\
	    $\infty$     &0.000&0.011&0.055&0.083&0.073&0.048&0.037&0.050&0.050&0.049&0.048&0.048\\
			
			\hline
		\end{tabular}
	}

\end{table}

%% file: Rejection_Rates_under_H0_MultiSmallXi_Application.tex
\begin{table}[h]
	
	\centering
	\caption{Application-based Rejection Rates under $H_{0}$ for Different $\xi$}
	\label{tab:Rej H0 multi Small xi Application}
	\scalebox{0.73}{
		\begin{tabular}{   c  c  c  c  c  c  c  c  c  c  c  c  c  c  c }
			\hline
			\hline
			\multirow{2}{*}{$\tau_n$} & \multicolumn{14}{c}{$\xi$ for $\delta_{\xi}$ ($Z=0$)} \\
			\cline{2-15}
		  & 0.01 & 0.02 & 0.03 & 0.04 & 0.07 & 0.1 & 0.13 & 0.16 & 0.19 & 0.22 & 0.25 & 0.28 & 0.3 & 1\\
			\hline
			
		$0.1$	&0.180&0.180&0.180&0.181&0.163&0.153&0.160&0.178&0.176&0.206&0.228&0.243&0.250&0.291\\
		$0.5$	&0.096&0.096&0.096&0.095&0.087&0.087&0.084&0.084&0.086&0.087&0.102&0.099&0.099&0.090\\
		$1$ 	&0.037&0.055&0.062&0.068&0.063&0.065&0.058&0.058&0.066&0.066&0.065&0.063&0.065&0.067\\
		$2$ 	&0.005&0.034&0.041&0.048&0.053&0.056&0.054&0.043&0.049&0.057&0.052&0.048&0.048&0.053\\
		$3$ 	&0.000&0.020&0.035&0.043&0.053&0.056&0.054&0.043&0.049&0.057&0.052&0.048&0.048&0.052\\
		$4$ 	&0.000&0.017&0.032&0.043&0.053&0.056&0.054&0.043&0.049&0.057&0.052&0.048&0.048&0.052\\
		$\infty$ 	&0.000&0.017&0.032&0.043&0.053&0.056&0.054&0.043&0.049&0.057&0.052&0.048&0.048&0.052\\

		\hline	
		\multirow{2}{*}{$\tau_n$} & \multicolumn{14}{c}{$\xi$ for $\delta_{\xi}$ ($Z=1$)} \\
		\cline{2-15}
		& 0.01 & 0.02 & 0.03 & 0.04 & 0.07 & 0.1 & 0.13 & 0.16 & 0.19 & 0.22 & 0.25 & 0.28 & 0.3 & 1\\
		\hline
			
		$0.1$	&0.153&0.153&0.153&0.145&0.131&0.128&0.112&0.111&0.132&0.142&0.168&0.161&0.174&0.232\\
		$0.5$	&0.084&0.084&0.084&0.083&0.080&0.072&0.067&0.061&0.055&0.052&0.074&0.073&0.066&0.088\\
		$1$ 	&0.020&0.060&0.077&0.072&0.067&0.061&0.054&0.048&0.051&0.042&0.042&0.042&0.052&0.054\\
		$2$ 	&0.003&0.031&0.046&0.049&0.063&0.056&0.047&0.037&0.036&0.037&0.026&0.030&0.037&0.037\\
		$3$ 	&0.000&0.005&0.044&0.048&0.063&0.056&0.047&0.037&0.035&0.037&0.026&0.030&0.037&0.037\\
		$4$ 	&0.000&0.005&0.035&0.048&0.051&0.056&0.047&0.037&0.035&0.037&0.026&0.030&0.037&0.037\\
		$\infty$ 	&0.000&0.005&0.035&0.048&0.051&0.056&0.047&0.037&0.035&0.037&0.026&0.030&0.037&0.037\\
			\hline
		\end{tabular}
	}

\end{table}

%% file: ApplicationResultsCard_ChosenXi.tex
\begin{table}[h]
	
	\centering
	\caption{$p$-values Obtained from the Proposed Test for Each Measure $\nu$ using Application-based $\Xi$}
	\scalebox{0.9}{
		\begin{tabular}{   c  c  c  c  c  c  c  c  c  c  c  }
			\hline
			\hline
			 \multicolumn{10}{c}{$\xi$ for $\delta_{\xi}$} &\multirow{2}{*}{$\bar{\nu}_{\xi}$} \\
			\cline{1-10}
		 0.03 & 0.04 & 0.07 & 0.1 & 0.13 & 0.16 & 0.19 & 0.22 & 0.3 & 1 & \\
			\hline
			 0.957 &  0.939 &  0.958 &  0.975 &  0.975 &  0.975  & 0.975 &  0.975 &  0.975  & 0.975 &  0.981\\

			\hline
		\end{tabular}
	}
	
	\label{tab:ApplicationCard_ChosenXi}
\end{table}

%% file: IVValidityforHeterogeneousCausalEffects.bbl
\begin{thebibliography}{}

\bibitem[Aliprantis and Border, 2006]{aliprantis2006infinite}
Aliprantis, C.~D. and Border, K. (2006).
\newblock {\em Infinite Dimensional Analysis: A Hitchhiker's Guide}.
\newblock Springer Science \& Business Media.

\bibitem[Bogachev, 2007]{bogachev2007measure}
Bogachev, V.~I. (2007).
\newblock {\em Measure Theory}, volume~2.
\newblock Springer Science \& Business Media.

\bibitem[Bugni et~al., 2017]{bugni2017inference}
Bugni, F.~A., Canay, I.~A., and Shi, X. (2017).
\newblock Inference for subvectors and other functions of partially identified
  parameters in moment inequality models.
\newblock {\em Quantitative Economics}, 8(1):1--38.

\bibitem[Card, 1993]{NBERw4483}
Card, D. (1993).
\newblock Using geographic variation in college proximity to estimate the
  return to schooling.
\newblock National Bureau of Economic Research.

\bibitem[Davydov et~al., 1998]{davydov1998local}
Davydov, Y.~A., Lifshits, M.~A., and Smorodina, N.~V. (1998).
\newblock {\em Local Properties of Distributions of Stochastic Functionals},
  volume 173.
\newblock American Mathematical Society.

\bibitem[Fang and Santos, 2019]{fang2014inference}
Fang, Z. and Santos, A. (2019).
\newblock Inference on directionally differentiable functions.
\newblock {\em The Review of Economic Studies}, 86(1):377--412.

\bibitem[Folland, 1999]{folland2013real}
Folland, G.~B. (1999).
\newblock {\em Real Analysis: Modern Techniques and Their Applications}.
\newblock John Wiley \& Sons.

\bibitem[Giacomini et~al., 2013]{giacomini2013warp}
Giacomini, R., Politis, D.~N., and White, H. (2013).
\newblock A warp-speed method for conducting {M}onte {C}arlo experiments
  involving bootstrap estimators.
\newblock {\em Econometric Theory}, 29(3):567--589.

\bibitem[Kitagawa, 2015]{kitagawa2015test}
Kitagawa, T. (2015).
\newblock A test for instrument validity.
\newblock {\em Econometrica}, 83(5):2043--2063.

\bibitem[Shapiro, 1990]{shapiro1990concepts}
Shapiro, A. (1990).
\newblock On concepts of directional differentiability.
\newblock {\em Journal of Optimization Theory and Applications},
  66(3):477--487.

\bibitem[van~der Vaart and Wellner, 1996]{van1996weak}
van~der Vaart, A.~W. and Wellner, J.~A. (1996).
\newblock {\em Weak Convergence and Empirical Processes}.
\newblock Springer.

\end{thebibliography}


\begin{thebibliography}{}

\bibitem[Abadie, 2002]{abadie2002bootstrap}
Abadie, A. (2002).
\newblock Bootstrap tests for distributional treatment effects in instrumental
  variable models.
\newblock {\em Journal of the American Statistical Association},
  97(457):284--292.

\bibitem[Abadie et~al., 2002]{abadie2002instrumental}
Abadie, A., Angrist, J., and Imbens, G. (2002).
\newblock Instrumental variables estimates of the effect of subsidized training
  on the quantiles of trainee earnings.
\newblock {\em Econometrica}, 70(1):91--117.

\bibitem[Aliprantis and Border, 2006]{aliprantis2006infinite}
Aliprantis, C.~D. and Border, K. (2006).
\newblock {\em Infinite Dimensional Analysis: A Hitchhiker's Guide}.
\newblock Springer Science \& Business Media.

\bibitem[Ananat and Michaels, 2008]{ananat2008effect}
Ananat, E.~O. and Michaels, G. (2008).
\newblock The effect of marital breakup on the income distribution of women
  with children.
\newblock {\em Journal of Human Resources}, 43(3):611--629.

\bibitem[Andrews, 2000]{andrews2000inconsistency}
Andrews, D.~W. (2000).
\newblock Inconsistency of the bootstrap when a parameter is on the boundary of
  the parameter space.
\newblock {\em Econometrica}, 68(2):399--405.

\bibitem[Andrews and Shi, 2013]{andrews2013inference}
Andrews, D.~W. and Shi, X. (2013).
\newblock Inference based on conditional moment inequalities.
\newblock {\em Econometrica}, 81(2):609--666.

\bibitem[Andrews and Soares, 2010]{andrews2010inference}
Andrews, D.~W. and Soares, G. (2010).
\newblock Inference for parameters defined by moment inequalities using
  generalized moment selection.
\newblock {\em Econometrica}, 78(1):119--157.

\bibitem[Angrist, 1990]{angrist1990lifetime}
Angrist, J.~D. (1990).
\newblock Lifetime earnings and the {V}ietnam era draft lottery: Evidence from
  social security administrative records.
\newblock {\em The American Economic Review}, 80(3):313--336.

\bibitem[Angrist and Imbens, 1995]{angrist1995two}
Angrist, J.~D. and Imbens, G.~W. (1995).
\newblock Two-stage least squares estimation of average causal effects in
  models with variable treatment intensity.
\newblock {\em Journal of the American Statistical Association},
  90(430):431--442.

\bibitem[Angrist et~al., 1996]{angrist1996identification}
Angrist, J.~D., Imbens, G.~W., and Rubin, D.~B. (1996).
\newblock Identification of causal effects using instrumental variables.
\newblock {\em Journal of the American Statistical Association},
  91(434):444--455.

\bibitem[Angrist and Krueger, 1991]{angrist1991does}
Angrist, J.~D. and Krueger, A.~B. (1991).
\newblock Does compulsory school attendance affect schooling and earnings?
\newblock {\em The Quarterly Journal of Economics}, 106(4):979--1014.

\bibitem[Angrist and Krueger, 1995]{angrist1995split}
Angrist, J.~D. and Krueger, A.~B. (1995).
\newblock Split-sample instrumental variables estimates of the return to
  schooling.
\newblock {\em Journal of Business \& Economic Statistics}, 13(2):225--235.

\bibitem[Angrist and Pischke, 2008]{angrist2008mostly}
Angrist, J.~D. and Pischke, J.-S. (2008).
\newblock {\em Mostly Harmless Econometrics: An Empiricist's Companion}.
\newblock Princeton University Press.

\bibitem[Angrist and Pischke, 2014]{angrist2014mastering}
Angrist, J.~D. and Pischke, J.-S. (2014).
\newblock {\em Mastering Metrics: The Path from Cause to Effect}.
\newblock Princeton University Press.

\bibitem[Armstrong, 2014]{armstrong2014weighted}
Armstrong, T.~B. (2014).
\newblock Weighted {KS} statistics for inference on conditional moment
  inequalities.
\newblock {\em Journal of Econometrics}, 181(2):92--116.

\bibitem[Armstrong and Chan, 2016]{armstrong2016multiscale}
Armstrong, T.~B. and Chan, H.~P. (2016).
\newblock Multiscale adaptive inference on conditional moment inequalities.
\newblock {\em Journal of Econometrics}, 194(1):24--43.

\bibitem[Balke and Pearl, 1997]{balke1997bounds}
Balke, A. and Pearl, J. (1997).
\newblock Bounds on treatment effects from studies with imperfect compliance.
\newblock {\em Journal of the American Statistical Association},
  92(439):1171--1176.

\bibitem[Barrett and Donald, 2003]{barrett2003consistent}
Barrett, G.~F. and Donald, S.~G. (2003).
\newblock Consistent tests for stochastic dominance.
\newblock {\em Econometrica}, 71(1):71--104.

\bibitem[Barrett et~al., 2014]{barrett2014consistent}
Barrett, G.~F., Donald, S.~G., and Bhattacharya, D. (2014).
\newblock Consistent nonparametric tests for {L}orenz dominance.
\newblock {\em Journal of Business \& Economic Statistics}, 32(1):1--13.

\bibitem[Beare and Fang, 2017]{Beare2016global}
Beare, B.~K. and Fang, Z. (2017).
\newblock Weak convergence of the least concave majorant of estimators for a
  concave distribution function.
\newblock {\em Electronic Journal of Statistics}, 11(2):3841--3870.

\bibitem[Beare and Moon, 2015]{beare2015nonparametric}
Beare, B.~K. and Moon, J.-M. (2015).
\newblock Nonparametric tests of density ratio ordering.
\newblock {\em Econometric Theory}, 31(3):471--492.

\bibitem[Beare and Shi, 2019]{Beare2015improved}
Beare, B.~K. and Shi, X. (2019).
\newblock An improved bootstrap test of density ratio ordering.
\newblock {\em Econometrics and Statistics}, 10:9--26.

\bibitem[Bogachev, 2007]{bogachev2007measure}
Bogachev, V.~I. (2007).
\newblock {\em Measure Theory}, volume~2.
\newblock Springer Science \& Business Media.

\bibitem[Bound et~al., 1995]{bound1995problems}
Bound, J., Jaeger, D.~A., and Baker, R.~M. (1995).
\newblock Problems with instrumental variables estimation when the correlation
  between the instruments and the endogenous explanatory variable is weak.
\newblock {\em Journal of the American Statistical Association},
  90(430):443--450.

\bibitem[Buckles and Hungerman, 2013]{buckles2013season}
Buckles, K.~S. and Hungerman, D.~M. (2013).
\newblock Season of birth and later outcomes: {O}ld questions, new answers.
\newblock {\em Review of Economics and Statistics}, 95(3):711--724.

\bibitem[Bugni et~al., 2017]{bugni2017inference}
Bugni, F.~A., Canay, I.~A., and Shi, X. (2017).
\newblock Inference for subvectors and other functions of partially identified
  parameters in moment inequality models.
\newblock {\em Quantitative Economics}, 8(1):1--38.

\bibitem[Card, 1993]{NBERw4483}
Card, D. (1993).
\newblock Using geographic variation in college proximity to estimate the
  return to schooling.
\newblock National Bureau of Economic Research.

\bibitem[Cawley and Meyerhoefer, 2012]{cawley2012medical}
Cawley, J. and Meyerhoefer, C. (2012).
\newblock The medical care costs of obesity: {An} instrumental variables
  approach.
\newblock {\em Journal of Health Economics}, 31(1):219--230.

\bibitem[Chernozhukov et~al., 2015]{chernozhukov2014implementing}
Chernozhukov, V., Kim, W., Lee, S., and Rosen, A.~M. (2015).
\newblock Implementing intersection bounds in {S}tata.
\newblock {\em The Stata Journal}, 15(1):21--44.

\bibitem[Chernozhukov et~al., 2013]{chernozhukov2013intersection}
Chernozhukov, V., Lee, S., and Rosen, A.~M. (2013).
\newblock Intersection bounds: Estimation and inference.
\newblock {\em Econometrica}, 81(2):667--737.

\bibitem[Chetverikov, 2018]{chetverikov2017adaptive}
Chetverikov, D. (2018).
\newblock Adaptive tests of conditional moment inequalities.
\newblock {\em Econometric Theory}, 34(1):186--227.

\bibitem[Davydov et~al., 1998]{davydov1998local}
Davydov, Y.~A., Lifshits, M.~A., and Smorodina, N.~V. (1998).
\newblock {\em Local Properties of Distributions of Stochastic Functionals},
  volume 173.
\newblock American Mathematical Society.

\bibitem[Donald and Hsu, 2016]{donald2016improving}
Donald, S.~G. and Hsu, Y.-C. (2016).
\newblock Improving the power of tests of stochastic dominance.
\newblock {\em Econometric Reviews}, 35(4):553--585.

\bibitem[D{\"u}mbgen, 1993]{dumbgen1993nondifferentiable}
D{\"u}mbgen, L. (1993).
\newblock On nondifferentiable functions and the bootstrap.
\newblock {\em Probability Theory and Related Fields}, 95(1):125--140.

\bibitem[Eren and Ozbeklik, 2014]{eren2014benefits}
Eren, O. and Ozbeklik, S. (2014).
\newblock Who benefits from {Job Corps}? {A} distributional analysis of an
  active labor market program.
\newblock {\em Journal of Applied Econometrics}, 29(4):586--611.

\bibitem[Fang and Santos, 2019]{fang2014inference}
Fang, Z. and Santos, A. (2019).
\newblock Inference on directionally differentiable functions.
\newblock {\em The Review of Economic Studies}, 86(1):377--412.

\bibitem[Folland, 1999]{folland2013real}
Folland, G.~B. (1999).
\newblock {\em Real Analysis: Modern Techniques and Their Applications}.
\newblock John Wiley \& Sons.

\bibitem[Fr{\"o}lich and Melly, 2013]{frolich2013unconditional}
Fr{\"o}lich, M. and Melly, B. (2013).
\newblock Unconditional quantile treatment effects under endogeneity.
\newblock {\em Journal of Business \& Economic Statistics}, 31(3):346--357.

\bibitem[Giacomini et~al., 2013]{giacomini2013warp}
Giacomini, R., Politis, D.~N., and White, H. (2013).
\newblock A warp-speed method for conducting {M}onte {C}arlo experiments
  involving bootstrap estimators.
\newblock {\em Econometric Theory}, 29(3):567--589.

\bibitem[Hansen, 2017]{hansen2017regression}
Hansen, B.~E. (2017).
\newblock Regression kink with an unknown threshold.
\newblock {\em Journal of Business \& Economic Statistics}, 35(2):228--240.

\bibitem[Heckman and Pinto, 2018]{heckman2018unordered}
Heckman, J.~J. and Pinto, R. (2018).
\newblock Unordered monotonicity.
\newblock {\em Econometrica}, 86(1):1--35.

\bibitem[Heckman et~al., 2006]{heckman2006understanding}
Heckman, J.~J., Urzua, S., and Vytlacil, E. (2006).
\newblock Understanding instrumental variables in models with essential
  heterogeneity.
\newblock {\em The Review of Economics and Statistics}, 88(3):389--432.

\bibitem[Heckman et~al., 2008]{heckman2008instrumental}
Heckman, J.~J., Urzua, S., and Vytlacil, E. (2008).
\newblock Instrumental variables in models with multiple outcomes: The general
  unordered case.
\newblock {\em Annales d'Economie et de Statistique}, pages 151--174.

\bibitem[Heckman and Vytlacil, 2005]{heckman2005structural}
Heckman, J.~J. and Vytlacil, E. (2005).
\newblock Structural equations, treatment effects, and econometric policy
  evaluation.
\newblock {\em Econometrica}, 73(3):669--738.

\bibitem[Heckman and Vytlacil, 2007]{heckman2007econometric}
Heckman, J.~J. and Vytlacil, E.~J. (2007).
\newblock Econometric evaluation of social programs, part {II}: Using the
  marginal treatment effect to organize alternative econometric estimators to
  evaluate social programs, and to forecast their effects in new environments.
\newblock In {\em Handbook of Econometrics}, pages 4875--5143. Amsterdam:
  Elsevier.

\bibitem[Hirano and Porter, 2012]{hirano2012impossibility}
Hirano, K. and Porter, J.~R. (2012).
\newblock Impossibility results for nondifferentiable functionals.
\newblock {\em Econometrica}, 80(4):1769--1790.

\bibitem[Hong and Li, 2018]{hong2014numerical}
Hong, H. and Li, J. (2018).
\newblock The numerical delta method.
\newblock {\em Journal of Econometrics}, 206(2):379--394.

\bibitem[Horv{\'a}th et~al., 2006]{horvath2006testing}
Horv{\'a}th, L., Kokoszka, P., and Zitikis, R. (2006).
\newblock Testing for stochastic dominance using the weighted {McFadden}-type
  statistic.
\newblock {\em Journal of Econometrics}, 133(1):191--205.

\bibitem[Hsu et~al., 2019]{hsu2019testing}
Hsu, Y.-C., Liu, C.-A., and Shi, X. (2019).
\newblock Testing generalized regression monotonicity.
\newblock {\em Econometric Theory}, 35(6):1146--1200.

\bibitem[Huber and Mellace, 2015]{huber2015testing}
Huber, M. and Mellace, G. (2015).
\newblock Testing instrument validity for {LATE} identification based on
  inequality moment constraints.
\newblock {\em Review of Economics and Statistics}, 97(2):398--411.

\bibitem[Huber and W{\"u}thrich, 2018]{huber2018local}
Huber, M. and W{\"u}thrich, K. (2018).
\newblock Local average and quantile treatment effects under endogeneity: A
  review.
\newblock {\em Journal of Econometric Methods}, 8(1).

\bibitem[Imbens, 2014]{imbens2014instrumental}
Imbens, G. (2014).
\newblock Instrumental variables: An econometrician's perspective.
\newblock National Bureau of Economic Research.

\bibitem[Imbens and Angrist, 1994]{imbens1994identification}
Imbens, G.~W. and Angrist, J.~D. (1994).
\newblock Identification and estimation of local average treatment effects.
\newblock {\em Econometrica}, 62(2):467--475.

\bibitem[Imbens and Manski, 2004]{imbens2004confidence}
Imbens, G.~W. and Manski, C.~F. (2004).
\newblock Confidence intervals for partially identified parameters.
\newblock {\em Econometrica}, 72(6):1845--1857.

\bibitem[Imbens and Rubin, 1997]{imbens1997estimating}
Imbens, G.~W. and Rubin, D.~B. (1997).
\newblock Estimating outcome distributions for compliers in instrumental
  variables models.
\newblock {\em The Review of Economic Studies}, 64(4):555--574.

\bibitem[Imbens and Rubin, 2015]{imbens2015causal}
Imbens, G.~W. and Rubin, D.~B. (2015).
\newblock {\em Causal Inference in Statistics, Social, and Biomedical
  Sciences}.
\newblock Cambridge University Press.

\bibitem[K{\'e}dagni and Mourifi{\'e}, 2020]{kedagni2020generalized}
K{\'e}dagni, D. and Mourifi{\'e}, I. (2020).
\newblock Generalized instrumental inequalities: Testing the instrumental
  variable independence assumption.
\newblock {\em Biometrika}, 107(3):661--675.

\bibitem[Kitagawa, 2015]{kitagawa2015test}
Kitagawa, T. (2015).
\newblock A test for instrument validity.
\newblock {\em Econometrica}, 83(5):2043--2063.

\bibitem[Koenker et~al., 2017]{koenker2017handbook}
Koenker, R., Chernozhukov, V., He, X., and Peng, L. (2017).
\newblock {\em Handbook of Quantile Regression}.
\newblock CRC Press.

\bibitem[Lee and Salani{\'e}, 2018]{lee2018identifying}
Lee, S. and Salani{\'e}, B. (2018).
\newblock Identifying effects of multivalued treatments.
\newblock {\em Econometrica}, 86(6):1939--1963.

\bibitem[Lee et~al., 2018]{lee2018testing}
Lee, S., Song, K., and Whang, Y.-J. (2018).
\newblock Testing for a general class of functional inequalities.
\newblock {\em Econometric Theory}, 34(5):1018--1064.

\bibitem[Linton et~al., 2010]{linton2010improved}
Linton, O., Song, K., and Whang, Y.-J. (2010).
\newblock An improved bootstrap test of stochastic dominance.
\newblock {\em Journal of Econometrics}, 154(2):186--202.

\bibitem[Liu et~al., 2020]{liu2020two}
Liu, S., Mourifi{\'e}, I., and Wan, Y. (2020).
\newblock Two-way exclusion restrictions in models with heterogeneous treatment
  effects.
\newblock {\em The Econometrics Journal}, 23(3):345--362.

\bibitem[Melly and W{\"u}thrich, 2017]{melly2017local}
Melly, B. and W{\"u}thrich, K. (2017).
\newblock Local quantile treatment effects.
\newblock In {\em Handbook of Quantile Regression}, pages 145--164. Chapman and
  Hall/CRC.

\bibitem[Mogstad et~al., 2021]{mogstad2021causal}
Mogstad, M., Torgovitsky, A., and Walters, C.~R. (2021).
\newblock The causal interpretation of two-stage least squares with multiple
  instrumental variables.
\newblock {\em American Economic Review}, 111(11):3663--98.

\bibitem[Mourifi{\'e} and Wan, 2017]{mourifie2016testing}
Mourifi{\'e}, I. and Wan, Y. (2017).
\newblock Testing local average treatment effect assumptions.
\newblock {\em Review of Economics and Statistics}, 99(2):305--313.

\bibitem[Pollard, 1990]{pollard1990empirical}
Pollard, D. (1990).
\newblock Empirical processes: Theory and applications.
\newblock In {\em NSF-CBMS Regional Conference Series in Probability and
  Statistics}, pages i--86. JSTOR.

\bibitem[Reed, 2001]{reed2001pareto}
Reed, W.~J. (2001).
\newblock The {P}areto, {Z}ipf and other power laws.
\newblock {\em Economics Letters}, 74(1):15--19.

\bibitem[Reed, 2003]{reed2003pareto}
Reed, W.~J. (2003).
\newblock The {P}areto law of incomes---{A}n explanation and an extension.
\newblock {\em Physica A: Statistical Mechanics and Its Applications},
  319:469--486.

\bibitem[Rubin, 1974]{rubin1974}
Rubin, D.~B. (1974).
\newblock Estimating causal effects of treatments in randomized and
  nonrandomized studies.
\newblock {\em Journal of Educational Psychology}, 66(5):688.

\bibitem[Seo, 2018]{Seo2016tests}
Seo, J. (2018).
\newblock Tests of stochastic monotonicity with improved power.
\newblock {\em Journal of Econometrics}, 207(1):53--70.

\bibitem[Shapiro, 1990]{shapiro1990concepts}
Shapiro, A. (1990).
\newblock On concepts of directional differentiability.
\newblock {\em Journal of Optimization Theory and Applications},
  66(3):477--487.

\bibitem[Splawa-Neyman et~al., 1990]{splawa1990application}
Splawa-Neyman, J., Dabrowska, D.~M., and Speed, T. (1990).
\newblock On the application of probability theory to agricultural experiments.
  {E}ssay on principles. {S}ection 9.
\newblock {\em Statistical Science}, 5(4):465--472.

\bibitem[Sun and Beare, 2021]{Beare2017improved}
Sun, Z. and Beare, B.~K. (2021).
\newblock Improved nonparametric bootstrap tests of {L}orenz dominance.
\newblock {\em Journal of Business \& Economic Statistics}, 39(1):189--199.

\bibitem[Toda, 2012]{toda2012double}
Toda, A.~A. (2012).
\newblock The double power law in income distribution: Explanations and
  evidence.
\newblock {\em Journal of Economic Behavior \& Organization}, 84(1):364--381.

\bibitem[van~der Vaart and Wellner, 1996]{van1996weak}
van~der Vaart, A.~W. and Wellner, J.~A. (1996).
\newblock {\em Weak Convergence and Empirical Processes}.
\newblock Springer.

\bibitem[Vytlacil, 2002]{vytlacil2002independence}
Vytlacil, E. (2002).
\newblock Independence, monotonicity, and latent index models: An equivalence
  result.
\newblock {\em Econometrica}, 70(1):331--341.

\bibitem[Vytlacil, 2006]{vytlacil2006ordered}
Vytlacil, E. (2006).
\newblock Ordered discrete-choice selection models and local average treatment
  effect assumptions: Equivalence, nonequivalence, and representation results.
\newblock {\em The Review of Economics and Statistics}, 88(3):578--581.

\end{thebibliography}
